%% file: oopsla24.tex
\documentclass[nonacm]{acmart}
\setcopyright{none}

\startPage{1}

\usepackage{listings}
\usepackage{adjustbox}
\usepackage{stmaryrd}
\usepackage[utf8]{inputenc}
\usepackage[T1]{fontenc}
\usepackage{microtype}
\usepackage{graphicx}
\usepackage{caption}
\usepackage{subcaption}

\usepackage[lined,linesnumbered]{algorithm2e} 
\SetAlCapSty{}
\SetKw{Continue}{continue}
\makeatletter
\renewcommand{\@algocf@capt@plain}{above} 
\makeatother

\usepackage{dblfloatfix}
\usepackage{cleveref}
\usepackage{amsmath}
\usepackage{relsize} 
\usepackage[makeroom]{cancel}
\usepackage[normalem]{ulem} 

\usepackage{tcolorbox}

\usepackage{listings}
\usepackage{geometry}
\usepackage{enumitem}
\setlist{leftmargin=14pt,topsep=2pt}

\SetCommentSty{mycommfont}

\usepackage{multirow,bigdelim}
\usepackage{tikz}
\usepackage{tabularx}
\usepackage{longtable}
\usepackage{array}
\usepackage{booktabs}
\usepackage{pifont}
\usepackage{mathpartir} 
\input{common_macros.tex}

\newfloat{Algorithm}{t!}{lop}

\lstdefinelanguage{IR}[]{C}{
    morekeywords={call,alloc,dealloc,ret}
}
\lstnewenvironment{myexamplesmall}{\lstset{basicstyle=\scriptsize\ttfamily}}{}
\lstnewenvironment{myexamplefnsize}{\lstset{basicstyle=\footnotesize\ttfamily}}{}
\lstnewenvironment{myexamplesmallc}{\lstset{basicstyle=\small\ttfamily}}{}
\lstnewenvironment{myexamplesmallir}{\lstset{language=IR,basicstyle=\scriptsize\ttfamily}}{}
\lstnewenvironment{myexamplesmallasm}{\lstset{language=[x86masm]Assembler,basicstyle=\scriptsize\ttfamily,morecomment = [l]{//}}}{}
\lstnewenvironment{myexamplescriptsz}{\lstset{basicstyle=\scriptsize\ttfamily}}{}

\usepackage{stackengine}

\usepackage{wrapfig}
\usepackage{mathtools}
\usepackage[h]{esvect}

\usetikzlibrary{arrows, quotes,shapes,positioning, calc}

\begin{document}

\title{Modeling Dynamic (De)Allocations of Local Memory for Translation Validation}

\author{Abhishek Rose}
\orcid{0009-0002-2222-8906}
\affiliation{%
  \institution{Indian Institute of Technology Delhi}
  \city{New Delhi}
  \country{India}
}
\email{abhishek.rose@cse.iitd.ac.in}

\author{Sorav Bansal}
\orcid{0009-0004-2006-9635}
\affiliation{%
  \institution{Indian Institute of Technology Delhi}
  \city{New Delhi}
  \country{India}
}
\email{sbansal@iitd.ac.in}

\begin{abstract}
End-to-End Translation
Validation is the problem of verifying the executable code generated
by a compiler against the
corresponding input source code for a single compilation.
This becomes particularly hard in the presence of dynamically-allocated
local memory where addresses of local memory
may be observed by the program. In
the context of validating the
translation of a C procedure to executable code, a validator needs to tackle
constant-length local arrays,
address-taken local variables, address-taken formal parameters,
variable-length local arrays,
procedure-call
arguments (including variadic arguments),
and the {\tt alloca()} operator.  We provide an execution
model, a definition of refinement, and an algorithm
to soundly convert a refinement check
into first-order logic queries that an off-the-shelf SMT solver can handle
efficiently.
In our experiments, we perform
blackbox translation validation of C procedures (with up to 100+ SLOC), involving these
local memory allocation constructs, against their corresponding
assembly implementations (with up to 200+ instructions)
generated by an optimizing compiler with complex loop and vectorizing transformations.
\end{abstract}


\maketitle

\section{Introduction}\label{sec:intro}

Compiler bugs can be catastrophic, especially for
safety-critical
applications.
End-to-End Translation Validation (TV for short) checks a
single compilation to ascertain
if the machine executable code generated
by a compiler agrees with the input source program.
In our work, we validate translations from {\em unoptimized} IR of a C program
to {\em optimized} executable (or assembly) code, which forms an overwhelming
majority of the
complexity in an end-to-end compilation pipeline.
In this setting,
the presence of dynamic allocations
and deallocations due to local variables and procedure-call
arguments in the IR program
presents a special challenge --- in these cases,
the identification and modeling of relations between a local
variable
(or a procedure-call argument)
in IR and its stack address in assembly is often
required to complete the validation proof.

Unlike IR-to-assembly, modeling
dynamic local memory allocations is significantly simpler for
IR-to-IR TV
\cite{tvi,alive2,tristan_tv_eqsat11,stepp_eqsat_llvm11,llvm_verify_zhao12,ssa_verify_zhao13,aliveFP,namjoshi13,llvm_tv21}.
For example,
(pseudo)register-allocation of local variables can
be tackled by
identifying
relational invariants
that equate the value contained in a local variable's memory region (in the original
program) with the value in the corresponding pseudo-register (in the transformed program)
\cite{crellvm18}.
If the address of a
local variable is observable by the C program (e.g., for an address-taken local variable),
we need to additionally relate the variable addresses across both programs.
These address correlations
can be achieved by first correlating the corresponding allocation statements
in both programs (e.g., through their names) and then assuming
that their return values are equal.
Provenance-based syntactic
pointer analyses,
that
show separation between distinct
variables \cite{andersen94programanalysis,steensgaard},
thus suffice for translation validation across
IR-to-IR transformations.

An IR-to-assembly transformation
involves the lowering of a memory allocation (deallocation) IR instruction
to a stackpointer decrement (increment) instruction in assembly.
Further, the stack space in assembly is shared by multiple local variables, procedure-call
arguments,
and by the potential intermediate values generated by the compiler, e.g., pseudo-register spills.
Provenance-based pointer analyses
are thus inadequate for showing separation in assembly.

Prior
work on IR-to-assembly and assembly-to-assembly TV
\cite{tv_oskernel,ddec,semalign,oopsla20} assumes that local
variables are either absent or their addresses
are not observed in the program and so they are removed through
(pseudo)register-allocation. Similarly, these prior
works assume that variadic parameters
(and other cases of address-taken parameters)
are absent in the program.

Prior work on certified compilation, embodied in CompCert \cite{compcert},
validates its own transformation passes from IR to assembly, and supports
both address-taken local variables and variadic parameters.
However, CompCert sidesteps the task of having to model dynamic
allocations by ensuring that the
generated assembly code {\em preallocates} the space for
all local variables and
procedure-call arguments at the beginning of a procedure's
body.
Because preallocation is not possible
if the size of an
allocation is not known at compile time,
CompCert does not support variable-sized local variables or {\tt alloca()}.
Moreover, preallocation is prone to
stack space wastage.
In contrast to a certified compiler,
TV needs to validate
the compilation of a third-party compiler, and thus
needs to support an arbitrary (potentially dynamic) allocation strategy.

\begin{figure}
\begin{minipage}[t]{0.4\textwidth} 
\begin{subfigure}[t]{\textwidth}
\begin{myexamplesmall}
~{\tiny \textcolor{mygray}{C0:}}~ int fib(int n, int m) {
~{\tiny \textcolor{mygray}{C1:}}~   int v[n+2];
~{\tiny \textcolor{mygray}{C2:}}~   v[0]=0; v[1]=1;
~{\tiny \textcolor{mygray}{C3:}}~   for(int i=2; i<=m; i++)
~{\tiny \textcolor{mygray}{C4:}}~     v[i]=v[i-1]+v[i-2];
~{\tiny \textcolor{mygray}{C5:}}~   printf("fib(%d) = %d", m, v[m]);
~{\tiny \textcolor{mygray}{C6:}}~   return v[m];
~{\tiny \textcolor{mygray}{C7:}}~ }
\end{myexamplesmall}
\caption{\label{fig:example1c}C Program with VLA.}
\end{subfigure}
\begin{subfigure}[t]{\textwidth}
\begin{myexamplesmallir}
~{\tiny \textcolor{mygray}{I0: }}~ int fib(int* n, int* m):
~{\tiny \textcolor{mygray}{I1: }}~   i=~$p_{\tt I1}$~=alloc 1,int,4;
~{\tiny \textcolor{mygray}{I2: }}~   v=~$p_{\tt I2}$~=alloc *n+2,int,4;
~{\tiny \textcolor{mygray}{I3: }}~   v[0]=0; v[1]=1; *i=2;
~{\tiny \textcolor{mygray}{I4: }}~   if(*i ~$>_s$~ *m) goto I7;
~{\tiny \textcolor{mygray}{I5: }}~     v[*i]=v[*i-1]+v[*i-2];
~{\tiny \textcolor{mygray}{I6: }}~     (*i)++; goto I4;
~{\tiny \textcolor{mygray}{I7: }}~   ~$p_{\tt I7}$~=alloc 1,char*,4;
~{\tiny \textcolor{mygray}{I8: }}~   ~$p_{\tt I8}$~=alloc 1,struct{int; int;},4;
~{\tiny \textcolor{mygray}{I9: }}~   *~$p_{\tt I7}$~=__S__; *~$p_{\tt I8}$~=*m; *~($p_{\tt I8}+4)$~=v[*m];
~{\tiny \textcolor{mygray}{I10:}}~   t=call int printf(~$p_{\tt I7}$~, ~$p_{\tt I8}$~);
~{\tiny \textcolor{mygray}{I11:}}~   dealloc I8;
~{\tiny \textcolor{mygray}{I12:}}~   dealloc I7;
~{\tiny \textcolor{mygray}{I13:}}~   r=v[*m];
~{\tiny \textcolor{mygray}{I14:}}~   dealloc I2;
~{\tiny \textcolor{mygray}{I15:}}~   dealloc I1;
~{\tiny \textcolor{mygray}{I16:}}~   ret r;
\end{myexamplesmallir}
\caption{\label{fig:example1i}(Abstracted) IR.}
\end{subfigure}
\end{minipage}
\hfill
\begin{minipage}[t]{0.55\textwidth}
\begin{subfigure}[t]{\textwidth}
\begin{myexamplesmallasm}
~{\tiny \textcolor{mygray}{A0: }}~ fib:
~{\tiny \textcolor{mygray}{A1: }}~   push ebp~;~ ebp = esp~;~
~{\tiny \textcolor{mygray}{A2: }}~   push {edi, esi, ebx}~;~ esp = esp-12~;~
~{\tiny \textcolor{mygray}{A3: }}~   eax = mem~$_4$~[ebp+8]~;~ ebx = mem~$_4$~[ebp+12]~;~
~{\tiny \textcolor{mygray}{A4: }}~   esp = esp-(0xFFFFFFF0 & (4*(eax+2)+15))~;~
~{\tiny \textbf{\textcolor{red}{A4.1: }}}~  ~\textbf{\textcolor{red}{v$_{\tt I1}$ = alloc$_v$ 4,4,I1; }}~
~{\tiny \textbf{\textcolor{red}{A4.2: }}}~  ~\textbf{\textcolor{red}{alloc$_s$ esp,4*(eax+2),4,I2; }}~
~{\tiny \textcolor{mygray}{A5: }}~   esi = ((esp+3)/4)*4~;~
~{\tiny \textcolor{mygray}{A6: }}~   mem~$_4$~[esi] = 0~;~ mem~$_4$~[esi+4] = 1~;~
~{\tiny \textcolor{mygray}{A7: }}~   if(ebx ~$\le_s$~ 1) jmp A12~;~
~{\tiny \textcolor{mygray}{A8: }}~   edi = 0~;~ edx = 1~;~ eax = 2~;~
~{\tiny \textcolor{mygray}{A9:}}~     ecx = edx+edi~;~ edi = edx~;~ edx = ecx~;~
~{\tiny \textcolor{mygray}{A10:}}~    mem~$_4$~[esi+4*eax] = ecx~;~ eax = eax+1~;~
~{\tiny \textcolor{mygray}{A11:}}~    if(eax ~$\le_s$~ ebx) jmp A9~;~
~{\tiny \textcolor{mygray}{A12:}}~   edi = mem~$_4$~[esi+4*ebx]~;~  esp = esp-4~;~
~{\tiny \textcolor{mygray}{A13:}}~   push {edi, ebx, __S__}~;~ //__S__ is the ptr to format string
~{\tiny \textbf{\textcolor{red}{A13.1:}}}~  ~\textbf{\textcolor{red}{alloc$_s$ esp,\ \ 4,4,I7; }}~
~{\tiny \textbf{\textcolor{red}{A13.2:}}}~  ~\textbf{\textcolor{red}{alloc$_s$ esp+4,8,4,I8; }}~
~{\tiny \textcolor{mygray}{A14:}}~   call ~\textbf{\textcolor{teal}{int}}~ printf~\textbf{\textcolor{teal}{(<char*> esp, <struct\{int; int;\}> esp+4)}}~
           ~\textcolor{teal}{$\pmb{G \cup \{ \heap, cl, \mathtt{I7}, \mathtt{I8} \}}$};~
~{\tiny \textbf{\textcolor{red}{A14.1:}}}~  ~\textbf{\textcolor{red}{dealloc$_s$ I8;}}~
~{\tiny \textbf{\textcolor{red}{A14.2:}}}~  ~\textbf{\textcolor{red}{dealloc$_s$ I7;}}~
~{\tiny \textcolor{mygray}{A15:}}~   eax = edi~;~
~{\tiny \textbf{\textcolor{red}{A15.1:}}}~  ~\textbf{\textcolor{red}{dealloc$_s$ I2;}}~
~{\tiny \textbf{\textcolor{red}{A15.2:}}}~  ~\textbf{\textcolor{red}{dealloc$_v$ I1;}}~
~{\tiny \textcolor{mygray}{A16:}}~   esp = ebp-12~;~ pop {ebx, esi, edi, ebp}~;~
~{\tiny \textcolor{mygray}{A17:}}~   ret~;~
\end{myexamplesmallasm}
\caption{\label{fig:example1a}(Abstracted) 32-bit x86 Assembly Code.}
\end{subfigure}
\end{minipage}
\caption{\label{fig:example1}Example program with VLA and its lowerings to IR and assembly.
Subscripts $_s$ and $_u$ denote signed and unsigned comparison respectively.
Bold font (parts of) instructions are added by our algorithm.
}
\end{figure}

{\em Example}: Consider a C and a 32-bit x86 assembly program
in \cref{fig:example1}.
The {\tt fib} procedure in \cref{fig:example1c}
accepts two integers {\tt n} and {\tt m},
allocates a variable-length array (VLA) {\tt v}
of {\tt n+2} elements,
computes the first {\tt m+1} fibonacci numbers in {\tt v},
calls {\tt printf()},
and returns the $m^{th}$ fibonacci number. Notice that for an execution
free of {\em Undefined Behaviour} (UB),
both {\tt n} and {\tt m} must be non-negative and {\tt m} must be less than {\tt (n+2)}.
Note that the memory for local variables ({\tt v} and {\tt i}) and
procedure-call arguments (for the call to {\tt printf}) is
allocated dynamically through the {\tt alloc} instruction
in the IR program (\cref{fig:example1i}). In the assembly
program (\cref{fig:example1a}), memory is allocated through
instructions that manipulate the stackpointer
register {\tt esp}.

If the IR program uses an address, say $\alpha$, of a local
variable (e.g., $\alpha{}\in{}\{p_{\tt I1},p_{\tt I2}\}$)
or a procedure-call
argument (e.g., $\alpha{}\in{}\{p_{\tt I7},p_{\tt I8}\}$)
in its computation (e.g.,
for pointer arithmetic at lines {\tt I3} and {\tt I5}, or for accessing
the variadic argument at $p_{\tt I8}$
within {\tt printf}), validation requires
a relation between $\alpha$ and its corresponding
stack address
in assembly (e.g., $p_{\tt I7}=\mathtt{esp}$ at line {\tt A14}).

{\em Contributions}: We formalize
IR and assembly execution semantics in the presence
of dynamically (de)allocated memory for local variables
and procedure-call arguments, define
a notion of correct translation,
and provide an algorithm that converts the correctness check
to first-order logic queries over bitvectors, arrays, and uninterpreted
functions.
Almost all production compilers (e.g., GCC) generate
assembly code to dynamically allocate
stack space for procedure-call arguments at the callsite, e.g.,
in \cref{fig:example1a}, the arguments to {\tt printf}
are allocated at line {\tt A13}.
Ours
is perhaps the first effort to enable validation of this common
allocation strategy.
Further, our work enables translation validation for programs with
dynamically-allocated fixed-length and variable-length local variables
for a wide set of allocation strategies used
by a compiler including stack merging,
stack reallocation (if the order of allocations
is preserved), and intermittent register allocation.

\section{Execution Semantics and Notion of Correct Translation}
\label{sec:definitions}
We are interested in showing that
an x86 assembly program \progA{} is a correct translation of
the unoptimized IR representation of
a C program \progC{}.
Prior TV efforts identify a lockstep correlation between
(potentially unrolled) iterations of loops in the two programs
to show equivalence \cite{semalign}.
These correlations can be represented through a {\em product program}
that executes \progC{} and
\progA{} in lockstep, using a careful choice of program path
correlations, to keep the machine states of both programs related at the
ends of correlated paths \cite{covac,oopsla20}.

Our TV algorithm additionally
attempts to identify a lockstep correlation between the dynamic
(de)allocation events and procedure-call events performed in both
programs, i.e.,
we require the order and values of these
execution events to be identical in both programs.
To identify a lockstep correlation, our algorithm
annotates \progA{}
with (de)allocation instructions
and procedure-call arguments.
Our key insight is to define
a {\em refinement relation} between \progC{}
and \progA{} through the existence
of an annotation in \progA{}.
We also generalize the definition of a product program so it can
be used to witness
refinement
in the presence of non-determinism
due to addresses of dynamically-allocated local memory, UB, and stack overflow.

{\em Overview through example}:
In \progC{}, an {\tt alloc}
instruction
returns a non-deterministic
address of the newly allocated region with non-deterministic
contents, e.g.,
in \cref{fig:example1i},
the address ($p_{\tt I2}$)
and initial contents of VLA {\tt v} allocated at {\tt I2} are
non-deterministic.
In \cref{fig:example1a},
our algorithm annotates an {\tt alloc$_s$} instruction at {\tt A4.2}
to correlate in lockstep with {\tt I2},
so that $p_{\tt I2}$'s determinized value is identified through its
first operand ({\tt esp}).
An {\tt alloc$_s$} instruction
allocates a contiguous address interval from
the stack, starting at {\tt esp}
in this case,
to a local variable.
The second ({\tt 4*(eax+2)}), third ({\tt 4}), and fourth
({\tt I2}) operands of {\tt alloc$_s$} specify the allocation size in bytes,
required alignment, and
the PC of the correlated allocation instruction in \progC{} (which also
identifies the local variable) respectively. The determinized
values of the initial contents of VLA {\tt v} at {\tt I2}
are identified to be equal to
the contents of the stack region {\tt [esp,esp+4*(eax+2)-1]} at {\tt A4.1}.
A corresponding {\tt dealloc$_s$} instruction, that
correlates in lockstep with {\tt I14}, is
annotated at {\tt A15.1} to free
the memory allocated by {\tt A4.2} (both have operand {\tt I2})
and return it to stack.

A procedure call appears as an x86 {\tt call} instruction
and we annotate the actual
arguments as its operands in \progA{}.
In \cref{fig:example1a},
the two operands ({\tt esp} and {\tt esp+4}) annotated at {\tt A14}
are the determinized
values of $p_{\tt I7}$ and $p_{\tt I8}$, as obtained
through x86 calling conventions. The last annotation
at {\tt A14}
is the set of memory regions (e.g.,
$G$, \heap{}, $cl$, \ldots{}, as described in \cref{sec:memregions})
observable by {\tt printf} in \progA{} ---
this is equal to the set of memory regions observable by
{\tt printf} in \progC{}, as
obtained
through an over-approximate points-to analysis.
Annotations of {\tt alloc$_s$} at
{\tt A13.\{1,2\}} and {\tt dealloc$_s$}
at {\tt A14.\{1,2\}}
identify the
memory regions occupied by {\tt printf}'s parameters
during {\tt printf}'s execution.

Consider the local variable {\tt i}, allocated at {\tt I1}, with address
$p_{\tt I1}$ in \cref{fig:example1i}. Because {\tt i}'s address is never taken
in the source program, a correlation of $p_{\tt I1}$
with its determinized value in \progA{}'s stack is not necessarily required.
Further, the compiler
may register-allocate {\tt i} in which case no
stack address exists for {\tt i}, e.g., {\tt i}
lives in {\tt eax} at {\tt A8-A11}
in \cref{fig:example1a}.
The
{\tt alloc$_v$}
instruction annotated at {\tt A4.1}
performs a ``virtual allocation'' for variable {\tt i} in lockstep with
{\tt I1}.
The first ({\tt 4}), second ({\tt 4}), and third ({\tt I1})
operands of {\tt alloc$_v$} indicate the allocation size, required alignment,
and the PC of the correlated allocation in \progC{}
respectively. The corresponding {\tt dealloc$_v$} instruction,
annotated at {\tt A15.2}, correlates in lockstep with
{\tt I15}.
The address and initial contents of the memory allocated by {\tt alloc$_v$}
are chosen non-deterministically in \progA{}, and are assumed to be equal to the
address and initial contents of memory allocated by a correlated {\tt alloc}
in \progC{}, e.g., $\mathrm{v}_{\tt I1}=p_{\tt I1}$ at {\tt A4.1}.
A ``virtually-allocated region'' is never used by \progA{}.
We introduce the {\tt (de)alloc$_{s,v}$} instructions formally in \cref{sec:refnDefn}.

Consider the memory
access {\tt v[*i]} at {\tt I5} in \cref{fig:example1i},
and assume we
identify a lockstep correlation of this
memory access
with the assembly program's access {\tt mem$_4$[esi+4*eax]}
at {\tt A10} in \cref{fig:example1a}, with value
relations {\tt esi}$=${\tt v} and {\tt eax}$=${\tt *i}. We
need to cater to the possibility where
{\tt *i$>_s$*n+2} (equivalently, {\tt eax} $>_s$ {\tt mem$_4$[ebp+12]+2}),
which would trigger UB in \progC{}, and
may go out of variable bounds in stack in assembly.
Our product program encodes the necessary
UB semantics that
allow anything to happen in assembly (including out of bound stack access)
if UB is
triggered in \progC{}.

Finally, consider the stackpointer decrement instruction at {\tt A4}
in \cref{fig:example1a}. If {\tt eax} (which corresponds to {\tt *n})
is too large, this instruction at {\tt A4} may potentially overflow
the stack space.  Our product program encodes the assumption that an
assembly program will have the necessary stack space required for execution,
which is necessary to
be able to validate a translation from IR to assembly.

Thus, we are interested in identifying {\em legal} annotations of {\tt (de)alloc$_{s,v}$}
instructions and operands of procedure-call instructions in \progA{}, such that the execution
behaviours of \progA{}
can be shown to refine the execution behaviours of \progC{}, assuming \progA{}
has the required stack space for execution.
We show refinement separately
for each procedure \Ck{} in \progC{} and its corresponding implementation
\Ak{} in \progA{}.
Thereafter, a coinductive argument
shows refinement
for full
programs \progC{} and \progA{} starting at the {\tt main()} procedure.
We do not support inter-procedural transformations.

{\em Paper organization}: \Cref{sec:IRandAsm,sec:tgraph,sec:trans}
describe a procedure's execution
semantics for both IR and assembly representations.
Refinement, through annotations, is defined in \cref{sec:refnDefn}.
\Cref{sec:product}
defines a product program and its associated requirements such that
refinement can be witnessed, and \cref{sec:algo} provides an
algorithm to automatically construct such a product program.

\subsection{Intermediate and Assembly Representations}
\label{sec:IRandAsm}
\subsubsection{IR}
The unoptimized IR used to represent \progC{}
is mostly a subset of LLVM --- it supports all the primitive
types (integer, float, code labels) and the derived
types (pointer, array, struct, procedure) of LLVM.
Being unoptimized, our IR does not need to
support LLVM's {\tt undef} and {\tt poison} values, it instead
treats all error conditions as UB.
Syntactic conversion of C to LLVM IR entails the usual conversion
of types/operators.
A global variable name \g{} or a parameter name \y{}
appearing in a C procedure body is translated to the variable's start
address in IR, denoted $\ghost{\LB{\g}}$ and $\ghost{\LB{\y}}$
respectively\footnote{As we will also
see later, \ghost{\LB{v}} denotes the {\em lower bound}
of the memory addresses occupied by variable with name $v$.}.
A local variable declaration or a call to C's {\tt alloca()} operator
is converted to LLVM's {\tt alloca} instruction, and to distinguish
the two, we
henceforth refer to the latter as the
``{\tt alloc}'' instruction.
Unlike LLVM, our IR also supports
a {\tt dealloc} instruction that deallocates a variable at the
end of its scope --- we use LLVM's {\tt stack\{save,restore\}} intrinsics
(that maintain equivalent scope information for a different purpose) to
introduce explicit {\tt dealloc} instructions in our IR.
Henceforth, we refer to our IR as \ourIR{} (for LLVM + {\tt dealloc}).

We discuss our logical model in the context of compilation to 32-bit
x86 for the relative simplicity of the calling conventions in 32-bit mode.
Like LLVM, a procedure definition in \ourIR{}
can only return a scalar value ---
aggregate return value is passed in memory.
Unlike LLVM which allocates memory for a parameter only if
its address is taken, \ourIR{} allocates
memory for all parameters ---
\ourIR{} thus takes all parameters through pointers,
e.g., both {\tt n} and {\tt m} are passed through pointers
in \cref{fig:example1i}.
This makes the translation of a procedure-call from C to \ourIR{} slightly more verbose, as
explicit instructions to (de)allocate memory for the arguments are required.
An example of this translation is shown in \cref{fig:example1}
where a call to {\tt printf} at {\tt C5} of \cref{fig:example1c}
translates to instructions {\tt I7}-{\tt I12} in \cref{fig:example1i}:
the \ourIR{} program performs two allocations, one for the format string,
and another for the variable argument list;
the latter represented as an object of ``{\tt struct}'' type
containing two {\tt int}s.
The call instruction takes the pointers returned by these allocations as operands.

\Cref{fig:c2IRVariadic} shows the C-to-\ourIR{} translations
for variadic macros.
The translation rules have template holes marked by \Itempl{}
for types and variables of C which are populated at the time of translation
with appropriate \ourIR{} entities. \ourIR{}'s {\tt va\_start\_ptr}
instruction returns the
first address of the current procedure's variable argument list.

\begin{figure}
  \begin{scriptsize}
  \[
  \mprset{flushleft}
  \inferrule*[vcenter]{\mathtt{va\_start}(\mathtt{ap},last)}
  {
    \EmitX{a \Assign \mathtt{va\_start\_ptr}} 
    \\\\
    \EmitX{\mathtt{store}\ \mathtt{void*}, \  4, \ a, \ \Itempl{ap}}
  }
  \quad
  \inferrule*[vcenter]{\mathtt{va\_arg}(\mathtt{ap},\,\tau)}
  {
    \EmitX{a \Assign \mathtt{load} \ \mathtt{void*}, \ 4, \ \Itempl{ap}} 
    \\\\
    \EmitX{result \Assign \mathtt{load}\ \Itempl{\tau}, \ \Itempl{\mathtt{alignof}(\tau)}, \ a} 
    \\\\
    \EmitX{a' \Assign a + \Itempl{\mathtt{roundup_4}(\mathtt{sizeof}(\tau))}} 
    \\\\
    \EmitX{\mathtt{store} \ \mathtt{void*}, \ 4, \ a', \, \Itempl{ap}} 
  }
  \quad
  \inferrule*[vcenter]{\mathtt{va\_copy}(\mathtt{aq},\mathtt{ap})}
  {
    \EmitX{a \Assign \mathtt{load} \ \mathtt{void*}, \ 4, \ \Itempl{ap}} 
    \\\\
    \EmitX{\mathtt{store} \ \mathtt{void*}, \ 4, \ a, \, \Itempl{aq}} 
  }
  \quad
  \inferrule*[vcenter]{\mathtt{va\_end}(\mathtt{ap})}
  {
    \EmitX{\mathtt{store} \ \mathtt{void*}, \ 4, \ 0, \, \Itempl{ap}} 
  }
  \]
\end{scriptsize}
\caption{\label{fig:c2IRVariadic} Translation of C's variadic macros to \ourIR{} instructions.
$\mathtt{roundup_4}(a)$ returns the closest multiple of 4 greater than or equal to $a$.
  }
\end{figure}

\subsubsection{Assembly}
Broadly, an assembly program \progA{} consists of a code
section (with a sequence of assembly instructions), a data section (with
read-only and read-write global variables), and a symbol table
that maps string symbols to memory addresses in code and data sections.
The validator checks that the regions specified
by the symbol table
are well-aligned and non-overlapping, and uses it
to relate a global variable or
procedure in \progC{} to its address or implementation in \progA{}.

We assume that the OS guarantees the caller-side contract of the
ABI calling conventions for the entry procedure, {\tt main()}.
For
32-bit x86, this means that at the start of program execution, the stackpointer
is available in register {\tt esp}, and
the return address and input parameters {\tt (argc,argv)} to {\tt main()} are available
in the stack
region just above the stackpointer.
For
other procedure-calls, the
validator verifies the adherence to calling conventions at a callsite (in
the caller)
and assumes adherence at procedure entry (in the callee).
Heap allocation procedures like {\tt malloc()} are left
uninterpreted, and so, the only compiler-visible way to allocate (and
deallocate) memory in \progA{}
is through the decrement (and increment) of the stackpointer stored in register {\tt esp}.

\subsubsection{Allocation and Deallocation}
\label{sec:allocdealloc}
Allocation and deallocation instructions appear only in \progC{}, and do not appear in \progA{}.
Let \Ck{} represent a procedure in program \progC{}.

An \ourIR{} instruction ``{\tt \PC{\Ck}{a}: v \Assign{} alloc n, $\tau$, align}''
allocates a contiguous region of local memory with space for {\tt n} elements
of type $\tau$ aligned by {\tt align}, and returns its start address
in {\tt v}. The PC, \PC{\Ck}{a}, of an {\tt alloc} instruction is
also called an {\em allocation site} (denoted by \z{}), and
let the set of allocation sites in \Ck{} be \Z{}.
During conversion of the C program to \ourIR{}, we distinguish
between allocation sites due
to the declaration of a local variable (or a procedure-call argument)
and allocation sites due
to a call to {\tt alloca()} --- we use \Zl{} for the former
and \Za{} for the latter, so that
$\Z{}=\Zl\cup{}\Za$.

The address of an
allocated region is non-deterministic, but is subject
to two {\em Well-Formedness (WF) constraints}: (1) the newly allocated memory region should be separate
from all currently allocated memory regions, i.e., there should be no overlap; and (2) the address of the newly
allocated memory region should be aligned by {\tt align}.

An \ourIR{} instruction ``{\tt \PC{\Ck}{d}: dealloc \z{}}'' deallocates {\em all}
local memory regions allocated due to the execution of
allocation site $\z{}\in{}\Z{}$.

\subsection{Transition Graph Representation}
\label{sec:tgraph}
An \ourIR{} or assembly instruction may mutate the machine state, transfer control, perform I/O, or
terminate the execution.
We represent a C procedure, \Ck{},
as a transition graph, $\Ck=(\mathcal{N}_{\Ck},\mathcal{E}_{\Ck})$, with
a finite set of nodes $\mathcal{N}_{\Ck}=\{n^s=n_1,n_2,\ldots,n_{m}\}$, and
a finite set of labeled directed edges $\mathcal{E}_{\Ck}$. A unique node $n^s$ represents
the start node or entry point of \Ck{}, and
every other node $n_j$ must be reachable from $n^s$. A node
with no outgoing edges is a {\em terminating node}.
A variable in \Ck{} is identified by its scope-resolved unique name.
The machine state of \Ck{} consists of the set of input
parameters $\vv{y}$,
set of temporary variables $\vv{t}$, and an explicit
array variable $M_{\Ck}$ denoting the current state of memory.
We use \bv{N} to denote a bitvector type of size $N>0$.
$M_{\Ck}$ is of type $\mathtt{T}(M_\Ck)=\memType{}$.

An assembly implementation of the C procedure \Ck{}, identified through the
symbol table, is the assembly procedure \Ak{}.
The machine state of \Ak{} consists of its hardware
registers $\vv{regs}$
and memory $M_{\Ak}$. Similarly
to \Ck{}, $\Ak=(\mathcal{N}_{\Ak},\mathcal{E}_{\Ak})$ is also represented as
a transition graph.

Let $P\in{}\{\Ck,\Ak\}$.
In addition to the memory (data) state $M_P$, we
also need to track the allocation
state, i.e., the set of intervals
of addresses that have been allocated by the procedure.
We use $\alpha$ (potentially with a subscript)
to denote a memory address of bitvector type.
Let $\intervalSym=[\alpha_b,\alpha_e]$ represent an {\em address interval} starting at $\alpha_b$
and ending at $\alpha_e$ (both inclusive),
such that $\alpha_b \le_u \alpha_e$
(where $\le_u$ is unsigned comparison operator).
Let $[\alpha]_{w}$ be a shorthand for the address interval $[\alpha,\alpha + w - 1_{\bv{32}}]$, where $n_{\bv{32}}$ is the two's complement representation of integer $n$ using $32$ bits.

\subsubsection{Address Set}
Let \asetSym{} (potentially with a sub- or superscript) represent a set of addresses, or an {\em address set}.
An empty address set is represented by $\emptyset$, and an
address set of contiguous
addresses is an interval $\intervalSym$.
Two address sets overlap, written $\Overlap(\asetSym_1,\asetSym_2)$,
iff $\asetSym_1\cap{}\asetSym_2\neq{}\emptyset{}$.
Extended to $m\geq{}2$ sets,
$\Overlap(\asetSym_1,\asetSym_2,\ldots,\asetSym_m)\Leftrightarrow{}\exists_{1\leq{}j_1<j_2\leq{}m}{\Overlap(\asetSym_{j_1},\asetSym_{j_2})}$.
$|\asetSym{}|$ represents the number of distinct addresses in \asetSym{}.
For a non-empty address set, \LBi{\asetSym{}} and
\UBi{\asetSym{}} represent the smallest and largest address respectively in \asetSym{}.
${\tt comp}(\asetSym{})$
represents the {\em complement} of \asetSym{}, so that:
$\forall\alpha:{(\alpha{}\in{}\asetSym{})\Leftrightarrow{}(\alpha{}\notin{}{\tt comp}(\asetSym{}))}$.

\subsubsection{Memory Regions}
\label{sec:memregions}
To support dynamic (de)allocation, an execution model
needs to individually track regions of memory belonging to each
variable, heap, stack, etc. We next describe
the memory regions tracked by our model.
\begin{enumerate}
\item Let $G$ be the
set of names of all global variables in \progC{}.
For each global variable $\g{}\in{}G$, we
track the memory region belonging to that variable.
We use the name of a global variable $\g \in G$ as its {\em region identifier}
to identify the region belonging to $g$ in both \Ck{} and \Ak{}.

\item For a procedure \Ck{},
  let \Y{} be the set of
names of formal parameters, including
the variadic parameter, if present. We use
the special name \yv{} for the variadic parameter.
The memory region belonging to a parameter $\y{}\in{}\Y{}$ is
called \y{} in both \Ck{} and \Ak{}.

\item The memory region allocated by allocation site $\z\in\Z$ is called \z{} in \Ck{}.
In \Ak{}, our algorithm potentially annotates allocation instructions corresponding to
an allocation site \z{} in \Ck{}.

\item $\heap$ denotes the region belonging to the {\em program heap} (managed by the OS)
in both \Ck{} and \Ak{}.

\item Local variables and actual arguments may be allocated by
the {\em call chain} of a procedure (caller, caller's caller, and so on).
The accessible subset
 (accessible to procedure \Ck{}%
      \footnote{A local variable or actual argument $v$ of procedure \CkP{} in the call chain of procedure \Ck{} is accessible in procedure \Ck{} only if the address of $v$ is accessible in \Ck{}, i.e., $v$ is address-taken in \CkP{}.})
is coalesced into a single region denoted by region $cl$,
or {\em callers' locals}, in both \Ck{} and \Ak{}.

\item In procedure \Ak{}, stack memory can be allocated and deallocated through stackpointer
decrement and increment. The addresses belonging to the
stack frame of \Ak{} (but not to a local variable $\z\in\Z$ or a parameter $\y{}\in{}\Y{}$)
belong to the \stk{} (stack) region in \Ak{}. The \stk{} region is absent in \Ck{}.

\item Similarly, in \Ak{}, we use $cs$ ({\em callers' stack})
to identify the
region that belongs to the stack space (but not to $cl$)
of the call chain
of procedure \Ak{}. $cs$ is absent in \Ck{}.

\item Program \progA{} may use more global memory than \progC{}, e.g., to
store pre-computed constants to implement vectorizing transformations.
Let \F{} be the
set of names of all {\em assembly-only global variables} in \progA{}.
For each $\f\in{}\F$, its
memory region in \Ak{} is identified by \f{}.

\item The region {\tt free} denotes the free space, that
does not belong to any of the aforementioned regions, in both \Ck{} and \Ak{},
 
\item The region $cv$%
\footnote{$cv$ stands for \emph{callers' virtual}.  The reason for tracking this region will become apparent when we discuss virtual allocation in \cref{sec:refnDefnVirtual}.}
denotes the inaccessible subset of local variables and actual arguments in the call chain
of \Ck{}.
$cv$ is present in both \Ck{} and \Ak{}. 

\end{enumerate}

Let $\Rall{}=G\cup\Y\cup\Z\cup\F\cup\{\heap,cl,cv,\stk,cs,\free\}$
represent all {\em region identifiers};
$\STACK{}=\{\stk,cs\}$ denote the stack regions in \Ak{}
and
$\NS{}=G\cup{}\Y\cup\Z\cup\{\heap,cl\}$
($= \Rall{} \setminus (\F \cup \STACK \cup \{ \free, cv\})$)
denote
the accessible regions in \emph{b}oth \Ck{} and \Ak{}.

Let $\Gro{} \subseteq G$ be the set of read-only global variables in \progC{};
and, let $\Grw{} = G \setminus \Gro{}$ denote the set of read-write global variables.
We define $\Fro \subseteq \F$ and $\Frw = \F \setminus \Fro$ analogously.

For each non-free region $r\in{}(\Rall{}\setminus{}\{\mathtt{free}\})$,
the machine state of a procedure $P$ includes a unique variable
$\il{P}{r}$ that tracks
region $r$'s
address set as $P$ executes. If $\il{P}{r}$ is a contiguous non-empty
interval, we also refer to it as $\ii{P}{r}$.
For $r\in{}G\cup{}\Y{}\cup\F\cup\{\heap,cl,cv,cs\}$
($r \in \Rall \setminus (\Z \cup \{ \free, \stk \})$),
$\il{P}{r}$ remains constant throughout $P$'s execution.
For $\R{}\subseteq{}\Rall{}$, we define an expression
$\il{P}{\R{}}=\bigcup_{r\in{}\R{}}{\il{P}{r}}$.
Because \Ck{} does not have a stack or an
assembly-only global variable, $\il{\Ck}{\F\cup{}\STACK{}}=\emptyset{}$
holds throughout \Ck{}'s execution.
At any point in $P$'s execution, the free space
can be computed as
$\il{P}{\free{}}={\tt comp}(\il{P}{\NS{}\cup{}\F\cup{}\STACK{}\cup\{cv\}})$. 
Notice that we do not use an explicit variable to
track $\il{P}{\free{}}$.

\subsubsection{Ghost Variables}
Our validator introduces {\em ghost variables} in a
procedure's execution semantics, i.e., variables that were not
originally present in $P$.
We
use $\ghost{x}$ to indicate
that $x$ is a ghost variable.
For each region $r\in G\cup\Y\cup\Z$
(resp.\ $r \in \F$),
we introduce
$\ghost{\Empty{r}}$,
$\ghost{\LB{r}}$, and $\ghost{\UB{r}}$ in \Ck{} (resp. \Ak{})
to track the
{\em emptiness} (whether the
region is empty), {\em lower bound} (smallest
address), and {\em upper bound} (largest
address) of $\il{\Ck}{r}$ (resp. $\il{\Ak}{r}$) respectively;
for $r\in G\cup\Y$ (resp. $r \in \F$),
$\ghost{\Sz{r}}$ tracks the size of $\il{\Ck}{r}$ (resp. $\il{\Ak}{r}$),
and for $\z \in \Z$, $\ghost{\LSz{\z}}$ tracks
the {\em size of last allocation} due to allocation-site \z{}.
\ghost{\rdAcc{P}} and \ghost{\wrAcc{P}} track the set of
addresses read and written by
$P$ respectively.
Let $\ghost{+}$ be the set of all ghost variables.

\subsubsection{Error Codes}
Execution of \Ck{} or \Ak{}
may terminate successfully,
may never terminate,
or
may terminate with an error.
We support two error codes to distinguish between two
categories of errors: $\mathscr{U}$ and
$\mathscr{W}$.
In \Ck{}: $\mathscr{U}$ represents an occurrence of
UB, and $\mathscr{W}$ represents a violation of a WF constraint that needs
to be ensured either by the compiler or the OS (both external
to the program itself).
In \Ak{}: $\mathscr{U}$ represents UB or a translation error, and
$\mathscr{W}$ represents occurrence of a condition that can be assumed to never
occur, e.g., if the OS ensures that it never occurs.
In summary, for a procedure $P$, $\mathscr{W}$ represents
an error condition that $P$ can {\em assume} to be absent
(because the external environment ensures it), while $\mathscr{U}$ represents
an error that $P$ must {\em ensure} to be absent.

\subsubsection{Outside World and Observable Trace}
Let \World{P} be a
state of the outside world (OS/hardware) for $P$ that supplies external
inputs whenever $P$ reads from it, and consumes external outputs
generated by $P$. \World{P} is
assumed to mutate arbitrarily
but deterministically based on the values consumed or produced due to
the I/O operations performed by $P$ during execution.
Let $\Trace{P}$ be a potentially-infinite sequence of observable
trace events generated by an execution of $P$.

\subsubsection{Expressions}
Let variable $v$
and variables $\vv{v}$ or $\vv{x}$ be drawn from
$\mathrm{{\tt Vars}}=(\vv{t},\vv{regs},M_P,\il{P}{r},\allowbreak \ghost{+})$
(for all $P\in\{\Ck,\Ak\}$ and for
all $r\in{}(\Rall{}\setminus{}\{\mathtt{free}\})$);
$e(\vv{x})$ be an expression over $\vv{x}$,
and $E(\vv{x})$ be
a list of expressions over $\vv{x}$.
An expression $e(\vv{x})$ is a well-formed combination of
constants, variables $\vv{x}$, and arithmetic, logical, relational, memory access, and address set operators.
For memory reads and writes,
{\tt select} ({\tt sel} for short)
and
{\tt store} ({\tt st} for short)
operations are used to access
and modify $M_P$ at a given address $\alpha$.
Further, the {\tt sel} and {\tt st} operators are associated with
a \size{} parameter: {\tt \select$_\sizeMath{}${}(arr,$\alpha$)}
returns a little-endian concatenation
of \size{} bytes starting at $\alpha$ in the array {\tt arr}.
Similarly, {\tt \store$_\sizeMath{}${}(arr,$\alpha$,data)} returns a new
array that has contents identical to {\tt arr} except for the \size{} bytes starting
at $\alpha$ which have been replaced by {\tt data} in little-endian format.
To encode reads/writes to a region of memory, we define projection and updation
operations.
\begin{definition}
\prjM{\asetSym}{M_P} denotes the {\em projection} of $M_P$ on addresses
in \asetSym{}, i.e., if $M'_P=\prjM{\asetSym}{M_P}$,
then $\forall\alpha\in{}\asetSym:{\selectMath_1(M'_P,\alpha)=\selectMath_1(M_P,\alpha)}$
and $\forall\alpha\notin{}\asetSym:{\selectMath_1(M'_P,\alpha)=0}$.  The sentinel
value $0$ is used for the addresses outside \asetSym{}.
We use $\prjMEq{\asetSym}{M_{P_1}}{M_{P_2}}$ as shorthand for
$(\prjM{\asetSym}{M_{P_1}}=\prjM{\asetSym}{M_{P_2}})$.
\end{definition}

\begin{definition}
\label{def:updM}
\updM{\asetSym}{M_P}{M} denotes the {\em updation} of $M_P$ on addresses
in \asetSym{} using the values in $M$. If $M'_P=\updM{\asetSym}{M_P}{M}$,
then $\prjMEq{\asetSym}{M'_P}{M}$
and $\prjMEq{{\tt comp}(\asetSym)}{M'_P}{M_{P}}$ hold.
\end{definition}

\subsubsection{Instructions}
\label{sec:instructions}
Each edge
$e_{P}\in{}\mathcal{E}_{P}$ is labeled with one of the following {\em graph instructions}:
\begin{enumerate}
\item A {\em simultaneous assignment} of the form $\vv{v} \Assign E(\vv{x})$.  Because variables
$\vv{v}$ and $\vv{x}$ may include $M_P$,
an assignment suffices for encoding memory loads and stores.
Similarly, because the variables may be drawn
from $\il{P}{\z}$ (for an allocation site \z{}), an assignment
is also used to encode the allocation of an interval $\intervalSym_{\tt new}$
($\il{P}{\z} \Assign \il{P}{\z} \cup \intervalSym_{\tt new}$)
and the deallocation of all addresses allocated due to \z{}
($\il{P}{\z} \Assign \emptyset{}$). Stack allocation and
deallocation in \Ak{} can be
similarly represented as $\il{\Ak}{\stk} \Assign \il{\Ak}{\stk} \cup \intervalSym_{{\tt new}}$
and $\il{\Ak}{\stk} \Assign \il{\Ak}{\stk} \setminus \intervalSym_{{\tt new}}$
respectively.

\item A {\em guard} instruction, $e(\vv{x})?$, indicates that when execution
  reaches its head, the
  edge is taken iff its {\em edge condition} $e(\vv{x})$ evaluates to {\tt true}.
  For every other instruction, the edge is always taken upon reaching
  its head, i.e., its edge condition is {\tt true}.
For a non-terminating node $n_P\in{}\mathcal{N}_P$, the guards of all edges departing
from $n_P$ must be mutually exclusive, and their disjunction
must evaluate to {\tt true}.

\item A type-parametric
  {\em choose} instruction $\Ichoose{\vv{\tau}}$. Instruction $\vv{v} \Assign \Ichoose{\vv{\tau{}}}$ non-deterministically
chooses values of types $\vv{\tau}$ and assigns
them to variables $\vv{v}$, e.g.,
a memory with non-deterministic contents is obtained by using
$\Ichoose{\bv{32}\rightarrow{}\bv{8}}$.

\item A {\em read} ({\tt rd}) or {\em write} ({\tt wr}) I/O instruction.
A read instruction $\vv{v} \Assign \Iread{\vv{\tau}}$
reads values of types $\vv{\tau}$
from the outside world into variables $\vv{v}$,
e.g., an address set is read using $\asetSym{}\Assign{}\Iread{\asetType{}}$.

A write instruction \Iwrite{V(E(\vv{x}))} writes the value constructed
by value constructor $V$ using $E(\vv{x})$ to the outside world.
A value constructor is defined for each type of observable event.
For a procedure-call, $\Vcall{\rho, \Tuple{v}, \vv{r}, M}$ 
represents a value constructed for a procedure-call to callee
with name (or address) $\rho$, the actual arguments $\Tuple{v}$,
callee-observable regions $\vv{r}$, and memory $M$.
Similarly,
$\Vret{E(\vv{x})}$ is a
value constructed during
procedure return that captures observable
values computed through $E(\vv{x})$.
Local (de)allocation events have their own
value
constructors, $\VallocA{\z,w,a}$, $\VallocB{\z,i,M}$, and $\Vdealloc{\z}$,
which represent (de)allocation due to allocation site \z{} with the
associated
observables $w$ (size), $a$ (alignment), $i$ (interval), and $M$ (memory).

A read or write instruction mutates \World{P} arbitrarily based on
the read and written values.
Further, the data items read or written are appended to the observable trace $T_P$.
Let $\mathtt{read}_{\vv{\tau}}(\World{P})$ be an uninterpreted
function that reads values of types $\vv{\tau}$ from $\World{P}$;
and \Iio{E(\vv{x})}{rw} be an uninterpreted function
that returns an updated state of $\World{P}$ after an I/O operation of type $\mathtt{rw}\in{}\{\mathtt{r},\mathtt{w}\}$ (read
or write)
with values $E(\vv{x})$.
Thus, in its explicit syntax, $\vv{v} \Assign \Iread{\vv{\tau}}$ translates
to a sequence of instructions: $\vv{v} \Assign \mathtt{read}_{\vv{\tau}}(\World{P});$
$\World{P} \Assign \Iio{\vv{v}}{r};$
$T_P \Assign T_P \cdot \vv{v}$,
where $\cdot$ is the trace concatenation operator.
Similarly, \Iwrite{V(E(\vv{x}))} translates
to: $\World{P} \Assign \Iio{V(E(\vv{x}))}{w};$
$T_P \Assign T_P \cdot V(E(\vv{x}))$.
Henceforth, we only use the implicit syntax for brevity.

\item An error-free and error-indicating {\em halt} instruction that
terminates execution. \Ihalt{\emptyset} indicates termination
without error, and \Ihalt{\mathscr{r}} indicates
termination with error code $\mathscr{r}\in{}\{\mathscr{U},\mathscr{W}\}$.
Upon termination without error, a special {\tt exit} event
is appended to observable trace $T_P$. Upon termination
with error, the error code is appended to $T_P$. Thus,
the destination of an edge with a {\tt halt} instruction is a terminating node.
We create a unique terminating node for an error-free exit.
We also create a unique terminating node for each error code, also
called an {\em error node}; an edge terminating at an error node is
called an {\em error edge}. $\mathscr{U}_P$ and
$\mathscr{W}_P$ represent error nodes in $P$ for errors
$\mathscr{U}$ and $\mathscr{W}$ respectively. Execution transfers to
an error node upon encountering the corresponding error.
Let $\NNP{P} = \NP{P} \setminus \{ \EU_{P}, \EW_{P} \}$ be the set of non-error nodes in $P$.

\end{enumerate}
In addition to the observable trace events generated by {\tt rd}, {\tt wr}, and {\tt halt}
instructions, the execution of every instruction in $P$ also appends an
observable {\em silent trace event}, denoted $\bot$, to $T_P$. Silent trace events
count the number of
executed instructions as a proxy for observing the passage of time.

\subsection{Translations of \Ck{} and \Ak{} to Their Graph Representations}
\label{sec:trans}

\begin{table}
  \caption{\label{tab:preds} Definitions of operators and predicates used in translations in \cref{fig:xlateRuleIR,fig:xlateRuleAsm,fig:xlateRuleAsmStackLocals,fig:xlateRuleAsmAllLocals}}
\begin{scriptsize}
  \begin{tabularx}{\textwidth}{@{}p{1.4cm}X@{}}
    \hline
    Operator & Definition \\
    \hline \hline

    $\mathtt{sz}(\tau)$ &
    Returns the size (in bytes) of type $\tau$. For example, $\mathtt{sz}(\mathtt{\bv{32}}) = 4$, $\mathtt{sz}(\mathtt{\bv{8}*}) = 4$, and $\mathtt{sz}(\mathtt{[80\ x\ \bv{8}]}) = 80$.
    \\ \hline

    $\mathtt{T}(a)$ &
    Returns the type $\tau$ of $a$ where $a$ can be a global variable, a parameter, or a register.  For example, $\mathtt{T}(\mathtt{eax}) = \bv{32}$.
    \\ \hline

    $\VretVal{\tau}$ &
    A macro operator which derives the return value of an assembly procedure with return type $\tau$ from input registers \texttt{eax} and \texttt{edx} using the calling conventions, e.g.,
    $\VretVal{\mathtt{\bv{8}}} = \mathtt{extract}_{7,0}(\mathtt{eax})$,
    $\VretVal{\mathtt{\bv{32}}} = \mathtt{eax}$,
    $\VretVal{\mathtt{\bv{64}}} = \mathtt{concat}(\mathtt{edx}, \mathtt{eax})$,
  where $\mathtt{extract}_{h,l}(a)$ extracts bits $h$ down to $l$ from $a$ and $\mathtt{concat}(a,b)$ returns the bitvector concatenation of $a$ and $b$ where $b$ takes the less significant position.
    \\ \hline

    $\VgetVal{\tau}{v}$ &
    Inverse of $\VretVal{\tau}$.  Distributes the packed bitvector $v$ of type $\tau$ into two bitvectors of 32 bit-width each, setting the bits not covered by $v$ to some non-deterministic value.
    \\ \hline

    $\roMem{P}{r}{\intervalSym}$  &
    Returns a memory array containing the contents of read-only global variable named $r$ in $P$.  The contents are mapped at the addresses in the provided interval \intervalSym{}.
    \\ \hline

    $\initAddrSets{\F}$  &
    Returns the address sets of the assembly-only global variables \F{} using the symbol table in \progA{}.
    \\ \hline

  \end{tabularx}
  \begin{tabularx}{\textwidth}{@{}p{2.7cm}X@{}}
    \hline
    Predicate & Definition \\
    \hline \hline

    $\mathtt{aligned}_{n}(a)$ &
    Bitvector $a$ is at least $n$ bytes aligned.  Equivalent to: $a \% n = 0$, where $\%$ is remainder operator.
    \\ \hline


    $\mathtt{isAlignedIntrvl}_{a}(p, w)$  &
    A $w$-sized sequence of addresses starting at $p$ is aligned by $a$ and does not wraparound.
    Equivalent to:
    $\mathtt{aligned}_{a}(p) \land (p \leq_u p+w-1_{\bv{32}})$.
    \\ \hline


    $\mathtt{accessIsSafeC}_{\tau,a}(p, \asetSym)$ &
    Equivalent to:
    $\mathtt{isAlignedIntrvl}_{\mathtt{align}(a)}(p, \mathtt{sz}(\tau))
    \land
    ([p]_{\mathtt{sz}(\tau)} \subseteq \asetSym)$.
    \\ \hline

    $\mathtt{addrSetsAreWF}(\il{P}{\heap}, \il{P}{cl}, \allowbreak \il{P}{cv}, \ldots, \ii{P}{\g}, \ldots, \il{P}{\f}, \ldots,\ii{P}{\y}, \allowbreak \ldots, \il{P}{\yv})$ &
    The address sets passed as parameter are well-formed with respect to C semantics.
    Equivalent to:
    $(0_{\bv{32}} \notin \il{P}{G \cup \F \cup \Y \cup \{\heap,cl,cv\}})
    \land \neg \Overlap(\il{P}{\heap},\il{P}{cl},\ldots,\ii{P}{\g},\ldots,\il{P}{\f},\ldots,\ii{P}{\y},\ldots,\il{P}{\yv})
    \land \neg \Overlap(\il{P}{G \cup \Y \cup \{ \heap, cl \}},\il{P}{cv}) 
    \land (\il{P}{\yv} \ne \emptyset \Rightarrow \mathtt{isInterval}(\il{P}{\yv}))
    \land \forall_{r \in G \cup (\Y \setminus \{ \yv \}) \cup \F} (|{\ii{P}{r}}| = \mathtt{sz}(\mathtt{T}(r))
    \land \mathtt{aligned}_{\mathtt{algnmnt}(r)}(\LBr{\ii{P}{r}}))$,
    where $\mathtt{isInterval}(\il{P}{\yv})$ holds iff the address set \il{P}{\yv} is an interval,
     $\mathtt{algnmnt}(r)$ returns the alignment of variable $r$.
  \\ \hline

    $\intervalContainedInAddrSet{\alpha_b}{\alpha_e}{\asetSym}$ &
    The pair $(\alpha_{b}, \alpha_{e})$ forms a valid interval inside the address set $\asetSym$.
    Equivalent to:
    $(\alpha_{b} \neq 0_{\bv{32}})
      \land (\alpha_{b} \leq_u \alpha_{e})
      \land ([\alpha_{b},\alpha_{e}] \subseteq \asetSym)$
    \\ \hline

    $\intervalContainedInAddrSetAndAligned{\alpha_b}{\alpha_e}{\asetSym}{a}$ &
    Equivalent to:
    $\mathtt{aligned}_{a}(\alpha_{b}) \land{} \intervalContainedInAddrSet{\alpha_b}{\alpha_e}{\asetSym}$
    \\ \hline
    $\mathtt{obeyCC}(e_{\mathtt{esp}}, \Tuple{\tau}, \Tuple{x})$ &
    Pointers \Tuple{x} match the expected addresses of arguments for a procedure-call in assembly.  Based on the calling conventions, $\mathtt{obeyCC}$ uses the value of the current stackpointer ($e_{\mathtt{esp}}$) and parameter types (\Tuple{\tau}) to obtain the expected addresses of the arguments.
    For example, $\mathtt{obeyCC}(\mathtt{esp},\allowbreak(\mathtt{\bv{8}},\mathtt{\bv{32}}),\allowbreak(\mathtt{esp}, \mathtt{esp}+4_{\bv{32}}))$ holds.
    \\ \hline

    $\mathtt{overflow}_{mul}(a, b)$ &
    Signed multiplication of bitvectors $a$, $b$ overflows.  E.g., $\mathtt{overflow}_{mul}(2147483647_{\bv{32}}, 2_{\bv{32}})$ holds.
    \\ \hline

    $\mathtt{stkIsWF}(\mathtt{esp}, \ghost{\spE}, \ghost{\stkE}, \Tuple{\tau},\allowbreak \il{\Ak}{\heap}, \il{\Ak}{cl}, \il{\Ak}{G\cup\F}, \ldots,\ii{\Ak}{\y},\ldots,\allowbreak \il{\Ak}{\yv})$ &
    The pairs $(\mathtt{esp}, \ghost{\spE})$, $(\ghost{\spE}, \ghost{\stkE})$ represent well-formed intervals for initial \stk{} and $cs$ regions with respect to parameter types \Tuple{\tau} and other (input) address sets in \Ak{}.
    Equivalent to:
    $\mathtt{aligned}_{16}(\mathtt{esp}+4_{\bv{32}}) \allowbreak
    \land (\mathtt{esp} \leq_u \mathtt{esp}+4_{\bv{32}})
    \land \neg \Overlap([\mathtt{esp}]_{4_{\bv{32}}}, \il{\Ak}{G \cup \F \cup \Y \cup \{ \heap, cl \}}) 
    \land \mathtt{obeyCC}(\mathtt{esp}+4_{\bv{32}}, \Tuple{\tau}, \ldots,\LBr{\ii{\Ak}{\y}},\ldots)
    \land (\ghost{\spE} <_u \ghost{\stkE})
    \land \neg \Overlap([\ghost{\spE}+1_{\bv{32}},\ghost{\stkE}], \il{\Ak}{G \cup \F \cup \{\heap\}})
    \land \il{\Ak}{cl} \subseteq [\ghost{\spE}+1_{\bv{32}}, \ghost{\stkE}]$
    \\ \hline

    $\mathtt{UB}_{P}(\mathtt{op}, \vv{x})$ &
    Application of operation {\tt op} of procedure $P$ on arguments $\vv{x}$ triggers UB.  E.g., $\mathtt{UB}_{\Ck}(\mathtt{udiv}, (1_{\bv{32}}, 0_{\bv{32}}))$ holds.
    \\ \hline

  \end{tabularx}
\end{scriptsize}
\end{table}

\Cref{fig:xlateRuleIR,fig:xlateRuleAsm} (and \cref{fig:xlateRuleAsmStackLocals,fig:xlateRuleAsmAllLocals} later) present the key translation rules
from \ourIR{} and (abstracted) assembly instructions to graph instructions.
Each rule is composed of three parts separated by a horizontal line segment:
on the left is the name of the rule, above the line segment is the \ourIR{}/assembly
instruction, and below the line segment is the graph instructions listing.
We describe the operators and predicates used in the rules in \cref{tab:preds}.
As an example,
the top right corner of \cref{fig:xlateRuleIR} shows the parametric \TRule{Op} rule
which gives the translation of an operation using arithmetic/logical/relational operator
{\tt op} in \ourIR{} to corresponding graph instructions.
We use C-like constructs in graph instructions as syntactic sugar for brevity, e.g.
`;' is used for sequencing, `?:' is used for conditional assignment,
and \underline{\tt if}, \underline{\tt else}, and \underline{\tt for}
are used for control flow transfer.
We highlight the
read and write
I/O instructions with a
\Shaded{} background, and use
\textbf{\textcolor{somered}{bold}}, \textbf{\textcolor{someviolet}{colored}}
fonts for {\tt halt} instructions.
We use macros \texttt{IF} and \texttt{ELSE} to choose translations based on a
boolean condition on the input syntax.

\subsubsection{Translation of \Ck{}}
\label{sec:ctrans}
\begin{figure}
  \begin{scriptsize}
  \begin{minipage}[t]{.6\textwidth}
  \[
  \mprset{flushleft}
  \inferrule*[left=(\Entry{\Ck}),vcenter]{\PC{\Ck}{j}: \mathtt{def}\ \Ck(\Tuple{\tau})}
  {
    \HLBox{\il{\Ck}{\heap}, \il{\Ck}{cl}, \il{\Ck}{cv}, \ldots, \ii{\Ck}{\g}, \ldots,\ii{\Ck}{\y},\ldots, \il{\Ck}{\yv} \Assign \IreadB{\asetType{},\asetType{},\ldots,\asetType{}}};
    \\\\
    \il{\Ck}{\stk}, \il{\Ck}{cs}, \ldots, \il{\Ck}{\f}, \ldots, \il{\Ck}{\z}, \ldots, \ghost{\rdAcc{\Ck}}, \ghost{\wrAcc{\Ck}} \Assign \emptyset, \emptyset, \ldots, \emptyset;
    \\\\
    \underline{\mathtt{if}}\ (\neg \mathtt{addrSetsAreWF}(\il{\Ck}{\heap}, \il{\Ck}{cl}, \il{\Ck}{cv}, \ldots, \ii{\Ck}{\g}, \ldots, \il{\Ck}{\f}, \ldots, \ii{\Ck}{\y},\ldots, \il{\Ck}{\yv}))
    \\\\ \quad \IhaltW{};
    \\\\
    M_{\Ck} \Assign \IchooseM{};
    \quad
    \HLBox{M_{\Ck} \Assign {\tt upd}_{\il{\Ck}{\NS{}}}(M_{\Ck}, \IreadB{\memType})}; 
    \\\\
    \underline{\mathtt{for}}\ \g\ \mathtt{in\ }\Gro{}\ \{\ M_{\Ck} \Assign {\tt upd}_{\ii{\Ck}{\g}}(M_{\Ck}, \roMem{\Ck}{\g}{\ii{\Ck}{\g}});\ \}
    \\\\
    \underline{\mathtt{for}}\ \z\ \mathtt{in\ }\Z\ \{\ \ghost{\Empty{\z}} \Assign \mathtt{true}; \quad \BasedOnM{\z} \Assign \emptyset; \}
    \\\\
    \underline{\mathtt{for}}\ r\mathtt{\ in\ }G\cup \Y\cup\{\heap,cl\}\ \{\ \BasedOnM{r} \Assign G\cup\{\heap,cl\};\ \} 
    \\\\
    \underline{\mathtt{for}}\ r\mathtt{\ in\ }G\cup{}\Y\ \{\ 
      \ghost{\Sz{r}},\ghost{\Empty{r}} \Assign |\il{\Ck}{r}|, (|\il{\Ck}{r}| = 0_{\bv{32}});
      \\\\ \quad
      \underline{\mathtt{if}}(\neg \ghost{\Empty{r}})\ \{\ 
      \ghost{\LB{r}}, \ghost{\UB{r}} \Assign \LBr{\il{\Ck}{r}}, \UBr{\il{\Ck}{r}};\ \}
      \quad
      \BasedOn{\ghost{\LB{r}}} \Assign \{ r \};
      \\\\
    \}
  }
  \]
  \[
  \mprset{flushleft}
  \inferrule*[left=(Alloc),vcenter]{\z: v \Assign\ \mathtt{alloc}\ n, \ \tau, \ a}
  {
    \mathtt{IF}\{\z \in \Zl\}\{\ \underline{\mathtt{if}}\ (n \le_s 0_{\bv{32}} \lor \mathtt{overflow}_{mul}(n, \mathtt{sz}(\tau)))\ \IhaltU{}; \}
    \\\\
    \IwriteB{\VallocA{\z, n{*}\mathtt{sz}(\tau), a}};
    \\\\
    \alpha_{b} \Assign \Ichoose{\bv{32}};
    \quad
    \alpha_{e} \Assign \alpha_{b} + n{*}\mathtt{sz}(\tau)-1_{\bv{32}}; 
    \\\\
    \underline{\mathtt{if}}\ (\neg \intervalContainedInAddrSetAndAligned{\alpha_{b}}{\alpha_{e}}{\il{\Ck}{\mathtt{free}}}{a})\ \IhaltW{}; 
    \\\\
    \il{\Ck}{\z} \Assign \il{\Ck}{\z} \cup [\alpha_{b},\alpha_{e}];
    \quad
    M_{\Ck} \Assign \updM{[\alpha_{b},\alpha_{e}]}{M_{\Ck}}{\IchooseM};
    \\\\
    \ghost{\LB{\z}} \Assign \ghost{\Empty{\z}} ?\ \alpha_b : \mathtt{min}(\ghost{\LB{\z}}, \alpha_b);
    \quad
    \ghost{\LSz{\z}} \Assign n{*}\mathtt{sz}(\tau);
    \\\\
    \ghost{\UB{\z}} \Assign \ghost{\Empty{\z}} ?\ \alpha_e : \mathtt{max}(\ghost{\UB{\z}}, \alpha_e);
    \quad
    \ghost{\Empty{\z}} \Assign \mathtt{false};
    \\\\
    v \Assign \alpha_{b};
    \quad
    \BasedOn{v} \Assign \{ \z \};
    \\\\
    \IwriteB{\VallocB{\z, [\alpha_{b},\alpha_{e}], \prjM{[\alpha_{b},\alpha_{e}]}{M_{\Ck}}}};
  }
  \]
  \end{minipage}
  \begin{minipage}[t]{.4\textwidth}
  \[
  \mprset{flushleft}
  \inferrule*[left=(Op),vcenter]{\PC{\Ck}{j}: v \Assign \mathtt{op}(\vv{x})}
  {
    \underline{\mathtt{if}}\ (\mathtt{UB}_{\Ck}(\mathtt{op}, \vv{x}))\ \IhaltU{};
    \\\\
    v \Assign \mathtt{op}(\vv{x});
    \\\\
    \ldots,x,\ldots \Assign \vv{x};
    \quad
    \BasedOn{v} \Assign \BasedOnOp(\ldots,\BasedOn{x},\ldots);
  }
  \]
  \[
  \mprset{flushleft}
  \inferrule*[left=(RetV),vcenter]{\PC{\Ck}{j}: \mathtt{ret}\ \mathtt{void}}
  {
    \IwriteB{\Vret{\prjM{\il{\Ck}{\NS{}}}{M_{\Ck}}}};
    \\\\
    \IhaltN{};
  }
  \]
  \[
  \mprset{flushleft}
  \inferrule*[left=(\Ret{\Ck}),vcenter]{\PC{\Ck}{j}: \mathtt{ret}\ v}
  {
    \IwriteB{\Vret{v, \prjM{\il{\Ck}{\NS{}}}{M_{\Ck}}}};
    \\\\
    \IhaltN{};
  }
  \]
  \[
  \mprset{flushleft}
  \inferrule*[left=(AssignConst),vcenter]{\PC{\Ck}{j}: v \Assign c}
  {
    v \Assign c;
    \quad
    \BasedOn{v} \Assign \emptyset;
  }
  \]
  \[
  \mprset{flushleft}
  \inferrule*[left=(Dealloc),vcenter]{\PC{\Ck}{j}: \mathtt{dealloc}\ \z}
  {
    \il{\Ck}{\z} \Assign \emptyset;
    \quad
    \ghost{\Empty{\z}} \Assign \mathtt{true};
    \\\\
    \IwriteB{\Vdealloc{\z}};
  }
  \]
  \[
  \mprset{flushleft}
  \inferrule*[left=(VaStartPtr),vcenter]{\PC{\Ck}{j}: p \Assign \mathtt{va\_start\_ptr}}
  {
    \underline{\mathtt{if}}\ (\il{\Ck}{\yv} = \emptyset)\ \{
      \\\\ \quad
      p \Assign 0_{\bv{32}};
      \quad
      \BasedOn{p} \Assign \emptyset;
      \\\\
    \}\ \underline{\mathtt{else}}\ \{
      \\\\ \quad
      p \Assign \ghost{\LB{\yv}};
      \ \ 
      \BasedOn{p} \Assign \{ \yv \};
      \\\\
    \}
  }
  \]
  \end{minipage}%
  \[
  \mprset{flushleft}
  \inferrule*[left=(\Load{\Ck}),vcenter]{\PC{\Ck}{j}: v \Assign \mathtt{load}\ \tau, \ a, \ p}
  {
    \underline{\mathtt{if}}\ (\neg \mathtt{accessIsSafeC}_{\tau,a}(p, \il{\Ck}{\BasedOn{p}}))\ \IhaltU{};
    \\\\
    v \Assign \select_{\mathtt{sz}(\tau)}(M_{\Ck}, p);
    \\\\
    \BasedOn{v} \Assign \BasedOnM{\BasedOn{p}};
    \quad
    \ghost{\rdAcc{\Ck}} \Assign \ghost{\rdAcc{\Ck}} \cup [p]_{\mathtt{sz}(\tau)}; 
  }
  \inferrule*[left=(\Store{\Ck}),vcenter]{\PC{\Ck}{j}: \mathtt{store}\ \tau, \  a,\ v,\,p} 
  {
    \underline{\mathtt{if}}\ (\neg \mathtt{accessIsSafeC}_{\tau,a}(p, \il{\Ck}{\BasedOn{p} \setminus \Gro}))\ \IhaltU{};
    \\\\
    M_{\Ck} \Assign \store_{\mathtt{sz}(\tau)}(M_{\Ck}, p, v);
    \\\\
    \BasedOnM{\BasedOn{p}} \Assign \BasedOnM{\BasedOn{p}} \cup \BasedOn{v};
    \quad
    \ghost{\wrAcc{\Ck}} \Assign \ghost{\wrAcc{\Ck}} \cup [p]_{\mathtt{sz}(\tau)}; 
  }
  \]
  \[
  \mprset{flushleft}
  \inferrule*[left=(CallV),vcenter]{\PC{\Ck}{j}: \mathtt{call}\ \mathtt{void}\ \rho(\Tuple{\tau}\ \Tuple{x})}
  {
    \CalleeRegions \Assign \BasedOnMS{\bigcup_{x \in \Tuple{x}}\BasedOn{x} \cup G \cup \{\heap\}};
    \\\\
    \IwriteB{\Vcall{\rho,\Tuple{x}, \CalleeRegions, \prjM{\il{\Ck}{\CalleeRegions}}{M_{\Ck}}}};
    \\\\
    \HLBox{M_{\Ck} \Assign \updM{\il{\Ck}{\CalleeRegions \setminus \Gro}}{M_{\Ck}}{\IreadB{\memType}}};
    \\\\
    \BasedOnM{\CalleeRegions \setminus \Gro} \Assign \CalleeRegions;
  }
  \qquad
  \inferrule*[left=(\Call{\Ck}),vcenter]{\PC{\Ck}{j}: v \Assign \mathtt{call}\,\gamma\ \rho(\Tuple{\tau}\ \Tuple{x}) \quad \gamma \ne \mathtt{void}}
  {
    \CalleeRegions \Assign \BasedOnMS{\bigcup_{x \in \Tuple{x}}\BasedOn{x} \cup G \cup \{\heap\}};
    \\\\
    \IwriteB{\Vcall{\rho, \Tuple{x}, \CalleeRegions, \prjM{\il{\Ck}{\CalleeRegions}}{M_{\Ck}}}};
    \\\\
    \HLBox{M_{\Ck} \Assign \updM{\il{\Ck}{\CalleeRegions \setminus \Gro}}{M_{\Ck}}{\IreadB{\memType}}};
    \\\\
    \HLBox{v \Assign \IreadB{\gamma}};
    \ 
    \BasedOn{v}, \BasedOnM{\CalleeRegions \setminus \Gro} \Assign \CalleeRegions, \CalleeRegions;
  }
  \]
  \end{scriptsize}
  \caption{\label{fig:xlateRuleIR} Translation rules for converting \ourIR{} instructions to graph instructions.
}
\end{figure}

\Cref{fig:xlateRuleIR} shows the translation rules for converting \ourIR{} instructions to graph instructions.
The \TRule{\Entry{\Ck}} rule presents the initialization performed
at the entry of a procedure \Ck{}.
The address sets and memory state of \Ck{} are initialized
using reads from the outside world $\World{\Ck}$.
The address sets that are read are checked for well-formedness
with respect to C semantics, or else error $\mathscr{W}$ is triggered.
The ghost variables are also appropriately initialized.

The \TRule{Alloc} and \TRule{Dealloc} rules provide semantics for the allocation
and deallocation of local memory at allocation site \z{}
--- if $\z \in \Zl$, $n$ (the number of elements allocated) has additional
constraints for a UB-free execution.
A (de)allocation instruction generates observable traces
using the \texttt{wr} instruction at the beginning and end of each execution
of that instruction. We will later use these traces to identify
a lockstep correlation of (de)allocation events between \Ck{}
and \Ak{}, towards validating a
translation.

In \TRule{Op},
an application of \texttt{op} may trigger UB for certain
inputs, as abstracted
through the $\mathtt{UB}_{\Ck}(\mathtt{op},\vv{x})$ operation.
While there are many UBs in the C standard, we model only
the ones that we have seen getting exploited by the compiler for optimization.
These include the UB associated with a
logical or arithmetic shift operation (second operand should be bounded by a limit which is determined by the size of the first operand), address computation (no over- and under-flow), and division operation (second operand should be non-zero).
In
\TRule{\Load{\Ck}} and \TRule{\Store{\Ck}}, a UB-free execution requires
the dereferenced pointer $p$ to be non-{\tt NULL} ($\ne 0_{\bv{32}}$ in our modeling),
aligned by $a$, and have its access interval belong to the
regions
which $p$ may {\em point to}, or $p$ may be {\em based on} ($\mathsection6.5.6p8$ of the C17 standard).

To identify the regions a pointer $p$ may point to,
we define two
maps:
(1) $\BasedOnSym{}: \mathrm{{\tt Vars}} \to 2^{\Rall{}}$, so that
for a (pointer) variable $x\in{}\mathtt{Vars}$, \BasedOn{x} returns the
set of regions $x$ may point to;
and
(2) $\BasedOnSym_M{}: \Rall{} \to 2^{\Rall{}}$, so that
for a region $r\in{}\Rall{}$,
\BasedOnM{r} returns the set of regions
that some (pointer) value stored in $\prjM{\il{\Ck{}}{r}}{M_{\Ck{}}}$
may point to.
$\BasedOn{\vv{x}}$ is equivalent to
$\bigcup_{x \in \vv{x}} \BasedOn{x}$,
and $\BasedOnM{\R{}}$ is equivalent to
$\bigcup_{r \in \R{}} \BasedOnM{r}$.
Similarly, $\BasedOnM{\RN{1}} \Assign \RN{2}$ is equivalent to
`$\mathtt{for}\ r_1\ \mathtt{in}\ \RN{1}\ \{\ \BasedOnM{r_1} \Assign \RN{2};\ \}$'.
The initialization and updation of $\BasedOnSym$ and $\BasedOnSym_M$ due to
each \ourIR{} instruction can be seen in \cref{fig:xlateRuleIR}.
For an operation $\mathtt{op}$,
$\BasedOnOp{}:(2^{\Rall} \times 2^{\Rall} \ldots \times 2^{\Rall})\to{}2^{\Rall{}}$
represents the over-approximate abstract transfer
function for $v \Assign{} \mathtt{op}(\vv{x})$,
that takes as input
$(\BasedOn{x_1},\BasedOn{x_2},\ldots,\BasedOn{x_m})$
for $\vv{x} = (x_1,x_2,\ldots,x_m)$
and returns $\BasedOn{v}$.
We use $\BasedOnOp{}(\R{})=\R{}$ if {\tt op}
is bitwise complement and unary negation.
We use $\BasedOnOp{(\RN{1},\ldots,\RN{m})}=\bigcup{}_{1\leq{}j\leq{}m}{\RN{j}}$
if {\tt op}
is bitvector addition, subtraction, shift, bitwise-\{and,or\}, extraction, or concatenation.
We use $\BasedOnOp{}(\RN{1},\ldots,\RN{m})=\emptyset{}$
if {\tt op} is bitvector multiplication, division, logical, relational or any other remaining operator.

The translation of an \ourIR{} procedure-call is given by the rules
\TRule{CallV} and \TRule{\Call{\Ck}}
and involves producing non-silent observable trace events using the
\texttt{wr} instruction for the callee name/address,
arguments, and callee-accessible regions and memory state.
To model return values
and side-effects to the memory state due to a callee,
\texttt{rd} instructions are used.
A callee may access a memory region iff it is {\em transitively
reachable} from a global variable $g\in{}G$, the heap \heap{}, or
one of the arguments $x\in{}\Tuple{x}$.
The transitively reachable memory regions are over-approximately
computed through a reflexive-transitive closure
of $\BasedOnSym_M$, denoted $\BasedOnSym_M^*$.

A {\tt rd} instruction clobbers
the callee-observable state elements arbitrarily.
Thus, if a callee procedure terminates normally (i.e., without error),
{\tt wr} and {\tt rd}
instructions over-approximately model the execution of a procedure-call.
Later, our definition of refinement (\cref{sec:refnDefn})
caters to the case when a callee procedure may not terminate
or terminates with error
(i.e., a termination with error is modeled identically to non-termination).

\subsubsection{Translation of \Ak{}}
\label{sec:translationA}

\begin{figure}
  \begin{scriptsize}
  \begin{minipage}[t]{.5\textwidth}
  \[
  \mprset{flushleft}
  \inferrule*[left=(Op-esp),vcenter]{\PC{\Ak}{j}: \mathtt{esp} \Assign \mathtt{op}(\vv{x})}
  {
    \underline{\mathtt{if}}\ (\mathtt{UB}_{\Ak}(\mathtt{op}, \vv{x}))\ \IhaltU{};
    \\\\
    t \Assign \mathtt{op}(\vv{x});
    \\\\
    \underline{\mathtt{if}}\ (\mathtt{isPush}(\PC{\Ak}{j}, \mathtt{esp}, t))\ \{
      \\\\ \quad
      \underline{\mathtt{if}}\ (\neg \intervalContainedInAddrSet{t}{\mathtt{esp}-1_{\bv{32}}}{\il{\Ak}{\free} \cup (\il{\Ak}{cv} \setminus \il{\Ak}{\F})})\ \IhaltW{};
      \\\\ \quad
      \il{\Ak}{\stk} \Assign \il{\Ak}{\stk} \cup [t, \mathtt{esp}-1_{\bv{32}}];
      \\\\ \quad
      M_{\Ak} \Assign \updM{[t,\mathtt{esp}-1_{\bv{32}}]}{M_{\Ak}}{\IchooseM}; 
      \\\\
    \}\ \underline{\mathtt{else\ if}}\ (t \ne \mathtt{esp})\{ 
      \\\\ \quad
      \underline{\mathtt{if}}\ (\neg \intervalContainedInAddrSet{\mathtt{esp}}{t-1_{\bv{32}}}{\il{\Ak}{\stk}})\ \IhaltU{};
      \\\\ \quad
      \il{\Ak}{\stk} \Assign \il{\Ak}{\stk} \setminus [\mathtt{esp}, t-1_{\bv{32}}];
      \\\\
    \}
    \\\\
    \mathtt{esp} \Assign t;
    \quad
    \ghost{\spV{\PC{\Ak}{j}}} \Assign t;
  }
  \]
  \end{minipage}
  \begin{minipage}[t]{.58\textwidth}
  \[
  \mprset{flushleft}
  \inferrule*[left=(\Load{\Ak}),vcenter]{\PC{\Ak}{j}: v \Assign \mathtt{load}\ w,\, a,\, p}
  {
    \underline{\mathtt{if}}\ (\phantom{{}\lor{}}\neg \mathtt{isAlignedIntrvl}_{a}(p, w)
    \\\\
    \phantom{\underline{\mathtt{if}}\ (}{} \lor \Overlap([p]_{w}, \il{\Ak}{\free} \cup (\il{\Ak}{cv} \setminus \il{\Ak}{\STACK \cup \F})))\ \IhaltU{};
    \\\\
    v \Assign \select_{w}(M_{\Ak}, p);
    \\\\
    \ghost{\rdAcc{\Ak}} \Assign \ghost{\rdAcc{\Ak}} \cup [p]_{w};
  }
  \]
  \[
  \mprset{flushleft}
  \inferrule*[left=(\Store{\Ak}),vcenter]{\PC{\Ak}{j}: \mathtt{store}\ w,\, a,\, p,\, v}
  {
    \underline{\mathtt{if}}\ (\phantom{{}\lor{}}\neg \mathtt{isAlignedIntrvl}_{a}(p, w)
    \\\\
    \phantom{\underline{\mathtt{if}}\ (}{} \lor \Overlap([p]_{w}, \il{\Ak}{\{\free\}\cup\Gro\cup\Fro} \cup (\il{\Ak}{cv} \setminus \il{\Ak}{\STACK \cup \Frw})))
    \\\\ \quad
    \IhaltU{};
    \\\\
    M_{\Ak} \Assign \store_{w}(M_{\Ak}, p, v);
    \\\\
    \ghost{\wrAcc{\Ak}} \Assign \ghost{\wrAcc{\Ak}} \cup [p]_{w};
  }
  \]
  \end{minipage}

  \hspace*{-.4cm}
  \begin{minipage}[t]{.7\textwidth}
  \[
  \mprset{flushleft}
  \inferrule*[left=(\Entry{\Ak}),vcenter]{\PC{\Ak}{j}: \mathtt{def}\ \Ak(\Tuple{\tau})}
  {
    \HLBox{\il{\Ak}{\heap}, \il{\Ak}{cl}, \il{\Ak}{cv}, \ldots, \ii{\Ak}{\g}, \ldots, \ii{\Ak}{\y},\ldots, \il{\Ak}{\yv} \Assign \IreadB{\asetType{},\asetType{},\ldots,\asetType{}}};
    \\\\
    \ldots, \il{\Ak}{\f}, \ldots \Assign \initAddrSets{\F};
    \qquad
    \ldots, \il{\Ak}{\z}, \ldots \Assign \ldots, \emptyset, \ldots;
    \\\\
    \underline{\mathtt{if}}\ (\neg \mathtt{addrSetsAreWF}(\il{\Ak}{\heap}, \il{\Ak}{cl}, \il{\Ak}{cv}, \ldots, \ii{\Ak}{\g}, \ldots, \il{\Ak}{\f}, \ldots,\ii{\Ak}{\y},\ldots, \il{\Ak}{\yv}))
    \\\\ \quad \IhaltW{};
    \\\\
    M_{\Ak} \Assign \IchooseM{};
    \quad
    \HLBox{M_{\Ak} \Assign {\tt upd}_{\il{\Ak}{\NS{}}}(M_{\Ak}, \IreadB{\memType})};
    \\\\
    \underline{\mathtt{for}}\ r\ \mathtt{in\ }\Gro \cup \Fro\ \{\ M_{\Ak} \Assign {\tt upd}_{\ii{\Ak}{r}}(M_{\Ak}, \roMem{\Ak}{r}{\ii{\Ak}{r}});\ \}
    \\\\
    \underline{\mathtt{for}}\ x\mathtt{\ in\ }\vv{regs}\ \{\ x \Assign \Ichoose{\mathtt{T}(x)};\ \}
    \\\\
    \ghost{\spE} \Assign \il{\Ak}{\Y} \ne \emptyset\ ?\ \UBr{\il{\Ak}{\Y}} : \mathtt{esp}+3_{\bv{32}};
    \qquad
    \ghost{\stkE} \Assign \Ichoose{\bv{32}};
    \\\\
    \underline{\mathtt{if}}\ (\neg \mathtt{stkIsWF}(\mathtt{esp}, \ghost{\spE}, \ghost{\stkE}, \Tuple{\tau}, \il{\Ak}{\heap}, \il{\Ak}{cl}, \il{\Ak}{G \cup \F}, \ldots,\ii{\Ak}{\y},\ldots, \il{\Ak}{\yv}))
    \\\\ \quad \IhaltW{};
    \\\\
    \il{\Ak}{\stk} \Assign [\mathtt{esp}, \ghost{\spE}] \setminus \il{\Ak}{\Y};
    \quad
    \il{\Ak}{cs} \Assign [\ghost{\spE}+1_{\bv{32}}, \ghost{\stkE}] \setminus \il{\Ak}{cl};
    \\\\
    \ghost{\spV{entry}} \Assign \mathtt{esp};
    \quad
    \ghost{M^{cs}} \Assign \prjM{\il{\Ak}{cs}}{M_{\Ak}};
    \quad
    \ghost{\rdAcc{\Ak}}, \ghost{\wrAcc{\Ak}} \Assign \emptyset, \emptyset;
    \\\\
    \ghost{ebp}, \ghost{esi}, \ghost{edi}, \ghost{ebx}, \ghost{eip} \Assign \mathtt{ebp}, \mathtt{esi}, \mathtt{edi}, \mathtt{ebx}, \select_4(M_{\Ak}, \mathtt{esp});
    \\\\
    \underline{\mathtt{for}}\ \f\mathtt{\ in\ }\F\ \{ 
      \ghost{\Sz{\f}},\ghost{\Empty{\f}}, \ghost{\LB{\f}}, \ghost{\UB{\f}} \Assign |\il{\Ak}{\f}|, |\il{\Ak}{\f}| = 0_{\bv{32}}, \LBr{\il{\Ak}{\f}}, \UBr{\il{\Ak}{\f}};\ 
    \}
  }
\]
  \end{minipage}
  \begin{minipage}[t]{.35\textwidth}
  \[
  \mprset{flushleft}
  \inferrule*[left=(Op-Nesp),vcenter]{\PC{\Ak}{j}: r \Assign \mathtt{op}(\vv{x}) \\ r \ne \mathtt{esp}}
  {
    \underline{\mathtt{if}}\ (\mathtt{UB}_{\Ak}(\mathtt{op}, \vv{x}))\ \IhaltU{};
    \\\\
    r \Assign \mathtt{op}(\vv{x});
  }
  \]
  \\
    \[
  \mprset{flushleft}
  \inferrule*[left=(\Ret{\Ak}),vcenter]{\PC{\Ak}{j}: \mathtt{ret}\ \tau}
  {
    \underline{\mathtt{if}}\ (\phantom{{}\lor{}} \ghost{\spV{entry}} \ne \mathtt{esp}
    \\\\
    \phantom{\underline{\mathtt{if}}\ (}\lor{} \ghost{ebp} \ne \mathtt{ebp} \lor \ghost{esi} \ne \mathtt{esi}
    \\\\
    \phantom{\underline{\mathtt{if}}\ (}\lor{} \ghost{edi} \ne \mathtt{edi} \lor \ghost{ebx} \ne \mathtt{ebx}
    \\\\
    \phantom{\underline{\mathtt{if}}\ (}\lor{} \ghost{eip} \ne \select_4(M_{\Ak}, \mathtt{esp})
    \\\\
  \phantom{\underline{\mathtt{if}}\ (}\lor{} \neg (\prjMEq{\il{\Ak}{cs}}{\ghost{M^{cs}}}{M_{\Ak}}))\ 
    \IhaltU{};
    \\\\
    \mathtt{IF}\{\tau = \mathtt{void}\}\{\ \IwriteB{\prjM{\il{\Ak}{\NS{}}}{M_{\Ak}}};\ \}
    \\\\
    \mathtt{ELSE}\{\
    \\\\ \ 
    \IwriteB{\Vret{\VretVal{\tau}, \prjM{\il{\Ak}{\NS{}}}{M_{\Ak}}}};
    \\\\
    \}
    \\\\
    \IhaltN{};
  }
  \]
  \end{minipage}%
  \end{scriptsize}
  \caption{\label{fig:xlateRuleAsm} Translation rules for converting pseudo-assembly instructions to graph instructions.
}
\end{figure}

The translation rules for converting assembly instructions to graph instructions are shown in \cref{fig:xlateRuleAsm}.
The assembly opcodes are abstracted to an IR-like syntax for ease of exposition. 
For example, in \TRule{\Load{\Ak}}, a memory read operation is represented by
a \texttt{load} instruction which is annotated with address $p$, access size $w$ (in bytes), and required alignment $a$ (in bytes).
Similarly, in \TRule{\Store{\Ak}}, a memory write operation is represented by
a \texttt{store} instruction with similar operands.
Both \TRule{\Load{\Ak}} and \TRule{\Store{\Ak}} translations update the ghost
address sets \ghost{\rdAcc{\Ak}} and \ghost{\wrAcc{\Ak}}, in the same manner as done in \Ck{}.
Exceptions like division-by-zero 
are modeled as UB in \progA{} through $\mathtt{UB}_{\Ak}$
(rules \TRule{Op-esp} and \TRule{Op-Nesp})

\TRule{Op-esp} shows the translation of an instruction
that updates the stackpointer.
Assignment to the stackpointer register \texttt{esp} may indicate allocation (push)
or deallocation (pop) of stack space.
A stackpointer assignment which corresponds to a stackpointer decrement (push) is
identified through predicate $\mathtt{isPush}(\PC{\Ak}{j}, \iota{}_{b}, \iota_{a})$
where $\iota{}_{b}$ and $\iota{}_{a}$
are the values of \texttt{esp} {\em b}efore and {\em a}fter the execution of the instruction.
We use
$\mathrm{\tt isPush}(\PC{\Ak}{j}, \iota{}_{b}, \iota_{a})\Leftrightarrow (\iota{}_{b} >_u \iota_{a})$.
While this choice of {\tt isPush} suffices for most TV settings, we show in \cref{app:isPush}
that if the translation
is performed by an adversarial compiler,
discriminating a stack push from a pop is trickier and may require external trusted guidance
from the user.
For a stackpointer decrement, a failure to allocate stack space,
either due to wraparound or overlap with other allocated space,
triggers $\mathscr{W}$, i.e., we expect the environment (e.g., OS)
to ensure that the required stack space is available to \Ak{}.
For a stackpointer increment, it
is a translation error if the stackpointer moves out of stack frame bounds (captured
by error code \EU{}).
The stackpointer value at the end of an assignment instruction at PC $\PC{\Ak}{j}$ is saved
in a ghost variable named \ghost{\spV{\PC{\Ak}{j}}}.
These ghost variables help with inference of invariants that relate
a local variable's address with stack addresses
(discussed in \cref{sec:invInferFormal}).
During push,
the initial contents of the newly allocated stack region are chosen non-deterministically using $\theta$
--- this admits the possibility of arbitrary clobbering of
the unallocated stack region below the stackpointer
due to asynchronous external interrupts, before it is allocated again.

\TRule{\Entry{\Ak}} shows the initialization of state elements of procedure \Ak{}.
For region $r\in{}\NS{}$,
the initialization of $\il{\Ak}{r}$ and
\prjM{\il{\Ak}{r}}{M_{\ddAk}}
is similar to \TRule{\Entry{\Ck{}}}.
The address sets of all assembly-only regions $f\in{}\F{}$
are initialized using \progA{}'s symbol table (\initAddrSets{\F}).
The memory contents of a read-only global variable $r\in{}\Gro{}\cup{}\Fro{}$ are initialized
using \roMem{\Ak}{r}{\ii{\Ak}{r}} (defined in \cref{tab:preds}). 
The machine registers
are initialized with arbitrary contents ($\theta$) --- the
constraints on the {\tt esp} register are
checked later, and \EW{} is generated if a constraint
is violated.
The x86 stack of an assembly procedure
includes the stack frame \il{\Ak}{\stk}
of the currently executing procedure \Ak{},
the parameters \il{\Ak}{\Y} of \Ak{}, and
the remaining space which includes caller-stack \il{\Ak}{cs}
and, possibly, the locals $\il{\Ak}{cl}$ defined in the call chain of \Ak{}.
Ghost variable $\ghost{\spV{entry}}$ holds the {\tt esp} value at entry of \Ak{}.
$\ghost{\spE}$ represents the largest address in \il{\Ak}{\Y\cup{}\{\stk{}\}}
so that at entry, $\il{\Ak}{\stk}=[\ghost{\spV{entry}},\ghost{\spE}]\setminus{}\il{\Ak}{\Y}$.
If there are no parameters, $\ghost{\spE}=\mathtt{esp}+3_{\bv{32}}$
represents the end of the region that holds the return address.
Ghost variable $\ghost{\stkE}$ holds the largest address
in $\il{\Ak}{\{\stk,cs,cl\}\cup{}\Y{}}$.
At entry, due to the calling conventions, we assume (through \texttt{stkIsWF()}) that:
(1) the parameters are laid out at addresses above the stackpointer as per
calling conventions ({\tt obeyCC});
(2) the value $\mathtt{esp}+4_{\bv{32}}$ is 16-byte aligned;
and (3) the caller stack is above \Ak{}'s stack frame \il{\Ak}{\stk}.
A violation of these conditions trigger $\mathscr{W}$.
Notice that unlike region $r \in \NS$, region $cv$ may potentially overlap with
assembly-only regions $\F \cup \STACK$.
Thus, while an address $\alpha \in \il{\Ck}{cv}$ is inaccessible in \Ck{},
it is potentially accessible in \Ak{} if $\alpha \in \F \cup \STACK$.

Upon return (rule \TRule{\Ret{\Ak}}), we require that the callee-save registers,
caller stack, and the return address remain preserved ---
a violation of these conditions trigger \EU{}.
For simplicity, we only tackle
scalar return values, and ignore
aggregate return values that need to be passed in memory.

\subsection{Observable Traces and Refinement Definition}
\label{sec:refnDefn}

Recall that a procedure
execution yields an observable
trace containing silent and non-silent events.
The
error code of a trace $T$, written $\eT{T}$, is either $\emptyset{}$ (indicating
either non-termination or error-free termination), or
one of
$\mathscr{r}\in{}\{\mathscr{U}, \mathscr{W}\}$ (indicating
termination with error code $\mathscr{r}$).
The
non-error part of a trace $T$, written $\neT{T}$, is $T$
when $\eT{T} = \emptyset$, and
$T'$ such that $T = T' \cdot \eT{T}$ otherwise.

\begin{definition}
$P \downarrow_{\World{}} T$ denotes the condition that for
an initial outside world $\World{}$, the
execution of a procedure $P$ may produce an observable trace
$T$ (for some sequence of non-deterministic choices).
\end{definition}

\begin{definition}
Traces $T$ and $T'$ are {\em stuttering equivalent}, written $\steq{T}{T'}$,
iff they differ only by finite sequences of silent events $\bot$.
A trace $T$ is a {\em stuttering prefix} of trace $T'$, written $\stprefix{T}{T'}$, iff
$(\steq{T'}{T}) \vee (\exists{T^{\tt r}}:{\steq{T'}{(T\cdot{}T^{\tt r})}})$.
\end{definition}

\begin{definition}
$U^{\World{},\Trace{\Ak}}_{\tt pre}(\Ck)$ denotes the condition: $\exists{\Trace{\Ck}}:{(\Ck\downarrow_{\World{}} \Trace{\Ck} \cdot \mathscr{U}) \land (\stprefix{\Trace{\Ck}}{\Trace{\Ak}})}$.
\end{definition}

\begin{definition}
$W^{\World{},\Trace{\Ak}}_{\tt pre}(\Ck)$ denotes the condition:
$(\eT{\Trace{\Ak}} = \mathscr{W}) \land (\exists{\Trace{\Ck}}: {(\Ck\downarrow_{\World{}} \Trace{\Ck}) \land (\stprefix{\neT{\Trace{\Ak}}}{\Trace{\Ck})})}$
\end{definition}

\begin{definition}
\label{defn:ref1}
$\Ck \refines{} \Ak$, read
\Ak{} refines \Ck{} (or \Ck{} is refined by \Ak{}), iff:
$$
\forall{\World{}}:
{(\execT{\Ak}{\World{}}{\Trace{\Ak}}) \Rightarrow (W^{\World{},\Trace{\Ak}}_{\tt pre}(\Ck)  \vee{} U^{\World{},\Trace{\Ak}}_{\tt pre}(\Ck) \vee{} (\exists{\Trace{\Ck}}:(\execT{\Ck}{\World{}}{\Trace{\Ck}}) \land (\steq{\Trace{\Ak}}{\Trace{\Ck}})))}
$$
\end{definition}
The $W^{\World{},\Trace{\Ak}}_{\tt pre}(\Ck)$ and $U^{\World{},\Trace{\Ak}}_{\tt pre}(\Ck)$ conditions
cater to the cases where \Ak{} triggers $\mathscr{W}$ and \Ck{} triggers $\mathscr{U}$ respectively;
the constituent $\stprefixSym$ conditions ensure that a procedure
call in \Ak{} has identical termination behaviour to a procedure-call in \Ck{} before
an error is triggered.
If neither \Ak{} triggers \EW{} nor \Ck{} triggers \EU{},
$\steq{T_{\Ak{}}}{T_{\Ck{}}}$ ensures that \Ak{} and \Ck{} produce identical
non-silent events at similar speeds.
In the absence of local variables and procedure-calls in \Ck{},
$\Ck \refines{} \Ak$ implies a correct translation from \Ck{} to \Ak{}.

\subsubsection{Refinement Definition in the Presence of Local Variables and Procedure-Calls When All Local Variables Are Allocated on the Stack in \Ak{}}

For each local variable (de)allocation and for each procedure-call,
our execution
semantics generate a {\tt wr} trace event in \Ck{} (\cref{fig:xlateRuleIR}).
Thus, to reason about refinement,
we require correlated and equivalent trace events to be generated in \Ak{}.
For this, we annotate \Ak{} with two types of annotations to obtain \dAk{}:
\begin{enumerate}
\item $\mathtt{alloc}_s$ and $\mathtt{dealloc}_s$ instructions are added to explicitly indicate the (de)allocation of a local variable $\z \in \Z$, e.g., a stack region may be marked as belonging to \z{} through these annotations.
\item A procedure-call, direct or indirect, is annotated with the types and addresses of the arguments and the set of memory
regions observable by the callee.
\end{enumerate} 
These annotations are intended to encode the correlations
with the corresponding allocation, deallocation, and procedure-call events in the
source procedure \Ck{}. For now, we assume that the locations and values of these
annotations in \dAk{} are coming from an oracle --- later in \cref{sec:algo},
we present an
algorithm to identify these annotations automatically in a best-effort
manner.

\begin{figure}
  \begin{scriptsize}
  \begin{minipage}[t]{.5\textwidth}
  \[
  \mprset{flushleft}
  \inferrule*[left=(AllocS),vcenter]{\PC{\dAk}{j}: \mathtt{alloc}_s\ e_v,\, e_w,\, a,\, \z} 
  {
    \IwriteB{\VallocA{\z, e_w, a}};
    \\\\
    v, w \Assign e_v, e_w;
    \\\\
    \underline{\mathtt{if}}\ (\neg \intervalContainedInAddrSetAndAligned{v}{v+w-1_{\bv{32}}}{\il{\dAk}{\stk}}{a})\ \IhaltU{};
    \\\\
    \underline{\mathtt{if}}\ (\Overlap([v]_{w}, \il{\dAk}{cv}))\ \IhaltW{};
    \\\\
    \il{\dAk}{\stk}, \il{\dAk}{\z} \Assign \il{\dAk}{\stk} \setminus [v]_{w}, \il{\dAk}{\z} \cup [v]_{w};
    \\\\
    \IwriteB{\VallocB{\z, [v]_{w}, \prjM{[v]_{w}}{M_{\dAk}}}};
  }
  \]
  \[
  \mprset{flushleft}
  \inferrule*[left=(DeallocS),vcenter]{\PC{\dAk}{j}: \mathtt{dealloc}_s\ \z}
  {
    \il{\dAk}{\z}, \il{\dAk}{\stk} \Assign \emptyset, \il{\dAk}{\stk} \cup \il{\dAk}{\z};
    \\\\
    \IwriteB{\Vdealloc{\z}};
  }
  \]
  \end{minipage}
  \begin{minipage}[t]{.5\textwidth}
  \[
  \mprset{flushleft}
  \inferrule*[left=(\Call{\dAk}),vcenter]{\PC{\dAk}{j}: \mathtt{call}\ \gamma\ \rho(\Tuple{\tau} \, \Tuple{x})\ \CalleeRegions}
  {
    \underline{\mathtt{if}}\ (\neg\mathtt{aligned}_{16}(\mathtt{esp}) \lor \neg\mathtt{obeyCC}(\mathtt{esp}, \Tuple{\tau}, \Tuple{x}))
    \\\\
    \quad \IhaltU{};
    \\\\
    \IwriteB{\Vcall{\rho, \Tuple{x}, \CalleeRegions, \prjM{\il{\dAk}{\CalleeRegions}}{M_{\dAk}}}};
    \\\\
    \HLBox{M_{\dAk} \Assign \updM{\il{\dAk}{\CalleeRegions \setminus \Gro}}{M_{\dAk}}{\IreadB{\memType}}};
    \\\\
    \mathtt{ecx} \Assign \Ichoose{\bv{32}};
    \\\\
    \mathtt{IF}\{\gamma = \mathtt{void}\}\{\ \mathtt{eax}, \mathtt{edx} \Assign \Ichoose{\bv{32},\bv{32}};\ \}
    \\\\
    \mathtt{ELSE}\{\ \HLBox{\mathtt{eax}, \mathtt{edx} \Assign \VgetVal{\gamma}{\IreadB{\gamma}}};\ \}
    \\\\
  }
  \]
  \end{minipage}%
  \end{scriptsize}
  \caption{\label{fig:xlateRuleAsmStackLocals}
    Additional translation rules for converting pseudo-assembly instructions to graph instructions for procedures with only stack-allocated locals.
}
\end{figure}

\Cref{fig:xlateRuleAsmStackLocals} presents three new instructions in \dAk{} ---
$\mathtt{alloc}_s$, $\mathtt{dealloc}_s$, and \texttt{call} ---
and their translations to graph instructions.

An instruction `$\PC{\dAk}{j}: \mathtt{alloc}_s\ e_v,\,e_w,\,a,\,\z$'
represents the stack allocation of a local variable identified by
allocation site \z{}.
$e_v$ is the expression for start address,
$e_w$ is the expression for allocation size, and
$a$ is the required alignment of the start address.
During stack allocation of
a local variable \TRule{AllocS}, the allocated address
must satisfy the required alignment and separation constraints, or
else $\mathscr{U}$ is triggered.
The allocated interval must be separate from region $cv$
\footnote{Recall that $cv$ may potentially overlap with \stk{}
unlike a region $r \in \NS{}$.},
otherwise \EW{} is triggered;
we explain the rationale for triggering \EW{}
in this case in next section when we discuss virtual allocation.
An allocation removes an address interval from
\il{\dAk}{\stk}
and adds it to
\il{\dAk}{\z}.

A `$\PC{\dAk}{j}: \mathtt{dealloc}_s\ \z{}$' instruction represents the deallocation of \z{}
and empties the address set $\il{\dAk}{\z{}}$, adding the removed
addresses to
$\il{\dAk}{\stk}$ \TRule{DeallocS}.

For procedure-calls \TRule{\Call{\dAk}},
we annotate the {\tt call} instruction in assembly to explicitly specify
the start addresses of the address regions belonging to the
arguments (shown as \Tuple{x} in \cref{fig:xlateRuleAsmStackLocals}).
The address region of an argument should have previously been demarcated using
an $\mathtt{alloc}_s$ instruction. Additionally, these address regions
should satisfy the constraints imposed by the calling conventions ({\tt obeyCC}).
The calling
conventions also require the {\tt esp} value to be 16-byte aligned.
A procedure-call is recorded as an observable event, along with the
observation of the
callee name (or address), the addresses of the arguments, callee-observable regions
and their memory contents.
The returned values, modeled through \Iread{\memType} and
\Iread{\gamma},
include the contents of the callee-observable memory regions and the scalar
values returned by the callee (in registers {\tt eax}, {\tt edx}).
The callee additionally clobbers the caller-save registers using $\theta$.

\begin{definition}[Refinement in the presence of only stack-allocated locals]
\label{defn:Ref2}
$\Ck \gtrdot{} \Ak$ iff: $\exists{\dAk}:\Ck \refines{} \dAk$
\end{definition}
$\Ck \gtrdot{} \Ak$ encodes the property that it is possible to annotate \Ak{}
to obtain \dAk{}
so that the local variable (de)allocation and
procedure-call events of \Ck{} and
the annotated
\dAk{} can be
correlated in lockstep.
In the presence of stack-allocated local variables and procedure-calls,
$\Ck \gtrdot{} \Ak$ implies a correct translation from \Ck{} to \Ak{}.
In the absence of local variables and procedure
calls, $\Ck \gtrdot{} \Ak$ reduces to $\Ck \refines{} \Ak$ with $\dAk=\Ak$.

\subsubsection{Capabilities and Limitations of $\Ck \gtrdot{} \Ak$}
$\Ck \gtrdot{} \Ak$ requires that for allocations and procedure calls
that reuse the same stack space, their relative order
remains preserved.
This requirement is sound but may be too strict for
certain (arguably rare) compiler transformations that may reorder the (de)allocation
instructions that reuse the same stack space.
Our refinement definition admits intermittent register-allocation of
(parts of) a local variable.

$\Ck \gtrdot{} \Ak$
supports {\em merging} of multiple allocations
into a single stackpointer
decrement instruction. Let \PC{\Ak}{s} be the PC of a single stackpointer
decrement instruction that implements multiple allocations.
Merging can be encoded by
adding multiple {\tt alloc}$_s$
instructions to \Ak{}, in the same order as
they appear in \Ck{}, to obtain \dAk{}, so that
these {\tt alloc}$_s$ instructions execute
only after \PC{\Ak}{s} executes; similarly,
the corresponding {\tt dealloc}$_s$
instructions must execute before a
stackpointer increment instruction deallocates
this stack space.

CompCert's preallocation
is a special case of merging where stack space for
all local variables is allocated in the assembly procedure's prologue and
deallocated in the epilogue (with no reuse
of stack space).
In this case, our approach
annotates \Ak{} with {\tt (de)alloc}$_s$
instructions, potentially in the middle of the
procedure body, such that they execute in lockstep with
the (de)allocations in \Ck{}.

A compiler may {\em reallocate} stack space by reusing the same
space for two or more local variables with non-overlapping
lifetimes (potentially without an intervening stackpointer increment
instruction). If the relative order
of (de)allocations is preserved,
reallocation can be encoded by annotating \dAk{}
with a {\tt dealloc}$_s$
instruction (for deallocating the first variable) immediately followed by an
{\tt alloc}$_s$ instruction, such that the allocated region potentially
overlaps with the previously deallocated region.
Our refinement
definition may not be able to cater to a translation that changes the relative
order of (de)allocation instructions during reallocation.

Because our execution model observes each (de)allocation event (due
to the {\tt wr} instruction), a successful
refinement check ensures that the allocation states of \dAk{}
and \Ck{} are identical at every correlated callsite. An inductive
argument over \progC{} and \progA{} is thus
used to show that
the address set for region identifier
$cl$ is identical
at the beginning of each correlated pair of procedures \Ck{} and \Ak{} (as modeled
through identical reads from the outside world
in \TRule{\Entry{P}} ($P \in \{ \Ck, \Ak \}$) of \cref{fig:xlateRuleIR,fig:xlateRuleAsm}).

\subsubsection{Refinement Definition in the Presence of Potentially Register-Allocated or Eliminated Local Variables in \Ak{}}
\label{sec:refnDefnVirtual}

\begin{figure}
  \begin{scriptsize}
  \begin{minipage}[t]{.50\textwidth}
  \[
  \mprset{flushleft}
  \inferrule*[left=(AllocV),vcenter]{\PC{\ddAk}{j}: v \Assign \mathtt{alloc}_v\ e_w,\, a,\,\zl}
  {
    \IwriteB{\VallocA{\zl, e_w, a}};
    \\\\
    v, w \Assign \Ichoose{\bv{32}}, e_w;
    \\\\
    \underline{\mathtt{if}}\ (\neg \intervalContainedInAddrSetAndAligned{v}{v+w-1_{\bv{32}}}{\mathtt{comp}(\il{\ddAk}{\NS \cup \{ cv \}})}{a}) 
    \\\\ \quad
    \IhaltW{};
    \\\\
    \ilzv{\ddAk}{\zl} \Assign \ilzv{\ddAk}{\zl} \cup [v]_{w};
    \\\\
    \IwriteB{\VallocB{\zl, [v]_{w}, \prjM{[v]_{w}}{M_{\ddAk}}}};
  }
  \]
  \[
  \mprset{flushleft}
  \inferrule*[left=(AllocS'),vcenter]{\PC{\ddAk}{j}: \mathtt{alloc}_s\ e_v,\,e_w,\,a,\,\z}
  {
    \ldots \\\\
    \underline{\mathtt{if}}\ (\Overlap([v]_{w}, \il{\ddAk}{cv} \HighlightMath{\cup \ilZv{\ddAk}}))\ \IhaltW{};
    \\\\
    \hcancel{\il{\dAk}{\stk}, \il{\dAk}{\z{}} \Assign \il{\dAk}{\stk} \setminus [v]_{w}, \il{\dAk}{\z{}} \cup [v]_{w};}
    \\\\
    \HighlightMath{\mathtt{IF}\{\z \in \Zl\}\ \{\ \il{\ddAk}{\stk}, \ilzs{\ddAk}{\z} \Assign \il{\ddAk}{\stk} \setminus [v]_{w}, \ilzs{\ddAk}{\z} \cup [v]_{w}; \}}
    \\\\
    \HighlightMath{\mathtt{ELSE}\ \{ \ \il{\ddAk}{\stk}, \il{\ddAk}{\z} \Assign \il{\ddAk}{\stk} \setminus [v]_{w}, \il{\ddAk}{\z} \cup [v]_{w}; \}}
    \\\\
    \ldots
  }
  \]
  \[
  \mprset{flushleft}
  \inferrule*[left=(DeallocS'),vcenter]{\PC{\ddAk}{j}: \mathtt{dealloc}_s\ \z{}}
  {
    \hcancel{\il{\dAk}{\z}, \il{\dAk}{\stk} \Assign \emptyset, \il{\dAk}{\stk} \cup \il{\dAk}{\z};}
    \\\\
    \HighlightMath{\mathtt{IF}\{\z \in \Zl\}\ \{}
    \\\\ \quad
    \HighlightMath{\underline{\mathtt{if}}\ (\ilzv{\ddAk}{\z} \ne \emptyset) \ \IhaltU{};}
    \\\\ \quad
  \HighlightMath{\ilzs{\ddAk}{\z},\ \il{\ddAk}{\stk} \Assign \emptyset,\ \il{\ddAk}{\stk} \cup \ilzs{\ddAk}{\z};\ \}}
    \\\\
    \HighlightMath{\mathtt{ELSE}\ \{\ \il{\ddAk}{\z},\ \il{\ddAk}{\stk} \Assign \emptyset,\ \il{\ddAk}{\stk} \cup \il{\ddAk}{\z};\ \}}
    \\\\
    \IwriteB{\Vdealloc{\z{}}};
  }
  \]
  \end{minipage}
  \begin{minipage}[t]{.45\textwidth}%
  \[
  \mprset{flushleft}
  \inferrule*[left=(DeallocV),vcenter]{\PC{\ddAk}{j}: \mathtt{dealloc}_v\ \zl}
  {
    \underline{\mathtt{if}}\ (\ilzs{\ddAk}{\zl} \ne \emptyset) \ \IhaltU{};
    \\\\
    \ilzv{\ddAk}{\zl} \Assign \emptyset;
    \\\\
    \IwriteB{\Vdealloc{\zl}};
  }
  \]
  \[
  \mprset{flushleft}
  \inferrule*[left=(Op-esp'),vcenter]{\PC{\ddAk}{j}: \mathtt{esp} \Assign \mathtt{op}(\vv{x})}
  {
    \ldots \\\\
    \intervalContainedInAddrSet{t}{\mathtt{esp}-1_{\bv{32}}}{\il{\ddAk}{\free} \cup ((\il{\ddAk}{cv}\HighlightMath{\cup \ilZv{\ddAk}}) \setminus \il{\ddAk}{\F})}
    \\\\
    \ldots
  }
  \]
  \[
  \mprset{flushleft}
  \inferrule*[left=(\Entry{\ddAk}),vcenter]{\PC{\ddAk}{j}: \mathtt{def}\ \ddAk(\Tuple{\tau})}
  {
    \ldots
    \\\\
    (\text{same as \cref{fig:xlateRuleAsm}})
    \\\\
    \ldots
    \\\\
    \hcancel{\ldots, \il{\ddAk}{\z}, \ldots \Assign \ldots, \emptyset, \ldots;}
    \\\\
    \HighlightMath{\ldots, \il{\ddAk}{\za}, \ldots \Assign \ldots, \emptyset, \ldots;}
    \\\\
    \HighlightMath{\underline{\mathtt{for}}\ \z\mathtt{\ in\ }\Zl\ \{\
        \ilzs{\ddAk}{\z}, \ilzv{\ddAk}{\z} \Assign \emptyset, \emptyset;
    \}}
  }
  \]
  \[
  \mprset{flushleft}
  \inferrule*[left=(\Load{\ddAk}),vcenter]{\PC{\ddAk}{j}: v \Assign \mathtt{load}\ w\ a\ p}
  {
    \ldots \\\\
    \Overlap([p]_{w}, \il{\ddAk}{\free} \cup ((\il{\ddAk}{cv} \HighlightMath{\cup (\ilZv{\ddAk})}) \setminus \il{\ddAk}{\F \cup \STACK{}}))
    \\\\
    \ldots
  }
  \]
  \[
  \mprset{flushleft}
  \inferrule*[left=(\Store{\ddAk}),vcenter]{\PC{\ddAk}{j}: \mathtt{store}\ w\ a\ p\ v}
  {
    \ldots \\\\
    \Overlap([p]_{w}, \il{\ddAk}{\{\free\}\cup\Gro\cup\Fro} \cup ((\il{\ddAk}{cv}\HighlightMath{\cup (\ilZv{\ddAk})}) \setminus \il{\ddAk}{\Frw \cup \STACK}))
    \\\\
    \ldots
  }
  \]
  \end{minipage}
  \end{scriptsize}
  \caption{\label{fig:xlateRuleAsmAllLocals}
    Additional and revised translation rules for converting pseudo-assembly instructions to graph instructions for procedures with both stack and register allocated (or eliminated) locals.
}
\end{figure}

If a local variable $\zl{}\in{}\Zl{}$ is either register-allocated or eliminated in \Ak{}, there exists
no stack region in \Ak{} that can be associated with \zl. However, recall that our execution
model observes each allocation event in \Ck{} through the {\tt wr} instruction.
Thus, for a successful refinement
check, a correlated allocation event still needs to be annotated in \Ak{}.
We pretend that a correlated allocation occurs in \Ak{} by
introducing
the notion of a {\em virtual allocation} instruction, called {\tt alloc}$_v$, in \Ak{}.
\Cref{fig:xlateRuleAsmAllLocals} shows the virtual (de)allocation
instructions,
{\tt alloc}$_v$ and
{\tt dealloc}$_v$, and the revised translations
of
procedure-entry and
{\tt alloc}$_s$, {\tt dealloc}$_s$, {\tt load}, {\tt store}, and {\tt esp}-modifying instructions.
Instead of reproducing the full translations, we only show the changes with appropriate
context.
The additions have a $\HighlightMath{\text{highlighted}}$ background and
deletions are \hcancel{striked out}.
We update and annotate \Ak{} with the translations and instructions in
\cref{fig:xlateRuleAsmStackLocals,fig:xlateRuleAsmAllLocals}
to obtain \ddAk{}.

A `$\PC{\ddAk}{j}: {\tt v} \Assign{} \mathtt{alloc}_v\ e_w,\,a,\,\z$' instruction
non-deterministically chooses the start address (using $\theta(\bv{32})$)
of a local variable \z{}
of size $e_w$ and alignment $a$, performs a virtual allocation,
and returns the start address in {\tt v}
(\TRule{AllocV} in \cref{fig:xlateRuleAsmAllLocals}
shows the graph translation).
The chosen start address is {\em assumed} to satisfy
the desired WF constraints, such as separation (non-overlap) and
alignment; error $\mathscr{W}$ is triggered otherwise. Notice that this
is in contrast to $\mathtt{alloc}_s$ where error $\mathscr{U}$ is
triggered on WF violation to indicate that it is the compiler's responsibility
to ensure the satisfaction of WF constraints.
Unlike a stack allocation where the
compiler chooses the allocated region (and the validator identifies it through
an {\tt alloc$_s$} annotation), a virtual allocation is only a
validation construct (the compiler is not involved) that is used only to enforce a
lockstep correlation of allocation events. By triggering $\mathscr{W}$ on a failure
during a virtual allocation, we effectively assume that allocation through {\tt alloc$_v$}
satisfies
the required WF conditions.

For simplicity, we support virtual allocations only for a variable declaration
$\zl{}\in{}\Zl{}$.
Thus, we expect a call to {\tt alloca()} at $\za{}\in{}\Za{}$
to always be stack-allocated in \ddAk{}.
In \ddAk{}, we replace the single variable \il{\ddAk}{\zl} with
two variables $\ilzs{\ddAk}{\zl}$ and $\ilzv{\ddAk}{\zl}$ that represent the
address sets corresponding to the stack
and
virtual-allocations due to allocation-site \zl{} respectively.
We compute $\il{\ddAk}{\zl}=\ilzs{\ddAk}{\zl} \cup \ilzv{\ddAk}{\zl}$ (but
we do not maintain a separate variable \il{\ddAk}{\zl}).
We also assume that a single variable declaration \zl{} in \Ck{} may
either correlate with only stack-allocations (through {\tt alloc}$_s$)
or only
virtual-allocations (through {\tt alloc}$_v$) in \ddAk{}\footnote{For simplicity,
  we do not tackle path-specializing transformations that may require, for
  a single variable declaration \zl, a stack-allocation on one
  assembly path and a virtual-allocation
  on another. Such transformations are arguably rare.},
i.e.,
$\ilzs{\ddAk}{\zl} \cap \ilzv{\ddAk}{\zl} = \emptyset$ holds at all times.
For convenience, we define $\ilZv{\ddAk}=\bigcup_{\zl \in \Zl}(\ilzv{\ddAk}{\zl})$.

Importantly, a virtual allocation must be
separate from other \Ck{} allocated regions ($\NS \cup \{ cv \}$) but
may overlap with assembly-only regions $(\F\cup{}\STACK{})$.
Thus,
in the revised semantics of \TRule{Op-esp'},
a stack push is allowed to overstep a virtually-allocated region.

An instruction `$\PC{\ddAk}{j}:\ \mathtt{dealloc}_v\ \zl$' in \ddAk{}
empties the address set \ilzv{\ddAk}{\zl} and
produces an observable event through \texttt{wr} instruction
(\TRule{DeallocV} in \cref{fig:xlateRuleAsmAllLocals}
shows the graph translation).
An execution of $\mathtt{dealloc}_v$ where
$\ilzs{\ddAk}{\zl}$ is non-empty
triggers error \EU{}, i.e.,
we require an error-free
execution of $\mathtt{dealloc}_v$ to ``empty''
the address set \il{\ddAk}{\zl} (defined as
$\il{\ddAk}{\zl} = \ilzs{\ddAk}{\zl} \cup \ilzv{\ddAk}{\zl}$).
Thus, we ensure the emptiness of \il{\ddAk}{\zl} before
producing the observable trace for deallocation of \zl{}
(similar to \hyperref[fig:xlateRuleIR]{\texttt{dealloc}} in \Ck{}).

The revised semantics of the {\tt alloc$_s$} instruction
\TRule{AllocS'} assume that stack-allocated local memory is separate from
virtually-allocated regions
(\ilzv{\ddAk}{\Zl}).
The revised semantics of memory access instructions
(\TRule{\Load{\ddAk}}
and \TRule{\Store{\ddAk}})
enforce that a virtually-allocated region must never be accessed in \ddAk{},
unless it also
happens to belong to the assembly-only regions $(\F \cup \STACK{})$.

Similarly to $\mathtt{dealloc}_v$,
in the revised semantics \TRule{DeallocS'},
$\mathtt{dealloc}_s$ triggers
\EU{} if \ilzv{\ddAk}{\zl} ($\zl \in \Zl$) is non-empty,
ensuring the execution of $\mathtt{dealloc}_s$ empties \il{\ddAk}{\zl}
($= \ilzs{\ddAk}{\zl} \cup \ilzv{\ddAk}{\zl}$).
Effectively, a
lockstep
correlation of virtual allocations in \ddAk{}
with allocations
in \Ck{} ensures
that the allocation states of both procedures
always agree for regions $r\in{}\NS{} \cup \{ cv \}$.

The purpose of the $cv$ or callers' virtual region should be clear now:
$cv$ or callers's virtual region of an assembly procedure \ddAk{}
is the set of virtually-allocated addresses in \ddAk{}'s call chain.
At a procedure-call, the address set \il{\ddAk}{cv} for a
callee is computed as $\il{\ddAk}{cv} \cup \ilzv{\ddAk}{\Zl}$.
The lockstep correlation of allocation states
(due to observation of (de)allocation) enables us to
define \il{\Ck}{cv} for a callee in \Ck{} using \il{\ddAk}{cv}.
As a virtual allocation is supposed to correspond to a
register-allocated or an eliminated local, region
$cv$ is assumed to be inaccessible in the callee%
\footnote{For a caller local to be accessible in a callee, it should have
its address taken.
An address-taken local cannot be register-allocated or eliminated.}.
This is sound because the set of observable regions for a callee constitute
an observable in the caller
and the equality of observables is required for establishing refinement.

\begin{definition}[Refinement with stack and virtually-allocated locals]
\label{defn:Ref3}
$\Ck \Supset{} \Ak$ iff: $\exists{\ddAk}:\Ck \refines{} \ddAk$
\end{definition}

Recall that
$\Ck \refines{} \ddAk$ requires that {\em for all} non-deterministic choices of
a virtually allocated local variable address in \ddAk{} ($v$ in \TRule{AllocV}),
there {\em exists} a non-deterministic choice
for the correlated local variable address in \Ck{} ($v$ in \TRule{Alloc} in \cref{fig:xlateRuleIR})
such that: if \ddAk{}'s execution is well-formed (does not trigger $\mathscr{W}$), and \Ck{}'s
execution is UB-free (does not trigger $\mathscr{U}$), then
the two allocated intervals are identical (the observable values created through \VallocASym{}
and \VallocBSym{} must be equal).

In the presence of potentially register-allocated and eliminated
local variables,
$\Ck \Supset{} \Ak$ implies a correct translation from \Ck{} to \Ak{}.
If all local variables are allocated in stack, $\Ck \Supset{} \Ak$ reduces to $\Ck \gtrdot{} \Ak$ with $\ddAk=\dAk$.
\Cref{fig:example1a} is an example of an annotated \ddAk{}.

\section{Witnessing Refinement through a Determinized Cross-Product $\ddAk\boxtimes{}\Ck$}
\label{sec:product}

We first introduce program paths and their properties.
Let $P\in\{\Ck,\ddAk\}$.
Let $e_{P}=(n_{P}\rightarrow{}n_{P}^t)\in{}\mathcal{E}_{P}$ represent an edge from node
$n_{P}$ to node $n_{P}^t$, both drawn from $\mathcal{N}_{P}$.
A {\em path} $\xi_{P}$ from $n_P$ to $n^t_P$,
written $\xi_{P}=n_P\twoheadrightarrow{}n^t_P$,
is a sequence
of $m\geq{}0$ edges $(e_P^1, e_P^2, \ldots, e_P^m)$
with $\forall_{1\leq{}j\leq{m}}:{e_P^j=(n_P^{f,j}\rightarrow{}n_P^{t,j})\in{}\mathcal{E}_P}$,
such that
$n_P^{f,1}=n_P$,
$n_P^{t,m}=n_P^t$,
and
$\bigwedge\limits_{j=1}^{m-1}(n_P^{t,j}=n_P^{f,j+1})$.
Nodes $n_P$ and $n_P^t$ are called
the {\em source} and {\em sink} nodes of
$\xi_{P}$ respectively. Edge $e_P^j$ (for some $1\leq{}j\leq{}m$) is said to
be present in $\xi_{P}$, written $e_P^j\in{}\xi_{P}$.
An
empty sequence, written $\epsilon$, represents the {\em empty path}.
The {\em path condition}
of a path $\xi_{P}=\PathT{n_P}{n^t_P}$,
written $pathcond(\xi_{P})$,
is a conjunction of the edge conditions of the
constituent edges. Starting at $n_P$,
$pathcond(\xi_{P})$ represents the condition that
$\xi_{P}$
executes to completion.

A sequence of edges corresponding to a \Shaded{} statement in the translations
(\cref{fig:xlateRuleIR,fig:xlateRuleAsm,fig:xlateRuleAsmStackLocals,fig:xlateRuleAsmAllLocals})
is distinguished and identified as an {\em I/O path}.
An I/O path must contain either a single {\tt rd} or a single {\tt wr} instruction.
For example, the sequence of edges corresponding to
``$\Iwrite{\Vcall{\rho,\Tuple{x}, \CalleeRegions, \prjM{(\il{\Ck}{\CalleeRegions}}{M_{\Ck}}}}$''
and
``$M_{\Ck} \Assign \updM{\il{\Ck}{\CalleeRegions \setminus \Gro}}{M_{\Ck}}{\Iread{\memType}}$''
in \TRule{\Call{\Ck}} (\cref{fig:xlateRuleIR})
refer to two separate I/O paths. A path without any {\tt rd} or {\tt wr} instructions
is called an {\em I/O-free path}.

\subsection{Determinized Product Graph as a Transition Graph}
\label{sec:detX}
A product program, represented as a {\em determinized product graph}, also called a comparison
graph or a cross-product, $\Xk=\ddAk\boxtimes{}\Ck=(\mathcal{N}_{\Xk},\mathcal{E}_{\Xk},\DXk{})$, is a directed multigraph with
finite sets of nodes $\mathcal{N}_{\Xk}$ and edges $\mathcal{E}_{\Xk}$,
and a {\em deterministic choice map} \DXk{}.
\Xk{} is used to encode a lockstep execution
of \ddAk{} and \Ck{}, such that $\mathcal{N}_{\Xk{}}\subseteq{}\mathcal{N}_{\ddAk}\times{}\mathcal{N}_{\Ck}$.
The start node of \Xk{} is $n^s_{\Xk{}}=(n^s_{\ddAk},n^s_{\Ck})$ and all nodes
in $\mathcal{N}_{\Xk{}}$ must be reachable from $n^s_{\Xk{}}$.
A node $n_{\Xk{}}=(n_{\ddAk},n_{\Ck})$
is an error node iff either $n_{\ddAk}$ or $n_{\Ck}$ is an error node.
$\NNP{\Xk{}}$ denotes the set of non-error nodes in \Xk{}, such that
$n_{\Xk{}}=(n_{\ddAk},n_{\Ck})\in{}\NNP{\Xk{}}\Leftrightarrow{}(n_{\ddAk}\in{}\NNP{\ddAk}\land{}n_{\Ck}\in{}\NNP{\Ck})$.

Let $n_{\Xk{}}=(n_{\ddAk},n_{\Ck})$ and $n_{\Xk{}}^t=(n^t_{\ddAk},n^t_{\Ck})$ be
nodes in $\mathcal{N}_{\Xk{}}$; let
$\xi_{\ddAk}=n_{\ddAk}\twoheadrightarrow{}n^t_{\ddAk}$
be a finite path in \ddAk{};
and let $\xi_{\Ck}=n_{\Ck}\twoheadrightarrow{}n^t_{\Ck}$
be a
finite path in \Ck{}.
Each edge,
$e_{\Xk{}}=(\XEdgeT{n_{\Xk{}}}{\xi_{\ddAk};\xi_{\Ck}}{n^t_{\Xk{}}})\in{}\mathcal{E}_{\Xk{}}$,
is defined as a sequential execution
of $\xi_{\ddAk}$
followed by $\xi_{\Ck}$.
The execution of $e_{\Xk{}}$ thus transfers
control of \Xk{}
from $n_{\Xk{}}$ to $n_{\Xk{}}^t$.
The machine state of \Xk{} is the concatenation of the machine states
of \ddAk{} and \Ck{}.
The outside world of \Xk{}, written $\World{\Xk{}}$, is
a pair of the outside worlds of \ddAk{} and \Ck{}, i.e.,
$\World{\Xk{}}=(\World{\ddAk},\World{\Ck})$.
Similarly, the trace generated by \Xk{}, written $\Trace{\Xk{}}$,
is a pair of the traces generated by \ddAk{} and \Ck{},
i.e., $\Trace{\Xk{}}=(\Trace{\ddAk},\Trace{\Ck})$.

During an execution
of $e_{\Xk{}}=(\XEdgeT{n_{\Xk{}}}{\xi_{\ddAk};\xi_{\Ck}}{n^t_{\Xk{}}})\in{}\mathcal{E}_{\Xk{}}$,
let $\vv{x}_{\ddAk}$ be variables in \ddAk{} just
at the end of the execution of path $\xi_{\ddAk}$ (at $n_{\ddAk}^t$)
but before the execution
of path $\xi_{\Ck}$ (recall, $\xi_{\ddAk}$ executes before $\xi_{\Ck}$).
$\DXk{}:((\EXk{}\times\EP{\Ck}\times{}\mathbb{N})\to\mathtt{ExprList})$,
called a {\em deterministic choice map}, is a partial function
that maps edge $e_{\Xk}\in{\mathcal{E}_{\Xk}}$,
and 
the $n^{th}$ (for $n\in{}\mathbb{N}$)
occurrence of
an edge `$e^{\theta}_{\Ck}\in{}\xi_{\Ck}$'
labeled with instruction $\vv{v} \Assign \Ichoose{\vv{\tau}}$
to a list of expressions $E(\vv{x}_{\ddAk})$.
The semantics of \DXk{} are such that,
if $\DXk{}(e_{\Xk}, e^{\theta}_{\Ck}, n)$ is defined,
then during an execution of $e_{\Xk}$,
an execution of the $n^{th}$ occurrence
of edge $e^{\theta}_{\Ck}\in{}\xi_{\Ck}$ labeled with
$\vv{v} \Assign \Ichoose{\vv{\tau}}$
is semantically equivalent to an
execution of $\vv{v} \Assign \DXk{}(e_{\Xk}, e^{\theta}_{\Ck}, n)$; otherwise,
the original non-deterministic semantics of $\theta$ are used.

\DXk{} determinizes (or refines) the non-deterministic choices in \Ck{}.
For example, in a product
graph \Xk{} that correlates the programs in \cref{fig:example1i}
and \cref{fig:example1a},
let $e^2_{\Xk}\in{}\EXk{}$ correlate single instructions {\tt I2} and {\tt A4.2}.
Let $e^{\mathtt{I2},\theta_a}_{\Ck}$ represent the edge labeled
with $\alpha_{b} \Assign \Ichoose{\bv{32}}$ as a part 
of the translation of the {\tt alloc} instruction at {\tt I2}, as seen in \TRule{Alloc}.
Then, 
$\DXk(e^2_{\Xk}, e^{\mathtt{I2},\theta_a}_{\Ck}, 1)={\tt esp}$ is identified by
the first operand of the annotated {\tt alloc$_s$} instruction at {\tt A4.2}.
Similarly, if another edge $e^{\mathtt{I2},\theta_m}_{\Ck}$ (in the
translation of {\tt alloc} at {\tt I2}) is labeled
with $\IchooseM$
(due to $M_{\Ck} \Assign \updM{[\alpha_{b},\alpha_{e}]}{M_{\Ck}}{\IchooseM}$),
then
$\DXk(e^2_{\Xk}, e^{\mathtt{I2},\theta_m}_{\Ck}, 1)=M_{\ddAk}$, i.e., the
initial contents of the newly-allocated region in \Ck{} are based on the
contents of the correlated uninitialized stack region in \ddAk{}.
Similarly,
let $e^1_{\Xk}\in{}\EXk{}$ correlate single instructions {\tt I1} and {\tt A4.1}
so that
$\DXk(e^1_{\Xk}, e^{\mathtt{I1},\theta_a}_{\Ck}, 1)=\mathtt{v}_{\tt I1}$ and
$\DXk(e^1_{\Xk}, e^{\mathtt{I1},\theta_m}_{\Ck}, 1)=M_{\ddAk}$.

For a path \CPath{} in \Ck{}, \CPathD{}
denotes a {\em determinized path} that is identical to \CPath{}
except that:
if $\DXk{}(e_{\Xk}, e^{\theta}_{\Ck}, n)$ is defined, then
the $n^{th}$ occurrence of edge $e^{\theta}_{\Ck}\in{}\CPath{}$,
labeled with $\vv{v} \Assign \Ichoose{\vv{\tau}}$,
is replaced with a new edge $e^{\theta_n^{\prime}}_{\Ck}$
labeled with $\vv{v} \Assign \DXk{}(e_{\Xk}, e^{\theta}_{\Ck}, n)$.

Execution of a product
graph \Xk{} must begin at node $n^s_{\Xk{}}$ in an initial machine state where
$\World{\ddAk}=\World{\Ck}$
and $\steq{T_{\ddAk}}{T_{\Ck}}$ hold.
Thus, \Xk{} is a transition
graph with its execution semantics derived from
the semantics of \ddAk{} and \Ck{}, and the map \DXk{}.

\subsection{Analysis of the Determinized Product Graph}
Let $\Xk{}=\ddAk\boxtimes{}\Ck=(\mathcal{N}_{\Xk{}},\mathcal{E}_{\Xk{}},\DXk{})$ be
a determinized product graph.
At each
non-error
node $n_{\Xk}\in \NNP{\Xk}$,
we
infer a node 
invariant, $\phi_{n_{\Xk}}$, which is
a first-order logic predicate
over state elements of \Xk{} at
node $n_{\Xk}$ that holds
for all possible executions of \Xk{}.
A node invariant $\phi_{n_{\Xk}}$ relates the values of state elements
of \Ck{} and \ddAk{} that can be observed at $n_{\Xk}$.

\begin{definition}[Hoare Triple]
\label{defn:hoareTriple}
Let $n_{\Xk}=(n_{\ddAk},n_{\Ck})\in{}\NNP{\Xk{}}$.
Let
$\APath{}=\PathT{n_{\ddAk}}{n^t_{\ddAk}}$ and
$\CPath{}=\PathT{n_{\Ck}}{n^t_{\Ck}}$ be
paths in \ddAk{} and \Ck{}.
A {\em Hoare triple},
written
$\hoareTriple{pre}{(\APath{};\CPath{})}{post}$,
denotes
the statement: if execution starts at
node $n_\Xk{}$ in
state $\sigma$ such
that
predicate
$pre(\sigma)$ holds,
and if paths $\APath{};\CPath{}$ are executed
in sequence to completion
finishing in state $\sigma'$,
then
predicate
$post(\sigma')$ holds.
\end{definition}

\begin{definition}[Path cover]
\label{defn:pathCover}
At a node $n_{\Xk}=(n_{\ddAk},n_{\Ck})\in{}\NXk{}$, for a path $\APath{}=\PathT{n_{\ddAk}}{n^t_{\ddAk}}$,
let $\forall_{1\leq{}j\leq{}m}:e^{j}_{\Xk}=\XEdgeT{n_{\Xk}}{\xi_{\ddAk};\xi^j_{\Ck}}{n_{\Xk}^{t_j}}$
be all edges in $\mathcal{E}_{\Xk}$,
such that $n_{\Xk}^{t_j}=(n_{\ddAk}^{t},n_{\Ck}^{t_j})$.
The set of edges $\{e^1_{\Xk},e^2_{\Xk},\ldots,e^m_{\Xk}\}$ {\em covers path $\xi_{\ddAk}$},
written
\pthcover{\XEdgeN}{n_\Xk}{\DXk}{\xi_{\ddAk}},
iff
$\hoareTriple{\phi_{n_{\Xk}}}{(\APath{}; \epsilon)}{\bigvee\limits_{j=1}^{m}{\pathcond{\CPathDN{j}}}}$ holds.
\end{definition}

\begin{definition}[Path infeasibility]
At a node $n_{\Xk}=(n_{\ddAk},n_{\Ck})\in{}\NXk{}$, a path $\APath = \PathT{n_{\ddAk}}{n^t_{\ddAk}}$ is {\em infeasible} at $n_{\Xk}$ iff
$\hoareTriple{\phi_{n_{\Xk}}}{(\APath{}; \epsilon)}{\mathtt{false}}$ holds.
\end{definition}

\begin{definition}[Mutually exclusive paths]
Two paths,
$\PPath{P}^{1} = n_P\twoheadrightarrow{}n^{t_1}_P$
and
$\PPath{P}^{2} = n_P\twoheadrightarrow{}n^{t_2}_P$,
both originating at
node $n_P$
are {\em mutually-exclusive},
written
$\PPath{P}^{1} \Bumpeq{} \PPath{P}^{2}$,
iff neither is a prefix of the other.
\end{definition}

\begin{definition}
A {\em pathset} $\Pathset{P}$ is a
set of pairwise
mutually-exclusive paths $\Pathset{P}=\{\PPath{P}^1,\PPath{P}^2,\ldots,\PPath{P}^m\}$
originating at the same node $n_P$, i.e.,
$\forall_{1\leq{}j\leq{}m}:{\PPath{P}^j=\PathT{n_P}{n^j_P}}$ and
$\forall_{1\leq{}j_1{}<j_2{}\leq{}m}:({\PPath{P}^{j_1}\Bumpeq{}\PPath{P}^{j_2}})$.
\end{definition}

\subsubsection{\Xk{} Requirements}
\label{sec:reqX}
\noindent
The following
requirements on \Xk{} help
witness $\Ck{}\refines{}\ddAk{}$:
\begin{enumerate}[label={\arabic*.}]
\item  (Mutex\ddAk): For each node $n_{\Xk}$
	with {\em all} outgoing edges $\{e^1_{\Xk},e^2_{\Xk},\ldots,e^m_{\Xk}\}$
	such that $\XEdge^j=(\XEdgeT{n_{\Xk}}{\xi^j_{\ddAk};\xi^j_{\Ck}}{n^j_{\Xk}})$ (for $1\leq{}j\leq{}m$),
	the following holds:
	$\forall_{1\leq{}j_1,j_2\leq{}m}:((\APathN{j_1}=\APathN{j_2})\vee{}(\APathN{j_1}\Bumpeq{}\APathN{j_2}))$.

\item  (Mutex\Ck):
 At each node $n_{\Xk}$, for a path \APath{},
	let $\{e^1_{\Xk},e^2_{\Xk},\ldots,e^m_{\Xk}\}$
	be a set of {\em all} outgoing edges such that
  $\XEdge^j=\XEdgeT{n_{\Xk}}{\APath; \CPathN{j}}{n^t_{\Xk}}$ (for $1\leq j \leq m$).
	Then,
  the set $\{\CPathN{1},\CPathN{2},\ldots,\CPathN{m}\}$ must be a pathset.

\item (Termination)
For each non-error node $n_{\Xk} = (n_{\ddAk},n_{\Ck}) \in \NNP{\Xk}$,
$n_{\ddAk}$ is a terminating node iff $n_{\Ck}$ is a terminating node.

\item  (SingleIO):
For each edge $e_{\Xk}=(\XEdgeT{n_{\Xk}}{\xi_{\ddAk};\xi_{\Ck}}{n_{\Xk}^{t}})\in{}\mathcal{E}_{\Xk{}}$,
either both \APath{} and $\CPath{}$ are I/O paths, or
both \APath{} and $\CPath{}$ are I/O-free.

\item  (Well-formedness): If a node of the form $n_{\Xk}=(\_, \mathscr{W}_{\Ck})$ exists in $\mathcal{N}_{\Xk}$,
then $n_{\Xk}$ must be $(\mathscr{W}_{\ddAk}, \mathscr{W}_{\Ck})$.

\item  (Safety): If a node of the form $n_{\Xk}=(\mathscr{U}_{\ddAk}, \_)$ exists in $\mathcal{N}_{\Xk}$,
then $n_{\Xk}$ must be $(\mathscr{U}_{\ddAk}, \mathscr{U}_{\Ck})$.

\item  (Similar-speed):
Let $(e^1_\Xk,e^2_\Xk,\ldots,e^m_\Xk)$ be a cyclic path, so that
$\forall_{1 \leq j \leq m}:{\XEdge^j=(\XEdgeT{n_{\Xk}^{f,j}}{\xi^j_{\ddAk};\xi^j_{\Ck}}{n_{\Xk}^{t,j}})\in{}\EXk{}}$;
$n_{\Xk}^{f,1}=n_{\Xk}^{t,m}$; and
$\bigwedge\limits_{j=1}^{m-1}{(n_{\Xk}^{t,j} = n_{\Xk}^{f,j+1})}$.
For each cyclic path,
$(\neg{}\bigwedge\limits_{j=1}^m(\xi^j_{\ddAk}=\epsilon))\land{}(\neg{}\bigwedge\limits_{j=1}^m(\xi^j_{\Ck}=\epsilon))$
holds.

\item  (Coverage\ddAk{}):
  For each non-error node $n_{\Xk} = (n_{\ddAk},n_{\Ck}) \in \NNP{\Xk}$ and
  for each possible outgoing path $\APathO{} = \PathT{n_{\ddAk}}{n^o_{\ddAk}}$,
  either
  $\APathO$ is infeasible at $n_{\Xk}$, or 
  there exists $\XEdge = (\XEdgeT{n_{\Xk}}{\APath;\CPath}{n^t_{\Xk}})\in\EXk$ such that either
      \APath{} is a prefix of $\APathO$ or
      $\APathO$ is a prefix of \APath{}.

\item  (Coverage\Ck{}):
 At node $n_{\Xk}$, for some \APath{},
	let $\{e^1_{\Xk},e^2_{\Xk},\ldots,e^m_{\Xk}\}$
	be the set of {\em all} outgoing edges such that
  $\XEdge^j=\XEdgeT{n_{\Xk}}{\APath; \CPathN{j}}{(n^t_{\ddAk}, n_{\Ck}^{t_j})}$ (for $1\leq{}j\leq{}m$).
	Then,
\pthcover{\XEdgeN}{n_\Xk}{\DXk}{\xi_{\ddAk}}
holds.

\item  (Inductive):
	For each non-error edge $e_{\Xk}=(\XEdgeT{n_{\Xk}}{\APath{};\CPath{}}{n_{\Xk}^{t}})\in\mathcal{E}_{\Xk}$,
$\hoareTriple{\phi_{n_{\Xk}}}{(\APath{};\CPathD{})}{\phi_{n_{\Xk}^t}}$ holds.

\item  (Equivalence): For each non-error node $n_{\Xk} = (n_{\ddAk},n_{\Ck}) \in \NNP{\Xk}$,
$\Omega_{\ddAk}=\Omega_{\Ck}$ must belong to $\phi_{n_{\Xk}}$.

\item  (Memory Access Correspondence) or (MAC):
For each edge $e_{\Xk}=(\XEdgeT{n_{\Xk}}{\APath{};\CPath{}}{n_{\Xk}^{t}})\in\mathcal{E}_{\Xk}$,
such that $n_{\Xk}^{t}\neq{}(\_, \mathscr{U}_{\Ck})$,
$\hoareTriple{\phi_{n_{\Xk}}\land{}(\ghost{\rdAcc{\ddAk}}=\ghost{\rdAcc{\Ck}}=\emptyset{})}{(\APath{};\CPathD{})}{(\ghost{\rdAcc{\ddAk}}\setminus{}\ghost{\rdAcc{\Ck}})\subseteq{}\il{\ddAk}{G \cup \F}\cup{}[\mathtt{esp},\ghost{\spE}]}$
and 
$\hoareTriple{\phi_{n_{\Xk}}\land{}(\ghost{\wrAcc{\ddAk}}=\ghost{\wrAcc{\Ck}}=\emptyset{})}{(\APath{};\CPathD{})}{(\ghost{\wrAcc{\ddAk}}\setminus{}\ghost{\wrAcc{\Ck}})\subseteq{}\il{\ddAk}{\Grw \cup \Frw}\cup{}[\mathtt{esp},\ghost{\spE}]}$ hold.

\item  (MemEq): For each non-error node
  $n_{\Xk} \in \NNP{\Xk}$,
  $\prjMEq{\il{\ddAk}{\NS{}} \setminus (\ilZv{\ddAk})}{M_{\ddAk}}{M_{\Ck}}$ must belong to $\phi_{n_{\Xk}}$.
\end{enumerate}

(MAC) effectively requires that for every access on path \APath{} to an address $\alpha$
belonging to region $r\in{}\{\heap{},cl\}$, there exists an access
to $\alpha$ of the same read/write type on path \CPathD{}. This requirement
allows us to soundly over-approximate
the set of addresses belonging to \heap{} and $cl$ for
a faster SMT encoding (\cref{theorem:fastEncoding,sec:intervalEncodingsHeapCl}).
For (MAC) to be meaningful,
\ghost{\rdAcc{{\ddAk,\Ck}}} and \ghost{\wrAcc{{\ddAk,\Ck}}}
must not be included in \Xk{}'s state elements over which a
node invariant $\phi_{n_{\Xk}}$ is inferred.

The first seven are {\em structural requirements} (constraints on the graph structure of \Xk)
and the remaining six are {\em semantic requirements} (require discharge of proof obligations).
The first eleven are {\em soundness requirements} (required
for \cref{theorem:witness}), the first twelve
are {\em fast-encoding requirements}, and
all thirteen are {\em search-algorithm requirements} (required
for search optimizations).
Excluding (Coverage\ddAk{}) and (Coverage\Ck{}),
the remaining eleven
are called {\em non-coverage requirements}.

\begin{theorem}
\label{theorem:witness}
If there exists $\Xk=\ddAk\boxtimes\Ck$ that satisfies the soundness requirements,
then
$\Ck \refines{} \ddAk$ holds.
\end{theorem}
\begin{proof}[Proof sketch]
  (Coverage\ddAk{}) and (Coverage\Ck{}) ensure the coverage of \ddAk{}'s and
  \Ck{}'s traces in \Xk{}.
  For an error-free execution of \Xk{}, (Equivalence) and (Similar-speed) ensure
  that the generated traces are stuttering equivalent;
  for executions terminating in an error, (SingleIO), (Well-formedness), and (Safety) ensure that
  $\Ck \refines \ddAk$ holds by definition.
  See \cref{app:soundnessProof} for the coinductive proof.
\end{proof}

\subsubsection{Callers' Virtual Smallest Semantics}
\label{sec:cvSmallest}

Construct \CkP{} and \AkP{} from \Ck{} and \Ak{}
by using new
\emph{callers' virtual smallest semantics}
such that
assignments to \il{\Ck}{cv} and \il{\Ak}{cv}
due to \TRule{\Entry{\Ck}} and \TRule{\Entry{\Ak}}
respectively
(\cref{fig:xlateRuleIR,fig:xlateRuleAsm})
are removed
and
uses of \il{\Ck}{cv} and \il{\Ak}{cv}
due to
\TRule{\Entry{\Ck}},
\TRule{\Entry{\Ak}},
\TRule{Op-esp'},
\TRule{\Load{\ddAk}},
\TRule{\Store{\ddAk}},
\TRule{AllocS'},
and
\TRule{AllocV}
are replaced with $\emptyset$:
\begin{enumerate}
  \item In \TRule{\Entry{\Ck}} and \TRule{\Entry{\Ak}},
    $\mathtt{addrSetsAreWF}(\il{P}{\heap}, \il{P}{cl}, \il{P}{cv}, \ldots, \ii{P}{\g}, \ldots, \il{P}{\f}, \ldots, \ii{P}{\y},\ldots, \il{P}{\yv})$
    is replaced with 
    $\mathtt{addrSetsAreWF}(\il{P}{\heap}, \il{P}{cl}, \allowbreak \ldots, \ii{P}{\g}, \ldots, \il{P}{\f}, \ldots, \ii{P}{\y},\ldots, \il{P}{\yv})$
    for $P \in \{ \Ck, \Ak\}$.

  \item In \TRule{Op-esp'},
    $\intervalContainedInAddrSet{t}{\mathtt{esp}-1_{\bv{32}}}{\il{\ddAk}{\{\free\}} \cup ((\il{\ddAk}{cv} \cup \ilZv{\ddAk}) \setminus \il{\ddAk}{\F})}$
    is replaced with 
    $\intervalContainedInAddrSet{t}{\mathtt{esp}-1_{\bv{32}}}{\il{\ddAk}{\free} \cup (\ilZv{\ddAk} \setminus \il{\ddAk}{\F})}$.

  \item In \TRule{\Load{\ddAk}},
    $\Overlap([p]_{w}, \il{\ddAk}{\free} \cup ((\il{\ddAk}{cv} \cup (\ilZv{\ddAk})) \setminus \il{\ddAk}{\F \cup \STACK{}}))$
    is replaced with
    $\Overlap([p]_{w}, \il{\ddAk}{\free} \cup ((\ilZv{\ddAk}) \setminus \il{\ddAk}{\F{} \cup \STACK}))$.

  \item In \TRule{\Store{\ddAk}},
    $\Overlap([p]_{w}, \il{\ddAk}{\{\free\}\cup\Gro\cup\Fro} \cup ((\il{\ddAk}{cv} \cup (\ilZv{\ddAk})) \setminus \il{\ddAk}{\Frw \cup \STACK}))$
    is replaced with
    $\Overlap([p]_{w}, \il{\ddAk}{\{\free\}\cup\Gro\cup\Fro} \cup ((\ilZv{\ddAk}) \setminus \il{\ddAk}{\Frw \cup \STACK}))$.

  \item In \TRule{AllocS'},
    $\Overlap([v]_{w}, \il{\ddAk}{cv} \cup \ilZv{\ddAk})$
    is replaced with
    $\Overlap([v]_{w}, \ilZv{\ddAk})$.

  \item In \TRule{AllocV},
    $\intervalContainedInAddrSetAndAligned{v}{v+w-1_{\bv{32}}}{\mathtt{comp}(\il{\ddAk}{\NS \cup \{ cv \}})}{a}$
    is replaced with
    $\intervalContainedInAddrSetAndAligned{v}{v+w-1_{\bv{32}}}{\mathtt{comp}(\il{\ddAk}{\NS})}{a}$.
\end{enumerate}

Essentially, with callers' virtual smallest semantics,
the $cv$ region is made empty ($\il{\Ck}{cv} = \il{\Ak}{cv} = \emptyset$).
With an empty $cv$, the address set of region \free{} is computed as
$\il{P}{\free} = {\tt comp}(\il{P}{\NS \cup \F \cup \STACK})$
for $P \in \{ \Ck, \Ak \}$.

Let \ddAkP{} be obtained
by annotating \AkP{} as described in \cref{sec:refnDefnVirtual}.
Let \ddAk{} be the annotated
version of \Ak{}, such that the annotations made in
\ddAk{} and \ddAkP{} are identical.

\begin{theorem}
\label{theorem:cvSmallest}
Given $\XkP=\ddAkP\boxtimes{}\CkP{}$ that
satisfies the fast-encoding requirements,
it is possible to construct $\Xk{}=\ddAk{}\boxtimes{}\Ck{}$
that also satisfies the fast-encoding requirements.
\end{theorem}
\begin{proof}[Proof sketch]
Start by constructing $\Xk=\XkP{}$.
With a non-empty $cv$, \ddAk{} may include more executions of a path
of form $\APath{} = \PathT{n_{\ddAk}}{\WAk}$;
add new edges to \EXk{}, where each new edge
correlates \APath{} with an empty \Ck{} path ($\CPath{} = \epsilon$).
\Xk{} should still satisfy the fast-encoding requirements.
See \cref{app:cvSmallestProof} for the proof.
\end{proof}

\subsubsection{Safety-Relaxed Semantics}
\label{sec:safetyUnchecked}
Construct $\AkP$ from \Ak{} (with callers' virtual smallest semantics)
by using new {\em safety-relaxed semantics} for the assembly procedure such
that: (1) a
$\varphi{}_l=\Overlap([p]_{w}, \il{\ddAk}{\free} \cup ((\ilZv{\ddAk}) \setminus \il{\ddAk}{\F \cup \STACK{}}))$
check in \TRule{\Load{\ddAk}} in \Ak{}
is replaced with
$\varphi{}_l'=\Overlap([p]_{w}, (\ilZv{\ddAk}) \setminus (\il{\ddAk}{\F} \cup [\mathtt{esp}, \ghost{\stkE}]))$ in \AkP{};
(2) a
$\varphi{}_s=\Overlap([p]_{w}, \il{\ddAk}{\{\free\}\cup\Gro\cup\Fro} \cup ((\ilZv{\ddAk}) \setminus \il{\ddAk}{\Frw \cup \STACK{}}))$
check in \TRule{\Store{\ddAk}} in \Ak{}
is replaced with
$\varphi{}_s'=\Overlap([p]_{w}, (\ilZv{\ddAk}) \setminus (\il{\ddAk}{\Frw} \cup [\mathtt{esp}, \ghost{\stkE}]))$ in \AkP{};
and (3) a $\varphi{}_r=\neg(\prjMEq{\il{\Ak}{cs}}{\ghost{M^{cs}}}{M_{\Ak}})$ check
in \TRule{\Ret{\Ak}} in \Ak{}
is replaced with $\varphi{}_r'=\mathtt{false}$ in \AkP{}. 
Let $\ddAkP$ be obtained
by annotating $\AkP$ using instructions described in \cref{sec:refnDefnVirtual}.
Let \ddAk{} be the annotated
version of $\Ak{}$, such that the annotations made in
$\ddAk{}$ and $\ddAk'$ are identical.
Let \Ck{} be the corresponding unoptimized IR procedure with the
callers' virtual smallest semantics.

\begin{theorem}
\label{theorem:fastEncoding}
Given $\XkP=\ddAkP\boxtimes{}\Ck{}$ that
satisfies the fast-encoding requirements,
it is possible to construct $\Xk{}=\ddAk{}\boxtimes{}\Ck{}$
that also satisfies the fast-encoding requirements.
\end{theorem}
\begin{proof}[Proof sketch]
Start by constructing $\Xk=\XkP{}$.
Because $\varphi{}_{l,s,r}'\Rightarrow{}\varphi{}_{l,s,r}$,
\ddAk{} may include more executions of a path
of form $\APath=\PathT{n_{\ddAk}}{\UAk}$.
Add new
edges to $\EXk{}$, where each new edge correlates \APath{} with some
$\CPath{}=\PathT{n_{\Ck}}{\UCk}$.
Because $\XkP$ satisfies (MAC),
the addition of such new edges will ensure that \Xk{} satisfies (Coverage\Ck{}).
See \cref{app:fastEncodingProof} for the proof.
\end{proof}

Using \cref{theorem:fastEncoding,theorem:cvSmallest}, hereafter,
we will use only the safety-relaxed and callers' virtual smallest semantics of the
unoptimized IR and assembly procedures.
We will continue to refer to the unoptimized IR
with the callers' virtual smallest semantics
and assembly procedure with the safety-relaxed
and callers' virtual smallest semantics
as \Ck{} and \Ak{} respectively.
The corresponding annotated procedure of \Ak{} will be referred as \ddAk{}.

\section{Automatic Construction of a Cross-Product}
\label{sec:algo}

We now describe \toolName{}, an algorithm that
takes as input, the
transition graphs corresponding to procedures \Ck{} and \Ak{}, and an {\em unroll factor} $\mu$,
and returns as output,
annotated \ddAk{} and
product
graph $\Xk=\ddAk\boxtimes{}\Ck=(\mathcal{N}_{\Xk{}},\mathcal{E}_{\Xk{}},\DXk{})$,
such that all thirteen search-algorithm requirements
are met.
It identifies an inductive invariant
network \IXk{} that maps each
non-error
node $n_{\Xk}\in{}\NNP{\Xk}$ to
its node invariant $\phi_{n_{\Xk}}$.
Given enough computational time,
\toolName{} is guaranteed to find the required $(\ddAk,\Xk)$
if: (a) \Ak{} is a translation of \Ck{} through
bisimilar transformations up to a maximum
unrolling of $\mu$;
(b) for two or more
allocations or procedure
calls that reuse stack space in \Ak{}, their
relative order in \Ck{} is preserved in \Ak{};
(c) the
desired annotation to $\ddAk$ is identifiable either
through search heuristics or through compiler hints; and (d) our
invariant inference procedure is able to identify
the required invariant network \IXk{} that captures the compiler
transformations across \Ck{} and \Ak{}.
\toolName{} constructs the solution
incrementally, by relying on the property that for a
non-coverage requirement to hold for
fully-annotated \ddAk{} and fully-constructed \Xk{},
it must also hold for partially-annotated \ddAk{} and
a partially-constructed subgraph of \Xk{}
rooted at its entry node
$n^s_{\Xk{}}$.

\toolName{} is presented in \cref{algo:constructX}.
The algorithm has two phases.
In the first phase,
identified by \texttt{CORRELATE\_AND\_ANNOTATE}
(\cref{lst:line:consXp1} in \cref{algo:constructX}),
\toolName{} attempts to correlate
the paths in \Ak{} with the paths in \Ck{}
while simultaneously identifying the required annotation for \Ak{}.
At the successful completion of the first phase,
all paths in the original, unannotated \Ak{}
are correlated.
However, recall that the annotation instructions,
$\mathtt{(de)alloc}_s$ and $\mathtt{(de)alloc}_v$,
have additional paths to error nodes \UAk{} and \WAk{}
(\cref{fig:xlateRuleAsmStackLocals,fig:xlateRuleAsmAllLocals}).
These paths to error nodes are not correlated
in the first phase.
The second phase of the algorithm,
identified by \texttt{CORRELATE\_NEW\_ERROR\_PATHS}
(\cref{lst:line:consXp2} in \cref{algo:constructX}),
correlates these additionally introduced (error) paths.

The sub-procedure \textit{constructX()},
used in both phases,
identifies the required correlations and annotation
and builds the product program \Xk{} incrementally.
It assumes the availability of an oracle
It assumes
the availability of a
{\bf chooseFrom} operator, such that
$\rho \mapsfrom{}\mathrm{\bf chooseFrom}\ \vv{\rho}$
chooses a quantity $\rho$ from
a finite set $\vv{\rho}$, such that
\toolName{}
is able to complete the refinement proof, if such a choice exists. If
the search space is limited, an exhaustive
search could be used to implement {\bf chooseFrom}.
Otherwise, a
counterexample-guided
best-first search procedure (described later) is employed
to approximate {\bf chooseFrom}.

\begin{algorithm}[!tb]
\caption{\label{algo:constructX}Automatic construction of \Xk{}}
\begin{footnotesize}
\SetKwProg{Fn}{Function}{}{end}
\Fn{\toolName{}$(\Ak{}, \Ck{}, \mu)$} {
  $\ddAk \mapsfrom \Ak$;
  \qquad
  $\NXk{} \mapsfrom \{(n^s_{\ddAk},n^s_{\Ck})\}$;
  \qquad
  $\EXk{} \mapsfrom \emptyset$;
  \qquad
  $\DXk{} \mapsfrom \emptyset$;
  \qquad
  $\IXk{} \mapsfrom \{ (n^s_{\ddAk},n^s_{\Ck}) \mapsto (\World{\ddAk}=\World{\Ck}) \}$;\\
  \uIf{$\neg \text{\textit{constructX}}(\ddAk{}, \Ck{}, \mu, \NXk{}, \EXk{}, \DXk{}, \IXk, \mathtt{CORRELATE\_AND\_ANNOTATE})$}{
  \nllabel{lst:line:consXp1}
    \Return{\textnormal{\texttt{Failure}}}
  }
  \uIf{$\neg \text{\textit{constructX}}(\ddAk{}, \Ck{}, \mu, \NXk{}, \EXk{}, \DXk{}, \IXk, \mathtt{CORRELATE\_NEW\_ERROR\_PATHS})$}{
  \nllabel{lst:line:consXp2}
    \Return{\textnormal{\texttt{Failure}}}
  }
  \uIf{$\neg \text{\textit{checkCoverageReqs}}(\NXk{}, \EXk{}, \DXk{}, $\IXk$, \ddAk{}, \Ck{})$}{
  	  \Return{\textnormal{\texttt{Failure}}}
  }
  \Return{\textnormal{\texttt{Success}$(\ddAk{}, (\NXk{}, \EXk{}, \DXk{}), \IXk)$}}
}
\Fn{constructX$(\ddAk{}, \Ck{}, \mu, \NXk{}, \EXk{}, \DXk{}, \IXk, phase)$} {
  $Q_{\ddAk}$ $\mapsfrom$ \textit{getCutPointsInRPO}(\ddAk{});\\
  \ForEach{$q_{\ddAk}$ {\tt in} $Q_{\ddAk}$}
  {
    \ForEach{$q^t_{\ddAk}$ {\tt in} cutPointSuccessors$(q_{\ddAk}, Q_{\ddAk}, \ddAk)$}
    {
      \ForEach{\APath{} {\tt in} getAllSimplePathsBetweenCutPoints$(q_{\ddAk}, q^t_{\ddAk}, \ddAk{})$}
      {
      \uIf{\textit{pathIsInfeasible}$(\APath, \NXk{}, \IXk)$}{
        \Continue
	    }
		  \uIf{\textit{pathExists}$(\APath{}, \EXk{})$}{
  			  \nllabel{lst:line:consXp3}
		    \Continue
	    }
      \ForEach{\CPath{} {\tt in} \textnormal{\textbf{chooseFrom}}\ \textit{correlatedPathsInCOptions}$(\APath, \mu, \NXk{}, \EXk{}, \ddAk{}, \Ck{})$}
      {
				  \If{$phase = \mathtt{CORRELATE\_AND\_ANNOTATE}$}{
	            $(\ddAk, \APath) \mapsfrom$ \textnormal{\textbf{chooseFrom}} \textit{asmAnnotOptions}$(\APath, \CPath, \ddAk, \Ck)$;\\
				  }
          
	        $\APathsP{},\CPathsP{} \mapsfrom \textit{breakIntoSingleIOPaths}(\APath), \textit{breakIntoSingleIOPaths}(\CPath{})$;\\
	        $\APathsN{*}, \CPathsN{*} \mapsfrom \textit{trimToMatchPathToErrorNode}(\APathsP, \CPathsP{})$;\\
	        \uIf{$\neg$\textit{haveSimilarStructure}$(\APathsN{*}, \CPathsN{*})$}{
	          \Return{\textnormal{\texttt{Failure}}}
	        }
      	  \ForEach{$\APathP{}=\PathT{n_{\ddAk}}{n^t_{\ddAk}},\CPathP{}=\PathT{n_{\Ck}}{n^t_{\Ck}}$ {\tt in} {\tt zip}$(\APathsN{*},\CPathsN{*})$}
      		{
          	$e_{\Xk} \mapsfrom (\APathP; \CPathP{})$;
          	\qquad
          	$n^t_{\Xk} \mapsfrom (n^t_{\ddAk},n^t_{\Ck})$;\\
          	\uIf{\textit{addingEdgeWillCreateEmptyCCycle}$(\NXk{}, \EXk{}, e_{\Xk})$}{\Return{\textnormal{\texttt{Failure}}}}
            $\EXk{} \mapsfrom \EXk{} \cup{}\{e_{\Xk}\}$;
            \qquad
          	$\NXk{} \mapsfrom \NXk{} \cup{}\{n^t_{\Xk}\}$;\\
          	$\DXk{} \mapsfrom$ \textit{addDetMappings}$(e_{\Xk}, \DXk{})$;\\
          	$\IXk \mapsfrom$ \textit{inferInvariantsAndCounterexamples}$(n^t_{\Xk},\NXk{}, \EXk{}, \DXk{}, \IXk, \ddAk{}, \Ck{})$;\\
          	\uIf{$\neg$\textit{checkSemanticReqsExceptCoverage}$(\NXk{}, \EXk{}, \DXk{}, \IXk, \ddAk{}, \Ck{})$}{\Return{\textnormal{\texttt{Failure}}}}
          }
        }
      }
    }
  }
  \Return{\textnormal{\texttt{Success}}}
}
\end{footnotesize}
\end{algorithm}

${\tt io}(n_P)$ evaluates to {\tt true} iff
$n_P$ is either a source or sink node of an I/O
path.
${\tt term}(n_P)$ evaluates to {\tt true} iff $n_P$ is a terminating node.
\toolName{} first identifies an ordered set of
nodes $Q_{P}\subseteq{}\mathcal{N}_{P}$,
called the {\em cut points}
in procedure $P$ ({\em getCutPointsInRPO}), such that
$Q_{P} \supseteq \{n_{P} : n_{P}\in{}\mathcal{N}_{P} \land{} (n_{P}=n^s_{P} \vee{} {\tt io}(n_{P}) \vee{} {\tt term}(n_{P}))\}$
and the maximum length of a path between
two nodes in $Q_{P}$ (not containing any other intermediate node
that belongs to $Q_P$) is finite.

The algorithm to identify $Q_{P}$ first initializes
$Q_{P} \Assign \{n_{P} : n_{P}\in{}\mathcal{N}_{P} \land{} (n_{P}=n^s_{P} \vee{} {\tt io}(n_{P}) \vee{} {\tt term}(n_{P}))\}$,
and then identifies all cycles in the transition
graph that do not already contain a cut point;
for each such cycle, the first node belonging to that
cycle in reverse postorder
is added to $Q_{P}$.
In \cref{fig:example1a}, $Q_{\ddAk}$ includes constituent nodes
of assembly instructions at
{\tt A1}, {\tt A9}, {\tt A14}, and {\tt exit}, where {\tt exit} is the destination node of the error-free {\tt halt} instruction
due to the procedure return at {\tt A17}.

A {\em simple path} $\PathT{q_P}{q^t_P}$ is a path connecting
two cut points $q_P,q^t_P\in{}Q_P$,
and not containing any other cut point as an intermediate node;
$q^t_P$ is called a {\em cut-point successor} of
$q_P$.
By definition, a simple path must be finite.
The {\em cutPointSuccessors()} function takes
a cut point $q_P$ and returns all its cut-point successors in reverse postorder.
In our example, the cut-point successors of a node at instruction {\tt A9}
are (constituent nodes of) {\tt A9}, {\tt A14}, $\mathscr{U}_{\ddAk}$,
and $\mathscr{W}_{\ddAk}$.
{\em getAllSimplePathsBetweenCutPoints($q_{P}$, $q^t_{P}$, $P$)} returns
{\em all} simple paths of the form $q_P\twoheadrightarrow{}q^t_P$, for
$q_{P},q^t_{P}\in{}Q_P$.
Given a simple path $\APath = \PathT{q_{\ddAk}}{q^t_{\ddAk}}$,
{\em pathIsInfeasible$(\APath, q_{\ddAk}, \NXk{}, \IXk{})$} returns {\tt true}
iff $\APath$ is infeasible at every
node $n_{\Xk}=(q_{\ddAk},\_)\in{}\NXk{}$;
our algorithm ensures there
can be at most one $n_{\Xk}=(q_{\ddAk},\_)\in{}\NXk{}$ for each
$q_{\ddAk}\in{}Q_{\ddAk}$.
Similarly, \textit{pathExists$(\APath, \EXk)$}
returns \texttt{true} iff \APath{} is already correlated with some
$\CPath{}=\PathT{q_{\Ck}}{q^t_{\Ck}}$ in \EXk{}
(i.e.,
$\exists \XEdge{}: \XEdge{} = \XEdgeT{(q_{\ddAk},q_{\Ck})}{\APath; \CPath}{(q^t_{\ddAk}, q^t_{\Ck})} \in \EXk{}$ holds).
Because the same \textit{constructX()} procedure is invoked in both phases,
the use of \textit{pathExists()} in \cref{lst:line:consXp3} of \cref{algo:constructX} is an optimization to avoid
correlating the same paths again in the second phase.
In the second phase, \textit{pathExists$(\APath,\EXk)$}
would return \texttt{false} only if \APath{}
corresponds to an error path due to an annotated $\mathtt{(de)alloc}_{s,v}$
instruction.

\subsubsection*{\bf correlatedPathsInCOptions()}
{\em correlatedPathsInCOptions($\xi_{\ddAk}$, \ldots{})} identifies options
for candidate pathsets
$[\Pathset{\Ck}]$, that can potentially be correlated with
$\xi_{\ddAk}=\PathT{q_{\ddAk}}{q^t_{\ddAk}}$, and the {\bf chooseFrom}
operator chooses a pathset $\Pathset{\Ck}$ from it.
A path $\PPath{\Ck}\in{}\Pathset{\Ck}$ need not be a simple path,
and can visit any node $n_{\Ck}\in{}\NCk{}$ up to $\mu$ times.
All paths in $\Pathset{\Ck}$ must
originate at a unique cut-point $q_{\Ck{}}$ such that
$(q_{\ddAk}, q_{\Ck}) \in \NXk$.
By construction, there will be exactly one such $(q_{\ddAk}, q_{\Ck})$ in \NXk.
Paths in $\Pathset{\Ck}$ may have different
end points however.
For example, $\Pathset{\Ck}=\{\epsilon{}\}$
and $\Pathset{\Ck}\!=\!\{\mathtt{I3}\!\!\rightarrow\!\!\mathtt{I4}\!\!\rightarrow\!\!\mathtt{I7}, \mathtt{I3}\!\!\rightarrow\!\!\UCk{},\mathtt{I3}\!\!\rightarrow\!\!\mathtt{I4}\!\!\rightarrow\!\!\UCk{}\}$
may be potential candidates for $\xi_{\ddAk}\!=\!\mathtt{A9}\!\!\rightarrow\!\!\mathtt{A10}\!\!\rightarrow\!\!\mathtt{A11}\!\!\rightarrow\!\!\mathtt{A9}$ in \cref{fig:example1}.

{\bf If $q^t_{\ddAk}\notin{}\{\UAk{},\WAk{}\}$},
{\em correlatedPathsInCOptions()}
returns candidates, where a candidate
pathset $\Pathset{\Ck}$ is
a maximal set
such that each path $\CPath\in\Pathset{\Ck}$
either
(a) ends at a unique non-error destination cut-point node, say $q^t_{\Ck}$ (i.e., all
paths $\CPath\in{}\Pathset{\Ck}$ ending
at a non-error node end at $q^t_{\Ck}$),
or (b) ends at
error node $\UCk$.
This path enumeration strategy is the same as the one
used in Counter \cite{oopsla20}; this strategy
supports path specializing compiler transformations like
loop peeling, unrolling, splitting, unswitching, etc., but does
not support a path de-specializing transformation like loop re-rolling.
{\bf If $q^t_{\ddAk}=\mathscr{U}_{\ddAk}$},
{\em correlatedPathsInCOptions()}
returns candidates, where a candidate pathset $\Pathset{\Ck}$ is a maximal
set such that
each path $\CPath \in\Pathset{\Ck}$ ends at
$\UCk$.
The algorithm identifies a correlation for a path $\APath=\PathT{q_{\ddAk}}{\WAk{}}$
only after correlations for all other
paths of the
form $\APathN{\cancel{\EW}}=\PathT{q_{\ddAk}}{q^{\cancel{\EW}}_{\ddAk{}}}$ (for $q^{\cancel{\EW}}_{\ddAk{}}\neq{}\WAk{}$)
have been identified:
a pathset candidate $\Pathset{\Ck}$
that has already been correlated with some
other path $\APathN{\cancel{\EW}}$
is then prioritized
for
correlation with $\APath{}$.

For example, in \cref{fig:example1a},
for a cyclic path $\APath$ from a node at {\tt A9} to itself,
one of the candidate pathsets, $\Pathset{\Ck}$, returned
by this procedure (at $\mu{}=1$) contains
eleven
paths originating at {\tt I4} in \cref{fig:example1i}:
one that cycles back to {\tt I4} and
ten
that terminate at \UCk{} (for each of the ten memory accesses in the path).
For example, to evaluate the expression
{\tt v[*i]}, two memory loads are required, one at address {\tt i} and
another at {\tt \&v[*i]}, and each such load may potentially
transition to \UCk{} due to the $\mathtt{accessIsSafeC}_{\tau,a}$
check evaluating to {\tt false} in \TRule{\Load{\Ck}}.
A path that terminates at \UCk{} represents
correlated transitions from node {\tt (A9,I4)} in \Xk{}
such that \ddAk{} remains error-free (to end at {\tt A9}) but
\Ck{} triggers \EU, e.g.,
if the memory access {\tt mem$_4$[esi+4*eax]} in \ddAk{} (corresponding
to {\tt v[*i]} in \Ck{})
overshoots the stack space corresponding
to variable {\tt v} but still lies within the stack region \stk{}.

\subsubsection*{\bf asmAnnotOptions()}
For each simple path $\APath$,
and each (potentially non-simple) path $\CPath$
in $\Pathset{\Ck}$
\footnote{The number of paths can be exponential
in procedure size, and so our implementation
represents a pathset using a series-parallel digraph \cite{oopsla20}
and annotates
a pathset in \ddAk{} in a single
step. Similarly, a pathset in \ddAk{} is
correlated with a pathset in \Ck{} in a single
step. For easier exposition, the presented algorithm
correlates each path individually.},
\textit{asmAnnotOptions()} enumerates the options for annotating
$\xi_{\ddAk}$ with {\tt alloc$_{s,v}$},
{\tt dealloc$_{s,v}$}
instructions
and operands for {\tt call} instructions,
and the {\bf chooseFrom} operator chooses one.

An annotation option
includes the positions and the operands of the (de)allocation instructions
(allocation site, alignment, address, and size).
For a procedure-call, an annotation option also includes the
arguments' types and values,
and the set of callee-observable regions.
The annotations for the callee name/address and the (de)allocations of procedure-call
arguments in $\xi_{\ddAk}$
are uniquely identified using the number and type of
arguments in the candidate correlated path $\xi_{\Ck}$ using the calling
conventions. Similarly, the annotation of callee-observable regions
follows from the regions observable by the correlated procedure call in \CPath{}.

These annotations thus update \Ak{}
to incrementally construct \ddAk{}.
If untrusted compiler hints are
available, they are used to precisely
identify these annotations.
In a blackbox setting, where no compiler
hints are available, we reduce the search space for
annotations (at the cost of reduced generality)
using the following three restrictions:
(1)  An {\tt alloc$_{s,v}$} ({\tt dealloc$_{s,v}$})
annotation is annotated in $\xi_{\ddAk}$ only
if an {\tt alloc} ({\tt dealloc}) instruction is
present
in $\xi_{\Ck}$;
(2) an {\tt alloc$_{s,v}$} ({\tt dealloc$_{s,v}$})
annotation is added only after (before) an instruction that updates
{\tt esp}; moreover, for {\tt alloc$_s$}, {\tt esp} is used as the
local variable's address expression;
(3) for a single allocation site
in \Ck{}, at most one {\tt alloc$_{s,v}$} instruction (but
potentially multiple {\tt dealloc$_{s,v}$} instructions) is
added to \ddAk{}.
Thus, in a blackbox setting, due
to the third restriction, a
refinement proof may fail if the compiler
specializes a path containing a local variable allocation.
Due to the second restriction, a refinement proof may
fail for certain (arguably rare)
types of order-preserving
stack reallocation and stack merging performed by the compiler.
Note that these limitations hold only for the blackbox setting.

\textit{asmAnnotOptions()} returns the (options for)
updated \ddAk{}
and \APath{} as output.
An annotation updates \APath{} by inserting edges
corresponding to the error-free execution of
a $\mathtt{(de)alloc}_{s,v}$ instruction
(recall the graph translations presented in
\cref{fig:xlateRuleAsmStackLocals,fig:xlateRuleAsmAllLocals}).
Thus, \APath{} only covers the error-free execution
of an annotated $\mathtt{(de)alloc}_{s,v}$ instruction and
the remaining error-going paths
of $\mathtt{(de)alloc}_{s,v}$
are not correlated
in this step of \textit{constructX()}.
These error-going paths are correlated
in the second call to \textit{constructX()}
when $phase = \mathtt{CORRELATE\_NEW\_ERROR\_PATHS}$;
because \ddAk{} is already annotated at this point,
\textit{asmAnnotOptions()} is not invoked in
this second call.

After annotations, \APath{} may become a non-simple path due
to the extra I/O instructions introduced by the annotations.
The (potentially non-simple) output path \APath{} is thus broken
into a sequence of constituent paths \APathsP{} 
through \textit{breakIntoSingleIOPaths()}
so that each I/O path appears by itself (and not as a sub-path
of a longer constituent path) in \APathsP{} --- this
caters to
the (SingleIO) requirement.
The same exercise is repeated for (also potentially non-simple) \CPath{} to obtain
\CPathsP{}.
(SingleIO) permits an I/O path to be correlated
with only an I/O path.
However, this may not be possible if one of the paths terminates early
due to error, e.g., if \CPathsP{} has fewer paths than \APathsP{}
because (the last path in) \CPathsP{} ends at \UCk{}
(similarly, if \APathsP{} ends at \WAk{} or \UAk{}).
Recall that our refinement definition does not impose
any requirement on \ddAk{} when \Ck{} terminates with error \EU{},
nor on \Ck{} when \ddAk{} terminates with \EW{}.
Therefore, \textit{trimToMatchPathToErrorNode$(\APathsP{}, \CPathsP{})$} trims the path sequences
\APathsP{} and \CPathsP{} to length of the shorter sequence
if \CPathsP{} ends at \EU{}
or \APathsP{} ends at \EW{}
(otherwise \APathsP{} and \CPathsP{} are returned unmodified).
A failure is returned if the potentially trimmed sequences \APathsN{*} and \CPathsN{*} do not
have similar structures (\textit{haveSimilarStructure()}).
Let
$\mathtt{pos(\xi,\vv{\xi})}$ represent the position
of path $\xi$ in a sequence of paths $\vv{\xi}$.
{\em haveSimilarStructure(\APathsN{*}, \CPathsN{*})} returns true
iff
\APathsN{*} and \CPathsN{*} are of the same size, and
for paths $\CPath' \in \CPathsN{*}$ and
$\APath' \in \APathsN{*}$, if
$\mathtt{pos}(\CPath',\CPathsN{*})=\mathtt{pos}(\APath',\APathsN{*})$, then either both $\CPath'$ and $\APath'$
are I/O paths of same structure
(i.e., they are either both reads or both writes for the same type of value)
or both are I/O free.

\subsubsection*{\bf Incremental Construction of $(\ddAk,\Xk)$}
For each simple path $\xi^{\prime}_{\ddAk}$ in $\vv{\xi}^{\prime}_{\ddAk}$
enumerated in execution order, \toolName{} correlates it with
$\xi^{\prime}_{\Ck}$, such that
$\mathtt{pos}(\xi^{\prime}_{\Ck},\vv{\xi^{\prime}_{\Ck}})=\mathtt{pos}(\xi^{\prime}_{\ddAk},\vv{\xi}^{\prime}_{\ddAk})$ (through {\tt zip} in \cref{algo:constructX}).
This candidate correlation $(\xi^{\prime}_{\ddAk};\xi^{\prime}_{\Ck})$
is checked against a violation of (Similar-speed)
(\textit{addingEdgeWillCreateEmptyCCycle()})
before getting added as an edge $e_{\Xk}$
to $\mathcal{E}_{\Xk}$, adding the destination node to
$\mathcal{N}_{\Xk}$ if not already present.

If $\xi^{\prime}_{\Ck}$ represents
a path between
\Iwrite{\VallocA{\ldots{}}}
and \Iwrite{\VallocB{\ldots{}}}
for an {\tt alloc} instruction in \Ck{},
and $\xi^{\prime}_{\ddAk}$ is a corresponding
path due to an
{\tt alloc$_{s,v}$} instruction,
and edges $e^{\theta_a}_{\Ck}$
and $e^{\theta_m}_{\Ck}$
in $\xi^{\prime}_{\Ck}$
are labeled with
instructions $\alpha_{b} \Assign \Ichoose{\bv{32}}$
and $\IchooseM$ respectively
due to \TRule{Alloc},
we add mappings $\DXk(e_{\Xk}, e^{\theta_a}, 1)=v$ and
$\DXk(e_{\Xk}, e^{\theta_m}, 1)=M_{\ddAk}$,
where
$v$ is the address defined in $\xi^{\prime}_{\ddAk}$ due to either
\TRule{AllocS} or \TRule{AllocV} ({\em addDetMappings}~($e_{\Xk}$)).
Notice that our algorithm
only populates $\DXk{}(e_{\Xk}, e^{\theta}_{\Ck}, n)$
for $n=1$, even though \cref{sec:detX} defines \DXk{} more generally.

If the destination node is not an error node, then
the {\em inferInvariantsAndCounterexamples()} procedure
updates the invariant network \IXk{} due to the addition
of this new edge.
The non-coverage requirements are checked
after invariant inference ({\em checkSemanticReqsExceptCoverage})
and a candidate is discarded if the check fails.

When all simple paths between
the cut points of \ddAk{} are exhausted,
the (Coverage\ddAk{}) requirement must
be satisfied by construction. {\em checkCoverageRequirements()} further
checks the satisfaction of (Coverage\Ck{}) before returning {\tt Success}.
\toolName{} is sound because it returns {\tt Success} only if all the
thirteen search-algorithm requirements are satisfied.

The {\bf chooseFrom}
operator must attempt to maximize the chances
of returning {\tt Success}, even if only a fraction of
the search space has been explored.
\toolName{} uses the
counterexamples generated when
a proof obligation is falsified (e.g., during invariant
inference)
to guide
the search towards the more promising options.
A counterexample is a proxy
for the machine states of \Ck{} and \ddAk{}
that may appear at a node $n_{\Xk}$ during the
lockstep execution encoded by \Xk{}.
Thus, if at any step during the construction
of \Xk{},
the execution of a counterexample for a
candidate partial
solution $(\ddAk,\Xk)$ results
in
the violation of a non-coverage requirement,
that candidate
is discarded. Further, counterexample execution
opportunistically weakens the node invariants in \Xk{}.
Like Counter,
we use
the number of live registers in \ddAk{} related
through the current invariants in \IXk{} to
rank the enumerated partial candidate solutions to implement a
best-first search.

\subsection{Invariant Inference}
\label{sec:invInferFormal}

We use
a counterexample-guided
inference algorithm
to identify node invariants \cite{oopsla20}.
Candidate invariants at a
node $n_{\Xk}$ of a partial product-graph
are formed by conjuncting
predicates drawn from the grammar shown in \cref{fig:invGrammar}.
Apart from affine (\pred{affine})
and inequality relations (\pred{ineq} and \pred{ineqC})
for relating values across \Ck{} and \ddAk{},
the guesses attempt
to equate the allocation and memory state
of common regions
across
the two procedures (\pred{AllocEq} and \pred{MemEq}).

\begin{figure}[t]
\begin{footnotesize}
  \setlength{\extrarowheight}{1.8pt}
  \begin{tabular}{@{}l@{}}
\pred{affine} $\sum_{i}{c_iv_i}=c$
\qquad \qquad
\pred{ineqC} $\pm v\le_s 2^c$ 
\qquad \qquad 
\pred{ineq} $v_1 \odot v_2$
\qquad \qquad
\pred{spOrd} $\ghost{\spV{\PC{\ddAk}{j_1}}} \le_u (\ghost{\spV{\PC{\ddAk}{j_2}}} - v)$
\\
\pred{AllocEq} $\forall_{r \in \NS{}} \il{\Ck}{r} = \il{\ddAk}{r}$
\qquad \qquad
\pred{MemEq} $\prjMEq{\il{\ddAk}{\NS{}} \setminus (\ilZv{\ddAk})}{M_{\Ck}}{M_{\ddAk}}$
\qquad \qquad
\pred{zEmpty} $\{\il{\Ck}{\z},\ilzs{\ddAk}{\z},\ilzv{\ddAk}{\z}\} \mathbin{\{ =, \ne \}} \emptyset$ 
\\
\pred{spzBd} $\ghost{\Empty{\z}} \lor (\ghost{\spV{\PC{\ddAk}{j}}} \odot \{ \ghost{\LB{\z}}, \ghost{\UB{\z}} \})$
\quad
\pred{spzBd'} $\ghost{\Empty{\z}} \lor (\ghost{\spV{\PC{\ddAk}{j}}} \le_u (\ghost{\LB{\z}} - \ghost{\LSz{\z}}))$
\\

\midrule

\apred{gfySz} $\forall_{r \in G \cup \F \cup \Y \setminus \{ \yv \}}(\ghost{\Sz{r}} = \mathtt{sz}(\mathtt{T}(r)))$
\ 
\apred{\yv{}Sz} $(\ghost{\Empty{\yv}} \Leftrightarrow \ghost{\Sz{\yv}} = 0)$
\ 
\apred{Empty} $\forall_{r \in G \cup \F \cup \Y \cup \Z} (\il{\ddAk}{r} = \emptyset \Leftrightarrow \ghost{\Empty{r}})$
\\
\apred{gfyIntvl} $\forall_{r \in G \cup \F \cup \Y}((\ghost{\Sz{r}} = 0) \lor ((\ghost{\LB{r}} \leq_u \ghost{\UB{r}}) \land (\ghost{\UB{r}} = \ghost{\LB{r}}+\ghost{\Sz{r}}-1_{\bv{32}}) \land ([\ghost{\LB{r}},\, \ghost{\UB{r}}] = \il{\ddAk}{r})))$
\\
\apred{zlIntvl} $\ghost{\Empty{\zl}} \lor ((\ghost{\LB{\zl}} \leq_u \ghost{\UB{\zl}}) \land (\ghost{\LB{\zl}}+\ghost{\LSz{\zl}} - 1_{\bv{32}} = \ghost{\UB{\zl}}) \land ([\ghost{\LB{\zl}}, \ghost{\LB{\zl}}] = \il{\Ck}{\zl})))$
\\
\apred{zaBd} $\ghost{\Empty{\za}} \lor ((\ghost{\LB{\za}} \leq_u \ghost{\UB{\za}}) \land (\ghost{\LB{\za}}+\ghost{\LSz{\za}} - 1_{\bv{32}} \le_u \ghost{\UB{\za}}) \land (\ghost{\LB{\za}} = \LBr{\il{\Ck}{\za}} \land \ghost{\UB{\za}} =\UBr{\il{\Ck}{\za}}))$
\\
\apred{StkBd} $\il{\ddAk}{\{\stk\}\cup{}\Y}\cup{}(\il{\ddAk}{\Z{}}\setminus{}(\ilZv{\ddAk})) = [{\tt esp}, \ghost{\spE}]$
\quad
\apred{$cs$Bd} $\il{\ddAk}{\{cs, cl\}} = [\ghost{\spE}+1,\ghost{\stkE}]$
\\
\apred{NoOverlap\Ck{}} $\neg \Overlap(\il{\ddAk}{\heap},\il{\ddAk}{cl},\il{\ddAk}{\yv},\ldots,\ii{\ddAk}{\g},\ldots,\ii{\ddAk}{\y},\ldots,\il{\ddAk}{\z})$
\quad
\apred{ROM\Ck} $\forall_{r \in \Gro}{\prjMEq{\ii{\Ck}{r}}{M_{\Ck}}{\roMem{\Ck}{r}{\ii{\Ck}{r}}}}$
\\
\apred{NoOverlap\Ak{}} $\neg \Overlap(\il{\ddAk}{\{\heap, cl\} \cup G \cup \Y},\ldots,\ilzs{\ddAk}{\z},\ldots,\il{\ddAk}{\f},\ldots,\il{\ddAk}{\stk},\il{\ddAk}{cs})$
\quad
\apred{ROM{\Ak}} $\forall_{r \in \Fro}{\prjMEq{\ii{\ddAk}{r}}{M_{\ddAk}}{\roMem{\ddAk}{r}{\ii{\ddAk}{r}}}}$
\\
\end{tabular}
\caption{\label{fig:invGrammar} Predicate grammar for constructing candidate invariants.  $v$ represents a bitvector variable (registers, stack slots, and ghost variables), $c$ represents a bitvector constant.  $\odot \in \{\leq_{s,u},<_{s,u},>_{s,u},\geq_{s,u}\}$.
}
\end{footnotesize}
\end{figure}

Recall that we save stackpointer value at the boundary of a stackpointer updating instruction
at PC $\PC{\ddAk}{j}$
in ghost variable \ghost{\spV{\PC{\ddAk}{j}}}
(\TRule{Op-esp} in \cref{fig:xlateRuleAsm}).
To
prove separation between different
local variables, we require invariants that
lower-bound the gap between two ghost variables,
say \ghost{\spV{\PC{\ddAk}{j_1}}} and \ghost{\spV{\PC{\ddAk}{j_2}}},
by
some value $v$ that depends on the allocation size operand of an
{\tt alloc}$_{s}$ instruction (\pred{spOrd}).
To capture the various relations between lower bounds, upper
bounds, region sizes, and \ghost{\spV{\PC{\ddAk}{j}}}, the guessing grammar
includes shapes \pred{spzBd} and \pred{spzBd'} that are of the form:
``either a local variable region is empty
or its bounds are related to \ghost{\spV{\PC{\ddAk}{j}}} in these possible ways''.
\pred{zEmpty} tracks the emptiness of the address-set of a local region.
Together, these predicate shapes
(along with \pred{affine} and \pred{ineq} relations between \ghost{\spV{\PC{\ddAk}{j}}})
enable disambiguation between stack writes involving spilled pseudo-registers
and stack-allocated locals.

The predicate shapes listed below the dividing line segment in \cref{fig:invGrammar}
encode the {\em global invariants} that hold by construction (due to our execution
semantics) at
every non-error product-graph node $n_{\Xk}$.
\apred{gfySz}, \apred{\yv{}Sz}, and \apred{gfyIntvl} together encode the fact
that the ghost variables associated with a region $r \in G \cup \F \cup \Y$
track its bounds, size, and that the address set of $r$ is an interval.
\apred{Empty} encodes that the ghost variable \ghost{\Empty{r}}
for $r \in G \cup \F \cup \Y \cup \Z$
tracks the emptiness of the region $r$.
\apred{zlIntvl} captures the property that a local variable region $\zl$,
if non-empty,
must be an interval of size $\ghost{\LSz{\zl}}$.
\apred{zaBd} captures a weaker property for a local region \za{} (allocated using {\tt alloca()}):
if non-empty, this region must be bounded by its ghost variables
and the region must be at least $\ghost{\LSz{\za}}$ bytes large.
\apred{StkBd} encodes the invariant that the interval
$[\mathtt{esp},\ghost{\spE}]$
represents the union of the address
sets of \stk{}, regions in \Y{}, and stack-allocated local regions ($\il{\ddAk}{\Z} \setminus (\ilZv{\ddAk})$);
\apred{$cs$Bd} is similarly shaped and encodes that the interval
$[\ghost{\spE}+1, \ghost{\stkE}]$ represents the union
of the address sets of regions $cs$ and $cl$.
\apred{NoOverlap\Ck{}} encodes the disjointedness of all regions $r \in \NS{}$.
\apred{NoOverlap\Ak{}} encodes the disjointedness of all regions
in \ddAk{} except
virtually-allocated regions.
Finally, \apred{ROM\Ck} and \apred{ROM\Ak} encode the preservation
of memory contents
of read-only regions in \Ck{} and \ddAk{}.

A dataflow analysis \cite{andersen94programanalysis} computes the possible states of \BasedOn{} and \BasedOnM{} maps
at each $n_{\Ck{}}\in{}\NCk{}$, and the over-approximate solution is added to $\phi_{n_{\Xk{}}}$ for each $n_{\Xk{}}=(\_,n_{\Ck{}})$.

\subsection{SMT Encoding}
\label{sec:smtEncoding}

At a non-error node $n_{\Xk}$,
a proof obligation
is represented as a
first-order logic predicate over the state elements
at $n_{\Xk}$ and
discharged
using an SMT solver.
The machine states of \Ck{} and \ddAk{} are represented
using bitvectors (for a register/variable), arrays (for memory),
and uninterpreted functions (for $\mathtt{read}_{\vv{\tau}}(\World{P})$ and
\Iio{\vv{v}}{rw}).
For address sets,
we encode the
set-membership predicate $\alpha \in \il{P}{r}$ for an arbitrary address $\alpha$,
region identifier $r$, and procedure $P\in{}\{\Ck,\ddAk{}\}$.
All other address set operations can be expressed in terms
of the set-membership
predicate (\cref{app:derived_encoding}). To simplify
the encodings, we rely on the correct-by-construction
invariants in \cref{fig:invGrammar} and assume that
$\phi_{n_{\Xk}}$ satisfies the (Equivalence), (MAC),
and (MemEq) requirements.
Notice that (Equivalence) implies \pred{AllocEq}.

Recall that for $\z \in \Zl$, at a node $n_{\Xk}\in{}\NXk{}$,
\ilzs{\ddAk}{\z} and \ilzv{\ddAk}{\z} represent the
address sets corresponding to the
stack and virtual allocations performed in \ddAk{}
for \z{}.
Let
$\Z{ls} = \{ \z\ |\ \z\in{}\Zl\land{}\ilzs{\ddAk}{\z} \ne \emptyset \}$
and
$\Z{lv} = \{ \z\ |\ \z\in{}\Zl\land{}\ilzv{\ddAk}{\z} \ne \emptyset \}$
represent the set of stack-allocated locals and virtually-allocated
at $n_{\Xk}$ respectively.
Recall that we restrict ourselves to only
those compiler
transformations that ensure the validity of
$\Z{ls}\cap{}\Z{lv}=\emptyset{}$ at each $n_{\Xk}$ (\cref{sec:refnDefnVirtual}).

\subsubsection{Representing Address-Sets Using Allocation State Array}

Let $\MemallocP{P}: \bv{32} \rightarrow{} \Rall{}$ be
an {\em allocation state array} that
maps an address to a region identifier in procedure $P$.
For $r\notin{}\Z{lv}$,
$\alpha \in \il{P}{r}$
is encoded as $\selectMath_1(\MemallocP{P}, \alpha) = r$.
Allocation of an address $\alpha$ to region $r$
($\il{P}{r} \Assign{} \il{P}{r} \cup \{ \alpha \}$)
is encoded as $\MemallocP{P} \Assign \store_1(\MemallocP{P}, \alpha, r)$.
Similarly, deallocation
($\il{P}{r} \Assign{} \il{P}{r} \setminus \{ \alpha \}$)
is encoded as $\MemallocP{P} \Assign \store_1(\MemallocP{P}, \alpha, \free{})$.

For $\z{lv}\in{}\Z{lv}$,
both $\alpha \in \il{\Ck}{\z{lv}}$ 
and
$\alpha \in \il{\ddAk}{\z{lv}}$
are encoded as $\selectMath_1(\MemallocP{\Ck}, \alpha) = \z{lv}$, i.e.,
the set-membership encodings for both procedures
use \MemallocP{\Ck} for virtually-allocated locals
(by relying on the \pred{AllocEq} invariant at $n_{\Xk}$).
In other words, \MemallocP{\ddAk} is not used to track the virtually-allocated
locals; instead, an address belonging to a virtually allocated-region maps
to one of $\{\mathtt{free}, \stk{}, cs\} \cup \F$ regions in \MemallocP{\ddAk}.
Consequently, the (de)allocation instructions
$\ilzv{\ddAk}{\z{lv}} \Assign \ilzv{\ddAk}{\z{lv}} \cup [v]_{w}$
and 
$\ilzv{\ddAk}{\z{lv}} \Assign \emptyset$
are vacuous in \ddAk{}, i.e., they do not change any state
element in \ddAk{} (\cref{fig:xlateRuleAsmAllLocals}).

This encoding, based on allocation state arrays $\MemallocP{\Ck{}}$
and $\MemallocP{\ddAk{}}$,
is called the {\em full-array encoding}.
The second and third columns of \cref{tab:intervalEncodings} describe
the full-array encoding for $P=\Ck$ and $P=\ddAk$.  In the table, we use \pred{AllocEq} to
replace $\selectMath_1(\MemallocP{\ddAk}, \alpha)$
with $\selectMath_1(\MemallocP{\Ck}, \alpha)$ for $r\in{}{\NS{}}$.
For example, in the full-array
encoding, the (MemEq) requirement $\prjMEq{\il{\ddAk}{\NS{}} \setminus (\ilZv{\ddAk})}{M_{\Ck}}{M_{\ddAk}}$
becomes
$\forall{\alpha{}}:{((\selectMath_1(\MemallocP{\Ck}, \alpha) \in{} G\cup{}\{\heap,cl\}\cup{}\Y}\cup{}\Z{ls}\cup\Za)\Rightarrow{}(\selectMath_1(M_{\Ck}, \alpha)=\selectMath_1(M_{\ddAk}, \alpha)))$.

\subsubsection{Interval Encodings for $r\in G\cup\F\cup\Y\cup\Zl\cup\{\stk\}$}

We use \apred{gfyIntvl}, \apred{zlIntvl}, and
\pred{AllocEq} invariants for
a more efficient {\em interval encoding}: for $r\in{}G \cup \F \cup \Y \cup \Zl$, we 
encode $\alpha \in \il{P}{r}$ as
$\neg \ghost{\Empty{r}} \land (\ghost{\LB{r}} \leq_u \alpha \leq_u \ghost{\UB{r}})$.
Moreover, if there are no local variables allocated due to the {\tt alloca()} operator (i.e.,
$\il{P}{\Za}=\emptyset$), then all local variables are contiguous, and
so, due to \apred{StkBd}, the \stk{} region can be identified
as $[\mathtt{esp},\ghost{\spE}]\setminus{}\il{\ddAk{}}{\Y\cup{}\Z{ls}}$ --- the
corresponding interval encoding is shown in the
right-most cell of $r=\stk{}$ row in \cref{tab:intervalEncodings}.

\subsubsection{Interval Encodings for $r\in{}\{\heap{},cl,cs\}$}
\label{sec:intervalEncodingsHeapCl}
Even though $\heap{},cl,cs$ can be discontiguous regions in general,
we over-approximate these regions to their contiguous
covers to be able to soundly encode them using intervals.
At a node $n_{\Xk}=(n_{\ddAk},n_{\Ck})$, \toolName{} may generate
a proof obligation $\proofObligation{}$ of the form
$\hoareTriple{pre}{(\APath{};\CPathD{})}{post}$
--- recall that path-cover and path-infeasibility
conditions are also represented as Hoare triples with $\CPath{}=\epsilon{}$.
If $\APath$ is an I/O path, its execution
interacts with the outside world, and so an over-approximation
of an externally-visible address set is unsound.
We thus restrict our attention to an I/O-free
$\APath$ for interval encoding.

Let $n^1_{\ddAk{}}, n^2_{\ddAk{}}, \ldots, n^m_{\ddAk{}}$
be the nodes on
path $\xi_{\ddAk{}}=(n_{\ddAk}\twoheadrightarrow{}n^t_{\ddAk})$,
such that $n^1_{\ddAk{}}=n_{\ddAk{}}$ and $n^m_{\ddAk{}}=n^t_{\ddAk{}}$.
Let $\espmin{}(\APath)$ represent the the
minimum value of $\mathtt{esp}$ observed at any
node $n^j_{\ddAk}$ ($1\leq{}j\leq{}m$) visited
during the execution of path $\xi_{\ddAk{}}$.
Similarly, let
$\zlunion(\APath)$ be the union of the values
of set $\il{\ddAk{}}{\Z{lv}}$ observed at any $n^j_{\ddAk}$ ($1\leq{}j\leq{}m$)
visited during $\xi_{\ddAk{}}$'s execution.

Let $\HP={\tt comp}(\il{\ddAk}{G\cup\F}\cup\zlunion(\APath)\cup[\espmin(\APath),\ghost{\stkE}])$,
$\CL=[\ghost{\spE{}}+1_{\bv{32}},\ghost{\stkE{}}]\setminus \zlunion(\APath)$,
and
$\CS=[\ghost{\spE{}}+1_{\bv{32}},\ghost{\stkE{}}]\cap \zlunion(\APath)$.

\begin{table}
  \caption{\label{tab:intervalEncodings} SMT encoding of $\alpha \in \protect\il{P}{r}$ for \toolName{}'s proof obligation $\proofObligation$ with outgoing assembly path $\APath$.}
  \begin{scriptsize}
  \setlength{\extrarowheight}{1pt}
    \begin{tabularx}{\textwidth}{|p{1.7cm}|p{.5cm}|p{.5cm}|p{3.8cm}|X|}
  \hline
    \multirow{2}{=}{$\alpha \in \il{P}{r}$}
  & \multicolumn{2}{c|}{Full-array encoding}
  & \multirow{2}{=}{Partial-interval encoding ($\il{P}{\Za}\ne{}\emptyset$)}
  & \multirow{2}{=}{\hfill Full-interval encoding ($\il{P}{\Za}=\emptyset$) \hfill}
  \\
  \cline{2-3}
  & \multicolumn{1}{c|}{$P = \Ck$} & \multicolumn{1}{c|}{$P = \Ak$} &  &
  \\

  \hline
  \hline

  $r = \heap$
      	    & \multicolumn{2}{c|}{}
            & \multicolumn{2}{c|}{
              $\alpha{}\notin{}(\il{\ddAk}{G\cup\F}\cup\zlunion(\APath)\cup{}[\espmin{}(\APath),\ghost{\stkE}])$
              }
            \\

  \cline{1-1}\cline{4-5}

  $r = cl$
      	    & \multicolumn{2}{c|}{$\selectMath_1(\MemallocP{\Ck}, \alpha) = r$}
            & \multicolumn{2}{c|}{
              $
              \alpha \in [\ghost{\spE}+1, \ghost{\stkE}]
              \land
              \alpha \notin \zlunion(\APath)
              $
              }
            \\

  \cline{1-1}\cline{4-5}

  $r \in G \cup \Z{lv}$
      	    & \multicolumn{2}{c|}{}
            & \multicolumn{2}{c|}{}
            \\
  \cline{1-1}\cline{4-4}

  $r \in \Y \cup \Za \cup \Z{ls}$
      	    & \multicolumn{2}{c}{}
      	    & 
            & $\neg \ghost{\Empty{r}} \land (\ghost{\LB{r}} \le_u \alpha \le_u \ghost{\UB{r}})$
            \\
  \cline{1-4}

  $r \in \F$
          &
          & 
          & \multicolumn{2}{c|}{}
          \\

  \cline{1-1}\cline{4-5}

  $r = cs$
      	  & \texttt{false}
      	  &
      	  & \multicolumn{2}{c|}{
      	    $
      	    \alpha \in [\ghost{\spE}+1, \ghost{\stkE}]
      	    \land
            \alpha \in \zlunion(\APath)
            $
      	  }
          \\

  \cline{1-1}\cline{4-5}

  $r = \stk$
      	    &
      	    & \multicolumn{2}{c|}{$\selectMath_1(\MemallocP{\ddAk}, \alpha) = r$}
            & $\alpha \in [\espmin{}(\APath), \ghost{\spE}] \land \bigwedge_{r \in \Y \cup \Z{ls}} (\alpha \notin \il{\ddAk}{r})$
            \\

  \hline
  \end{tabularx}
  \end{scriptsize}
\end{table}

\begin{theorem}
\label{theorem:rewriteHeapCl}
Let $\proofObligation{}=\hoareTriple{pre}{(\APath{};\CPathD{})}{post}$
be a proof obligation generated by \toolName{}.
Let $\proofObligation{}^{\prime}$ be obtained from $\proofObligation{}$ by
strengthening
precondition $pre$
to $pre^{\prime}=pre
\land
(\il{\ddAk}{\heap}=\HP)
\land
(\il{\ddAk}{cl}=\CL)
\land
(\il{\ddAk}{cs} = \CS)$.
If $\APath{}$ is I/O-free,
$\proofObligation{}\Leftrightarrow{}\proofObligation{}^{\prime}$ holds.
\end{theorem}
\begin{proof}[Proof sketch]
$\proofObligation{}\Rightarrow{}\proofObligation{}^{\prime}$ is trivial.
The proof for
$\proofObligation{}^{\prime}\Rightarrow{}\proofObligation{}$,
available in \cref{app:intervalEncodingProof}, relies on the limited shapes of predicates that
may appear in $pre$, $post$ --- for I/O-free $\xi_{\ddAk}$, these
shapes are limited by our invariant grammar (\cref{fig:invGrammar}),
and the edge conditions
appearing in our execution semantics
(\cref{fig:xlateRuleIR,fig:xlateRuleAsm,fig:xlateRuleAsmStackLocals,fig:xlateRuleAsmAllLocals}).
The proof holds only if the
safety-relaxed semantics are used for \ddAk{}.
\end{proof}

Using \cref{theorem:rewriteHeapCl},
we rewrite $\alpha \in \il{P}{\heap}$
to $\alpha \in \HP$,
$\alpha \in \il{P}{cl}$
to $\alpha \in \CL$,
and
$\alpha \in \il{P}{cs}$
to $\alpha \in \CS$
in proof obligation $\proofObligation{}$.
As shown in \cref{tab:intervalEncodings},
if $\il{P}{\Za}=\emptyset$ holds at $n_{\Xk}$, we encode all non-free regions
using intervals (called {\em full-interval encoding}); else,
we encode regions in $Y\cup\Za\cup\Z{ls}\cup\{\stk{}\}$ using an allocation state array,
and $G\cup\F\cup\Z{lv}\cup\{\heap,cl,cs\}$ using intervals (called {\em partial-interval encoding}).

\section{Experiments}
\label{sec:experiments}
\toolName{}
uses four SMT solvers running in parallel for discharging
proof obligations: {\tt z3-4.8.7}, {\tt z3-4.8.14},
{\tt Yices2-45e38fc},
and {\tt cvc4-1.7}. Unless otherwise specified, we use $\mu=64$,
a timeout of ten minutes for an SMT query, and
a timeout of eight hours for a
refinement check.

Before checking refinement, if
the address of a local variable $l$
is never taken in \Ck{}, we transform
\Ck{} to register-allocate $l$ (LLVM's
{\tt mem2reg}). This
reduces the proof effort, at the cost of having
to trust the pseudo-register allocation logic.
{\tt mem2reg} does not register-allocate local arrays and structs in \ourIR{},
even though an optimizing compiler may register-allocate them
in assembly --- virtual allocations help validate such translations.

\begin{table}
\caption{\label{tab:benchmarks}Benchmarks and their programming patterns. $N$ in \texttt{vil}$N$ is substituted to obtain {\tt vil1}, {\tt vil2}, and {\tt vil3}.  Program listings available in \cref{app:benchmarks_source_code}.}
\begin{scriptsize}
\begin{tabular}{l|l}
\hline
  Name  & Programming pattern \\
\hline
  ats   & \inv{Address-taken local scalar} \fbox{\tt int ats() \{ int ret; foo(\&ret); return ret; \}} \\
  atc   & \inv{Address taken conditionally} \fbox{\tt int atc(int* p) \{ int x; if (!p) p = \&x; foo(p); return *p \}}\\
  ata   & \inv{Local array} \fbox{\tt int ata() \{ char ret[8]; foo(ret); return bar(ret, 0, 16); \}}\\
  vwl   & \inv{Variadic procedure} \fbox{{\tt int vwl(int n, ...) \{ va\_list a; va\_start(a, n);} {\tt for(...)\{/* read va\_arg(a,int) */\}...\}}}\\
  as    & \inv{GCC {\tt alloca()}} \fbox{\tt int as(int n)\{...int* p=alloca(n*sizeof(n)); for(...)\{/*write to p*/\}...\}}\\
  vsl   & \inv{VLA with loop} \fbox{\tt int vsl(int n)\{... int v[n]; for(...)\{/*write to v*/\}...\}} \\
  vcu   & \inv{VLA conditional use} \fbox{\tt int vcu(int n,int k)\{ int a[n]; if (...) \{ /*rd/wr to a*/\}...\}}\\
  min   & \inv{{\tt minprintf} procedure from K\&R \cite{knr}} \\
  ac    & \inv{{\tt alloca()} conditional use} \fbox{\tt int ac(char*a) \{..if (!a) a=alloca(n); for(...)/*r/w to a*/\}}\\
  all   & \inv{
            \begin{tabular}{@{}l@{}}
            {\tt alloca()} in a loop\\
            to form a linked list
            \end{tabular}}
          \fbox{
            \begin{tabular}{@{}l@{}}
            {\tt all()\{..hd=NULL; for(...)\{..n=alloca(..);..n->nxt=hd; hd=n;\}}\\
            {\tt \ \ \ \ \ \ \ \ while(...)\{/* traverse the list starting at hd */\}\}}\\
            \end{tabular}
          }
\\
  atail & \inv{Local array alloc. in loop} \fbox{\tt int atail(..)\{..for(..)\{ char a[4096]; f(a..); b(a..);...\}...\}}\\
  vil$N$   & \inv{$N$ VLA(s) in a loop} \fbox{\tt int vil$N$(..)\{..for(i=1;i<n;++i) \{ int v1[4*i], v2[4*i], ... v$N$[4*i]; foo$N$(...); ...\}}\\
  vilcc & \inv{VLA in loop with {\tt continue}} \fbox{\tt int vilcc(..)\{..while(i<n)\{ char v[i];...if(..) continue;..\}..\}}\\
  fib   & \inv{Program from \cref{fig:example1}} \\
  vilce & \inv{VLA in loop with {\tt break}} \fbox{\tt int vilce(..)\{..while(i<n)\{ char v[i];...if(..) break;..\}..\}}\\
  rod   & \inv{A local char array initialized using a string and a VLA and a for loop} {\tt Available in \cref{app:benchmarks_source_code}.}\\
\end{tabular}
\end{scriptsize}
\end{table}

We first evaluate the efficacy of our implementation
to handle the diverse
programming patterns seen with local allocations (\cref{tab:benchmarks}).
These include variadic procedures,
VLAs allocated in loops, {\tt alloca()} in loops, etc.
\Cref{fig:graph_lt} shows the results of our experiments for these
18
programming patterns from \cref{tab:benchmarks}
and three compilers, namely
Clang/LLVM v12.0.0, GCC v8.4.0, and ICC v2021.8.0,
to generate 32-bit x86
executables
at {\tt -O3} optimization with inter-procedural analyses disabled using the compilers' command-line flags.
The X-axis lists the benchmarks and
the Y-axis represents the total time taken in seconds (log scale) for
a refinement check ---
to study the performance
implications,
we run a check with all three encodings for these benchmarks.
The filled and empty bars represent the time taken with full-interval and partial-interval SMT encodings respectively. The figure does not show the
results for the full-array encoding.
A missing bar represents a failure to compute the proof.
Of 54 procedure pairs, our implementation is able to check refinement
for 45, 43, and 37
pairs while using
full-interval, partial-interval, and full-array encodings respectively.
For benchmarks where a refinement check succeeds
for all encodings,
the full-interval encoding performs 1.7-2.2x and 3.5-4.9x
faster on average (for each compiler)
than the partial-interval and full-array encodings respectively.
The reasons for nine failures are:
(a) limitation of the blackbox annotation algorithm for one procedure-pair;
(b) incompleteness of invariant inference for six procedure-pairs  (e.g., requirement of non-affine invariants, choice of program variables); and
(c) SMT solver timeouts for two procedure-pairs.
{\tt vilcc} and {\tt vilce}
require multiple {\tt dealloc}$_s$ instructions to be added to \Ak{} for a
single {\tt dealloc} in \Ck{}.
An {\tt alloc$_v$} annotation
is required for the `{\tt va\_list a}' variable in the GCC and ICC
compilations of {\tt vwl}
(see \cref{tab:benchmarks}) ---
while GCC and ICC register-allocate {\tt a}, it is allocated
in memory using {\tt alloc} in \ourIR{} (even after {\tt mem2reg}).
The average number of
best-first search backtrackings across all
benchmarks is only 2.8.
The time spent in constructing
the correct product graph forms
around 70-80\% of the total search time.

\begin{figure}
\begin{subfigure}[b]{.48\textwidth}
  \includegraphics[width=\textwidth]{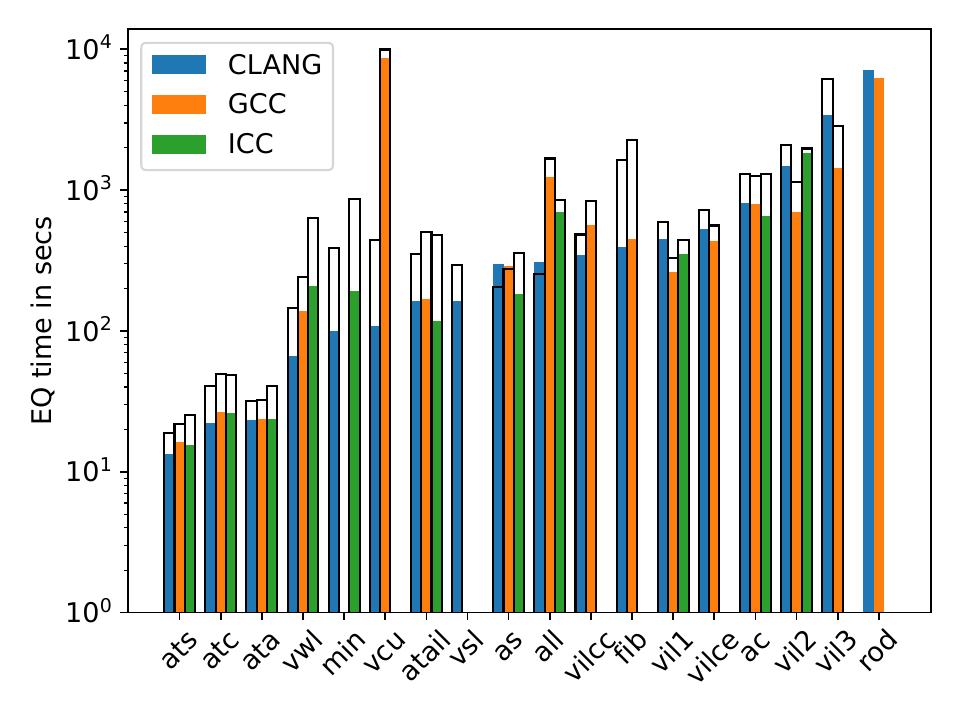}
\caption{\label{fig:graph_lt}Comparison of running times with full- (filled bars) and partial- (empty bars) interval encoding.}
\end{subfigure}
\begin{subfigure}[b]{.50\textwidth}
  \includegraphics[width=\textwidth]{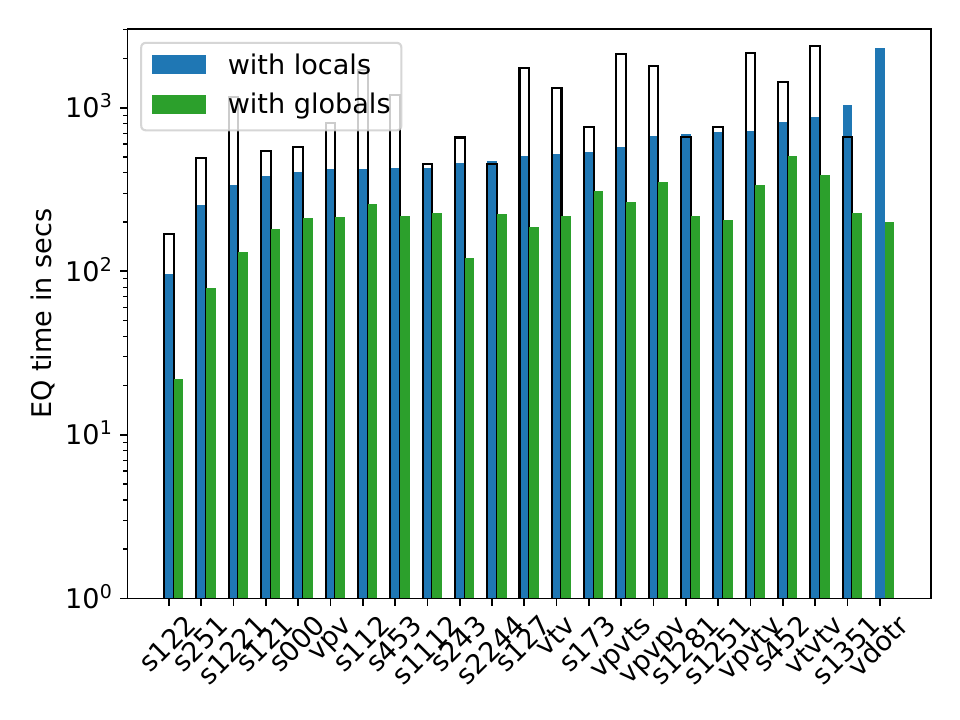}
\caption{\label{fig:graph_tsvc}Comparison of running times of benchmarks with exactly same code modulo allocation.}
\end{subfigure}
\caption{\label{fig:graphs} Experiments with procedures in \cref{tab:benchmarks} and TSVC.  Y-axis is logarithmically scaled.}
\end{figure}

We next evaluate
\toolName{} on the TSVC suite
of vectorization benchmarks with arrays
and loops \cite{tsvc}, also used in previous work
\cite{semalign,oopsla20}.
We use two versions of these
benchmarks: (1) `{\tt globals}' where global variables are
used for storing the output array values, and (2) `{\tt locals}' where local
array variables are used for storing the output
values and a procedure call is
added at the end of the procedure body to print the contents
of the local array variables. The compiler
performs the same vectorizing transformations on
both versions. Unlike {\tt globals},
{\tt locals} additionally
requires the automatic identification of required annotations.

\Cref{fig:graph_tsvc} shows the execution
times of \toolName{} for validating
the compilations produced by
Clang/LLVM v12.0.0 (at {\tt -O3})
for these two versions of the TSVC benchmarks. \toolName{}
can successfully validate these compilations.
Compared to {\tt globals},
refinement checks are 2.5x slower for {\tt locals} (on average) due to
the extra overhead of identifying the required annotations. 

\begin{table}[t]
\caption{\label{tab:bzip2}Statistics obtained by running \toolName{} on procedures in the {\tt bzip2} program.} 
\begin{scriptsize}
\centering
\begin{tabularx}{\textwidth}{@{}p{3.0cm}ccccccccccccccc@{}}
\toprule
Name                                       & {\tt SLOC} & {\tt ALOC} & \#$_{al}$ & \#$_{loop}$ & \#$_{fcall}$ & {\tt D} & {\tt eqT} & Nodes  & Edges  & {\tt EXP} & {\tt BT} & \#$_{q}$ & Avg. {\tt qT} \\
\midrule                                                                                                                   
{\tt generateMTFValues}                    & 76         & 144        & 1         & 6           &   1          & 2       & 4k        & 14     & 30     & 60        & 16       & 3860     & 0.56   \\
{\tt recvDecodingTables}                   & 70         & 199        & 2         & 14          &  10          & 3       & 3k        & 38     & 66     & 102       & 15       & 5611     & 0.21     \\
{\tt undoReversible\-Transformation\_fast} & 116        & 221        & 1         & 7           &   6          & 2       & 2k        & 21     & 34     & 43        & 6        & 2998     & 0.23    \\
\bottomrule
\end{tabularx}
\end{scriptsize}
\end{table}

Our third experiment is on
SPEC CPU2000's {\tt bzip2}\cite{spec:cint2000} program compiled
using Clang/LLVM v12.0.0 at three optimization levels: {\tt O1},
{\tt O2}, and {\tt O1-}.  {\tt O1-} is a custom optimization
level configured by us that enables all optimizations at {\tt O1} {\em except}
(a) merging of multiple procedure calls on different paths into a single call,
(b) early-CSE (common subexpression elimination), (c) loop-invariant code motion
at both LLVM IR and Machine IR, (d) dead-argument elimination, (e) inter-procedural
sparse conditional constant propagation, and (f) dead-code elimination of
procedure calls. {\tt bzip2} runs 2\% slower with {\tt O1-} than with {\tt O1};
this is still 5\% faster than the executable produced by CompCert, for example.
Of all 72 procedures in {\tt bzip2}, \toolName{} successfully
validates the translations for 64, 60, and 54 procedures at {\tt O1-}, {\tt O1},
and {\tt O2} respectively at $\mu=1$. At {\tt O1-}, \toolName{} takes around
six CPU hours to compute refinement proofs for the 64 procedures. \toolName{}
times out
for the remaining eight procedures, all of which are bigger than 190 ALOC.

Three of {\tt bzip2}'s procedures for which refinement proofs are successfully
computed at both {\tt O1-} and {\tt O1} contain at least one local array, and
\cref{tab:bzip2}
presents statistics for the {\tt O1-} validation experiments for these
procedures.
For each procedure, we show
the number of source lines of code in \Ck{} ({\tt SLOC}), the
number of assembly instructions in \Ak{} ({\tt ALOC}), the
number of local variables ({\tt \#$_{al}$}), the number
of loops ({\tt \#$_{loop}$}), the number of procedure calls
({\tt \#$_{fcall}$}),
and the maximum loop nest depth ({\tt D}).
The {\tt eqT} column shows the validation times (in seconds).
The Nodes and
Edges columns show
the number of nodes and edges in the final product graph,
and {\tt BT} and {\tt EXP} is
the number of backtrackings and the
number of (partial) candidate product graphs explored by \toolName{} respectively.
{\tt \#$_q$}
is the total number of SMT queries discharged,
and {\tt Avg. qT} is
the average time taken by an SMT query in seconds for the refinement check.

In a separate experiment, we split the large procedures in {\tt bzip2} into
smaller procedures, so that \toolName{} successfully validates
the {\tt O1-} compilation of the full modified {\tt bzip2} program: the splitting
disables some compiler transformations and also reduces the correlation search space.

Through our experiments, we uncovered and reported a bug in recent
versions of {\tt z3}, including {\tt z3-4.8.14} and {\tt z3-4.12.5},
where for an input satisfiability query $\Psi$,
the SMT solver returns an unsound model (counterexample) that evaluates
$\Psi$ to false \cite{z3bugreport_model2024}.
When a
modern SMT solver is used to validate compilations produced by
a mature compiler, a bug may be found on either side.

\section{Related Work and Conclusions}
\label{sec:relwork}
CoVaC \cite{covac} automatically
identifies a product program that demonstrates observable
equivalence for deterministic programs.
Counter \cite{oopsla20} extends
CoVaC to support path-specializing transformations, such as loop
unrolling, through counterexample-guided
search heuristics.
We extend these prior works to support refinement
between programs
performing dynamic allocations with non-deterministic
addresses for local variables and stack.

Recent work on bounded TV
\cite{alive2_llvm_mem_model} models allocations through
{\em separate blocks},
so a pointer is represented as a combination of a block-ID
and an offset into a block. While
this suffices for the bounded TV setting, our problem setting requires
a more general representation of a dynamically-allocated
variable (e.g., allocation-site)
and a
more general SMT encoding.

CompCert
provides axiomatic semantics for memory (de)allocation in the
source Clight program, and proves their
preservation along the compilation pipeline
\cite{compcertMemModel}.
They restrict their proof method to
CompCert's preallocation strategy for local variables,
possibly to avoid the manual effort required to write
mechanized proofs
for a more general allocation strategy.
Preallocation of local variables has also been used in
prior work on TV for a verified OS
kernel \cite{tv_oskernel}.
Preallocation can be
space inefficient and cannot support VLAs and {\tt alloca()}. Further,
TV for a third-party compiler
cannot assume a particular allocation strategy.

We provide a semantic model, refinement
definition, and
an algorithm
to determine the correctness of a third-party
translation from
an unoptimized high-level representation of a C program to an
optimized assembly program in the presence of dynamically-allocated
local memory. Our semantic model and
definition of refinement require that for allocations
and procedure calls that reuse stack space, their relative
order is preserved in both programs.
While our experiments show that this suffices
in practice, a more general definition of refinement, that admits
transformations that may reorder (de)allocations while
reusing stack space, is perhaps a good
candidate for future work.

\section*{Data-Availability Statement}

The \toolName{} tool that supports~\cref{sec:experiments}
is available on Zenodo
~\cite{oopsla24_artifact}
with instructions
for complete reproducibility of
the presented results.

\begin{acks}
  We thank Shubhani Gupta for contributing towards
  scalability improvements of the translation validation tool.
  We thank Abhishek Dang for carefully reading previous drafts of the paper,
  and pointing out several errors, and making important suggestions that improved
  the paper significantly.
\end{acks}

\bibliography{oopsla24}

\clearpage
\appendix

\section{Appendix}

\subsection{Conversion of C to \ourIR{} for Procedure Definitions and Calls}
\label{app:c2llvmFcalls}

For a C procedure definition,
parameters are passed through pointers of corresponding \ourIR{} types.
This includes both scalar and aggregate parameters.
For example,
a C procedure with parameters {\tt int}, {\tt struct bar}, and {\tt struct baz*}
(pointer to {\tt struct baz}) respectively
is translated to parameters of corresponding \ourIR{} types of
{\tt int*}, {\tt struct bar*}, and {\tt struct baz**}
respectively in \ourIR{}.

A procedure with aggregate ({\tt struct}) return value is translated
to have the return value passed through memory.
For a return value of {\tt struct} type, say `{\tt struct ret}', of
a C procedure {\tt foo()}, the \ourIR{} implementation assumes
that the caller has allocated a `{\tt struct ret}'-sized memory region and has passed its pointer
as the first argument.
The body of {\tt foo}$_{\ourIR{}}$ then populates the contents
of this memory region with
field values computed by it, before returning.

\begin{figure}
  \begin{scriptsize}
  \[
  \mprset{flushleft}
  \inferrule*[vcenter]{\rho(e_1, e_2, \ldots, e_m) \\ \gamma\ \rho(\tau_1, \ldots, \tau_n) \text{ is the type signature}}
  {
    argsP \Assign []; \qquad // \text{empty list}
    \\\\
    \mathtt{IF}\ \mathtt{is\_aggregate\_type}(\gamma)\ \{
      \\\\ \qquad
      \Emit{p_r\ \Assign\ \mathtt{alloc}\ 1,\,\Itempl{\gamma}, \ \Itempl{\mathtt{ALIGNOF}(\gamma)};}
      \\\\ \qquad
      \mathtt{APPEND}(argsP, p_r); \qquad // \text{pass pointer to allocated region as first argument}
      \\\\
    \}
    \\\\
    \mathtt{FOR}\ i\ \mathtt{in}\ 1 \ldots n\ \{
      \\\\ \qquad
      \Emit{p_i\ \Assign\ \mathtt{alloc}\ 1,\,\Itempl{\tau_i}, \ \Itempl{\mathtt{ALIGNOF}(\tau_i)};}
      \\\\ \qquad
      \Emit{\mathtt{store}\ \Itempl{\tau_i}, \ \Itempl{\mathtt{ALIGNOF}(\tau_i)}, \ \Itempl{\mathtt{GEN}(e_i)},\, p_i;}
      \\\\ \qquad
      \mathtt{APPEND}(argsP, p_i);
      \\\\
    \}
    \\\\
    \mathtt{IF}\ \mathtt{is\_variadic}(\rho)\ \{
      \\\\ \qquad
      \ldots, \kappa_i, \ldots \Assign \mathtt{promoted\_type}(e_{n+1}),\ldots,\mathtt{promoted\_type}(e_{m});
      \\\\ \qquad
      \eta \Assign \mathtt{mk\_struct\_x86\_cc}(\ldots, \kappa_i, \ldots); 
      \\\\ \qquad
      \Emit{pvar\ \Assign\ p_v \Assign \mathtt{alloc}\ 1,\,\Itempl{\eta}, \ \Itempl{\mathtt{ALIGNOF}(\eta)};}
      \\\\ \qquad
      \mathtt{APPEND}(argsP, pvar);
      \\\\ \qquad
      \mathtt{FOR}\ i\ \mathtt{in}\ (n+1) \ldots m\ \{
        \\\\ \qquad \qquad
        \Emit{\mathtt{store}\ \Itempl{\kappa_i}, \ \Itempl{\mathtt{ALIGNOF}(\kappa_i)}, \Itempl{\mathtt{GEN}(e_i)}, \  p_v;}
        \\\\ \qquad \qquad 
        \Emit{p_v \Assign p_v + \Itempl{\mathtt{OFFSETOF}(\eta, i)};}
      \\\\ \qquad
      \}
      \\\\
    \}
    \\\\
    \mathtt{IF}\ \gamma = \mathtt{void}\ \{
      \\\\ \qquad
      \Emit{\mathtt{call}\ \mathtt{void} \ \rho(\Itempl{argsP});}
      \\\\
    \} \ 
    \mathtt{ELSE\ IF}\ \mathtt{is\_aggregate\_type}(\gamma)\ \{
      \\\\ \qquad
      \Emit{\mathtt{call}\ \Itempl{\gamma} \ \rho(\Itempl{argsP});}
      \\\\ \qquad
      \Emit{\vv{result} \Assign \mathtt{AGG2REG}(p_r);} \qquad // \text{distribute the populated aggregate into scalar variables}
      \\\\
    \} \ 
    \mathtt{ELSE}\ \{
      \\\\ \qquad
      \Emit{result \Assign \mathtt{call}\ \Itempl{\gamma} \ \rho(\Itempl{argsP});}
      \\\\
    \}
    \\\\
    \mathtt{FOR}\ a\ \mathtt{in}\ \mathtt{reverse}(argsP)\ \{
      \\\\ \qquad
      \Emit{\mathtt{dealloc}\ \Itempl{a};}
      \\\\
    \}
  }
  \]
  \end{scriptsize}
  \caption{
    \label{fig:callToIRAsm} Pseudo-code for translation of a C procedure-call expression to \ourIR{} instructions.
  }
\end{figure}

The translation of a procedure-call from C to \ourIR{} is more complex, as we
generate explicit instructions to (de)allocate memory for the actual arguments,
including a variadic argument. \Cref{fig:callToIRAsm} shows the translation of a C procedure
call to \ourIR{} where $\gamma$ represents
the return value's type and $\tau_1,\ldots,\tau_n$ represents
the parameters' types.
The statements with a \Shaded{} background represent
the generated translation template with template slots
marked by \Itempl{}.
${\tt is\_aggregate\_type}(\gamma)$ returns {\tt true} iff $\gamma$ is an aggregate
({\tt struct} or {\tt union}) type.
For return value of aggregate type, the caller allocates space
for the return value
and passes the start address of allocated region as first argument to the callee
(ensuring the caller side contract of the scheme described in previous paragraph).
$\mathtt{ALIGNOF}(\tau)$ returns the alignment of C type $\tau$ and
$\mathtt{GEN}(e)$ returns the \ourIR{} variable holding value of expression $e$.
${\tt is\_variadic}(\rho)$ returns {\tt true} iff $\rho$ is variadic
and $\mathtt{promoted\_type}(e)$ returns the promoted type of 
C expression $e$ obtained after application of
{\em default argument promotion} rules (see C17 standard).
\texttt{mk\_struct\_x86\_cc($\ldots$)} returns a C `\texttt{struct}' type whose
member fields' alignment matches the calling conventions' requirements of 32-bit x86
and
$\mathtt{OFFSETOF}(\eta, i)$ returns the offset (in bytes) of $i^{th}$ member field in
\texttt{struct} type $\eta$.
$\mathtt{AGG2REG}(p)$ returns the scalar values in aggregate pointed to by $p$.

For example, a call to {\tt printf(fmt, (char)a, (int)b);} translates to:

\begin{tabular}{c}
\begin{myexamplesmallc}
p1 := alloc 1, (char const*), 4;
store (char const*), 4, fmt~$_{\ourIR{}}$~, p1;
p2 := alloc 1, struct{char; int;}, 4;
pv := p2;
store char, 1, a~$_{\ourIR{}}$~, pv;
pv := pv + OFFSETOF(struct{char; int;}, 1);
store int, 4, b~$_{\ourIR{}}$~, pv;
pv := pv + OFFSETOF(struct{char; int;}, 2);
result := call int printf(p1, p2);
dealloc p1;
dealloc p2;
\end{myexamplesmallc}
\end{tabular}
\newline

{\tt result} holds the returned value and $e_{\ourIR{}}$ represents the \ourIR{} variable corresponding to expression $e$ in C.

\subsection{Path enumeration algorithm}
\label{app:pathEnum}
While
enumerating paths terminating
at a non-error node $q^t_{\ddAk}$,
our path enumeration algorithm is similar to the one used in Counter \cite{oopsla20}.

Recall that a path in $\Pathset{\Ck}$ need not be a simple path,
and can visit any node $n_{\Ck}\in{}\NCk{}$ up to $\mu$ times.
All paths in $\Pathset{\Ck}$ must
originate at a unique cut-point $q_{\Ck{}}$ such that
$(q_{\ddAk}, q_{\Ck}) \in \NXk$.
{\em correlatedPathsInCOptions()}
returns candidates, where a candidate
pathset is
a maximal set $\Pathset{\Ck}$
such that each path $\CPath\in{}\Pathset{\Ck}$
either
(a) ends at a unique non-error destination cut-point node, say $q^t_{\Ck}$ (i.e., all
paths $\CPath\in{}\Pathset{\Ck}$ ending
at a non-error node end at $q^t_{\Ck}$),
or (b) ends at
error node \UCk{}.

For a path
$\CPath\in{}\Pathset{\Ck}$, let $\delta_{\CPath}$ be the
number of times the unique non-error destination
node $q^t_{\Ck}$ appears in $\CPath$. Then, due to
the maximality, mutual-exclusion,
and unique non-error destination properties, there must exist a unique value
$\delta_{\Pathset{\Ck}}\leq{}\mu$, such that:
\begin{itemize}
\item For a
path $\xi_{\Ck}\in{}\Pathset{\Ck}$, if $\xi_{\Ck}$ ends
in the unique non-error node node $q^t_{\Ck}$, then $\delta_{\xi_{\Ck}}=\delta_{\Pathset{\Ck}}$.
\item For a
path $\xi_{\Ck}\in{}\Pathset{\Ck}$, if $\xi_{\Ck}$ ends
in $\mathscr{U}_{\Ck}$, then $\delta_{\xi_{\Ck}}<\delta_{\Pathset{\Ck}}$.
\end{itemize}
This $\delta_{\Pathset{\Ck}}$ is the same as the $\delta$ described
in \cite{oopsla20}.

\subsection{Global invariants in \ddAk{} and \Ck{}}
\label{app:globalInv}

\begin{definition}[Non-entry Node]\label{def:nonEntryNode}
Let $P \in \{ \ddAk, \Ck \}$.
A node $n_P \in \NP{P}$ is called a
\textbf{non-entry node}
iff it does not correspond to a node due to
\TRule{\Entry{\Ck}} and \TRule{\Entry{\ddAk}}
(\cref{fig:xlateRuleIR,fig:xlateRuleAsmAllLocals})
in $P$.
A node
$n_{\Xk} = (n_{\ddAk}, n_{\Ck}) \in \NXk$
is called a non-entry node iff both
$n_{\ddAk}$ and $n_{\Ck}$
are non-entry nodes.
\end{definition}

Due to the execution semantics of \ddAk{} and \Ck{},
certain invariants hold by construction
in \ddAk{} and \Ck{}.
We call these invariants \emph{global invariants}
as they hold at each error-free, non-entry node.

\begin{theorem}[Global invariants in \ddAk{} and \Ck{}]
\label{theorem:globalInv}
The following
invariants hold
at each error-free, non-entry node $n_{\Ck}\in{}\NNP{\Ck}$: 
\begin{itemize}
\item (\ghost{\Empty{r}} tracks emptiness) $\il{P}{r} = \emptyset \Leftrightarrow \ghost{\Empty{r}}$, for $r \in G \cup \Y \cup \Z{}$.
\item (\ghost{\Sz{r}} tracks size) $\ghost{\Sz{r}} = |\il{P}{r}|$ for $r \in G \cup \Y$.  In particular, $\ghost{\Sz{r}} = \mathtt{sz}(\mathtt{T}(r))$, for $r \in G \cup (\Y \setminus \{\yv{}\})$.
\item (\ghost{\LSz{\zl}} tracks size) $\ghost{\LSz{\zl}} = |\il{P}{\zl}|$ for $\zl \in \Zl$.
\item ($\ghost{\LB{r}},\ghost{\UB{r}}$ track bounds) $\il{P}{r} = \emptyset \lor (\ghost{\LB{r}} = \LBi{\il{P}{r}} \land \ghost{\UB{r}} = \UBi{\il{P}{r}})$, for $r \in G \cup \Y \cup \Z{}$.
\item (Address sets of $G, \Y, \Zl$ are intervals) $(\ghost{\Empty{r}} \lor [\ghost{\LB{r}},\ghost{\UB{r}}] = \il{P}{r})$, for $r \in G \cup \Y \cup \Zl$.
  As a consequence, we have: $\ghost{\Empty{r}} \lor ((\ghost{\LB{r}} \le_u \ghost{\UB{r}}) \land (\ghost{\UB{r}} = \ghost{\LB{r}} + \ghost{\Sz{r}} - 1_{\bv{32}}))$, for $r \in G \cup \Y$.
  And, $\ghost{\Empty{\zl}} \lor ((\ghost{\LB{\zl}} \le_u \ghost{\UB{\zl}}) \land (\ghost{\UB{\zl}} = \ghost{\LB{\zl}} + \ghost{\LSz{\zl}} - 1_{\bv{32}}))$.
\item (Alignment of \g{} and \y{}) $\mathtt{aligned}_{\mathtt{algnmnt}(r)}(\ghost{\LB{r}})$, for $r \in G \cup (\Y \setminus \{ \yv \})$, where $\mathtt{algnmnt}(r)$ returns the alignment of variable $r$.
\item (Disjoint regions in \Ck{}) $\neg{}\Overlap(\il{\Ck}{\heap},\il{\Ck}{cl},\il{\Ck}{cv},\il{\Ck}{\yv},\ldots,\il{\Ck}{\g},\ldots,\il{\Ck}{\y},\ldots,\il{\Ck}{\z{}},\ldots)$.  
\item (Read-only memory in \Ck{}) ${\prjMEq{\ii{\Ck}{r}}{M_{\Ck}}{\roMem{\Ck}{r}{\ii{\Ck}{r}}}}$ for $r \in \Gro$.
\end{itemize}

The following
invariants hold
at each error-free, non-entry node $n_{\ddAk}\in{}\NNP{\ddAk}$: 

\begin{itemize}
\item (\ghost{\Empty{\f}} tracks emptiness) $\il{P}{\f} = \emptyset \Leftrightarrow \ghost{\Empty{\f}}$, for $\f \in \F$.
\item (\ghost{\Sz{\f}} tracks size) $\ghost{\Sz{\f}} = |\il{P}{\f}| = \mathtt{sz}(\mathtt{T}(\f))$ for $\f \in \F$.
\item (Address sets of $\F$ are intervals) $(\ghost{\Empty{\f}} \lor [\ghost{\LB{\f}},\ghost{\UB{\f}}] = \il{P}{\f})$, for $\f \in \F$.
\item (Alignment of \f) $\mathtt{aligned}_{\mathtt{algnmnt}(\f)}(\ghost{\LB{\f}})$, for $\f \in \F$, where $\mathtt{algnmnt}(f)$ returns the alignment of variable $f$.
\item (Stack bounds) $\il{\ddAk}{\{\stk\}\cup{}\Y}\cup{}(\il{\ddAk}{\Z{}}\setminus{}(\ilZv{\ddAk})) = [{\tt esp}, \ghost{\spE}]$.
\item ($cs$ and $cl$) $\il{\ddAk}{\{cs,cl\}} = [\ghost{\spE}+1, \ghost{\stkE}]$
\item (Heap subset) $\il{\ddAk}{\heap} \subseteq {\tt comp}(\il{\ddAk}{G \cup \F} \cup \ilZv{\ddAk} \cup [\mathtt{esp}, \ghost{\stkE}])$
\item (Disjoint regions in \ddAk{}) $\neg \Overlap(\il{\ddAk}{\heap},\il{\ddAk}{cl},\il{\ddAk}{cv},\il{\ddAk}{\yv},\ldots,\il{\ddAk}{\g},\ldots,\il{\ddAk}{\y},\ldots,\il{\ddAk}{\z{}},\ldots) \land \newline \neg \Overlap(\il{\ddAk}{\heap},\il{\ddAk}{cl},\il{\ddAk}{\yv},\ldots,\il{\ddAk}{\g},\ldots,\il{\ddAk}{\y},\ldots,\ilzs{\ddAk}{\z},\ldots,\il{\ddAk}{\f},\ldots,\il{\ddAk}{\stk},\il{\ddAk}{cs})$
\item (Read-only memory in \ddAk{}) ${\prjMEq{\ii{\ddAk}{r}}{M_{\ddAk}}{\roMem{\ddAk}{r}{\ii{\ddAk}{r}}}}$ for $r \in \Fro$.
\end{itemize}
\end{theorem}
\begin{proof}
By induction on the number of transitions
executed in \Ck{} (\ddAk{}), with the base case defined
by the first outgoing edge from the last instruction due to
\TRule{\Entry{\Ck}} in \cref{fig:xlateRuleIR}
(\TRule{\Entry{\ddAk}} in \cref{fig:xlateRuleAsmAllLocals}).
\end{proof}


\begin{theorem}[Global Invariants in \Xk]\label{theorem:globalInvX}
The following invariants hold at each error-free, non-entry node $n_{\Xk} = (n_{\ddAk},n_{\Ck}) \in \NNP{\Xk}$ of \Xk{}.
\begin{enumerate}
  \item \label{itm:globInvCA} The invariants stated in \cref{theorem:globalInv}.
  \item \label{itm:stkX} (Stack subset) $\il{\ddAk}{\stk} \subseteq \il{\Ck}{\{cv,\free{}\}} \cup \ilzv{\ddAk}{\Zl}$
\end{enumerate}
\end{theorem}
\begin{proof}
  \Cref{itm:globInvCA} follows from \cref{theorem:globalInv}
  as $n_{\Xk}$ is a non-error iff both $n_{\ddAk}$ and $n_{\Ck}$ are non-error nodes.

  \Cref{itm:stkX} follows from (Disjoint regions in \ddAk{}) of \cref{itm:globInvCA} and (Equivalence).
\end{proof}

\subsection{Soundness of $\Xk{}$ requirements}
\label{app:soundnessProof}

Let $\Xk = \ddAk \times \Ck$ be a product-graph
that satisfies the soundness requirements in
\cref{sec:reqX}.

\begin{lemma}[\Xk's execution]
\label{lemma:XkExec}
The following holds for an execution of \Xk{}:
\begin{equation*}
\begin{split}
\forall{\World{},\TraceP{\ddAk},\TraceP{\Ck}}:
\execT{\Xk}{\World{}}{(\TraceP{\ddAk},\TraceP{\Ck})}
\Rightarrow & \phantom{{} \lor {} }\steq{\TraceP{\ddAk}}{\TraceP{\Ck}} \\
            & \lor (\eT{\TraceP{\ddAk}} = \EW
                    \land \stprefix{\neT{\TraceP{\ddAk}}}{\TraceP{\Ck}}) \\
            & \lor (\eT{\TraceP{\Ck}} = \EU
                    \land \stprefix{\neT{\TraceP{\Ck}}}{\TraceP{\ddAk}})
\end{split}
\end{equation*}
\end{lemma}
\begin{proof}[Proof of \cref{lemma:XkExec}]
The proof proceeds through a coinduction on the number of edges
executed by \Xk{}.
We prove that the execution of a single edge
$\XEdge = \XEdgeT{n_{\Xk}}{\APath; \CPath}{n^t_{\Xk}} \in \EXk$,
starting at a non-error node $n_{\Xk} \in \NNP{\Xk}$ in a state that satisfies
$\steq{\TraceP{\ddAk}}{\TraceP{\Ck}}$,
either reaches a terminating node $n^t_{\Xk}$,
such that final state satisfies the RHS of the $\Rightarrow$ in the statement,
or reaches a non-terminating node $n^t_{\Xk}$, such that
$\steq{\TraceP{\ddAk}}{\TraceP{\Ck}}$ holds at the end of execution of $e_{\Xk}$.

Let edges $\{e^1_{\Xk},e^2_{\Xk},\ldots,e^m_{\Xk}\}$,
such that $\forall_{1 \leq j \leq m}:{\XEdge^j=(\XEdgeT{n_{\Xk}}{\APathN{j};\CPathN{j}}{n^j_{\Xk}})\in{}\EXk}$,
be the outgoing edges of a non-error node $n_{\Xk} \in \NNP{\Xk}$.
There can be two cases:

\begin{enumerate}

  \item $\APathN{j}$ and $\CPathN{j}$ are I/O paths.
Because I/O paths are straight-line
sequences of instructions (with no branching), due to
(SingleIO), it must be true that $j=m=1$. Further, an
I/O path can only end at a non-error node $n^j_{\Xk}$.
Because (Equivalence)
requires $\Omega_{\ddAk}=\Omega_{\Ck}$,
implying production of identical non-silent trace events,
the claim holds.

  \item $\APathN{j}$ and $\CPathN{j}$ are I/O free.
Due to (Mutex\ddAk{}) and (Coverage\ddAk{}), it must be possible to execute
a path $\APathN{j}$ to completion.
Due to (Coverage\Ck{}), there exists some outgoing edge
$\XEdge^j=(\XEdgeT{n_{\Xk}}{\APathN{j};\CPathN{j}}{n^j_{\Xk}}) \in \mathcal{E}_{\Xk}$
that is executed to completion.
Further, due to (Mutex\Ck{}), such an edge $e^j_{\Xk}$ must be unique.
The execution of \APathN{j} followed by execution of \CPathN{j} effectively
causes \Xk{} to execute $e^j_{\Xk}$ and reach node $n^j_{\Xk} = (n^j_{\ddAk},n^j_{\Ck})$.

The execution of $\APathN{j}$ may
end at either:
(1) the error node \WAk{},
(2) the error node \UAk{},
(3) a non-error node $n^j_{\ddAk}$.

\begin{itemize}
  \item In case (1), the execution ends at an error node \WAk{}.
Because the traces were stuttering equivalent before the execution of $\XEdge^j$
and the execution of \APathN{j} must only produce the \EW trace event
(due to \APathN{j} being I/O free and (SingleIO) requirement),
$(\eT{\TraceP{\ddAk}} = \EW \land \stprefix{\neT{\Trace{\ddAk}}}{\TraceP{\Ck}})$
will hold in case (1).

  \item In case (2), due to the (Safety) requirement, execution of $\XEdge^j$ must reach node
$n^{t_j}_{\Xk} = (\UAk, \UCk)$.
Moreover, the execution \APathN{j} and $\CPathN{j}$ 
must only generate the error code \EU{} as a trace event
(recall that both \APathN{j} and \CPathN{j} are I/O free
and (SingleIO) forbids {\tt rd}, {\tt wr} instructions in I/O free paths).
Because
the traces were stuttering equivalent before the execution of $\XEdge^j$,
$(\eT{\TraceP{\Ck}} = \EU \land \stprefix{\neT{\TraceP{\Ck}}}{\TraceP{\ddAk}})$
will hold in case (2).

\item In case (3), we analyze each possibility of
$n^{j}_{\Xk}$ separately which must be
one of the following forms:
(a) $(n^j_{\ddAk},\WCk)$,
(b) $(n^j_{\ddAk},\UCk)$, or
(c) a non-error node $(n^j_{\ddAk},n^{j}_{\Ck})$,
where $n^{j}_{\Ck}$ is a non-error node
(recall that $n^j_{\ddAk}$ is a non-error node in this case).
Case (a) cannot occur due to the (Well-formedness) requirement.
In case (b),
$(\eT{\TraceP{\Ck}} = \EU \land \stprefix{\neT{\TraceP{\Ck}}}{\TraceP{\ddAk}})$
holds due to
(SingleIO) and inductive assumption
(similar reasoning as case (2) above).
In case (c), due to the (Equivalence) requirement,
the sequence of non-silent trace events produced in both
executions must be identical.
Further, (Similar-speed) ensures that the silent events
in both traces differ only by a finite amount.
Thus, $\steq{\TraceP{\ddAk}}{\TraceP{\Ck}}$ must hold
at $n^j_{\Xk}$.
\end{itemize}
\end{enumerate}
\end{proof}

\begin{lemma}[\ddAk's traces are in \Xk{}]
\label{lemma:AkTracesInXk}
The following holds for an execution of \ddAk{}:
\begin{equation}
\begin{split}
\forall{\World{},\Trace{\ddAk}}:
\execT{\ddAk}{\World{}}{\Trace{\ddAk}}
\Rightarrow
\exists{\TraceP{\ddAk},\TraceP{\Ck}}: & \phantom{{} \land {}} \execT{\Xk}{\World{}}{(\TraceP{\ddAk},\TraceP{\Ck})} \\
                                    & \land (\phantom{{} \lor {}} \steq{\Trace{\ddAk}}{\TraceP{\ddAk}} \\
                                    & \phantom{{} \land {}(} \lor (      (\eT{\TraceP{\Ck}} = \EU)
                                                                  \land (\eT{\TraceP{\ddAk}} \neq \EW)
                                                                  \land (\stprefix{\neT{\TraceP{\Ck}}}{\Trace{\ddAk}})))
\end{split}
\label{eqn:Acovered}
\end{equation}
\end{lemma}
\begin{proof}
Consider an execution of \Xk{} that is currently at
a non-error node $n_{\Xk}=(n_{\ddAk},n_{\Ck})\in{}\NNP{\Xk}$.
We show by coinduction on the number
of edges executed in \ddAk{} starting
at $n_{\ddAk}$, that \cref{eqn:Acovered}
holds
The proof of the lemma follows
by using $n_{\Xk}=n_{\Xk}^s=(n_{\ddAk}^s,n_{\Ck}^s)\in{}\NXk{}$.

By (Coverage\ddAk{}) and (Coverage\Ck{}).
there exists $\XEdge = (\XEdgeT{n_{\Xk}}{\APath;\CPath}{n^t_{\Xk}})\in\EXk$ such that
\APath{} and \CPath{} execute to completion to reach $n^t_{\Xk}=(n^t_{\ddAk},n^t_{\Ck})$.
\begin{itemize}
\item If $\APath{}\neq{}\epsilon{}$:  If $n^t$ is a non-error node,
the lemma holds by the coinductive hypothesis.
If $n^t_{\Ck}=\WCk{}$, then $n^t_{\ddAk}$ must also be
\WAk{} due to (Well-formedness), and
$\steq{\Trace{\ddAk}}{\TraceP{\ddAk}}$
holds due to (SingleIO).
If $n^t_{\Ck}=\UCk{}$ and $n^t_{\ddAk}=\WAk{}$,
$\steq{\Trace{\ddAk}}{\TraceP{\ddAk}}$ holds due to (SingleIO).
If $n^t_{\Ck}=\UCk{}$ and $n^t_{\ddAk}\neq{}\WAk{}$,
the lemma holds by definition and due to (SingleIO).
$n^t_{\Ck}\neq{}\UCk{}$ and $n^t_{\ddAk}=\UAk{}$ is not
possible due to (Safety).
\item If $\APath{}=\epsilon{}$: execute $k$ edges in \Xk{}
before a non-$\epsilon{}$ path is encountered, where
$k$ is the length of the longest sequence of edges in \Xk{}
such that an edge
$\XEdge = (\XEdgeT{n_{\Xk}}{\APath;\CPath}{n^t_{\Xk}})$
with $\APath{}\neq{}\epsilon{}$ is reached; then repeat
the co-inductive step above. Due
to (Similar-speed), $k$ must be defined.
\end{itemize}


\end{proof}

\begin{lemma}[\Xk{}'s trace is derived from \Ck{}'s trace]
\label{lemma:XkTracesC}
The following holds for an execution of \Xk{}:
\begin{equation*}
\begin{split}
\forall{\World{},\TraceP{\ddAk},\TraceP{\Ck}}:
\execT{\Xk}{\World{}}{(\TraceP{\ddAk},\TraceP{\Ck})}
\Rightarrow
  \exists{\Trace{\Ck}} : & \phantom{{} \land {}} \execT{\Ck}{\World{}}{\Trace{\Ck}} \\
                          & \land (\phantom{{} \lor {}} \steq{\TraceP{\Ck}}{\Trace{\Ck}} \\
                          & \phantom{{} \land {} (} \lor (    (\eT{\TraceP{\ddAk}} = \EW)
                                                        \land (\stprefix{\neT{\TraceP{\ddAk}}}{\Trace{\Ck}})))
\end{split}
\end{equation*}
\end{lemma}
\begin{proof}

The proof proceeds through a coinduction on the number of edges
executed by \Xk{}.
Suppose \Xk{} and \Ck{} start execution with states
$\sigma_{\Xk} = (\sigma_{\ddAk},\sigma_{\Ck})$, $\sigma_{\Ck}$
at non-error nodes
$n_{\Xk} = (n_{\ddAk},n_{\Ck})$, $n_{\Ck}$
respectively,
such that
$\steq{\Trace{\Ck}}{\TraceP{\Ck}}$,
where $\Trace{\Ck} \in \sigma_{\Ck}$
and  $(\TraceP{\ddAk},\TraceP{\Ck}) \in \sigma_{\Xk}$,
holds.

Consider the execution of edge
$\XEdge = \XEdgeT{n_{\Xk}}{\APath; \CPath}{n^t_{\Xk}} \in \EXk$,
starting at non-error node $n_{\Xk} \in \NNP{\Xk}$ on state $\sigma_{\Xk}$.
If \CPath{} is executed, as part of \XEdge{}'s execution, using
some sequence of non-deterministic choices determined by \DXk{},
the same path \CPath{} can be executed in \Ck{} for the same sequence of
non-deterministic choices.
As both executions start in identical states, they will produce
identical sequence of trace events till execution
reaches the sink node $n^t_{\Ck}$ where
$\steq{\Trace{\Ck}}{\TraceP{\Ck}}$ will hold
(note that execution of \APath{} may not modify the state elements
of \Ck{} in $\sigma_{\Xk}$ as both have disjoint state space).
If $n^t_{\Xk} = (n^t_{\ddAk},n^t_{\Ck})$ is a non-terminating node,
then the claim holds due to the coinduction hypothesis.
Similarly, if both $n^t_{\Xk}$ and $n^t_{\Ck}$ are terminating nodes,
then the claim holds by definition.

Consider the case when $n^t_{\Xk} = (n^t_{\ddAk},n^t_{\Ck})$ is a terminating node
but $n^t_{\Ck}$ is not a terminating node.
There are three possibilities for $n^t_{\Xk} = (n^t_{\ddAk},n^t_{\Ck})$ in this case:
\begin{enumerate}
  \item $n^t_{\ddAk} = \WAk$:
    Due to (Equivalence) and (Similar-speed),
    $\steq{\TraceP{\ddAk}}{\TraceP{\Ck}}$ holds at $n_{\Xk}$.
    Further, due to (SingleIO), \APath{} cannot produce any non-silent trace event
    other than \EW{}.
    Hence, $\stprefix{\TraceP{\ddAk}}{\Trace{\Ck}}$ holds due to inductive assumption.

  \item $n^t_{\ddAk} = \UAk$:
    Due to (Safety), $n^t_{\Xk} = (n^t_{\ddAk},n^t_{\Ck})$ must be of the form $(\UAk, \UCk)$.
    However, this violates the assumption that $n^t_{\Ck}$ is a non-terminating node.

  \item $n^t_{\ddAk}$ is a non-error terminating node:
    This case is not possible due to (Termination) requiring
    $n^t_{\Ck}$ to be non-error terminating node whenever $n^t_{\ddAk}$
    is a non-error terminating node.

\end{enumerate}

\end{proof}

\begin{proof}[Proof of \cref{theorem:witness}]

Consider an execution of \ddAk{} under world \World{}.
Using \cref{lemma:AkTracesInXk}, we have:
\begin{equation*}
\begin{split}
\forall{\World{},\Trace{\ddAk}}:
\execT{\ddAk}{\World{}}{\Trace{\ddAk}}
\Rightarrow
\exists{\TraceP{\ddAk},\TraceP{\Ck}}: & \phantom{{} \land {}} \execT{\Xk}{\World{}}{(\TraceP{\ddAk},\TraceP{\Ck})} \\
                                    & \land (\phantom{{} \lor {}} \steq{\Trace{\ddAk}}{\TraceP{\ddAk}} \\
                                    & \phantom{{} \land {}} \lor (      (\eT{\TraceP{\Ck}} = \EU)
                                    \land (\eT{\TraceP{\ddAk}} \neq \EW)
                                    \land (\stprefix{\neT{\TraceP{\Ck}}}{\Trace{\ddAk}})))
\end{split}
\end{equation*}
  
  Instantiating \cref{lemma:XkTracesC}, we have:
\begin{equation}\label{eq:Combined}
\begin{split}
\forall{\World{},\Trace{\ddAk}}:
\execT{\ddAk}{\World{}}{\Trace{\ddAk}}
\Rightarrow
\exists{\TraceP{\ddAk},\TraceP{\Ck}}: & \phantom{{} \land {}} \execT{\Xk}{\World{}}{(\TraceP{\ddAk},\TraceP{\Ck})} \\
                                    & \land (\phantom{{} \lor {}} \steq{\Trace{\ddAk}}{\TraceP{\ddAk}} \\
                                    & \phantom{{} \land {}} \lor (      (\eT{\TraceP{\Ck}} = \EU)
                                    \land (\eT{\TraceP{\ddAk}} \neq \EW)
                                    \land (\stprefix{\neT{\TraceP{\Ck}}}{\Trace{\ddAk}}))) \\
                                    & \land (\exists{\Trace{\Ck}} : \execT{\Ck}{\World{}}{\Trace{\Ck}} \\
  & \phantom{{} \land {}} \land (\phantom{{} \lor {}} \steq{\TraceP{\Ck}}{\Trace{\Ck}} \\
  & \phantom{{} \land {}  \land (} \lor (  (\eT{\TraceP{\ddAk}} = \EW)
    \land (\stprefix{\neT{\TraceP{\ddAk}}}{\Trace{\Ck}}))))
\end{split}
\end{equation}
  Consider each minterm in the sum-of-products representation
  of the conjunction of the RHS of the equations in \cref{lemma:AkTracesInXk,lemma:XkTracesC}:
  \begin{enumerate}
  \item $(\steq{\Trace{\ddAk}}{\TraceP{\ddAk}}) \land (\steq{\TraceP{\Ck}}{\Trace{\Ck}})$ holds.

    Instantiating \cref{lemma:XkExec} in \cref{eq:Combined}, there are three cases:
    \begin{itemize}
      \item $\steq{\TraceP{\ddAk}}{\TraceP{\Ck}}$ holds.

        Due to $\steq{}{}$ being an equivalence relation, we have
        $\steq{\Trace{\ddAk}}{\Trace{\Ck}}$
        and therefore $\Ck \refines \ddAk$ holds.

      \item $\eT{\TraceP{\ddAk}} = \EW \land \stprefix{\neT{\TraceP{\ddAk}}}{\TraceP{\Ck}}$ holds.

        As $\steq{}{}$ is congruent with respect to $\stprefix{}{}$, we have
        $\eT{\Trace{\ddAk}} = \EW \land \stprefix{\neT{\Trace{\ddAk}}}{\Trace{\Ck}}$,
        which is equivalent to $W_{\tt pre}^{\World{},\Trace{\ddAk}}(\Ck)$.
        Therefore, $\Ck \refines \ddAk$ holds.

      \item $\eT{\TraceP{\Ck}} = \EU \land \eT{\TraceP{\ddAk}} \neq \EW \land{} \stprefix{\neT{\TraceP{\Ck}}}{\TraceP{\ddAk}}$ holds.

        Using congruence of $\steq{}{}$ with respect to $\stprefix{}{}$, we have
        $\eT{\Trace{\Ck}} = \EU \land \stprefix{\neT{\Trace{\Ck}}}{\Trace{\ddAk}}$,
        which is equivalent to $U_{\tt pre}^{\World{},\Trace{\ddAk}}(\Ck)$.
        Therefore, $\Ck \refines \ddAk$ holds.
    \end{itemize}

  \item $(\steq{\Trace{\ddAk}}{\TraceP{\ddAk}}) \land ((\eT{\TraceP{\ddAk}} = \EW) \land (\stprefix{\neT{\TraceP{\ddAk}}}{\Trace{\Ck}}))$ holds.

    Using definition of $\steq{}{}$ and congruence of $\steq{}{}$ with respect to $\stprefix{}{}$,
    we have 
    $(\eT{\Trace{\ddAk}} = \EW) \land (\stprefix{\neT{\Trace{\ddAk}}}{\Trace{\Ck}})$,
    which is equivalent to $W_{\tt pre}^{\World{},\Trace{\ddAk}}(\Ck)$.
    Therefore, $\Ck \refines \ddAk$ holds.

  \item $((\eT{\TraceP{\Ck}} = \EU) \land (\stprefix{\neT{\TraceP{\Ck}}}{\Trace{\ddAk}})) \land (\steq{\TraceP{\Ck}}{\Trace{\Ck}})$ holds.

    Using definition of $\steq{}{}$ and congruence of $\steq{}{}$ with respect to $\stprefix{}{}$,
    we have 
    $(\eT{\Trace{\Ck}} = \EU) \land (\stprefix{\neT{\Trace{\Ck}}}{\Trace{\ddAk}})$,
    which is equivalent to $U_{\tt pre}^{\World{},\Trace{\ddAk}}(\Ck)$.
    Therefore, $\Ck \refines \ddAk$ holds.

  \item $((\eT{\TraceP{\Ck}} = \EU) \land (\eT{\TraceP{\ddAk}} \neq \EW)  \land (\stprefix{\neT{\TraceP{\Ck}}}{\Trace{\ddAk}})) \land ((\eT{\TraceP{\ddAk}} = \EW) \land (\stprefix{\neT{\TraceP{\ddAk}}}{\Trace{\Ck}}))$ holds.


    This case is not possible due to the mutually unsatisfiable clauses
    $\ldots \land (\eT{\TraceP{\ddAk}} \neq \EW) \land \ldots \land (\eT{\TraceP{\ddAk}} = \EW) \land \ldots$.
    
  \end{enumerate}
\end{proof}

\subsection{Soundness of Callers' Virtual Smallest semantics}
\label{app:cvSmallestProof}

Let \Ak{} and \Ck{} be transition graphs obtained due to original semantics described in
\cref{fig:xlateRuleIR,fig:xlateRuleAsm,fig:xlateRuleAsmStackLocals,fig:xlateRuleAsmAllLocals}.
Let \AkP{} and \CkP{} be obtained from \Ak{} and \Ck{} respectively
by applying the callers' virtual smallest semantics described in
\cref{sec:cvSmallest}.
Let \ddAkP{} be obtained by annotating \AkP{} as described in \cref{sec:refnDefn}.
Let \ddAk{} be obtained by annotating \Ak{} such that annotations
made in \ddAkP{} and \ddAk{} are identical.

Let $\XkP = \ddAkP \boxtimes \CkP = (\NXkP{}, \EXkP{}, \DXkP{})$
be a product
graph such that $\XkP$ satisfies the search-algorithm requirements.
We prove that there exists a product graph
$\Xk{} = \ddAk{} \boxtimes \Ck = (\NXk{},\EXk{},\DXk{})$
such that \Xk{} satisfies the search-algorithm requirements.

\begin{definition}[(Coverage\Ck{}) holds for $\APath$ at $n_{\Xk}$ in \Xk{}]
  \label{defn:covCkForAPath}
  At a node $n_{\Xk} \in \NXk{}$,
  let $\{\XEdgeN{1},\XEdgeN{2},\ldots,\XEdgeN{m}\}$
  be the set of \emph{all} outgoing edges such that
  $\XEdgeN{j}=\XEdgeT{n_{\Xk}}{\APath; \CPathN{j}}{(n^t_{\ddAk}, n_{\Ck}^{t_j})}$ (for $1\leq{}j\leq{}m$).
	Then, \emph{(Coverage\Ck{}) holds for \APath{} at $n_{\Xk}$ in \Xk{}} iff
  \pthcover{\XEdgeN}{n_{\Xk}}{\DXk}{\APath{}}
  holds.
\end{definition}
Notice that this definition is identical to the (Coverage\Ck{}) definition in \cref{sec:reqX},
except that it defines (Coverage\Ck{}) for a specific path \APath{} starting at a
specific node $n_{\Xk}$.
We define (Coverage\ddAk{}) at node $n_{\Xk}$ analogously.

\begin{proof}[Proof of \cref{theorem:cvSmallest}]

Construct $\Xk{} = \XkP{}$.
Add extra edges in \Xk{} to nodes $(\WAk,n_{\Ck})$
where $n_{\Ck}$ is an error-free node
such that (Mutex\ddAk) is not violated.
These extra edges help in ensuring (Coverage\ddAk{}) in \Xk{}.

As the use of callers' virtual smallest semantics does not affect the graph structure of
\Ak{} and \Ck{}
(recall that the changes were limited to modifications to instructions of an edge),
the seven structural requirements,
(Mutex\ddAk{}), (Mutex\Ck{}), (Termination), (SingleIO), (Well-formedness), (Safety), and (Similar-speed),
should continue to hold for \Xk{}.

Let $n_{\XkP} = (n_{\ddAkP},n_{\CkP}) \in \NP{\XkP}$ be a node in \XkP{}
and
let $n_{\Xk} = (n_{\ddAk}, n_{\Ck}) \in \NXk$ be its corresponding node in \Xk{}.
Let \PPath{\ddAkP} be an outgoing path at $n_{\ddAkP}$ in \ddAkP{}
and let \APath{} be its structurally similar path originating at $n_{\ddAk}$ in \ddAk{}. 
Let
$\{ \XPEdgeN{1},\ldots,\XPEdgeN{m}\}$
be the set of all outgoing edges at
$n_{\XkP}$ such that
$\forall_{1 \leq j \leq m}:e^j_{\XkP} = \XEdgeT{n_{\XkP}}{\PPath{\ddAkP}; \CPPathN{j}}{n^t_{\XkP}} \in \EP{\XkP}$.
Let the set
$\{ \XPEdgeN{1},\ldots,\XPEdgeN{m}\}$
be defined analogously for \Xk{}.
Our proof completes by induction on the number of edges executed in \Xk{},
starting at $n_{\Xk}$.

We analyze the instructions in \ddAk{} and \Ck{} affected by the semantics change
and consider the case when an edge
$e_{\ddAk} \in \APath$ or $e_{\Ck} \in \CPPathN{j}$
corresponds to it
\footnote{
  Note that
  \TRule{\Load{\Ck}}, \TRule{\Store{\Ck}},
  \TRule{CallV}, and \TRule{\Call{\Ck}},
  are not affected as the $cv$ region is inaccessible in \Ck{}
  and cannot be returned by $\BasedOn{x}$ for any variable $x$
  and \BasedOnM{r} for any region $r$.
}.
\begin{itemize}
  \item \TRule{\Entry{\Ck}} and \TRule{\Entry{\ddAk}}:
    The $\neg \mathtt{addrSetsAreWF}(\ldots)$ condition
    is weaker in \ddAk{} and \Ck{} than \ddAkP{} and \CkP{} respectively.
    Consequently, the path condition for paths
    $\APath{} = \PathT{n_{\ddAk}}{n_{\ddAk}^{\cancel{\EW}}}$
    (where $n_{\ddAk}^{\cancel{\EW}} \in \NP{\ddAk} \setminus \WAk$)
    and
    $\CPath{} = \PathT{n_{\Ck}}{n_{\Ck}^{\cancel{\EW}}}$
    (where $n_{\Ck}^{\cancel{\EW}} \in \NP{\Ck} \setminus \WCk$)
    that do not go to \WAk{} and \WCk{} respectively
    is stronger in \ddAk{} and \Ck{}
    than \ddAkP{} and \CkP{} respectively.

    Because the address sets returned by the \texttt{rd} instruction
    are arbitrary and identical across \Ck{}
    and \ddAk{}, due to (Equivalence),
    (Coverage\Ck{}) holds by construction in this case.

    As the results of the \texttt{rd} instruction are arbitrary,
    the difference in infeasibility of
    $\PPath{\ddAkP} = \PathT{n_{\ddAkP}}{\WP{\ddAkP}}$
    and
    structurally similar
    $\APath = \PathT{n_{\ddAk}}{\WAk}$
    can only be due to the address set of regions in \F{}
    (see definition of $\mathtt{addrSetsAreWF}(\ldots)$ in \cref{tab:preds})
    As $\il{\ddAkP}{\F} = \il{\ddAk}{\F}$,
    (Coverage\ddAk{}) at $n_{\Xk}$
    should continue to hold in this case.

  \item \TRule{Alloc}, \TRule{AllocV}, and \TRule{AllocS'}:
    As $(\il{\ddAk}{cv} = \il{\Ck}{cv}) \supseteq (\il{\ddAkP}{cv} = \il{\CkP}{cv} = \emptyset)$,
    the $\neg \mathtt{intrvlInSet}_{a}(\ldots)$ condition of \TRule{Alloc} and \TRule{AllocV}
    and $\Overlap{}(\ldots)$ condition of \TRule{AllocS'}
    is weaker in \ddAk{} and \Ck{}
    than \ddAkP{} and \CkP{} respectively.
    Consequently, similarly to previous case, the path condition for paths
    that do not go to \WAk{} and \WCk{} respectively
    is stronger in \ddAk{} and \Ck{}
    than \ddAkP{} and \CkP{} respectively.
    
    Due to (SingleIO), the nodes $n_{\ddAk}$ and $n_{\Ck}$
    must either correspond to PCs due to:
    (1) \TRule{AllocV} and \TRule{Alloc}; or
    (2) \TRule{AllocS'} and \TRule{Alloc}.
    Due to (Equivalence),
    $\il{\ddAk}{\mathtt{comp}(\NS \cup \{cv\})} = \il{\Ck}{\mathtt{comp}(\NS \cup \{cv\})} = \il{\Ck}{\free}$
    must hold at $n_{\Xk}$.
    As, for $P \in \{ \ddAk, \ddAkP, \Ck, \CkP \}$,
    \il{P}{\{\heap,cl\}} is assigned arbitarily at entry,
    the set of possible values for \il{P}{\mathtt{comp}(\NS \cup \{cv\})}
    (note $\il{\ddAkP}{cv} = \il{\CkP}{cv} = \emptyset$)
    remain identical in $P$ at an error-free node $n_{\Xk}$ and $n_{\XkP}$.
    Thus, in case (1), the affected $\neg \mathtt{intrvlInSet}_{a}(\ldots)$
    condition should have identical semantics in both \XkP{} and \Xk{}
    and (Coverage\ddAk{}) and (Coverage\Ck{}) should continue to hold.

    In case (2), a path $\PPath{\ddAkP} = \PathT{n_{\ddAkP}}{\WP{\ddAkP}}$
    with an edge with the $\Overlap{}(\ldots)$ condition
    could be provably infeasible at $n_{\XkP}$ in \XkP{}
    but a similarly structured path \APath{}
    could potentially be feasible at $n_{\Xk}$ in \Xk{} ---
    e.g., when $\ilzv{\ddAkP}{\Zl} = \emptyset$.
    To ensure (Coverage\ddAk{}), we introduce edge
    $e^{\prime}_{\Xk} = \XEdgeT{(n_{\ddAk},n_{\Ck})}{\APath{};\epsilon}{(\WAk,n_{\Ck})}$
    for each such path \APath{} in \Xk{}.
    Notice that (Coverage\Ck) holds for \APath{} at $n_{\Xk}$.
    Because \APath{} does not contain any memory access,
    introduction of $e^{\prime}_{\Xk}$ would not disturb (MAC).

    For a path $\PathT{n_{\ddAk}}{n^{\cancel{\EW}}_{\ddAk}}$
    (where $n_{\ddAk}^{\cancel{\EW}} \in \NP{\ddAk} \setminus \{ \WAk \}$),
    (Coverage\Ck{})
    holds due to (Stack subset) invariant (\cref{theorem:globalInvX})
    and by using identical reasoning as case (1) above.

  \item \TRule{Op-esp'}:
    The condition $\mathtt{intrvlInSet}()$
    is not affected by the semantics change as the address sets
    $\il{\ddAk}{\free} \cup ((\il{\ddAk}{cv} \cup \ilZv{\ddAk}) \setminus \il{\ddAk}{\F})$
    and
    $\il{\ddAkP}{\free} \cup (\ilZv{\ddAkP} \setminus \il{\ddAkP}{\F})$
    must evaluate to identical values
    (on states $\sigma$ and $\sigma'$ at nodes $n_{\Xk}$ and $n_{\XkP}$ in \Xk{} and \XkP{} resp. such that
    $\Inv{n_{\Xk}}(\sigma)$ and $\Inv{n_{\XkP}}(\sigma')$ hold)
    due to new definition of \il{\ddAkP}{\free} in \ddAkP{}.

  \item \TRule{\Load{\ddAk}} and \TRule{\Store{\ddAk}}:
    Identical reasoning as \TRule{Op-esp'} case; the address set expressions
    should evaluate to identical values.
    Hence, no change in semantics for this case too.

\end{itemize}

As the path condition to an error-free node
is only stronger (or equivalent) in \ddAk{} and \Ck{},
the remaining  semantic requirements,
(Inductive), (Equivalence), (MAC), and (MemEq)
should also continue to hold in \Xk{}.

\end{proof}

\subsection{Soundness of Safety-Relaxed Semantics for $\Ak$}
\label{app:fastEncodingProof}

Let \Ak{} be the transition graph obtained due to the callers' virtual smallest
semantics of the assembly procedure, as presented in \cref{sec:cvSmallest}.
Let $\Ak^{\prime}$ be the transition graph obtained due to the safety-relaxed
semantics in \cref{sec:safetyUnchecked}.
Let $\ddAk^{\prime}$ be obtained
by annotating $\Ak^{\prime}$ as described in \cref{sec:refnDefn}.

Let $\XkP = \ddAkP \boxtimes \Ck$ be a product
graph 
such that $\XkP$ satisfies the fast-encoding requirements.
Let
$e_{\XkP} = (\XEdgeT{n_{\XkP}}{\PPath{\ddAkP}; \CPath}{n^t_{\XkP}}) \in \EP{\XkP}$,
be an edge in \XkP{}.

\begin{lemma}[Paths containing memory accesses do not modify allocation state of common regions]
  \label{lemma:memAccessAllocPreserved}
  If $\PPath{\ddAkP}$ contains an edge corresponding to
  \TRule{\Load{\ddAkP}} or \TRule{\Store{\ddAkP}}
  (i.e., a {\tt load} or {\tt store} instruction),
  then $\PPath{\ddAkP}$ does not modify the address sets
  corresponding to regions in $\NS$,
  \il{\ddAkP}{\g} (for each $\g \in G$),
  \il{\ddAkP}{\heap},
  \il{\ddAkP}{cl},
  \il{\ddAkP}{\y} (for each $\y \in \Y$),
  and
  \il{\ddAkP}{\z} (for each $\z \in \Z$).
\end{lemma}
\begin{proof}[Proof of \cref{lemma:memAccessAllocPreserved}]
  Once initialized in \TRule{\Entry{\Ak}} in an I/O path that does not contain any
  {\tt load} or {\tt store} instruction (\cref{fig:xlateRuleAsm}),
  the address sets corresponding to regions $\NS \setminus \Z$ 
  are not modified during the entire execution of \ddAkP{}.

  The address set corresponding to a region $\z \in \Z$
  may only be modified by the {\tt (de)alloc$_{s,v}$} instructions.
  Due to (SingleIO) requirement, these {\tt (de)alloc$_{s,v}$} instructions
  cannot exist as a part of longer paths that may contain {\tt load} or {\tt store} instructions
  (as evident from translations given in \cref{fig:xlateRuleAsmStackLocals,fig:xlateRuleAsmAllLocals}).
\end{proof}

As a corollary, due to (SingleIO), $\CPath$ also does not modify the address
sets corresponding to regions in $\NS$.

\begin{lemma}[\prjM{\il{\ddAkP}{cs}}{M_{\ddAkP}} is not modified in \XkP]
  \label{lemma:csIsPreserved}
  Let $\XkP = \ddAk^{\prime} \boxtimes \Ck$ be a product
  graph for a lockstep execution between $\ddAk^{\prime}$ and \Ck{}.
  If $\XkP$ satisfies the fast-encoding requirements,
  then
  $\prjMEq{\il{\ddAk^{\prime}}{cs}}{\ghost{M^{cs}}}{M_{\ddAk^{\prime}}}$ holds at each non-start,
  non-error node $n_{\XkP} \in \NNP{\XkP}$.
\end{lemma}
\begin{proof}[Proof of \cref{lemma:csIsPreserved}]
  For simplicity, let's first assume that there is only one outgoing edge $e^s_{\XkP}$ from
  the start node $n^s_{\XkP}$ to a non-error node $n^{s2}_{\XkP}$, such that
  $e^s_{\XkP} = \XEdgeT{n^s_{\XkP}}{\PPath{\ddAk^{\prime}}^{s}; \CPathN{s}}{n^{s2}_{\XkP}}$;
  where $\PPath{\ddAk^{\prime}}^{s}$ and $\CPathN{s}$ represent
  the program paths corresponding to \TRule{\Entry{\ddAk{}}} and
  \TRule{\Entry{\Ck{}}} respectively. Let's call this the {\em start-edge} assumption.

  The proof proceeds by
  induction over the number of edges executed in \XkP{}.

  $\prjMEq{\il{\ddAk^{\prime}}{cs}}{\ghost{M^{cs}}}{M_{\ddAk^{\prime}}}$ holds
  at $n^{s2}_{\XkP}$ due to \TRule{\Entry{\ddAk{}}}, which forms our base case.

  Consider a node $n_{\XkP}$ such that
  $\prjMEq{\il{\ddAk^{\prime}}{cs}}{\ghost{M^{cs}}}{M_{\ddAk^{\prime}}}$ holds at $n_{\XkP}$,
  and let $e_{\XkP} = \XEdgeT{n_{\XkP}}{\PPath{\ddAk^{\prime}}; \CPath{}}{n^{t}_{\XkP}}\in{}\mathcal{E}_{\XkP}$ such that $n^t_{\XkP} \in \NNP{\XkP}$ is a non-error node.

  If path $\PPath{\ddAk^{\prime}}$ does not contain a {\tt store} instruction, then
  $\prjMEq{\il{\ddAk^{\prime}}{cs}}{\ghost{M^{cs}}}{M_{\ddAk^{\prime}}}$ holds trivially at $n^t_{\XkP}$.

  If path $\PPath{\ddAk^{\prime}}$ contains a {\tt store} instruction, then
  this path cannot modify the allocation state of
  common regions (\NS{}) in
  \ddAkP{}
  (due to \cref{lemma:memAccessAllocPreserved}).
  Let $\alpha$ be an address such that a {\tt store} is performed to
  $\alpha$ in $\PPath{\ddAk^{\prime}}$, such that $\PPath{\ddAk^{\prime}}$ does
  not modify the allocation state of common memory regions (\NS) in $\ddAkP$.
  Similarly, $\CPath{}$ also does not modify the allocation state of common memory regions in \Ck{}.

  If $\alpha \in \il{\ddAk^{\prime}}{cs}$,
  then due to (MAC), there must be a {\tt store} to the
  same address in \Ck{} before execution may reach $n^t_{\Ck}$.
  Due to the global invariants,
  $\il{\ddAk^{\prime}}{cs}\subseteq{}(\ilZv{\ddAk^{\prime}}\cup \il{\Ck}{\free})\cap{}[\ghost{\spE}+1,\ghost{\stkE}]$ 
  must hold during the execution of $e_{\XkP}$,
  and so $\alpha{}\in{}(\ilZv{\ddAk^{\prime}}\cup{}\il{\Ck}{\free})\cap{}[\ghost{\spE}+1,\ghost{\stkE}]$. 
  However, $\alpha \in (\ilZv{\ddAk^{\prime}})\setminus (\il{\ddAk^{\prime}}{\Frw}\cup[\mathtt{esp},\ghost{\stkE}])$ is not
  possible due to \TRule{\Store{\ddAk}} with the safety-relaxed semantics.
  Thus,
  $\alpha{}\in{}(\il{\Ck}{\free}\cap{}[\ghost{\spE}+1,\ghost{\stkE}])$ must hold. 
  However, this is not possible due to \TRule{\Store{\Ck{}}}.
  Thus, by contradiction, a {\tt store} to address $\alpha \in \il{\ddAkP}{cs}$ is infeasible
  in $\ddAkP$.
  Thus, $\prjMEq{\il{\ddAkP}{cs}}{\ghost{M^{cs}}}{M_{\ddAkP}}$ holds at $n^t_{\XkP}$.

  To generalize beyond the start-edge assumption, we only need to show
  that for any outgoing edge
  of the start node $e^s_{\XkP} = \XEdgeT{n^s_{\XkP}}{\PPath{\ddAkP}^{s}; \CPathN{s}}{n^{t}_{\XkP}}$,
  $\prjMEq{\il{\ddAkP}{cs}}{\ghost{M^{cs}}}{M_{\ddAkP}}$ holds at $n^{t}_{\XkP}$.
  We observe that there must exist a node $n^{s2}_{\ddAkP}$
  in $\PPath{\ddAkP}^{s}$ where
  $\prjMEq{\il{\ddAkP}{cs}}{\ghost{M^{cs}}}{M_{\ddAkP}}$ holds. The rest of
  the argument remains identical for the path $\PPath{\ddAkP}^{s2}=n^{s2}_{\ddAkP} \twoheadrightarrow n^t_{\ddAkP}$.
\end{proof}

\begin{proof}[Proof of \cref{theorem:fastEncoding}]
  Construct \Xk{} = \XkP{} with some extra edges from nodes in \Xk{} to
  the error-node $(\UAk{},\UCk{})$ such that (Mutex\ddAk{})
  and (Mutex\Ck{}) are not violated. We later describe what edges are added to \Xk{}
  and why \Xk{} continues to satisfy the fast-encoding requirements even after the addition
  of these edges. It is already possible to
  see that the structural requirements will hold for \Xk{} even after the
  addition of such edges.

  Let \APath{} be a path in $\ddAk$ on which there exists an
  overlap check
  $\varphi=\Overlap([p]_{w}, \il{\ddAk}{\free} \cup ((\ilZv{\ddAk}) \setminus \il{\ddAk}{\F \cup \STACK{}}))$ 
  (for triggering \EU{})
  due to a \TRule{\Load{\ddAk}} instruction
  (or,
  an overlap check
  $\varphi=\Overlap([p]_{w}, \il{\ddAk}{\free} \cup ((\ilZv{\ddAk}) \setminus \il{\ddAk}{\Frw \cup \STACK{}}))$ 
  (for triggering \EU{})
  due to a \TRule{\Store{\ddAk}} instruction).
  In \ddAkP,
  $\varphi$ is replaced by
  $\varphi^{\prime}=\Overlap([p]_{w}, (\ilZv{\ddAk}) \setminus (\il{\ddAk}{\F} \cup [\mathtt{esp}, \ghost{\stkE}]))$, in case of a \TRule{\Load{\ddAk}},
  (or,
  $\varphi^{\prime}=\Overlap([p]_{w}, (\ilZv{\ddAk}) \setminus (\il{\ddAk}{\Frw} \cup [\mathtt{esp}, \ghost{\stkE}]))$, in case of a \TRule{\Store{\ddAk}})
  to obtain \PPath{\ddAk'}.
  Recall that \Ak{}'s translation has ``{\tt if $\varphi$ \Ihalt{\mathscr{U}}}''
  while $\Ak{}^{\prime}$'s translation has ``{\tt if $\varphi{}^{\prime}$ \Ihalt{\mathscr{U}}}''.
  Because $\varphi{}^{\prime}\Rightarrow{}\varphi{}$,
  \ddAk{} may trigger \EU{} when \ddAkP{} would simply execute
  the {\em non-error path}  in \TRule{\Load{\ddAk}} (or, \TRule{\Store{\ddAk}})
  (a path that does not terminate
  in an error node after executing the instructions in \TRule{\Load{\ddAk}} or \TRule{\Store{\ddAk}}
  Conversely, if \ddAk{} executes
  a non-error path
  (of \TRule{\Load{\ddAk}} or \TRule{\Store{\ddAk}})
  on an initial state $\sigma{}$, then \ddAkP{} will also
  execute the same non-error path on $\sigma{}$.

  Similarly, let $\Phi=\neg \prjMEq{\il{\ddAk}{cs}}{\ghost{M^{cs}}}{M_{\ddAk}}$ be a
  check in \ddAk{} (due to \TRule{\Ret{\Ak{}}}), that has been
  replaced with $\Phi^{\prime}=\mathtt{false}$ in $\ddAk{}^{\prime}$.
  Again, if \ddAk{} executes
  a non-error path of \TRule{\Ret{\Ak}} on an initial state $\sigma{}$, then \ddAkP{} will also
  execute the same non-error path on $\sigma{}$.

  Thus, it can be shown through induction that
  four of the six non-structural requirements --- (Inductive), (Equivalence), (MAC), (MemEq) ---
  hold on \Xk{} if they hold on \XkP{} with $\IXk{}=\IXkP{}$.  The common argument in this part of the proof is that the
  path condition of a non-error path in \Xk{} (containing $\neg{}\varphi{}$ or $\neg{}\Phi{}$) is always stronger than the path condition
  of a non-error path in \XkP{} (containing $\neg{}\varphi{}^{\prime}$ or $\neg{}\Phi{}^{\prime}$).

  We next show that if (Coverage\Ck{}) holds for path $\PPath{\ddAkP}$ starting
  at node $n_{\XkP}$ in \XkP{},
  (Coverage\Ck) also holds for corresponding path $\APath{}$
  starting at corresponding node $n_{\Xk}$ in \Xk{}
  (\cref{defn:covCkForAPath}).
  For an edge $\XEdge^j=\XEdgeT{n_{\Xk}}{\APath; \CPathN{j}}{(n^t_{\ddAk}, n_{\Ck}^{t_j})}$,
  if \APath{} ends at a node $n^t_{\ddAk}\neq{}\UAk{}$,
  then this is easy to show by induction on the number of edges executed on a path:
  because the path condition of \APath{} in \ddAk{} is always equal
  or stronger than the path condition
  of a corresponding (structurally identical) path $\PPath{\ddAkP{}}$ in \ddAkP{},
  if (Coverage\Ck{}) holds for
  $\PPath{\ddAkP{}}$ at a node $n_\XkP{}$ in \XkP{}, it must also hold for
  \APath{} at the corresponding node $n_\Xk{}$ in \Xk{}.  We next ensure
  that (Coverage\Ck{}) holds for a path \APath{} terminating in \UAk{}.

Consider a path \APath{} in \ddAk{} and the corresponding path $\PPath{\ddAkP{}}$
in \ddAkP{}. If on a machine state $\sigma$, both paths \APath{}
and $\PPath{\ddAkP{}}$ transition to \UAk{} and \UAkP{} respectively,
then because \XkP{} satisfies (Coverage\Ck{}), $\sigma$ must execute one
of $\CPathN{j}$ (for $1\leq{}j\leq{}m$) to completion, thus
satisfying (Coverage\Ck{}) in \Xk{} in this case. Thus, we only need to
cater to the two situations where execution on \ddAk{} may deviate from \ddAkP{}:
\begin{itemize}
\item \TRule{\Ret{\Ak}}:
  Let $\Phi=\neg \prjMEq{\il{\ddAk}{cs}}{\ghost{M^{cs}}}{M_{\ddAk}}$ be the
  check in \ddAk{} (due to \TRule{\Ret{\Ak{}}}), that has been
  replaced with $\Phi'=\mathtt{false}$ in $\ddAkP{}$.
  We show that
  $\Phi$ must evaluate to {\tt false}
  in \Xk{} at procedure return. In other words, the \ddAk{} path ``{\tt if $\Phi$ \Ihalt{\mathscr{U}}}''
  is infeasible (and so \ddAk{} does not deviate from \ddAkP{} in this case).

  By \cref{lemma:csIsPreserved},
  $\prjMEq{\il{\ddAk^{\prime}}{cs}}{\ghost{M^{cs}}}{M_{\ddAk^{\prime}}}$
  holds at every non-error node $n_{\Xk{}} \in{} \NNP{\Xk}$.
  Further,
  using the (MAC) requirement at the non-error terminating node {\tt exit}, this can be generalized to show that
  $\prjMEq{\il{\ddAk^{\prime}}{cs}}{\ghost{M^{cs}}}{M_{\ddAk^{\prime}}}$ holds
  at the beginning of the path corresponding to \TRule{\Ret{\Ak}} in \ddAkP{}.
  Thus, because the \ddAk{} path ``{\tt if $\Phi$ \Ihalt{\mathscr{U}}}'' is infeasible,
  (Coverage\Ck{}) holds trivially for this path at $n_{\Xk}$ in \Xk{}.

\item \TRule{\Load{\ddAk}} or \TRule{\Store{\ddAk}}:
  Let $\PPath{\ddAk}^{U} =\PathT{n_{\ddAk}}{\UAk}$
  be a path that terminates with \UAk{}.
  \begin{lemma}
  \label{lemma:CtermAtU}
  Let $\sigma$ be a state at a non-error node $n_{\Xk}=(n_{\ddAk},n_{\Ck})\in{}\NNP{\Xk}$ such that $\Inv{n_{\Xk}}(\sigma)$ holds
  and $\sigma$ executes $\PPath{\ddAk}^{U}=\PathT{n_{\ddAk}}{\UAk{}}$ to completion.
  Then $\sigma$ must execute some path
  $\CPath=\PathT{n_{\Ck}}{\UCk{}}$ to completion in \Ck{}.
  \end{lemma}
  \begin{proof}[Proof of \cref{lemma:CtermAtU}]

  Consider the execution of $\sigma$ on \XkP{} starting at $n_{\XkP}=(n_{\ddAkP},n_{\Ck})$,
  such that $n_{\XkP}$ in \XkP{} is structurally identical to $n_{\Xk}$ in \Xk{}.
  Due to (Mutex\ddAk) and (Coverage\ddAk),
  there can be only two cases:
\begin{enumerate}
  \item $\sigma$ executes some path $\PPath{\ddAkP}^{x}=\PathT{n_{\ddAkP}}{\UAk{}}$ to completion in \ddAkP{}.
  In this case, due to (Coverage\Ck{}) and (Safety),
  some $\CPathN{x}=\PathT{n_{\Ck}}{\UCk{}}$ must be executed to completion on $\sigma$ in \Ck{}.
  In this case, the lemma holds with $\CPath=\CPathN{x}$.




  \item $\sigma$ executes some path $\PPath{\ddAkP}^{x}=\PathT{n_{\ddAkP}}{n^{x}_{\ddAkP}}$
    to completion in \ddAkP{}, where $n^{x}_{\ddAkP}\neq{}\EU_{\ddAkP{}}$
    and $e^{x_v}_{\XkP} = (\XEdgeT{n_{\XkP}}{\PPath{\ddAkP}^{x}; \CPathN{x_v}}{n^{x_v}_{\XkP}}) \in \EP{\XkP}$
    (for $1 \leq v \leq w$) are $w\geq{}1$ edges in \XkP{}, where $n^{x_v}_{\XkP}=(n^{x}_{\ddAkP},n^{x_v}_{\Ck})$.
Because \XkP{} satisfies (Coverage\Ck{}),
$\sigma$ must execute a path $\CPathN{x_v}=\PathT{n_{\Ck}}{n^{x_v}_{\Ck}}$ to completion in \Ck{}, for some $1\leq{}v\leq{}w$.
We show by contradiction that $\forall{1\leq{}v\leq{}w}:{n^{x_v}_{\Ck}=\UCk{}}$ must hold.

Assume $n^{x_v}_{\Ck}\neq{}\UCk{}$.
Let memory access instructions $d_1$, $d_2$, \ldots, $d_k$ exist on path
$\PPath{\ddAkP}^{x}$, such that
$\PPath{\ddAkP}^{x}$ deviates from
$\PPath{\ddAk}^{U}$
on one of these memory access instructions $d_{r}$ ($1\leq{}r\leq{}k$), so
that $\PPath{\ddAk}^U$ transitions to \UAk{}
due to $\varphi{}$ evaluating to
{\tt true} in a check ``{\tt if $\varphi{}$ \Ihalt{\mathscr{U}}}'' in a
\TRule{\Load{\ddAk}} or
\TRule{\Store{\ddAk}} in \ddAk{},
while $\PPath{\ddAkP}^{x}$
continues execution to reach $n^{x}_{\ddAkP}\neq{}\EU_{\ddAkP}$
due to $\varphi{}'$ evaluating to
{\tt false} in a corresponding check ``{\tt if $\varphi{}'$ \Ihalt{\mathscr{U}}}''
in \ddAkP{}.

Let $[p]_{w}$ represent
the addresses being accessed by the memory access
instruction $d_r$.
It must be true that
$\exists{\alpha{}\in{}[p]_{w}:\alpha \in \mathtt{comp}(\il{\ddAk'}{\NS \cup \F \cup \STACK})}$
if $d_r$ is a {\tt load} instruction
and
$\exists{\alpha{}\in{}[p]_{w}:\alpha \in \mathtt{comp}(\il{\ddAk'}{(\NS \setminus \Gro) \cup \Frw \cup \STACK})}$
if $d_r$ is a {\tt store} instruction;
this is because
$\varphi^{\prime}$ evaluates to {\tt false} but $\varphi$ evaluates to {\tt true}
(for \TRule{\Load{\ddAk}} and \TRule{\Store{\ddAk}} instructions).
Because \XkP{} satisfies (MAC), the execution of $\sigma$ starting
at $n_{\Ck{}}$ must cause all addresses in
$[p]_{w}$ to be accessed
before execution can reach $n^{x_v}_{\Ck}$ in \Ck{} (and $n^{x_v}_{\XkP}$ in \XkP{}).
Further, because $\PPath{\ddAkP}^x$ contains a memory access instruction,
due to \cref{lemma:memAccessAllocPreserved},
both $\PPath{\ddAkP}^x$ and $\CPath^x$ cannot modify the address sets of common regions \NS{}.
Thus, during the execution of $\sigma$ starting at $n_{\Ck{}}$,
the $\mathtt{accessIsSafeC}_{\tau,a}$ check must necessarily evaluate to {\tt false}
and the execution must transition to \UCk{}. This is a contradiction, and
so it must be true that $n^{x_v}_{\Ck}=\UCk{}$.
Hence, the lemma holds in this case with
$\CPath{}=\CPathN{x}=\PathT{n_{\Ck}}{\UCk{}}$.
\end{enumerate}

\end{proof}

Using \cref{lemma:CtermAtU},
we enumerate all such paths $\CPath=\PathT{n_{\Ck}}{\UCk{}}$ that
can be executed in \Ck{} if $\PPath{\ddAk}^{U}=\PathT{n_{\ddAk}}{\UAk{}}$
is executed in \ddAk{}
starting at node $n_{\Xk{}}\in{}\NXk{}$.
As described in the proof of \cref{lemma:CtermAtU},
there are only a finite
number of such paths.
For each such path $\CPath$,
we add an
edge $e^{x}_{\Xk}=(\XEdgeT{n_{\Xk}}{\PPath{\ddAk}^{U}; \CPath{}}{(\EU_{\ddAkP}, \UCk{})})$
to \EXk{} if it does not exist already.
(Coverage\Ck{}) thus follows from \cref{lemma:CtermAtU}.
Further, (Coverage\ddAk{}) also holds for \Xk{} because all assembly paths that exist in \XkP{} also exist in \Xk{} and additional paths, only potentially feasible in \ddAk{}, are added.


\end{itemize}
\end{proof}

\subsection{Soundness of Interval Encoding}
\label{app:intervalEncodingProof}


Let the
Hoare triple representation of a proof obligation $\proofObligation{}$
generated by \toolName{}
be $\hoareTriple{pre}{(\APath;\CPath)}{post}$, where
$\APath = \PathT{n_{\ddAk}}{n^t_{\ddAk}}$,
and
either $\CPath = \epsilon$ or
$\CPath = \PathT{n_{\Ck}}{n^t_{\Ck}}$,
$n_{\Xk} = (n_{\ddAk},n_{\Ck}) \in \NNP{\Xk}$ is a non-error node,
$n_{\Xk}^t = (n^t_{\ddAk},n^t_{\Ck}) \in \NXk$,
if $\CPath = \PathT{n_{\Ck}}{n^t_{\Ck}}$, then
$\XEdge = (\XEdgeT{n_{\Xk}}{\APath; \CPath}{n^t_{\Xk}}) \in \EXk$,
and,
$\APath$ and $\CPath$ are I/O-free
execution paths
in \ddAk{} and \Ck{} respectively.

Let $n^0_{\ddAk}, n^1_{\ddAk}, n^2_{\ddAk}, \ldots, n^m_{\ddAk}$
be the nodes on
path $\APath=\PathT{n_{\ddAk}}{n^t_{\ddAk}}$,
such that $n^0_{\ddAk}=n_{\ddAk}$ and $n^m_{\ddAk}=n^t_{\ddAk}$.
Let $\espmin(\APath)$ represent the the
minimum value of $\mathtt{esp}$ observed at any
node $n^j_{\ddAk}$ ($0 \leq j\leq m$) visited
during the execution of path \APath{}.
Similarly, let
$\zlunion(\APath)$ be the union of the values
of set $\il{\ddAk}{\Z{lv}}$ observed at any $n^j_{\ddAk}$ ($0\leq j\leq m$)
visited during \APath{}'s execution.

Let $\HP={\tt comp}(\il{\ddAk}{G\cup\F}\cup\zlunion(\APath)\cup[\espmin(\APath),\ghost{\stkE}])$,
$\CL=[\ghost{\spE}+1_{\bv{32}},\ghost{\stkE}]\setminus \zlunion(\APath)$,
and
$\CS=[\ghost{\spE}+1_{\bv{32}},\ghost{\stkE}]\cap \zlunion(\APath)$.

Let $\proofObligation^{\prime}=\hoareTriple{pre}{(\APath; \CPath)}{post}$ be
obtained by strengthening precondition $pre$ to
$pre^{\prime}=pre
\land (\il{\ddAk}{\heap}=\HP)
\land (\il{\ddAk}{cl}=\CL)
\land (\il{\ddAk}{cs} = \CS)$
in $\proofObligation{}^{\prime}$.
We need to show that $\proofObligation{}\Leftrightarrow{}\proofObligation{}^{\prime}$ holds.

($\Rightarrow$) Proving $\proofObligation{}\Rightarrow{}\proofObligation{}^{\prime}$ is trivial,
as $\proofObligation{}^{\prime}$ requires a stronger
precondition than $\proofObligation{}$ (with everything else identical).

($\Leftarrow$) Assume that $\proofObligation{}^{\prime}$ holds. We are
interested in showing that $\proofObligation{}$ holds.
Assume a machine state $\sigma$ of product program \Xk{}
that satisfies the weaker precondition $pre$, and
executes to
completion over $\APath$
and $\CPath$.
We are interested in showing that $\sigma$ satisfies the postcondition
$post$
after completing the execution.

We define ``error-free execution'' to be the case where the execution on a state
$\sigma$ across $(\APath;\CPath)$ does not end at an error node in \Xk{}.

\begin{lemma}[$\HP$,$\CL$ overapproximate $\heap$,$cl$]
  \label{lemma:HCLlargest}
  $(\il{\ddAk}{\heap}\subseteq{}\HP)\land{}(\il{\ddAk}{cl}\subseteq{}\CL)$ holds on $\sigma{}$ for an error-free execution.
\end{lemma}
\begin{proof}
  Recall that $pre\Rightarrow{}\phi_{n_\Xk{}}$.
  If $\il{\ddAk}{\heap}\supset{}\HP$
  or $\il{\ddAk}{cl}\supset{}\CL$,
  then either at least one of
  \apred{NoOverlap\Ak{}} or \apred{NoOverlap\Ck{}} will
  evaluate to {\tt false} in $\phi_{n_{\Xk{}}}$ (and $pre$),
  or during
  the execution of path \APath; error
  $\mathscr{W}$ will be triggered in \ddAk{} because either the allocation
  of stack space through stackpointer decrement
  will overstep $\il{\ddAk}{\{\heap,cl\}}$ \TRule{Op-esp'},
  or the virtual allocation
  of a local variable will overstep $\il{\ddAk}{\{\heap,cl\}}$ \TRule{AllocV}.
  However, by assumption, $\sigma{}$ satisfies $pre$ (and $\phi_{n_{\Xk{}}}$) and
  executes
  $\APath$ and $\CPath$ to completion to a non-error node; thus proved by contradiction.
\end{proof}

\begin{lemma}[$\CS$ underapproximates $cs$]
  \label{lemma:CSSmallest}
  $(\il{\ddAk}{cs}\supseteq \CS)$ holds on $\sigma$ for an error-free execution.
\end{lemma}
\begin{proof}
Follows from
\cref{lemma:HCLlargest} and $\il{\ddAk}{cs}=[\ghost{\spE}+1_{\bv{32}}, \ghost{\stkE}] \setminus \il{\ddAk}{cl}$ \TRule{\Entry{\Ak}}.
\end{proof}

\begin{lemma}[$\HP$ and $\CL$ borrow from the \free{} and $cs$ regions]
  \label{lemma:HminusHpIsFree}

  The following hold on $\sigma$ for an error-free execution.
  \begin{enumerate}
    \item $(\HP\setminus{}\il{\ddAk}{\heap})\subseteq{}\il{\ddAk}{\free{}}\subseteq{}\il{\Ck}{\free}$
    \item $(\CL\setminus{}\il{\ddAk}{cl})\subseteq{}\il{\ddAk}{cs}\subseteq{}\il{\Ck}{\free}$
  \end{enumerate}
\end{lemma}
\begin{proof}
  The proof follows from the definition of $\HP$ and $\CL$, as these sets are not allowed to overlap
  with $\il{\ddAk}{{\NS}\cup\F\cup\STACK}$ or $\il{\Ck}{{\NS}\cup\F\cup\STACK}$.
\end{proof}

Construct a state $\sigma'$ that is identical to $\sigma$ with the
following modifications made in sequence:
\begin{enumerate}
  \item The region identified by addresses (that would
    belong to region \free{} in \Ck{} by \cref{lemma:HminusHpIsFree})
    $(\HP\cup{}\CL)\setminus{}\il{\ddAk}{\{\heap,cl\}}$ in $\sigma'$'s $M_\Ck{}$
    is updated through $M_{\Ck{}} := \updM{(\HP\cup{}\CL)\setminus{}\il{\ddAk}{\{\heap,cl\}}}{M_{\Ck{}}}{M_{\ddAk{}}}$.
  \item The address sets
    $\il{\ddAk}{\heap}$,
    $\il{\ddAk}{cl}$,
    $\il{\Ck}{\heap}$, and
    $\il{\Ck}{cl}$
    are expanded
    and the address set \il{\ddAk}{cs} is shrunk
    so that
    $\il{\ddAk}{\heap}=\il{\Ck}{\heap}=\HP$,
    $\il{\ddAk}{cl}=\il{\Ck}{cl}=\CL$,
    and
    $\il{\ddAk}{cs}=\CS$
    (this involves
    the transfer of addresses from the {\tt free} region (\cref{lemma:HminusHpIsFree})
    to $\heap$ and $cl$ regions in \Ck{}, and from the {\tt free} and $cs$ regions
    to $\heap$ and $cl$ regions respectively in \ddAk{}).
\end{enumerate}
The constructed state $\sigma'$ thus satisfies
the stronger precondition $pre^{\prime}$.

Let $\il{\sigma{}}{\heap}$ ($\il{\sigma'}{\heap}$),
$\il{\sigma{}}{cl}$ ($\il{\sigma'}{cl}$),
$\il{\sigma{}}{cs}$ ($\il{\sigma'}{cs}$), and
$\il{\sigma{}}{\mathtt{free}}$ ($\il{\sigma'}{\mathtt{free}}$)
denote the values of
$\il{\ddAk{}}{\heap}$,
$\il{\ddAk{}}{cl}$,
$\il{\ddAk{}}{cs}$, and
$\il{\ddAk{}}{\mathtt{free}}$ in state $\sigma$ ($\sigma'$) respectively.
Similarly, let $M^{\sigma}_{\ddAk}$ ($M^{\sigma}_{\Ck}$)
and $M^{\sigma'}_{\ddAk}$ ($M^{\sigma'}_{\Ck}$) represent the state of
procedure \ddAk{}'s (\Ck{}'s) memory $M_{\ddAk}$ ($M_{\Ck{}}$)
in machine states $\sigma$ and $\sigma'$ respectively.

To relate $\sigma$ and $\sigma'$, we define relation
$sim(\sigma{},\sigma{}')$
as the conjunction
of the following conditions:
\begin{enumerate}
  \item (\heap{} subset in $\sigma{}$) $\il{\sigma{}}{\heap} \subseteq{}\il{\sigma{}'}{\heap}$.
  \item ($cl$ subset in $\sigma{}$) $\il{\sigma{}}{cl} \subseteq{}\il{\sigma{}'}{cl}$.
  \item ($cs$ superset in $\sigma{}$) $\il{\sigma{}}{cs} \supseteq{}\il{\sigma{}'}{cs}$.
  \item ($\mathtt{free}$ superset in $\sigma{}$) $\il{\sigma{}}{\mathtt{free}} \supseteq{}\il{\sigma{}'}{\mathtt{free}}$.
  \item (\ddAk{}'s memory states are equal)
    $M^{\sigma}_{\ddAk} = M^{\sigma{}'}_{\ddAk}$
  \item (\Ck{}'s memory states are equal except at the updated regions)
    $\prjMEq{\mathtt{comp}(\il{\sigma'}{\{\heap,cl\}}\setminus{}\il{\sigma{}}{\{\heap,cl\}})}{M^{\sigma{}}_{\Ck{}}}{M^{\sigma{}'}_{\Ck{}}}$.
  \item The remaining state elements have equal values in $\sigma$ and $\sigma{}'$.
\end{enumerate}
By construction, $sim(\sigma{},\sigma')$ holds.

\begin{lemma}[$sim(\sigma,\sigma')$ is preserved for error-free execution across all non-I/O edges in $\mathcal{E}_{\ddAk{}}$]
  \label{lemma:singleAsmInstruction}
  If a non-I/O edge $e_{\ddAk{}}\in{}\mathcal{E}_{\ddAk{}}$
  is executed
  on both machine states $\sigma$ and $\sigma'$, and
  if $sim(\sigma,\sigma')$ holds before the execution,
  and if the execution on $\sigma$ completes without error,
  then there exists a sequence of non-deterministic choices during the
  execution on $\sigma'$ such that the execution is error-free
  and
  $sim(\sigma,\sigma')$ holds at the end of both error-free executions.
\end{lemma}
\begin{proof}
  For each non-I/O \ddAk{} instruction that does not refer to the
  $\{\heap,cl,cs,\mathtt{free}\}$ regions
  (\TRule{Op-Nesp},\TRule{AllocS}, \TRule{DeallocS},
  \TRule{\Call{\dAk}},
  \TRule{\Ret{\Ak}},
  \TRule{DeallocV}),
the execution will have identical behaviour on
both $\sigma$ and $\sigma'$, as identical values will be observed in
$\sigma$ and $\sigma'$. Thus, if an execution on $\sigma'$ makes
the same non-deterministic choice as the execution on $\sigma{}$, the execution
on $\sigma'$ will complete without error and $sim(\sigma,\sigma')$ will
hold at the end of both executions.

We consider each remaining non-I/O instruction in \ddAk{} below:
\begin{itemize}
  \item \TRule{\Entry{\ddAk}}.
  Consider the overlap conditions
  $\Upsilon{}_1=\Overlap(\il{\ddAk}{\heap},\il{\ddAk}{cl},\ldots,\ii{\ddAk}{\g}, \ldots, \il{\ddAk}{\f},\ldots,\ii{\ddAk}{\y},\ldots, \il{\ddAk}{\yv})$ (due to $\neg{}\mathtt{addrSetsAreWF}$),
  $\Upsilon{}_2=\Overlap([\mathtt{esp},\mathtt{esp}+3_{\bv{32}}], \il{\ddAk}{\NS{}\cup \F})$,
  $\Upsilon{}_3=\Overlap([\ghost{\spE}+1_{\bv{32}},\ghost{\stkE}], \il{\ddAk}{\{\heap\} \cup G \cup \F})$,
  and
  $\Upsilon{}_4=\Overlap(\il{\ddAk}{cl}, \mathtt{comp}([\ghost{\spE}+1_{\bv{32}}, \ghost{\stkE}]))$ (due to $\mathtt{stkIsWF}$).
During an execution
on $\sigma$, all four conditions must evaluate to {\tt false}, as
we assume an error-free execution on $\sigma$.
For the same non-deterministic
choices made in both executions (over $\sigma$ and $\sigma'$), by
the definitions of $\HP$ and $\CL$, $\Upsilon{}_1$, $\Upsilon{}_2$, $\Upsilon{}_3$, and $\Upsilon{}_4$
will also evaluate
to {\tt false} for an execution on $\sigma'$
--- recall that \HP{} cannot overlap with $[\mathtt{esp},\ghost{\stkE}]$ (which
includes the arguments)
and global variable regions
(due to \cref{lemma:HminusHpIsFree});
and
$\CL$ is a subset of $[\ghost{\spE}+1_{\bv{32}}, \ghost{\stkE}]$ (by definition).
Further, because all other
state elements observed during the execution of the non-I/O edges in \TRule{\Entry{\ddAk}}
are identical in both $\sigma$ and $\sigma'$, $sim(\sigma,\sigma')$ will
hold at the end of error-free executions.

\item \TRule{Op-esp}.
  The negated subset
check $\Upsilon{}=\neg{}([t, \mathtt{esp}-1_{\bv{32}}]\subseteq{}\il{\ddAk}{\mathtt{free}}\cup{}\ilZv{\ddAk})$
(due to $\neg{}\intervalContainedInAddrSet{t}{\allowbreak \mathtt{esp}-1_{\bv{32}}}{\il{\ddAk}{\mathtt{free}}\cup{}\ilZv{\ddAk}}$)
depends (indirectly) on the addresses of the set $\il{\ddAk}{\{\heap,cl\}}$
(as {\tt free} is defined as complement of the allocated region).
The execution on $\sigma$ must evaluate $\Upsilon{}$
to {\tt false} as we assume an error-free execution.
By the definitions of $\HP$ and $\CL$,
for the same non-deterministic choices made in both executions (over $\sigma$ and $\sigma'$),
$\Upsilon{}$ will
also evaluate to {\tt false} for an execution on $\sigma'$ --- recall that $(\HP\cup{}\CL)$
cannot overlap with $[\espmin{}(\APath),\ghost{\spE{}}]$, and the latter includes 
$[t, \mathtt{esp}-1_{\bv{32}}]$.
All other state elements observed in the other instructions of \TRule{Op-esp}
are identical in both $\sigma$, $\sigma'$ and $sim(\sigma,\sigma')$ will hold
at the end of error-free executions.

\item \TRule{AllocV}.
Consider the negated subset
check $\Upsilon{}=\neg{}([v]_{w}\subseteq{}\il{\ddAk}{\mathtt{comp}(\il{\ddAk}{\NS})})$
(due to $\neg{}\intervalContainedInAddrSetAndAligned{v}{v+w-1_{\bv{32}}}{\il{\ddAk}{\mathtt{comp}(\il{\ddAk}{\NS})}}{a}$).
The execution on $\sigma$ must evaluate $\Upsilon{}$
to {\tt false} as we assume an error-free execution.
By the definitions of $\HP$ and $\CL$,
for the same non-deterministic choices made in both executions (over $\sigma$ and $\sigma')$,
$\Upsilon{}$ will
also evaluate to {\tt false} for an execution on $\sigma'$ --- recall that $(\HP\cup{}\CL)$
cannot overlap with
$\zlunion{}(\xi_{\ddAk})$,
and the latter includes the interval $[v]_{w}$.
All other state elements observed in the other instructions of \TRule{AllocV}
are identical in both $\sigma$, $\sigma'$
and $sim(\sigma,\sigma')$ will hold at the end of
error-free executions.

\item \TRule{\Load{\ddAk}} and \TRule{\Store{\ddAk}}.
The overlap checks,
$\Overlap([p]_{w}, (\ilZv{\ddAk}) \setminus (\il{\ddAk}{\F} \cup [\mathtt{esp},\ghost{\stkE}]))$
for \TRule{\Load{\ddAk}}
and
$\Overlap([p]_{w}, (\ilZv{\ddAk}) \setminus (\il{\ddAk}{\Frw} \cup [\mathtt{esp},\ghost{\stkE}]))$
for \TRule{\Store{\ddAk}},
in the modified semantics
of \TRule{\Load{\ddAk}} and \TRule{\Store{\ddAk}}
will evaluate to {\tt false} for $\sigma$
due to the assumption
of error-free execution.
As these checks do not refer to the potentially
modified regions $\{\heap,cl,cs,\mathtt{free}\}$,
$\sigma'$ must also evaluate the check to {\tt false}
(for the same sequence of non-deterministic choices).
Notice that this reasoning relies on the safety-relaxed semantics,
and would not hold on the original semantics.
All other state elements observed in the other instructions of \TRule{\Load{\ddAk}}
and \TRule{\Store{\ddAk}} are identical in both $\sigma$, $\sigma'$
and $sim(\sigma,\sigma')$ will hold at the end of
error-free executions.
\end{itemize}
\end{proof}

Recall that the \toolName{} algorithm
populates the deterministic choice map \DXk{}
such that the result of the {\em choose} instruction (\Ichoose{\bv{32}})
for $\alpha_b$ in an
{\tt alloc} instruction in \CPath{}
matches the address $v$ in an $\mathtt{alloc}_{s,v}$ instruction in \APath{}
and
the result of the {\em choose} instruction for memory contents (\IchooseM{})
of the freshly allocated interval $[\alpha_b, \alpha_e]$
matches the memory contents of the interval $[v]_w$
(in the {\tt alloc} and $\mathtt{alloc}_{s,v}$ instructions respectively).
We use this fact in the following theorem on the execution of \CPathD{}.

\begin{lemma}[$sim(\sigma,\sigma')$ is preserved for error-free execution across all non-I/O edges in $\mathcal{E}_{\Ck{}}$]
\label{lemma:singleCInstruction}
If a non-I/O edge $e_{\Ck{}} \in \mathcal{E}_{\Ck}$ in the path \CPathD{}
is executed
on both machine states $\sigma$ and $\sigma'$, and
if $sim(\sigma,\sigma')$ holds before the execution,
and if the execution on $\sigma$, with non-deterministic choices determinized by \DXk, completes without error,
then, for the same sequence of non-deterministic choices, the
execution on $\sigma'$ completes without error
and
$sim(\sigma,\sigma')$ holds at the end of both error-free executions.
\end{lemma}
\begin{proof}
For a non-I/O \Ck{} instruction that does not refer to the $\{\heap,cl,cs,\mathtt{free}\}$
regions
(\TRule{Op}, \TRule{AssignConst}, \TRule{Dealloc}, \TRule{VaStartPtr},
\TRule{CallV}, \TRule{\Call{\Ck}},
\TRule{\Ret{\Ck}}, \TRule{RetV}),
the execution will have identical behaviour on
both $\sigma$ and $\sigma'$ as identical values will be observed in
both $\sigma$ and $\sigma'$. Thus, if an execution on $\sigma'$ makes
the same non-deterministic choice as the execution on $\sigma{}$, the exection
on $\sigma'$ will complete without error and $sim(\sigma,\sigma')$ will
hold at the end of both executions.

We consider each remaining non-I/O instruction in \Ck{} below:
\begin{itemize}

  \item \TRule{\Entry{\Ck}}
  Consider the overlap check $\Upsilon{}=\Overlap(\il{\Ck}{\heap},\il{\Ck}{cl},\ldots,\ii{\Ck}{\g},\ldots,\il{\Ck}{\f},\ldots,\ii{\Ck}{\y},\ldots, \il{\Ck}{\yv})$ (due to $\neg{}\mathtt{addrSetsAreWF}$).
During an execution on $\sigma$,
this condition must evaluate to {\tt false}, as
we assume an error-free execution on $\sigma$.
For the same non-deterministic
choices made in both executions (over $\sigma$ and $\sigma'$), by
the definitions of $\HP$ and $\CL$, $\Upsilon{}$
will also evaluate
to {\tt false} for an execution on $\sigma'$
--- recall that $(\HP\cup{}\CL)$ cannot overlap with other allocated regions
(due to \cref{lemma:HminusHpIsFree}).
Further, because all other
state elements observed during the execution of the non-I/O edges in \TRule{\Entry{\Ck}}
are identical in both $\sigma$ and $\sigma'$, $sim(\sigma,\sigma')$ will
hold at the end of error-free executions.

\item \TRule{Alloc}
Consider the negated subset
check $\Upsilon=\neg([\alpha_b, \alpha_e]\subseteq\il{\Ck}{\mathtt{free}})$
(due to
$\neg\intervalContainedInAddrSetAndAligned{\alpha_b}{\alpha_e}{\il{\Ck}{\mathtt{free}}}{a}$).
The execution on $\sigma$ must evaluate $\Upsilon{}$
to {\tt false} as we assume an error-free execution.
By the definitions of $\HP$ and $\CL$,
for the same non-deterministic choices made in both executions (over $\sigma$ and $\sigma'$),
$\Upsilon$ will
also evaluate to {\tt false} for an execution on $\sigma'$ --- recall that
during execution on $\sigma$,
the deterministic choice map \DXk{} will be used
for the non-deterministc choices
of address $\alpha_b$ and memory \prjM{[\alpha_b,\alpha_e]}{M_{\Ck}}
such that the 
freshly allocated interval $[\alpha_b,\alpha_e]$ matches
(in both address and data)
the allocated interval $[v]_w$ in an $\mathtt{alloc}_{s,v}$
instruction in \APath{};
because
the same \DXk{} is used in both $\sigma$ and $\sigma'$
executions, $\Upsilon$ will also evaluate to \texttt{false} in $\sigma'$.
All other state elements observed in the other instructions of \TRule{Alloc}
are identical in both $\sigma$, $\sigma'$.

\item \TRule{\Load{\Ck}} and \TRule{\Store{\Ck}}.
An $\mathtt{accessIsSafeC}_{\tau,a}()$ check must evaluate
to {\tt true} for $\sigma$ due to the
assumption of error-free execution.  Because the allocated
space $\il{\Ck}{\NS}$ can only be
bigger in $\sigma'$ (by \cref{lemma:HCLlargest}), the {\tt accessIsSafeC} check will
also evaluate to {\tt true} for $\sigma'$ (for the
same sequence of non-deterministic choices).
Further, for an execution on $\sigma$,
the contents
of the memory region
$\prjM{\il{\sigma'}{\{\heap,cl\}}\setminus\il{\sigma}{\{\heap,cl\}}}{M^\sigma_{\Ck}}$
cannot be observed on an error-free path; and because
all other state elements observed in \TRule{\Load{\Ck}}
and \TRule{\Store{\Ck}} are identical in both $\sigma$ and $\sigma'$,
the contents
of the memory region
$\prjM{(\il{\sigma'}{\{\heap,cl\}}\setminus\il{\sigma}{\{\heap,cl\}}}{M^{\sigma'}_{\Ck}}$
will also remain unobserved during an execution on $\sigma'$ (that uses
the same sequence of non-deterministic choices as an execution on $\sigma$).
All other state elements observed in the other instructions of \TRule{\Load{\Ck}}
and \TRule{\Store{\Ck}} are identical in both $\sigma$, $\sigma'$.

\end{itemize}
\end{proof}

\begin{lemma}[$sim(\sigma,\sigma')$ is preserved for error-free execution across $\APath;\CPath$]
\label{lemma:simPreservedAcrossPath}
Recall that $\xi_{\ddAk}$ contains only non-I/O instructions (by assumption). Thus,
due to the (SingleIO)
requirement, $\xi_{\Ck}$ also contains only non-I/O instructions.

If $\xi_{\ddAk}$ is executed on machine
states $\sigma$ and $\sigma'$, and if
the execution of $\sigma$ completes without error,
then there exists a sequence of non-deterministic choices during the execution
of $\sigma'$ such that the execution is error-free
and $sim(\sigma{},\sigma')$
holds
at the end of both error-free executions.

Similarly,
if $\xi_{\Ck}$ is next executed on machine
states $\sigma$ and $\sigma'$, and if
the execution of $\sigma$ completes without error,
then there exists a sequence of non-deterministic choices during the execution
of $\sigma'$ such that the execution is error-free
and $sim(\sigma{},\sigma')$
holds
at the end of both error-free executions.
\end{lemma}
\begin{proof}
To show this, we execute the
sequence of paths ($\xi_{\ddAk};\xi_{\Ck{}}$) in lockstep
on both $\sigma$ and $\sigma'$, i.e., in a single
step, one instruction is executed
on both states modifying the states in place.
The proof proceeds by induction on the number of steps. The base case holds by assumption.
For the inductive step, we rely on \cref{lemma:singleAsmInstruction,lemma:singleCInstruction}.
\end{proof}

\begin{lemma}[$\sigma$ and $\sigma{}'$ execute the same path in \ddAk{}]
\label{lemma:Apath}
If \APath{} executes to completion on state $\sigma$, it will also execute to
completion on $\sigma{}'$.
\end{lemma}
\begin{proof}
By case analysis on all edge conditions in \cref{fig:xlateRuleAsm}.
For $\APath{}=n_{\ddAk}\twoheadrightarrow{}\UAk{}$,
the proof relies on the safety-relaxed semantics,
and would not hold on the original semantics.
\end{proof}

\begin{lemma}[$\sigma$ and $\sigma{}'$ execute the same non-\EU{} path in \Ck{}]
\label{lemma:Cpath}
If \CPath{} does not terminate
in \UCk{}, and $\sigma$ executes \CPath{} to completion, then $\sigma{}'$ will also
execute \CPath{} to completion.
\end{lemma}
\begin{proof}
By case analysis on all edge conditions in \cref{fig:xlateRuleIR}
with same arguments as used in \cref{lemma:singleCInstruction}.
\end{proof}

\begin{lemma}[$post(\sigma') \land sim(\sigma, \sigma') \Rightarrow post(\sigma)$ holds for a non-error node $(n_{\ddAk}^t,n_{\Ck}^t)$]
\label{lemma:postHolds}
For two states $\sigma$ and $\sigma'$ at node $(n^t_{\ddAk},n^t_{\Ck})$,
where $n^t_{\ddAk}$ and $n^t_{\Ck}$ are non-error nodes,
$post(\sigma')\land{}sim(\sigma,\sigma')\Rightarrow{}post(\sigma)$ holds.
\end{lemma}
\begin{proof}
  The $post$ condition that may appear in
  a Hoare Triple proof obligation
  generated by \toolName{} can be one of the following:
  \begin{itemize}
    \item (Coverage\Ck) where
      $post=\bigvee_{1 \leq j \leq m}{\pathcond{\CPathDN{j}}}$
      for $e^j_{\Xk} = \XEdgeT{(n_{\ddAk},n_{\Ck}}{\APath{};\CPathN{j}}{(n^t_{\ddAk},n^t_{\Ck})} \in \EXk$ ($1 \leq j \leq m$).

    \item (Inductive) where $post$ is
      one of the predicate shapes listed in \cref{fig:invGrammar}.
      Note that the \pred{MemEq} shape in \cref{fig:invGrammar} represents the
      proof obligation for the (MemEq) requirement.

    \item (Equivalence) where $post$ is either $\World{\ddAk}=\World{\Ck}$ or
      $\steq{\Trace{\ddAk}}{\Trace{\Ck}}$.
      I/O free paths do not mutate world states so 
      $\World{\ddAk}=\World{\Ck}$ cannot appear as $post$.
      Further,
      the only I/O free paths that may modify trace
      must contain {\tt halt} instruction, appearing as
      the last edge of the sequence.
      As the generated trace event for {\tt halt}
      does not observe any procedure state variable,
      we ignore this case.

    \item (MAC) where $post$ checks the
      address of each memory access in \ddAk{}
      against the addresses of a set of memory accesses in \Ck{} for equality.
      Also, (MAC) checks if a memory access overlaps with address
      regions $\il{\ddAk}{G\cup{}F}\cup{}[\mathtt{esp},\ghost{\spE}]$
      or $\il{\ddAk}{G_w\cup{}F_w}\cup{}[\mathtt{esp},\ghost{\spE}]$.

  \end{itemize}

  \noindent
  {\em Case: When $post$ is one of the predicate shapes in \cref{fig:invGrammar} or is a (MAC) proof obligation}.
  \begin{itemize}
  \item The predicate shapes
  \pred{affine}, \pred{ineqC}, \pred{ineq}, \pred{spOrd}, \pred{zEmpty}, \pred{spzBd}, \pred{spzBd'},
  and a (MAC) proof obligation do not involve operations over address sets
  $\{ \heap, cl, cs, \mathtt{free} \}$ or
  memory operations in the updated region
  $\il{\sigma'}{\{\heap,cl\}} \setminus \il{\sigma'}{\{\heap,cl\}}$.
  Thus, $post(\sigma')\land{}sim(\sigma,\sigma')\Rightarrow{}post(\sigma)$ holds in this case.

  \item Consider the case when $post$ is \pred{AllocEq}.
    Due to (Equivalence), \pred{AllocEq} is guaranteed to in $pre$
    and therefore
  $\il{\ddAk}{\heap} = \il{\Ck}{\heap}$ and $\il{\ddAk}{cl} = \il{\Ck}{cl}$
  must
  hold for $\sigma'$.
  Due to $sim(\sigma, \sigma')$, $\sigma$ and $\sigma'$ agree on the remaining state elements,
  including the address sets for each region $\z{} \in \Z{}$.
  Thus, $post(\sigma')\land{}sim(\sigma,\sigma')\Rightarrow{}post(\sigma)$ holds in this case.

  \item Consider the case when $post$ is \pred{MemEq}.
  $sim(\sigma, \sigma')$ ensures that the address sets of regions
  $\{ \heap, cl \}$ in $\sigma$ are a subset of respective address sets in $\sigma'$.
  Further, due to $sim(\sigma,\sigma')$,
  the memory states of \Ak{} in $\sigma$ and $\sigma'$ are identical,
  $M^{\sigma}_{\ddAk} = M^{\sigma'}_{\ddAk}$,
  and the memory states of \Ck{} in $\sigma$ and $\sigma'$ disagree
  only over the updated (expanded) address sets,
  $\prjMEq{\mathtt{comp}(\il{\sigma'}{\{\heap,cl\}} \setminus \il{\sigma}{\{\heap,cl\}})}{M^{\sigma}_{\Ck}}{M^{\sigma'}_{\Ck}}$.
  Because the allocated regions in $\sigma{}$ belong to these
  addresses, $post(\sigma)$ follows from $post(\sigma')$.

  \end{itemize}

  \noindent
  {\em Case: When $post$ is a proof obligation for (Coverage\Ck)}.
  In this case, $post$ must be of the form
  $\bigvee_{1 \leq j \leq m}{\pathcond{\CPathDN{j}}}$
  for $e^j_{\Xk} = \XEdgeT{(n_{\ddAk},n_{\Ck})}{\APath{};\CPathN{j}}{(n^t_{\ddAk},n^t_{\Ck})} \in \EXk$ ($1 \leq j \leq m$).
  The
  edge conditions in \Ck{} are independent of
  the
  regions $\{ \heap, cl, cs, \mathtt{free} \}$, except for \TRule{\Load{\Ck}} and \TRule{\Store{\Ck}}.
  If the edge condition is independent of these address regions, then $post(\sigma)$
  follows trivially from $post(\sigma')$.
  For a non-error node, the maximal set of paths
  $\{\CPathN{1},\ldots,\CPathN{m}\}$ includes both the
  paths that evaluate
$\mathtt{accessIsSafeC}_{\tau,a}$ to {\tt true} and {\tt false} respectively.  Thus, even in
  this case, $post(\sigma{})$ holds if $post(\sigma{}')$ holds.

\end{proof}

\begin{lemma}[$post(\sigma') \Rightarrow post(\sigma)$ for $n^t_{\ddAk}=\WAk{}$]
\label{lemma:postHoldsForW}
For two states $\sigma$ and $\sigma'$ at node $(\WAk{},n_{\Ck}^t)$,
$post(\sigma')\Rightarrow{}post(\sigma)$ holds.
\end{lemma}
\begin{proof}
  The $post$ condition of this type may appear in
  a Hoare Triple proof obligation
  generated by \toolName{} for one of the following:
  \begin{itemize}
    \item (Coverage\Ck) where
      $post=\bigvee_{1 \leq j \leq m}{\pathcond{\CPathDN{j}}}$
      for $e^j_{\Xk} = \XEdgeT{(n_{\ddAk},n_{\Ck}}{\APath{};\CPathN{j}}{(n^t_{\ddAk},n^t_{\Ck})} \in \EXk$ ($1 \leq j \leq m$).

    \item (MAC) where $post$ checks the
      address of each memory access in \ddAk{}
      against the addresses of a set of memory accesses in \Ck{} for equality.
      Also, (MAC) checks if a memory access overlaps with address
      regions $\il{\ddAk}{G\cup{}F}\cup{}[\mathtt{esp},\ghost{\spE}]$
      or $\il{\ddAk}{G_w\cup{}F_w}\cup{}[\mathtt{esp},\ghost{\spE}]$.
  \end{itemize}
  The proof arguments for both these cases are identical to the ones made in the proof for \cref{lemma:postHolds}.
\end{proof}

\begin{lemma}[$post(\sigma') \Rightarrow post(\sigma)$ for $n^t_{\ddAk}=\UAk{}$]
\label{lemma:postHoldsForError}
For two states $\sigma$ and $\sigma'$ at node $(\UAk{},n_{\Ck}^t)$,
$post(\sigma')\Rightarrow{}post(\sigma)$ holds.
\end{lemma}
\begin{proof}
  The $post$ condition of this type may appear in
  only one type of proof obligation
  generated by \toolName{}:
  \begin{itemize}
    \item (Coverage\Ck) where
      $post=\bigvee_{1 \leq j \leq m}{\pathcond{\CPathDN{j}}}$
      for $e^j_{\Xk} = \XEdgeT{(n_{\ddAk},n_{\Ck})}{\APath{};\CPathN{j}}{(n^t_{\ddAk},n^t_{\Ck})} \in \EXk$ ($1 \leq j \leq m$).
  \end{itemize}



  Let the (Coverage\Ck{}) proof obligation
  be $\hoareTriple{\phi_{n_{\Xk}}}{(\APath;\epsilon{})}{\bigvee_{1 \leq j \leq m}{\pathcond{\CPathDN{j}}}}$,
  where $n_{\Xk}=(n_{\ddAk},n_{\Ck})$.
  Due to (Safety),
  each path $\CPathN{j}$ must end at \UCk{}.

From the semantics in \cref{fig:xlateRuleIR},
if the path condition for $\CPathN{j}$
evaluates to {\tt true} on
$\sigma{}'$ (for some $j$), it must also
evaluate to {\tt true} on $\sigma$ --- in other words, whenever
$\sigma{}'$ transitions to \UCk{}, $\sigma{}$ is guaranteed
to transition to \UCk{}.
This is because the edge
conditions in \Ck{} will evaluate either identically on $\sigma$
and $\sigma'$
(due to $\{\CPathN{1},\ldots,\CPathN{m}\}$ being a maximal set),
or in the case of
$\neg{}\mathtt{accessIsSafeC}_{\tau,a}()$,
the edge condition will evaluate to {\tt true} on $\sigma$
if it evaluates to {\tt true} on $\sigma'$
(due to \heap{},$cl$ subset in $\sigma$).
%

Thus, if $post(\sigma')$ evaluates to {\tt true}, $post(\sigma)$ also
evaluates to {\tt true}.
\end{proof}

\begin{proof}[Proof for ($\Leftarrow$)]
  Follows from \cref{lemma:simPreservedAcrossPath,lemma:Apath,lemma:Cpath,lemma:postHolds,lemma:postHoldsForW,lemma:postHoldsForError}.
\end{proof}

\begin{proof}[Proof of \cref{theorem:rewriteHeapCl}]
Follows from ($\Rightarrow$) and ($\Leftarrow$).
\end{proof}

\subsection{Encoding of address set relations}
\label{app:derived_encoding}

\begin{table}
\renewcommand{\arraystretch}{1.2} 
\begin{footnotesize}
  \begin{tabularx}{\textwidth}{|p{6cm}|X|}
  \hline
  Relation & Encoding using $\alpha \in \il{P}{r}$
  \\
  \hline
  \hline
  $\alpha \in \il{P}{\R{}}  \hfill \R{} \subseteq \Rall{}$
            & $\bigvee\limits_{r\in \R{}} \alpha \in \il{P}{r}$
  \\
  \hline
  $\forall_{r \in \NS}{\il{\Ck}{r} = \il{\ddAk}{r}}$
   	        & $\forall{\alpha}: (\alpha \in \il{\Ck}{\NS{}} \Leftrightarrow \alpha \in \il{\ddAk}{\NS{}})$
  \\
  \hline
  $\il{P}{r} = \emptyset$
            & $\forall{\alpha}: \neg (\alpha \in \il{P}{r})$
  \\
  \hline
  $(\ghost{\LB{\z{}}} = \LBr{\il{\Ck}{\z{}}} \land \ghost{\UB{\z{}}} = \UBr{\il{\Ck}{\z{}}})$
   	        & $\forall{\alpha}: \alpha \in \il{\Ck}{\z{}} \Rightarrow (\ghost{\LB{\z{}}} \le_u \alpha \le_u \ghost{\UB{\z{}}})$
  \\
  \hline
  $\Overlap([\alpha_{b},\alpha_{e}], \il{P}{\R{}}) \hfill \R{} \subseteq \Rall{}$
            & $\exists{\alpha}:({\alpha_b \leq_u \alpha \leq_u \alpha_e }) \land \alpha \in \il{P}{\R{}}$
  \\
  \hline
  $[\alpha_{b},\alpha_{e}] \subseteq \il{P}{\R{}} \hfill \R{} \subseteq \Rall{}$
            & $\forall{\alpha}:({\alpha_b \leq_u \alpha \leq_u \alpha_e }) \Rightarrow \alpha \in \il{P}{\R{}}$
  \\
  \hline
  $[\alpha_{b},\alpha_{e}] = \il{P}{r}$
            & $\forall{\alpha}:(\alpha_{b} \leq_u \alpha \leq_u \alpha_{e}) \Leftrightarrow \alpha \in \il{P}{r})$
  \\
  \hline
  $\il{\ddAk}{\{\stk\}\cup\Y} \cup (\il{\ddAk}{\Z{}} \setminus (\ilZv{\ddAk})) = [\mathtt{esp},\ghost{\spE}]$
            & $\forall{\alpha}: (\alpha \in \il{\ddAk}{\{\stk\} \cup \Y} \lor (\alpha \in \il{\ddAk}{\Z{}} \land \neg (\alpha \in \ilZv{\ddAk}))) \Leftrightarrow (\mathtt{esp} \leq_u \alpha \leq_u \ghost{\spE})$
  \\
  \hline
  $\il{\ddAk}{\{cs,cl\}} = [\ghost{\spE}+1,\ghost{\stkE}]$
            & $\forall{\alpha}: (\alpha \in \il{\ddAk}{\{cs,cl\}}) \Leftrightarrow (\ghost{\spE}+1 \leq_u \alpha \leq_u \ghost{\stkE})$
  \\
  \hline
\end{tabularx}
\end{footnotesize}
\caption{\label{tab:addressSetEncodings} Encodings of address set relations using the address set membership predicate.
  $\Rall{}$ is the set of all region identifiers.
}
\end{table}

\Cref{tab:addressSetEncodings} shows the encodings of various address set relations
using the address set membership predicate, $\alpha \in \il{P}{r}$.
These encodings follow from the definition of each relation in a straightforward manner.

\begin{table}
  \begin{footnotesize}
\renewcommand{\arraystretch}{1.2} 
    \begin{tabularx}{\textwidth}{|p{0.40\textwidth}|X|}
  \hline
  Instruction & SMT encoding using \MemallocP{P} \\
  \hline
  \hline
  $\il{P}{r} \Assign \il{P}{r} \cup [\alpha_b,\alpha_e];$ \hfill $r \in \{ \stk \} \cup \Z{}$
      & ${\MemallocP{P}}' = \updMA{\MemallocP{P}}{x \in [\alpha_b,\alpha_e]}{r}$
      \\
  \hline
  $\il{P}{\z{}} \Assign \emptyset;$
      & ${\MemallocP{P}}' = \updMA{\MemallocP{P}}{\selectMath_1(\MemallocP{P}, x) = \z{}}{\mathtt{free}}$
      \\
  \hline
  $\il{\ddAk}{\stk} \Assign \il{\ddAk}{\stk} \setminus [\alpha_b,\alpha_e];$
      & ${\MemallocP{\ddAk}}' = \updMA{\MemallocP{\ddAk}}{x \in [\alpha_b,\alpha_e]}{\mathtt{free}}$
      \\
  \hline
  $\il{\ddAk}{\stk} \Assign \{ [\mathtt{esp}, \ghost{\spE}]\} \setminus \il{\ddAk}{\Y}; $
      & ${\MemallocP{\ddAk}}' = \updMA{\MemallocP{\ddAk}}{x \in [\mathtt{esp},\ghost{\spE}] \land \bigwedge_{\y \in \Y} x \notin \il{\ddAk}{\y} }{\stk}$
      \\
  \hline
\end{tabularx}
  \end{footnotesize}
\caption{\label{tab:instructionEncodings}
SMT encoding of address set updating instructions using allocation state array \MemallocP{P}.  $P \in \{ \Ck, \ddAk \}$.  ${\MemallocP{P}}'$ is the allocation state array after executing the instruction.}
\end{table}

\Cref{tab:instructionEncodings} shows the allocation state array \MemallocP{P} based SMT encoding
of the transfer functions of the transition graph instructions that involve address sets --- these encodings follow
from the definition of an allocation state array.
The interval SMT encodings
utilize ghost variables $\ghost{\Empty{\z}}$, $\ghost{\LB{\z}}$, $\ghost{\UB{\z}}$ (as shown in \cref{tab:intervalEncodings})
and the update logic for these ghost variables is available in \cref{fig:xlateRuleIR}.

Given an input allocation state array \MemallocP{P}, an address set updating instruction produces a new allocation state array
${\MemallocP{P}}'$.
To show the encodings in \cref{tab:instructionEncodings}, we use an auxiliary operator, \updMCSym{},
to encode the update of an allocation state array $\MemallocP{P}$:
if ${\MemallocP{P}}' = \updMA{\MemallocP{P}}{c}{v}$,
then,
\begin{align*}
  \forall{\alpha}: \phantom{{} \land{} \neg} (\lambda x {.} c) (\alpha) &\Rightarrow \selectMath_1({\MemallocP{P}}',\alpha) = v
  \\
  {} \land{} \neg (\lambda x {.} c) (\alpha) &\Rightarrow \selectMath_1({\MemallocP{P}}',\alpha) = \selectMath_1(\MemallocP{P}, \alpha)
\end{align*}
Here, $(\lambda x {.} c)$ represents a function that takes as input value $x$ and returns a boolean
value evaluated through expression $c$, and $(\lambda x {.} c) (\alpha)$ represents the application of
this function to value $\alpha$.
Thus, \updMA{\MemallocP{P}}{c}{v} represents the modification of allocation
state array $\MemallocP{P}$ to value $v$ for all
addresses $\alpha$ that satisfy
the boolean condition $c$.
In other words, \updMA{\MemallocP{P}}{c}{v} is equivalent to
$$
\store_1(\ldots\store_1(\ldots\store_1(\MemallocP{P}, \alpha_1, v),\ldots, \alpha_i, v),\ldots, \alpha_n, v)
$$
for all $\alpha_1,\ldots,\alpha_i,\ldots,\alpha_n$ where the predicate $c$ holds.

As an example, in \cref{tab:instructionEncodings}, $\il{P}{\z{}} \Assign \il{P}{\z{}} \cup [\alpha_b,\alpha_e]$ is encoded as
${\MemallocP{P}}' = \updMA{\MemallocP{P}}{x \in [\alpha_b,\alpha_e]}{\z{}}$
which translates to
``mark the addresses in interval $[\alpha_b,\alpha_e]$ as belonging to region \z{} in ${\MemallocP{P}}'$''.

\subsection{Reasons for failures}
\label{app:failReasons}

\begin{table}[H]
    \centering
    \begin{tabularx}{\textwidth}{|p{2.5cm}|p{2cm}|X|}
        \hline
        Benchmark & Compiler & Failure reason \\
        \hline
        \hline
        \texttt{vcu}       & GCC & \multirow{2}{=}{SMT query timeout during affine invariant inference}  \\
        \cline{1-2}
        \texttt{vcu}       & ICC & \\
        \hline
        \texttt{vsl}       & GCC & Limitation of $\mathtt{dealloc}_s$ annotation --- see \cref{sec:asmAnnotLimitation} \\
        \hline
        \texttt{vsl}       & \multirow{6}{=}{ICC} & \multirow{6}{=}{Non-affine invariant required --- see \cref{sec:nonAffineInv}} \\
        \cline{1-1}
        \texttt{vilcc}     & & \\
        \cline{1-1}
        \texttt{vilce}     & & \\
        \cline{1-1}
        \texttt{fib}       & & \\
        \cline{1-1}
        \texttt{rod}       & & \\
        \hline
        \texttt{min}       & GCC & Incompleteness in affine invariant inference due to the chosen set of procedure variables --- see \cref{sec:invInterestingExprs} \\
        \hline
    \end{tabularx}
    \caption{\label{tab:fails}Failure reasons for benchmarks shown in \cref{fig:graph_lt}.}
\end{table}

\Cref{tab:fails} lists the failures and their reasons for benchmarks in \cref{fig:graph_lt}.
We discuss the reasons for failures in detail in following sections.

\subsubsection{Limitation of the $\mathtt{alloc}_s/\mathtt{dealloc}_s$ annotation algorithm in the blackbox setting}
\label{sec:asmAnnotLimitation}

As mentioned in \cref{sec:algo}, in the blackbox setting,
when hints from the compiler are not available, the annotation algorithm ({\em asmAnnotOptions})
limits the insertion of a $\mathtt{dealloc}_s$ instruction to only those PCs that occur
just before an instruction that updates the stackpointer register {\tt esp}.
This limitation may cause a refinement proof to fail in some (not all)
of the situations
where a compiler implements merging of multiple allocations (deallocations) into a single stackpointer
decrement (increment) instruction. This is the reason for the failure to validate GCC's compilation
of {\tt vsl}.

\begin{figure}[H]
\begin{subfigure}[b]{.35\textwidth}
\begin{myexamplefnsize}
~{\tiny \textcolor{mygray}{C0:}}~ int vsl(int n)
~{\tiny \textcolor{mygray}{C1:}}~ {
~{\tiny \textcolor{mygray}{C2:}}~   if (n <= 0)
~{\tiny \textcolor{mygray}{C3:}}~     return 0;
~{\tiny \textcolor{mygray}{C4:}}~   int v[n];
~{\tiny \textcolor{mygray}{C5:}}~   for (int i = 0; i < n; ++i) {
~{\tiny \textcolor{mygray}{C6:}}~     v[i] = i*(i+1);
~{\tiny \textcolor{mygray}{C7:}}~   }
~{\tiny \textcolor{mygray}{C8:}}~   return v[0]+v[n-1];
~{\tiny \textcolor{mygray}{C9:}}~ }
\end{myexamplefnsize}
	\caption{\label{fig:deallocSunkC} C source code}
\end{subfigure}
\begin{subfigure}[b]{.6\textwidth}
	\centering
  	\includegraphics{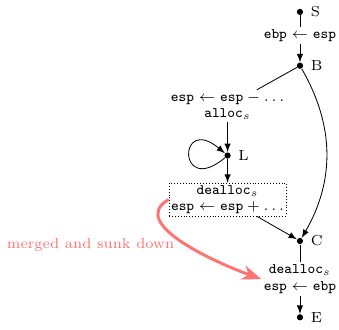}
	\caption{\label{fig:deallocSunkCFG} Abbreviated control-flow graph (CFG) of GCC compiled assembly}
\end{subfigure}
\caption{\label{fig:deallocSunk} {\tt vsl} procedure from \cref{tab:benchmarks} and its CFG of GCC compiled assembly.}
\end{figure}

\Cref{fig:deallocSunk} shows the {\tt vsl} procedure (\cref{fig:deallocSunkC}) and a sketch of the
CFG of the assembly procedure generated by GCC at O3 optimization level (\cref{fig:deallocSunkCFG}).
The assembly path $S\rightarrow{}B\rightarrow{}C\rightarrow{}E$ represents the
case when $n<=0$ and the procedure exits early (without allocating any local variable).
PC with label $L$ represents the loop head in the assembly procedure, and the allocation and
deallocation
of the VLA {\tt v} is supposed to happen before entering the loop and after exiting the loop respectively.

On the assembly procedure path $L\rightarrow{}C$, the assembly instruction that reclaims stack space (by
incrementing the stackpointer) for deallocating {\tt v} has been merged with an instruction that
restores the stackpointer to its value at the beginning of the function ({\tt ebp}). Thus, while
the original stackpointer increment instruction would have been at the end of the $L\rightarrow{}C$
path, the
merged instruction is sunk by the compiler to lie within the path $C\rightarrow{}E$.
As can be seen, this transformation saves an extra instruction
to update the stack pointer on the path $L\rightarrow{}C\rightarrow{}E$.

In the absence of compiler hints (blackbox setting), our tool considers
the annotation of a {\tt dealloc$_s$} instruction in assembly
only at a PC that immediately precedes an instruction that updates the stackpointer.
In this example, the only candidate PC for annotating {\tt dealloc$_s$} (considered by
our blackbox algorithm) is on
the path $C\rightarrow{}E$.
However,
the required position of the {\tt dealloc$_s$} instruction was at
the end of the path $L\rightarrow{}C$ (which is not considered
because there is no instruction that updates the
stackpointer at the end of the path $L\rightarrow{}C$). Thus, our blackbox
algorithm cannot find a refinement proof.
On the other hand, providing a manual hint to the tool that it should consider annotating a {\tt dealloc$_s$}
instruction at the end of the $L\rightarrow{}C$ path causes the algorithm
to successfully return a refinement proof for GCC's compilation of {\tt vsl}.

It is worth asking the question: what happens if the tool simply annotates a {\tt dealloc$_s$} instruction
just before the instruction that updates the stackpointer to {\tt ebp} on the $C\rightarrow{}E$ path?
Such an annotation violates the stuttering trace equivalence condition on the procedure
path $S\rightarrow{}B\rightarrow{}C\rightarrow{}E$: in the C procedure, there is no deallocation (or allocation)
on the early exit path ($n<=0$), but this annotation will
cause a {\tt dealloc$_s$} instruction to be executed on the correlated path ($S\rightarrow{}B\rightarrow{}C\rightarrow{}E$)
in
the assembly procedure. Because a {\tt dealloc$_s$} instruction
generates a trace event through the {\tt wr} instruction, this candidate annotation
therefore fails to show the equivalence of traces on at least one pair of correlated paths. Thus, this
candidate annotation
is discarded by our algorithm.

\subsubsection{Non-affine invariant shape requirement in some ICC benchmarks}
\label{sec:nonAffineInv}

Some compilations of VLA containing code by ICC have a certain assembly code pattern
which require a particular non-affine invariant shape for completing the
refinement proof.

For allocation of a VLA, ICC uses the following sequence of instructions for decrementing the stackpointer:
\begin{align*}
  reg_1 &\gets \mathtt{``Allocation\ size\ in\ bytes"}
  \\
  reg_2 &\gets (reg_1 + \mathtt{C}) \mathrel{\&} {\sim \mathtt{C}}
  \\
  \mathtt{esp} &\gets \mathtt{esp} - reg_2
\end{align*}
Here, $reg_1$ and $reg_2$ are assembly registers (other than $\mathtt{esp}$),
$\mathtt{esp}$ is the stackpointer register, and $\mathtt{C}$ is a bitvector constant. 
The value in $reg_1$ is the allocation size of VLA in bytes; it matches the corresponding allocation
size in the C procedure.  For example, for a VLA declaration \texttt{int v[n]}, $reg_1$ would have value
\texttt{n*4} (recall that 4 is the size of an \texttt{int} in 32-bit configuration).
The value in $reg_2$ is the allocation amount after adjusting for alignment requirements, e.g.,
\texttt{v} (of \texttt{int v[n]}) would have an alignment of at least 4 in 32-bit x86.

At time of deallocation, the stackpointer register is simply incremented by the same value as during allocation:
\begin{align*}
  reg_1 &\gets \mathtt{``Allocation\ size\ in\ bytes"}
  \\
  reg_2 &\gets (reg_1 + \mathtt{C}) \mathrel{\&} {\sim \mathtt{C}}
  \\
  \mathtt{esp} &\gets \mathtt{esp} + reg_2
\end{align*}

Thus, in order to prove that the stack deallocation consumes only the stack region (and
procedure does not go to \UAk{}),
we must have an invariant stating that the current stackpointer value, $\mathtt{esp}$,
is at least $reg_2$ bytes below the stackpointer value at assembly procedure entry,
$\ghost{\spV{entry}}$ (recall that $\ghost{\spV{entry}}$ holds the {\tt esp} value at entry
of the assembly procedure $\ddAk$ and is guaranteed to be a part of the stack region).
However, the value in $reg_2$ has a non-affine relationship with the allocation size,
(which is tracked in a ghost variable).
This non-affine relationship cannot be captured by any shape in our predicate grammar for candidate
invariants at a product graph node.

Therefore, we fail to prove that the deallocated region was part of stack and, consequently,
fail to
prove that the assembly procedure will not go to \UAk{} during deallocation
(if the \Ck{} procedure does not go to \UCk{}).

Note that, however, our invariants shapes are able to capture the invariant
that stack is large enough to allocate $reg1$ bytes through the \pred{spOrd}
shape in the predicate grammar (\cref{fig:invGrammar}).
This is required to ensure that \ddAk{} does not go to \UAk{} during allocation.


\subsubsection{Choice of program variables for invariant inference}
\label{sec:invInterestingExprs}

In the \pred{affine} invariant shape $\sum_{i}{c_iv_i}=c$ of the predicate grammar
(\cref{fig:invGrammar}),
the program variables $v_i$ are chosen from a
set $V$ that includes the pseudo-registers in $\Ck$
and registers and stack slots in $\Ak$.
The candidate variables for correlation in $V$ do not include
``memory slots'' in $\Ck$ of the
shape {\tt \select$_\sizeMath{}${}(mem,$\alpha$)}
(little-endian concatenation
of \size{} bytes starting at $\alpha$ in the array {\tt mem})
to avoid
an explosion in the number of candidate invariants (and consequently
the running time of the algorithm).

This causes a failure while
validating the GCC compilation (at O3)
of the {\tt min} benchmark ({\tt minprintf} function from K\&R \cite{knr}).
GCC register-allocates the \texttt{va\_list} variable (that maintains the
current position in the variadic argument). On the
other hand, the LLVM$_d$ IR maintains this
pointer value in a local variable (allocated using an {\tt alloca} instruction)
--- the loads and stores to this local variable
$\llangle ap \rrangle$ can be seen in \cref{fig:c2IRVariadic}.
Thus,
for a refinement proof to succeed,
a validator must relate the assembly register's value with the value stored inside
the local variable's memory region ({\tt \select$_\sizeMath{}${}(mem,\&va\_list\_var)}).
Because our invariant inference algorithm does
not consider memory slots, this required relation is not identified, resulting in
a proof failure.

It may be worth asking the question: why does our choice of program
variables work for the other benchmarks?
Due to the \texttt{mem2reg} pass used in $\Ck$
before computing equivalence,
the only memory
slots that remain in procedure $\Ck$ pertain to potentially
address-taken variables.
Our requirements on the product graph \Xk{}
ensure that the memory regions corresponding to address-taken
local variables (and global variables) of $\Ck$ and $\Ak$
are equated in \Xk{}.
Thus, relating the {\em addresses} of potential memory accesses in $\Ck$
and $\Ak$ using affine invariants
and considering only the memory slots from $\Ak$
largely suffices for invariant inference to validate most compilations
(but not for GCC's compilation of {\tt min}).

\subsection{Evaluation of programs with multiple VLAs}

\begin{table}[t]
  \begin{footnotesize}
\centering
\begin{tabularx}{\textwidth}{lccccccccc}
\toprule
Name                       & {\tt ALOC} & \# of locals & Time (s) & Nodes  & Edges  & \# of Qs & Avg. Q time (s) & Max. Q time (s) & Avg. inf. inv./node \\
\midrule
{\tt vil1}                 & 33         & 1            & 305      & 8      & 9      & 1923     & 0.09            & 6.40            & 75.5  \\
{\tt vil2}                 & 35         & 2            & 692      & 12     & 13     & 2498     & 0.17            & 4.37            & 93.6  \\
{\tt vil3}                 & 37         & 3            & 1295     & 16     & 17     & 3468     & 0.23            & 16.83           & 120.7 \\
{\tt vil4}                 & 41         & 4            & 6617     & 20     & 21     & 10026    & 0.37            & 129.31          & 166.7 \\
\bottomrule
\end{tabularx}
\caption{\label{tab:vil}Statistics obtained by running \toolName{} on functions with variable number of locals in a loop.} 
  \end{footnotesize}
\end{table}

\Cref{tab:vil} shows the quantitative results for validating the GCC O3 translations
of \texttt{vil1}, \texttt{vil2}, \texttt{vil3}, and \texttt{vil4}, containing one, two, three, and four VLA(s) in a loop respectively.
The general structure of the programs is shown in \cref{fig:vilN}.
\Cref{tab:vil} shows the run time in seconds (Time (s)), number of product graph nodes and edges (Nodes and Edges),
number of SMT queries (\# of Qs), average and maximum query time in seconds (Avg. Q time (s) and Max. Q time (s)),
and average number of invariants inferred at a product graph node (Avg. inf. inv./node) for the four programs. As the search space increases, the search algorithm takes longer.  Further, each SMT proof obligation also increases in size (and complexity) as the number of inferred invariants (at a node) that participate in an SMT query (translated from a Hoare triple) increase with the number of VLAs.

\begin{figure}[H]
  \begin{subfigure}[t]{.3\textwidth}
    \hfill
  \end{subfigure}
  \begin{subfigure}[t]{.7\textwidth}
\begin{myexamplesmall}
int vil~$\mathscr{N}$~(unsigned n)
{
  int r = 0;
  for (unsigned i = 1; i < n; ++i) {
    int v1[4*i], v2[4*i], ~{\ldots}~, v~$\mathscr{N}$~[4*i];
    r += foo~$\mathscr{N}$~(v1,v2,~{\ldots}~,v~$\mathscr{N}$~, i);
  }
  return r;
}
\end{myexamplesmall}
  \end{subfigure}
  \caption{\label{fig:vilN} General structure of the programs in \cref{tab:vil}. $\mathscr{N}$ can be substituted with 1,2,3,4 to obtain \texttt{vil1}, \texttt{vil2}, \texttt{vil3}, \texttt{vil4} respectively.}
\end{figure}

We discuss the validation of \texttt{vil3} in more detail through \cref{fig:example2}.  The addresses of \texttt{v2} and \texttt{v3} depend on the address and allocation size of
\texttt{v1}, which are different in each loop iteration. Our algorithm identifies an annotation of the assembly program such that relations between local variable addresses (in C) and stack addresses (in assembly) can be identified. These address relations rely on a lockstep correlation of the annotated (de)allocation instructions in assembly with the originally present (de)allocation instructions in C.  The positions and the arguments of the annotated {\tt alloc$_s$} and {\tt dealloc$_s$} instructions in \cref{fig:example2a} determine these required address relations.

\begin{figure}
\begin{minipage}[t]{0.4\textwidth} 
\begin{subfigure}[t]{\textwidth}
\begin{myexamplesmall}
int vil3(unsigned n)
{
  int r = 0;
  for (unsigned i = 1; i < n; ++i) {
   int v1[4*i];
   int v2[4*i];
   int v3[4*i];
   r += foo(v1,v2,v3,i);
  }
  return r;
}
\end{myexamplesmall}
\caption{\label{fig:example2c}{\tt vil3} C Program}
\end{subfigure}
\begin{subfigure}[t]{\textwidth}
\begin{myexamplesmallir}
~{\tiny \textcolor{mygray}{I0: }}~ vil3(unsigned* n):
~{\tiny \textcolor{mygray}{I1: }}~   r=0;
~{\tiny \textcolor{mygray}{I2: }}~   i=1;
~{\tiny \textcolor{mygray}{I3: }}~   if(i ~$>=_u$~ *n) goto I22;
~{\tiny \textcolor{mygray}{I4: }}~     ~$p_{I4}$~=alloc 4*i,int;
~{\tiny \textcolor{mygray}{I5: }}~     ~$p_{I5}$~=alloc 4*i,int;
~{\tiny \textcolor{mygray}{I6: }}~     ~$p_{I6}$~=alloc 4*i,int;
~{\tiny \textcolor{mygray}{I7: }}~     ~$p_{I7}$~=alloc 1,int*;
~{\tiny \textcolor{mygray}{I8: }}~     ~$p_{I8}$~=alloc 1,int*;
~{\tiny \textcolor{mygray}{I9: }}~     ~$p_{I9}$~=alloc 1,int*;
~{\tiny \textcolor{mygray}{I10:}}~     ~$p_{I10}$~=alloc 1,int;
~{\tiny \textcolor{mygray}{I11:}}~     *~$p_{I7}$~=~$p_{I4}$~; *~$p_{I8}$~=~$p_{I5}$~; *~$p_{I9}$~=~$p_{I6}$~; *~$p_{I10}$~=i;
~{\tiny \textcolor{mygray}{I12:}}~     t=call int foo(~$p_{I7}$~, ~$p_{I8}$~, ~$p_{I9}$~, ~$p_{I10}$~);
~{\tiny \textcolor{mygray}{I13:}}~     dealloc ~$p_{I7}$~;
~{\tiny \textcolor{mygray}{I14:}}~     dealloc ~$p_{I8}$~;
~{\tiny \textcolor{mygray}{I15:}}~     dealloc ~$p_{I9}$~;
~{\tiny \textcolor{mygray}{I16:}}~     dealloc ~$p_{I10}$~;
~{\tiny \textcolor{mygray}{I17:}}~     r = r + t;
~{\tiny \textcolor{mygray}{I18:}}~     dealloc ~$p_{I6}$~;
~{\tiny \textcolor{mygray}{I19:}}~     dealloc ~$p_{I5}$~;
~{\tiny \textcolor{mygray}{I20:}}~     dealloc ~$p_{I4}$~;
~{\tiny \textcolor{mygray}{I21:}}~     i++; goto I3;
~{\tiny \textcolor{mygray}{I22:}}~   ret r;
\end{myexamplesmallir}
\caption{\label{fig:example2i}(Abstracted) IR of {\tt vil3} (after the mem2reg pass)}
\end{subfigure}
\end{minipage}
\hfill
\begin{minipage}[t]{0.55\textwidth}
\begin{subfigure}[t]{\textwidth}
\begin{myexamplesmallasm}
~{\tiny \textcolor{mygray}{A0: }}~ vil3:
~{\tiny \textcolor{mygray}{A1: }}~   push ebp~;~ ebp = esp~;~
~{\tiny \textcolor{mygray}{A2: }}~   push {edi, esi, ebx}~;~ esp = esp-12~;~ esi = 0~;~
~{\tiny \textcolor{mygray}{A3: }}~   if(mem~$_4$~[ebp+8] ~$\le_u$~ 1) jmp A18
~{\tiny \textcolor{mygray}{A4: }}~   ebx = 1~;~
~{\tiny \textcolor{mygray}{A5: }}~     edi = esp~;~
~{\tiny \textcolor{mygray}{A6: }}~     eax = ebx~;~ eax = eax << 4~;~ ; eax = 4*4*i
~{\tiny \textcolor{mygray}{A7: }}~     esp = esp - eax~;~
~{\tiny \textbf{\textcolor{red}{A7.1: }}}~    ~\textbf{\textcolor{red}{alloc$_s$ esp,eax,4,I4; }}~ ; allocation of v1
~{\tiny \textcolor{mygray}{A8: }}~     edx = esp~;~
~{\tiny \textcolor{mygray}{A9: }}~     esp = esp - eax~;~
~{\tiny \textbf{\textcolor{red}{A9.1: }}}~    ~\textbf{\textcolor{red}{alloc$_s$ esp,eax,4,I5; }}~ ; allocation of v2
~{\tiny \textcolor{mygray}{A10:}}~     ecx = esp~;~
~{\tiny \textcolor{mygray}{A11:}}~     esp = esp - eax~;~
~{\tiny \textbf{\textcolor{red}{A11.1: }}}~   ~\textbf{\textcolor{red}{alloc$_s$ esp,eax,4,I6; }}~ ; allocation of v3
~{\tiny \textcolor{mygray}{A12:}}~     eax = esp~;~
~{\tiny \textcolor{mygray}{A13:}}~     push {ebx, eax, ecx, edx}~;~
~{\tiny \textbf{\textcolor{red}{A13.1:}}}~    ~\textbf{\textcolor{red}{alloc$_s$ esp,\ \ \ \ 4, 4, I7; }}~
~{\tiny \textbf{\textcolor{red}{A13.2:}}}~    ~\textbf{\textcolor{red}{alloc$_s$ esp+4,\ \ 4, 4, I8; }}~
~{\tiny \textbf{\textcolor{red}{A13.3:}}}~    ~\textbf{\textcolor{red}{alloc$_s$ esp+8,\ \ 4, 4, I9; }}~
~{\tiny \textbf{\textcolor{red}{A13.4:}}}~    ~\textbf{\textcolor{red}{alloc$_s$ esp+12, 4, 4, I10; }}~
~{\tiny \textcolor{mygray}{A14:}}~     call ~\textbf{\textcolor{teal}{int}}~ foo~\textbf{\textcolor{teal}{(int* esp, int* esp+4, int* esp+8, unsigned esp+12)}};~
~{\tiny \textbf{\textcolor{red}{A14.1:}}}~    ~\textbf{\textcolor{red}{dealloc$_s$ I10;}}~
~{\tiny \textbf{\textcolor{red}{A14.2:}}}~    ~\textbf{\textcolor{red}{dealloc$_s$ I9;}}~
~{\tiny \textbf{\textcolor{red}{A14.3:}}}~    ~\textbf{\textcolor{red}{dealloc$_s$ I8;}}~
~{\tiny \textbf{\textcolor{red}{A14.4:}}}~    ~\textbf{\textcolor{red}{dealloc$_s$ I7;}}~
~{\tiny \textcolor{mygray}{A15:}}~     esi = esi + eax~;~ ebx = ebx + 1~;~
~{\tiny \textbf{\textcolor{red}{A15.1:}}}~    ~\textbf{\textcolor{red}{dealloc$_s$ I6;}}~
~{\tiny \textbf{\textcolor{red}{A15.2:}}}~    ~\textbf{\textcolor{red}{dealloc$_s$ I5;}}~
~{\tiny \textbf{\textcolor{red}{A15.3:}}}~    ~\textbf{\textcolor{red}{dealloc$_s$ I4;}}~
~{\tiny \textcolor{mygray}{A16:}}~     esp = edi~;~
~{\tiny \textcolor{mygray}{A17:}}~     if(mem~$_4$~[ebp+8] ~$\ne$~ ebx) jmp A5~;~
~{\tiny \textcolor{mygray}{A18:}}~   esp = ebp-12~;~ eax = esi~;~
~{\tiny \textcolor{mygray}{A19:}}~   pop {ebx, esi, edi, ebp}~;~
~{\tiny \textcolor{mygray}{A20:}}~   ret~;~
\end{myexamplesmallasm}
\caption{\label{fig:example2a}(Abstracted) 32-bit x86 Assembly Code for {\tt vil3}.}
\end{subfigure}
\end{minipage}
\caption{\label{fig:example2}{\tt vil3} program with three VLAs in a loop and its lowerings to IR and assembly.
Subscript $_u$ denotes unsigned comparison.
Bold font (parts of) instructions are added by our algorithm.
}

\end{figure}

 \subsection{Soundness and Completeness Implications of {\tt isPush()} Choice}
 \label{app:isPush}

An update to the stackpointer {\tt esp} in the assembly procedure \Ak{}
can be through any arbitrary instruction, such as {\tt esp $\Assign$ $\mathbb{Y}$}.
If the previous {\tt esp} value, just before this instruction was executed, was
$\mathbb{X}$,
then the stackpointer update distance
is $\mathbb{D} = \mathbb{X}-\mathbb{Y}$.  In general,
it is impossible to tell whether this instruction intends a stack growth
by $\mathbb{D}$ bytes (push) or a shrink by $(2^{32}-\mathbb{D})$ bytes (pop).
The modeling for the two cases is different: for stack push,
an overlap of the interval representing the push with non-stack region causes a
\EW{} error,
while for stack pop, the stackpointer going outside stack region causes \EU{} error.
Refinement is trivially proven if \Ak{} terminates with \EW{} error.
Unfortunately, this seems impossible to disambiguate just
by looking at the assembly code --
to tackle this dilemma, we assume an oracle
function,
${\tt isPush}(\PC{\Ak}{j}, \mathbb{X}, \mathbb{Y})$, that
returns {\tt true} iff the assembly instruction at PC \PC{\Ak}{j}
represents a stack push.

In \cref{sec:translationA}, we define an
$\mathtt{isPush}(\PC{\Ak}{j}, \mathbb{X}, \mathbb{Y})$
operator for an assembly instruction at \PC{\Ak}{j}
based on thresholding of the update distance $\mathbb{D} = \mathbb{X} - \mathbb{Y}$
by a threshold value $\mathbb{K} = 2^{31}-1$:
$$
\mathrm{\tt isPush}(\PC{\Ak}{j}, \mathbb{X}, \mathbb{Y}) \Leftrightarrow  \mathbb{X} - \mathbb{Y} \le_u \mathbb{K}
$$
Here, $\mathbb{K}$ represents the threshold value for the stack update
distance $\mathbb{X} - \mathbb{Y}$, below which we consider the update to be a push.

If $\mathbb{K}$ is smaller than required, then we risk misclassifying
stack pushes (stack growth) as stack pops (stack shrink).  On the other
hand, if $\mathbb{K}$ is bigger than required, then we risk
misclassifying stack pops (stack shrink) as stack pushes (stack growth).
In the latter case (when $\mathbb{K}$ is bigger than required), we would
incorrectly trigger \EW{}, instead of \EU{},
and that would cause the refinement proof to complete incorrectly (soundness
problem).  In the extreme case, if $\mathbb{K}=2^{d}-1$ (where the address
space has size $2^d$),
then even 4-byte stack pops (e.g., through the x86 {\tt pop} instruction)
would be considered as stack pushes (growth), and we would incorrectly
trigger in every situation where \EU{} was expected,
and the refinement proof would complete
trivially (and unsoundly).

On the other hand, if $\mathbb{K}$ is smaller than required, we may
incorrectly count some stack growth operations as stack pops.
In these cases, we will have show to absence of \EU{}
(as part of (Safety)) for a stack pop for which a stack push
never happened.
This would result in an refinement failure (completeness problem).

\subsubsection{$\mathbb{K}$ needs to be at least $2^{d-1}$ in the presence of VLAs}

Consider a VLA declaration, ``{\tt char v[n]}'' in \Ck{}.
In this case,
{\tt n} could be any positive integer $\le_u${\tt INT\_MAX}; this upper
bound of {\tt INT\_MAX} comes from the variable size limits
imposed by the C language.  The corresponding allocation statement
in assembly code would be something like ``{\tt $\PC{\Ak}{j}$: esp $\Assign$ esp $-$ n}''.
The resulting condition for not triggering \EU{} is
(from \TRule{Op-esp} of \cref{fig:xlateRuleAsm}):
\begin{equation*}
  \begin{split}
    \neg(
    \begin{aligned}[t]
     & \neg \mathtt{isPush}(\PC{\Ak}{j}, \mathtt{esp}, \mathtt{esp} - \mathtt{n}) \\
      \land \, & \mathtt{esp} \ne \mathtt{esp} - \mathtt{n} \\
      \land \, & \neg \intervalContainedInAddrSet{\mathtt{esp}}{\mathtt{esp}-\mathtt{n}}{\il{\Ak}{\stk}})
    \end{aligned}
  \end{split}
\end{equation*}
or equivalently,
\begin{equation}\label{eqn:leastk}
\begin{split}
(\mathtt{n} >_u \mathbb{K}) \Rightarrow (&
\begin{aligned}[t]
           & \mathtt{n} = 0_{\bv{32}} \\
  \lor \,  & (
  \begin{aligned}[t]
            & \mathtt{esp} \ne 0_{\bv{32}} \\
   \land \, & (\mathtt{esp} \le_u \mathtt{esp}-\mathtt{n}) \\
   \land \, & [\mathtt{esp}, \mathtt{esp}-\mathtt{n}] \subseteq \il{\Ak}{\stk}))
  \end{aligned}
\end{aligned}
\end{split}
\end{equation}
Now, if $\mathbb{K}$ is smaller than the biggest possible value of {\tt n}, then
there exist values of {\tt n} where the left clause
(left of $\Rightarrow$)
of \cref{eqn:leastk}
would evaluate to \texttt{true}.
Consequently, there exist values of \texttt{n} for which the right clause has to be proven
\texttt{true},
i.e., prove that the stack region is at least $2^d - \mathtt{n}$ bytes large.
It may not be possible to prove such strong conditions in all cases and thus we
get false refinement check failures.
Because the C language
constrains {\tt n} to be $\le_u \mathrm{\tt INT\_MAX} (= 2^{d-1}-1)$,
$\mathbb{K}\ge_u 2^{d-1}-1$ seems sufficient to
be able to validate such translations.

However, $\mathbb{K} = 2^{d-1}-1$ is also insufficient, because typically the
code generated by a compiler for
``{\tt char v[n]}'' also aligns {\tt n} using
a rounding factor $r=2^i$:
``${\tt esp} \Assign \mathtt{esp}-(\lceil \frac{\mathtt{n}}{r} \rceil\cdot r)$''.
In this scenario, even though $\mathtt{n} \le_u (2^{d-1}-1)$,  it is possible for
$\mathbb{D} = \lceil \frac{\mathrm{\tt n}}{r} \rceil\cdot r$ to be
greater than $(2^{d-1}-1)$. Thus, if $\mathbb{K} = 2^{d-1}-1$, there exist legal
values of {\tt n} for which stack region is at least
$2^d - \mathtt{n}$
bytes large has to be proven
to demonstrate absence of \EU{}.
The choice $\mathbb{K} = 2^{d-1}$ allows for such alignment
padding, and thus allows the refinement proof to be completed in these situations.

\subsubsection{$\mathbb{K} = 2^{d-1}$ can still lead to completeness problems}
\label{sec:completenessProblem}
{\em If a single stack update allocates two VLAs at once, we can incorrectly classify
a stack growth as a stack shrink}.

Consider two C statements in sequence, ``{\tt char v1[m]; char v2[n];}''.  In this
case both {\tt m} and {\tt n} can individually be as large as $2^{d-1}-1$.  If the
compiler decides to use a single assembly instruction to allocate both these
variables, then it is possible for a single stack update distance $\mathbb{D}$
to be greater than $\mathbb{K} = 2^{d-1}$.
Thus, in these cases,
the refinement
proof may fail if we are not able to prove that stack is large enough
to contain $2^d - \mathbb{D}$ bytes (for the classified stack pop).
This is a completeness problem.

\subsubsection{$\mathbb{K} = 2^{d-1}$ can also lead to soundness problems}
{\em If a single stack update deallocates two VLAs at once, we can incorrectly
classify a stack shrink as a stack growth}.

Consider two C statements in sequence, ``{\tt char v1[$2^{d-1}-1$]; char v2[2];}''.
If during deallocation,
the compiler decides to use a single instruction to deallocate both the
arrays, e.g., ``{\tt  esp $\Assign$ esp + (2$^{d-1}$-1) + 2}''
for a total update distance of:
$$
\mathbb{D} = -((2^{d-1}-1) + 2) = 2^{d-1}-1\ \ \ \ \ \ \  ({\tt mod}\ 2^d)
$$
Here,
because $2^{d-1}-1 \leq_u \mathbb{K}$
we will classify this ``deallocation'' as a stack push
(allocation) of $(2^{d-1}-1)$ bytes
and trigger \EW{} if allocation of $(2^{d-1}-1)$ bytes is not possible.
This is a soundness problem
because triggering \EW{} under such a weaker condition
may lead the refinement proof to succeed incorrectly.

\subsubsection{Solution}
Thus, it seems impossible in general to be able to distinguish a push
from a pop in a sound manner.  This problem is unavoidable in the
presence of VLAs. CompCert side-stepped this problem by disabling VLA support
and thus being able to statically bound the overall
stack size.  For a bounded stack, it becomes possible to distinguish pushes
from pops.  But it is not possible to bound the stack in the presence of a VLA.

Thus we propose that
the compiler must explicitly emit trustworthy
information that distinguishes a push from a pop.  Hence, {\tt isPush()}
can simply leverage this information emitted by the compiler.

As explained in \cref{sec:translationA}, in our work, we use
a threshold of $2^{31}-1$ on the update distance
to disambiguate stack pushes from pops.
We rely on manual verification for soundness.

\subsection{Running the Validator on the Bugs Identified by Compiler Fuzzing Tools (involving address-taken local variables)}
\label{app:emiBugs}
We discuss the
operation of our validator
on two bugs reported by Sun et. al. \cite{emi16} in GCC-4.9.2.
Each of these bugs is representative of a class of bugs found in
modern compilers, and it is interesting to see how the validator
behaves for each of them.
\subsubsection{Incorrect Hoisting of Local Variable Access}

\begin{figure}
\begin{subfigure}[c]{.45\textwidth}
\begin{myexamplesmall}
int a;
int main()
{
  int b = -1, d, e = 0;
  int f[2] = { 0 };
  unsigned short c = b;

  for (; e < 3; e++)
    for (d = 0; d < 2; d++)
      /* a=0, b=-1, c=65535,
         d={0,1}, e={0,1,2},
         f[0]=0 */
      if (a < 0)
        for (d = 0; d < 2; d++)
          if (f[c])
            break;
  return 0;
}
\end{myexamplesmall}
\end{subfigure}
\begin{subfigure}[c]{.45\textwidth}
\begin{myexamplesmall}
main:
        leal    4(%esp), %ecx
        andl    $-8, %esp
        movl    $3, %eax
        pushl   -4(%ecx)
        pushl   %ebp
        movl    %esp, %ebp
        pushl   %ecx
        subl    $12, %esp
        movl    a, %ecx
.L2:
        testl   %ecx, %ecx
        movl    262124(%ebp), %edx
        #FIX: the above instruction should
        #be sunk to the beginning of .L3
        js      .L3
.L5:
        subl    $1, %eax
        jne     .L2
        addl    $12, %esp
        xorl    %eax, %eax
        popl    %ecx
        popl    %ebp
        leal    -4(%ecx), %esp
        ret
.L3:
        #FIX: movl 262124(%ebx),%edx
        #     should be here
        testl   %edx, %edx
        je      .L5
        jmp     .L3
\end{myexamplesmall}
\end{subfigure}
\caption{\label{fig:emi1}GCC-4.9.2 bug reproduced from Figure 1 of \cite{emi16}. The assembly code is generated using {\tt -O3} for 32-bit x86.}
\end{figure}
\begin{figure}
\begin{subfigure}[c]{.45\textwidth}
\begin{myexamplesmall}
struct S {
  int f0;
  int f1;
};
int b;
int main()
{
  struct S a[2] = { 0 };
  struct S d = { 0, 1 };
  for (b = 0; b < 2; b++) {
    a[b] = d;
    d = a[0];
  }
  return d.f1 != 1;
}
\end{myexamplesmall}
\end{subfigure}
\begin{subfigure}[c]{.45\textwidth}
\begin{myexamplesmall}
main:
        subl    $16, %esp
        movl    $1, %eax #FIX: 1->0
        movl    $2, b
        addl    $16, %esp
        ret
\end{myexamplesmall}
\end{subfigure}
\caption{\label{fig:emi9f}GCC-4.9.2 bug reproduced from Figure 9f of \cite{emi16}. The assembly code is generated using {\tt -O3} for 32-bit x86.}
\end{figure}

\Cref{fig:emi1} shows the C code and its (incorrect) assembly implementation generated using {\tt gcc-4.9.2 -O3} for 32-bit x86
ISA.  The problem occurs because the ``{\tt movl 262124(\%ebp), \%edx}'' instruction (in the basic block
starting at {\tt .L2}) {\em reads} from the local variable at {\tt f[c]} but does that even if the branch
condition {\tt a < 0} (implemented by the {\tt testl} and {\tt js} instructions in the {\tt .L2} basic block)
evaluates to {\tt false}. Consider what happens when {\tt a = 0} --- the memory access to {\tt f[c]} is out-of-bounds and thus this
compilation could potentially trigger a segmentation fault (or other undefined behavior) in the assembly
code when the source code would expect an error-free execution.
The assembly code can be fixed by sinking the {\tt movl 262124(\%ebp), \%edx} to
the beginning of the basic block starting at {\tt .L3}.

Our validator is able to compute the equivalence proof for
the fixed program at unroll-factor three
or higher in less than five minutes.
The resulting product-graph has five nodes and five edges. The only loop in the
resulting product-graph correlates the second inner loop (on {\tt d})
with the path {\tt .L2 $\rightarrow$ .L5 $\rightarrow$ .L2} in the assembly
program. Both
the top-level loop (on {\tt e}) and the inner-most loop on {\tt d} are
completely unrolled in the product-graph (which is supported at unroll factors
of three or higher).

For
the original program, our validator fails to compute equivalence at all unroll
factors due to the violation of the (MAC) constraint
in the correlated
path for the second inner loop (that iterates on the variable {\tt d}), i.e.,
{\tt .L2 $\rightarrow$ .L5 $\rightarrow$ .L2} in the original (unfixed) assembly program.

\subsubsection{Incorrect Final Value of Local Variable of Aggregate Type after Loop Unrolling}
\label{sec:emi9f}

\Cref{fig:emi9f} shows the C code and its (incorrect) assembly implementation generated using {\tt gcc-4.9.2 -O3} for 32-bit x86
ISA.  The compiler
fully unrolls the loop in this program
to generate a straight-line sequence
of assembly instructions that
directly sets the return values to the statically-computed
constants.  However, the correct return value in the {\tt eax} register should be {\tt 0} while the generated assembly code
sets it to 1.  Our validator fails to compute equivalence for
this pair of programs because it is unable to prove the observable
equivalence of return values.

When the assembly
code is fixed to set {\tt eax} to {\tt 0} (instead of 1),
the validator is correctly
able to prove equivalence at unroll factors of three or higher.
The validator is able to compute equivalence for the fixed
assembly program within around two minutes and the resulting product-graph has six
nodes and six edges.  There are no cycles in the resulting product-graph, i.e., all the loops
are fully unrolled in the product-graph (at an unroll factor of three or higher).

\subsection{System components, trusted computing base, and overview of contributions}
\label{app:sysComponents}

\begin{figure}
	\center
	\includegraphics[width=\textwidth]{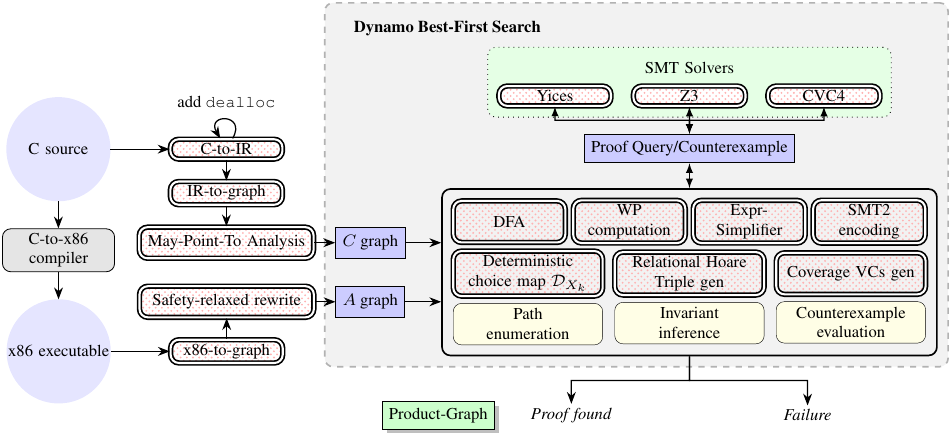}
	\caption{\label{fig:blocks}High-level components, including the TCB, of our counterexample-guided best-first search algorithm.}
\end{figure}

The soundness of a verification effort is critically dependent
on the correctness of the {\em trusted computing base} (TCB)
of the verifier.  \Cref{fig:blocks} shows the high-level
components and the flowchart of our best-first search algorithm,
where the components belonging to the TCB are marked with double
borders and a dotted background pattern.
Roughly speaking,
\toolName{} is around 400K Source Lines of Code (SLOC) in C/C++, of which
the TCB is around 70K SLOC.
Within the TCB, around 30K SLOC is
due to the expression handling and
simplification logic (Simplifier and
SMT encoding logic in \cref{fig:blocks}), 
less than 10K SLOC
for the graph representation
and the weakest-precondition logic (WP computation in \cref{fig:blocks}), 
less than 3K SLOC for the may-point-to analysis and simplifications, 
and less than 1K SLOC for a common dataflow analysis framework.
The x86-to-graph translation is around 18K SLOC of C code (for disassembly)
and 5K SLOC of OCaml code (for logic encoding), and the
IR-to-graph translation is around 3K SLOC of C++ code (including the addition of {\tt dealloc}).
We rely on the Clang framework for the C-to-IR translation --- one can
imagine replacing this with a verified frontend, such as CompCert's.
The modeling of the deterministic choice map,
Hoare triple and coverage 
verification conditions
is relatively simple (less than 1K SLOC total).

A soundness bug is a bug that causes the equivalence proof to succeed
incorrectly.
Over several person years of development, we have rarely
encountered soundness bugs in x86-to-graph or IR-to-graph --- it is unlikely that
both pipelines, written independently, have the same bug that results in an unsound
equivalence proof.  Similarly, it has been
rare to find a soundness bug in the SMT solvers --- we once discovered a bug
in {\tt Yices v2.6.1}, but it was easily caught because the other SMT solvers
disagreed with {\tt Yices}'s result. The {\tt Yices} bug was fixed upon our
reporting. For each proof obligation
generated by the equivalence checker as a Hoare triple, we
check the weakest-precondition and
SMT encoding logic by confirming that the counterexamples generated
by the SMT solver
satisfy
the pre- and post-conditions of the Hoare triple.
A rare soundness bug in the expression simplification,
may-point-to analysis, and graph translation
was relatively more common in the early stages of \toolName{}'s development.
As development matures, soundness bugs in an equivalence checker become scarce.
Compared to a
modern optimizing compiler, an equivalence checker's TCB is roughly 1000x smaller.



\subsection{Command-line used for compiling benchmarks in experiments}

\begin{enumerate}
  \item Programs in \cref{tab:benchmarks}
    \begin{itemize}
      \item Clang/LLVM v12.0.0
      \begin{lstlisting}[breaklines=true]
    clang -m32 -S -no-integrated-as -g -Wl,--emit-relocs -fdata-sections -g -fno-builtin -fno-strict-aliasing -fno-optimize-sibling-calls -fwrapv -fno-strict-overflow -ffreestanding -fno-jump-tables -fcf-protection=none -fno-stack-protector -fno-inline -fno-inline-functions -D_FORTIFY_SOURCE=0 -D__noreturn__=__no_reorder__ -I/usr/include/x86_64-linux-gnu/c++/9/32 -I/usr/include/x86_64-linux-gnu/c++/9  -mllvm -enable-tail-merge=false -mllvm -nomerge-calls -std=c11 -O3 <file.c> -o <file.s>
      \end{lstlisting}
      \item GCC v8.4.0
      \begin{lstlisting}[breaklines=true]
    gcc-8 -m32 -S -g -Wl,--emit-relocs -fdata-sections -g -no-pie -fno-pie -fno-strict-overflow -fno-unit-at-a-time -fno-strict-aliasing -fno-optimize-sibling-calls -fkeep-inline-functions -fwrapv -fno-reorder-blocks -fno-jump-tables -fno-caller-saves -fno-inline -fno-inline-functions -fno-inline-small-functions -fno-indirect-inlining -fno-partial-inlining -fno-inline-functions-called-once -fno-early-inlining -fno-whole-program -fno-ipa-sra -fno-ipa-cp -fcf-protection=none -fno-stack-protector -fno-stack-clash-protection -D_FORTIFY_SOURCE=0 -D__noreturn__=__no_reorder__ -fno-builtin-printf -fno-builtin-malloc -fno-builtin-abort -fno-builtin-exit -fno-builtin-fscanf -fno-builtin-abs -fno-builtin-acos -fno-builtin-asin -fno-builtin-atan2 -fno-builtin-atan -fno-builtin-calloc -fno-builtin-ceil -fno-builtin-cosh -fno-builtin-cos -fno-builtin-exp -fno-builtin-fabs -fno-builtin-floor -fno-builtin-fmod -fno-builtin-fprintf -fno-builtin-fputs -fno-builtin-frexp -fno-builtin-isalnum -fno-builtin-isalpha -fno-builtin-iscntrl -fno-builtin-isdigit -fno-builtin-isgraph -fno-builtin-islower -fno-builtin-isprint -fno-builtin-ispunct -fno-builtin-isspace -fno-builtin-isupper -fno-builtin-isxdigit -fno-builtin-tolower -fno-builtin-toupper -fno-builtin-labs -fno-builtin-ldexp -fno-builtin-log10 -fno-builtin-log -fno-builtin-memchr -fno-builtin-memcmp -fno-builtin-memcpy -fno-builtin-memset -fno-builtin-modf -fno-builtin-pow -fno-builtin-putchar -fno-builtin-puts -fno-builtin-scanf -fno-builtin-sinh -fno-builtin-sin -fno-builtin-snprintf -fno-builtin-sprintf -fno-builtin-sqrt -fno-builtin-sscanf -fno-builtin-strcat -fno-builtin-strchr -fno-builtin-strcmp -fno-builtin-strcpy -fno-builtin-strcspn -fno-builtin-strlen -fno-builtin-strncat -fno-builtin-strncmp -fno-builtin-strncpy -fno-builtin-strpbrk -fno-builtin-strrchr -fno-builtin-strspn -fno-builtin-strstr -fno-builtin-tanh -fno-builtin-tan -fno-builtin-vfprintf -fno-builtin-vsprintf -fno-builtin -I/usr/include/x86_64-linux-gnu/c++/9/32 -I/usr/include/x86_64-linux-gnu/c++/9  -fno-tree-tail-merge --param max -tail-merge-comparisons=0 --param max-tail-merge-iterations=0  -std=c11 -O3 <file.c> -o <file.s>
      \end{lstlisting}
      \item ICC v2021.8.0
      \begin{lstlisting}[breaklines=true]
    icc -m32 -D_Float32=__Float32 -D_Float64=__Float64 -D_Float32x=__Float32x -D_Float64x=__Float64x -S -g -Wl,--emit-relocs -fdata-sections -g -no-ip -fno-optimize-sibling-calls -fargument-alias -no-ansi-alias -falias -fno-jump-tables -fno-omit-frame-pointer -fno-strict-aliasing -fno-strict-overflow -fwrapv -fabi-version=1 -nolib-inline -inline-level=0 -fno-inline-functions -finline-limit=0 -no-inline-calloc -no-inline-factor=0 -fno-builtin-printf -fno-builtin-malloc -fno-builtin-abort -fno-builtin-exit -fno-builtin-fscanf -fno-builtin-abs -fno-builtin-acos -fno-builtin-asin -fno-builtin-atan2 -fno-builtin-atan -fno-builtin-calloc -fno-builtin-ceil -fno-builtin-cosh -fno-builtin-cos -fno-builtin-exp -fno-builtin-fabs -fno-builtin-floor -fno-builtin-fmod -fno-builtin-fprintf -fno-builtin-fputs -fno-builtin-frexp -fno-builtin-isalnum -fno-builtin-isalpha -fno-builtin-iscntrl -fno-builtin-isdigit -fno-builtin-isgraph -fno-builtin-islower -fno-builtin-isprint -fno-builtin-ispunct -fno-builtin-isspace -fno-builtin-isupper -fno-builtin-isxdigit -fno-builtin-tolower -fno-builtin-toupper -fno-builtin-labs -fno-builtin-ldexp -fno-builtin-log10 -fno-builtin-log -fno-builtin-memchr -fno-builtin-memcmp -fno-builtin-memcpy -fno-builtin-memset -fno-builtin-modf -fno-builtin-pow -fno-builtin-putchar -fno-builtin-puts -fno-builtin-scanf -fno-builtin-sinh -fno-builtin-sin -fno-builtin-snprintf -fno-builtin-sprintf -fno-builtin-sqrt -fno-builtin-sscanf -fno-builtin-strcat -fno-builtin-strchr -fno-builtin-strcmp -fno-builtin-strcpy -fno-builtin-strcspn -fno-builtin-strlen -fno-builtin-strncat -fno-builtin-strncmp -fno-builtin-strncpy -fno-builtin-strpbrk -fno-builtin-strrchr -fno-builtin-strspn -fno-builtin-strstr -fno-builtin-tanh -fno-builtin-tan -fno-builtin-vfprintf -fno-builtin-vsprintf -fno-builtin -D_FORTIFY_SOURCE=0 -D__noreturn__=__no_reorder__ -qno-opt-multi-version-aggressive -ffreestanding -unroll0 -no-vec -I/usr/include/x86_64-linux-gnu/c++/9/32 -I/usr/include/x86_64-linux-gnu/c++/9  -std=c11 -O3 <file.c> -o <file.s>
      \end{lstlisting}
    \end{itemize}

  \item TSVC
    \begin{lstlisting}[breaklines=true]
    clang -m32 -S -no-integrated-as -g -Wl,--emit-relocs -fdata-sections -g -fno-builtin -fno-strict-aliasing -fno-optimize-sibling-calls -fwrapv -fno-strict-overflow -ffreestanding -fno-jump-tables -fcf-protection=none -fno-stack-protector -fno-inline -fno-inline-functions -D_FORTIFY_SOURCE=0 -D__noreturn__=__no_reorder__ -I/usr/include/x86_64-linux-gnu/c++/9/32 -I/usr/include/x86_64-linux-gnu/c++/9 -msse4.2  -mllvm -enable-tail-merge=false -mllvm -nomerge-calls -std=c11 -O3 <file.c> -o <file.s>
    \end{lstlisting}

  \item bzip2 \texttt{O1-}
    \begin{lstlisting}[breaklines=true]
    clang bzip2.c -Wl,--emit-relocs -fno-unroll-loops -fdata-sections -fno-inline -fno-inline-functions -fcf-protection=none -fno-stack-protector -mllvm -enable-tail-merge=false -O1 -mllvm -nomerge-calls -mllvm -no-early-cse -mllvm -no-licm -mllvm -no-machine-licm -mllvm -no-dead-arg-elim -mllvm -no-ip-sparse-conditional-constant-prop -mllvm -no-dce-fcalls -mllvm -replexitval=never -std=c11 -fno-builtin -fno-strict-aliasing -fno-optimize-sibling-calls -fwrapv -fno-strict-overflow -ffreestanding -fno-jump-tables -D_FORTIFY_SOURCE=0 -D__noreturn__=__no_reorder__ -fno-builtin-printf -fno-builtin-malloc -fno-builtin-abort -fno-builtin-exit -fno-builtin-fscanf -fno-builtin-abs -fno-builtin-acos -fno-builtin-asin -fno-builtin-atan2 -fno-builtin-atan -fno-builtin-calloc -fno-builtin-ceil -fno-builtin-cosh -fno-builtin-cos -fno-builtin-exp -fno-builtin-fabs -fno-builtin-floor -fno-builtin-fmod -fno-builtin-fprintf -fno-builtin-fputs -fno-builtin-frexp -fno-builtin-isalnum -fno-builtin-isalpha -fno-builtin-iscntrl -fno-builtin-isdigit -fno-builtin-isgraph -fno-builtin-islower -fno-builtin-isprint -fno-builtin-ispunct -fno-builtin-isspace -fno-builtin-isupper -fno-builtin-isxdigit -fno-builtin-tolower -fno-builtin-toupper -fno-builtin-labs -fno-builtin-ldexp -fno-builtin-log10 -fno-builtin-log -fno-builtin-memchr -fno-builtin-memcmp -fno-builtin-memcpy -fno-builtin-memset -fno-builtin-modf -fno-builtin-pow -fno-builtin-putchar -fno-builtin-puts -fno-builtin-scanf -fno-builtin-sinh -fno-builtin-sin -fno-builtin-snprintf -fno-builtin-sprintf -fno-builtin-sqrt -fno-builtin-sscanf -fno-builtin-strcat -fno-builtin-strchr -fno-builtin-strcmp -fno-builtin-strcpy -fno-builtin-strcspn -fno-builtin-strlen -fno-builtin-strncat -fno-builtin-strncmp -fno-builtin-strncpy -fno-builtin-strpbrk -fno-builtin-strrchr -fno-builtin-strspn -fno-builtin-strstr -fno-builtin-tanh -fno-builtin-tan -fno-builtin-vfprintf -fno-builtin-vsprintf -fno-builtin -I/usr/include/x86_64-linux-gnu/c++/9/32 -I/usr/include/x86_64-linux-gnu/c++/9 -o bzip2.c.o -c -g -m32
    \end{lstlisting}

  \item bzip2 \texttt{O1}
    \begin{lstlisting}[breaklines=true]
    clang -m32 -S -no-integrated-as -g -Wl,--emit-relocs -fdata-sections -g -fno-builtin -fno-strict-aliasing -fno-optimize-sibling-calls -fwrapv -fno-strict-overflow -ffreestanding -fno-jump-tables -fcf-protection=none -fno-stack-protector -fno-inline -fno-inline-functions -D_FORTIFY_SOURCE=0 -D__noreturn__=__no_reorder__ -I/usr/include/x86_64-linux-gnu/c++/9/32 -I/usr/include/x86_64-linux-gnu/c++/9 -fno-unroll-loops -mllvm -enable-tail-merge=false -mllvm -nomerge-calls -std=c11 -O1 bzip2.c -o bzip2.s
    \end{lstlisting}

  \item bzip2 \texttt{O2}
    \begin{lstlisting}[breaklines=true]
    clang -m32 -S -no-integrated-as -g -Wl,--emit-relocs -fdata-sections -g -fno-builtin -fno-strict-aliasing -fno-optimize-sibling-calls -fwrapv -fno-strict-overflow -ffreestanding -fno-jump-tables -fcf-protection=none -fno-stack-protector -fno-inline -fno-inline-functions -D_FORTIFY_SOURCE=0 -D__noreturn__=__no_reorder__ -I/usr/include/x86_64-linux-gnu/c++/9/32 -I/usr/include/x86_64-linux-gnu/c++/9 -fno-unroll-loops -mllvm -enable-tail-merge=false -mllvm -nomerge-calls -std=c11 -O2 bzip2.c -o bzip2.s
    \end{lstlisting}
\end{enumerate}

\clearpage

\subsection{More details on {\tt bzip2} experiments}
\label{app:bzip2}


\begin{figure}
\begin{subfigure}[t]{.48\textwidth}
  \includegraphics[width=\textwidth]{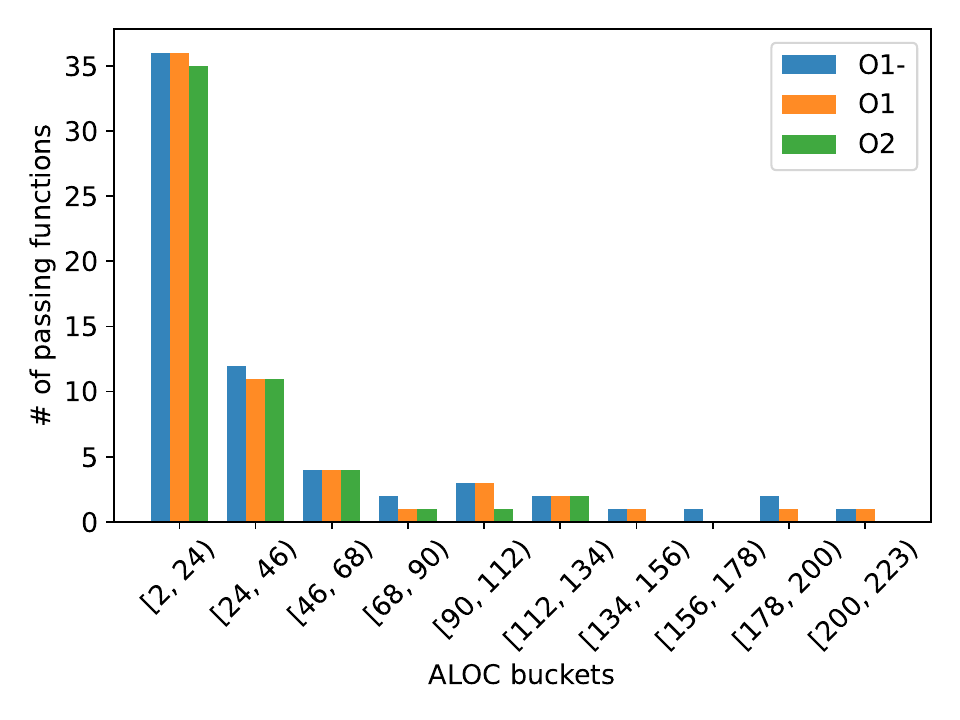}
\caption{\label{fig:histALOC} Histogram for the number of equivalence check passing functions in an ALOC range for {\tt bzip2} functions compiled with three different optimization levels: {\tt O1-}, {\tt O1}, {\tt O2}}
\end{subfigure}
\begin{subfigure}[t]{.48\textwidth}
  \includegraphics[width=\textwidth]{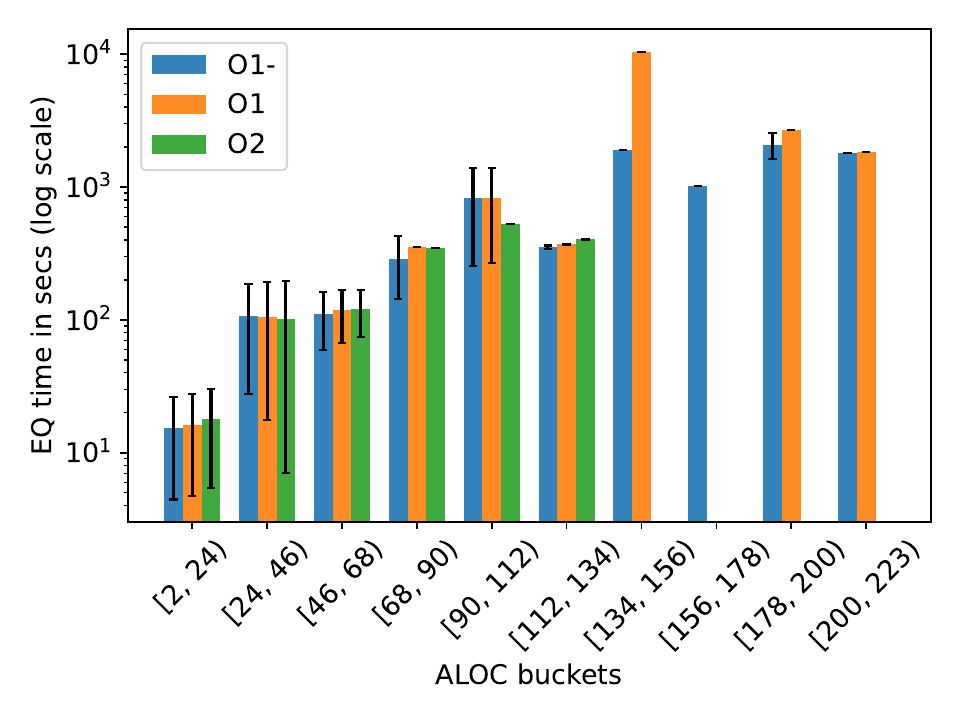}
\caption{\label{fig:histALOCvsEqT} Graph for ALOC range vs mean equivalence time (in seconds) for {\tt bzip2} functions compiled with three different optimization levels: {\tt O1-}, {\tt O1}, {\tt O2}.  Y-axis is logarithmically scaled.  The black lines indicate standard deviation.}
\end{subfigure}
\begin{subfigure}[t]{.48\textwidth}
  \includegraphics[width=\textwidth]{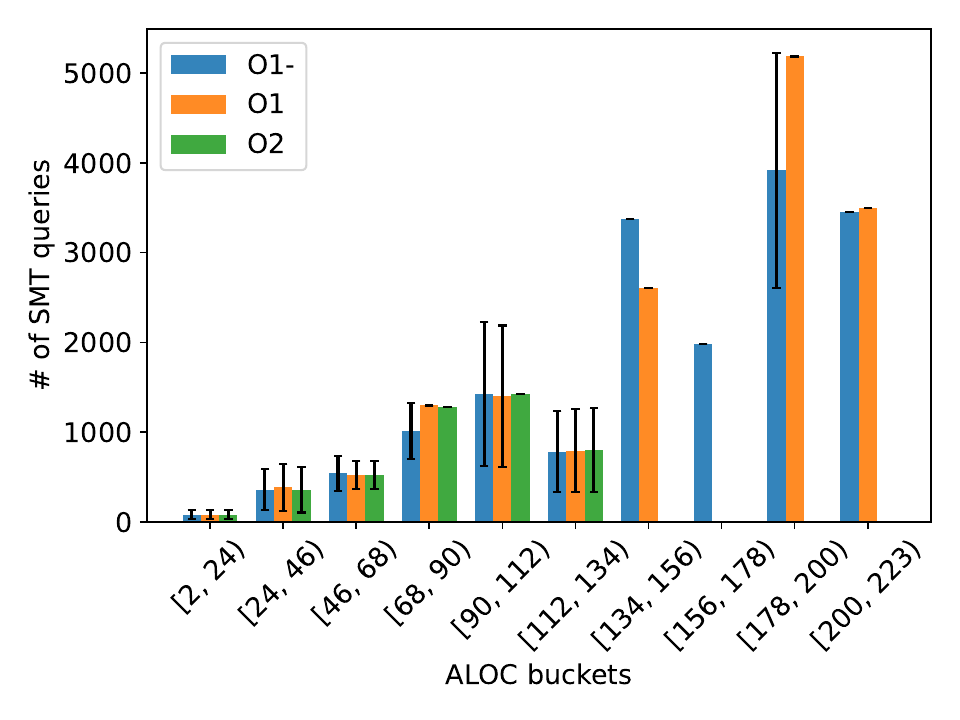}
\caption{\label{fig:histALOCvsNumQ} Graph for ALOC range vs mean number of SMT queries made for {\tt bzip2} functions compiled with three different optimization levels: {\tt O1-}, {\tt O1}, {\tt O2}.  The black lines indicate standard deviation.}
\end{subfigure}
\end{figure}

\Cref{fig:histALOC} shows the histogram for the number of equivalence check passing functions in a given ALOC (assembly lines of code) range for SPEC CPU2000's {\tt bzip2} functions compiled with three different optimization levels: {\tt O1-}, {\tt O1}, and {\tt O2} (the optimization levels are discussed in \cref{sec:experiments}).
The number of successful equivalence check passes start falling with the increase in optimization level and ALOC: at {\tt -O2}, we do not get any passes beyond 134 ALOCs.

\Cref{fig:histALOCvsEqT,fig:histALOCvsNumQ} show the mean equivalence time (in seconds) and mean number of SMT queries made by {\tt bzip2} functions grouped by ALOC ranges for each of the three optimization levels.  A missing bar indicates that no equivalence check passing function lies in that range for that particular optimization level.
The time taken for successful passes is almost similar across all three optimization levels (with the exception of a single function in the [134, 156) range).
A similar pattern is observed for the number of SMT queries made.

\Cref{tab:bzip2_full} shows the full list of {\tt bzip2} functions with their assembly lines of code (ALOC) and equivalence check times (in seconds) for the three Clang/LLVM compiler configurations ({\tt O1-}, {\tt O1}, {\tt O2}).

\begin{footnotesize}
\begin{longtable}{lllllll}
    \toprule
    Name & \multicolumn{3}{c}{ALOC} & \multicolumn{3}{c}{Equivalence time (seconds)} \\
    \midrule
         & {\tt O1-} & {\tt O1} & {\tt O2} & {\tt O1-} & {\tt O1} & {\tt O2}\\
    \midrule

{\tt allocateCompressStructures}  & 47 & 47 & 51 & 43.2 & 47.2 & 50.6 \\
{\tt badBGLengths}  & 13 & 13 & 13 & 23.2 & 25.3 & 28.5 \\
{\tt badBlockHeader}  & 13 & 13 & 13 & 21.9 & 23.0 & 27.3 \\
{\tt bitStreamEOF}  & 13 & 13 & 13 & 22.7 & 21.2 & 26.1 \\
{\tt blockOverrun}  & 13 & 13 & 13 & 23.5 & 25.3 & 27.2 \\
{\tt bsFinishedWithStream}  & 22 & 22 & 25 & 24.0 & 23.2 & 26.0 \\
{\tt bsGetInt32}  & 4 & 4 & 4 & 6.0 & 6.0 & 7.2 \\
{\tt bsGetIntVS}  & 6 & 6 & 6 & 7.3 & 8.1 & 9.8 \\
{\tt bsGetUChar}  & 5 & 5 & 5 & 6.0 & 7.4 & 7.1 \\
{\tt bsGetUInt32}  & 24 & 24 & 24 & 13.3 & 16.9 & 20.6 \\
{\tt bsPutInt32}  & 6 & 6 & 6 & 6.6 & 7.6 & 9.0 \\
{\tt bsPutIntVS}  & 6 & 6 & 6 & 9.6 & 9.8 & 11.1 \\
{\tt bsPutUChar}  & 8 & 8 & 8 & 7.8 & 8.4 & 10.9 \\
{\tt bsPutUInt32}  & 32 & 32 & 32 & 30.8 & 30.7 & 36.4 \\
{\tt bsR}  & 46 & 46 & 46 & 42.8 & 42.8 & 51.0 \\
{\tt bsSetStream}  & 9 & 9 & 9 & 3.7 & 3.7 & 4.7 \\
{\tt bsW}  & 31 & 32 & 36 & 34.4 & 33.8 & 40.0 \\
{\tt cadvise}  & 6 & 6 & 6 & 20.3 & 20.2 & 26.9 \\
{\tt cleanUpAndFail}  & 48 & 46 & 46 & 187.9 & 179.0 & 227.8 \\
{\tt compressOutOfMemory}  & 14 & 14 & 14 & 33.6 & 33.0 & 42.6 \\
{\tt compressStream}  & 124 & 124 & 124 & 342.0 & 369.2 & 402.6 \\
{\tt compressedStreamEOF}  & 16 & 16 & 16 & 21.5 & 24.0 & 27.6 \\
{\tt crcError}  & 15 & 15 & 15 & 31.4 & 30.9 & 36.6 \\
{\tt debug\_time}  & 2 & 2 & 2 & 1.6 & 1.8 & 2.0 \\
{\tt doReversibleTransformation}  & 48 & 49 & 47 & 93.3 & 102.9 & 129.2 \\
{\tt fullGtU}  & 120 & 113 & 113 & 363.0 & 375.4 & 404.4 \\
{\tt generateMTFValues}  & 144 & 144 & 166 & 1909.3 & 10441.4 & \xmark \\
{\tt getAndMoveToFrontDecode}  & 299 & 296 & 305 & \xmark & \xmark & \xmark \\
{\tt getFinalCRC}  & 3 & 3 & 3 & 1.9 & 2.2 & 2.3 \\
{\tt getGlobalCRC}  & 2 & 2 & 2 & 2.1 & 1.8 & 2.2 \\
{\tt getRLEpair}  & 72 & 73 & 73 & 144.6 & \xmark & \xmark \\
{\tt hbAssignCodes}  & 37 & 37 & 37 & 296.4 & 325.4 & 330.7 \\
{\tt hbCreateDecodeTables}  & 94 & 94 & 107 & 1610.3 & 1622.3 & \xmark \\
{\tt hbMakeCodeLengths}  & 261 & 249 & 292 & \xmark & \xmark & \xmark \\
{\tt indexIntoF}  & 23 & 23 & 23 & 30.6 & 32.0 & 41.2 \\
{\tt initialiseCRC}  & 2 & 2 & 2 & 2.3 & 1.9 & 2.3 \\
{\tt ioError}  & 15 & 15 & 15 & 17.9 & 18.3 & 23.1 \\
{\tt loadAndRLEsource}  & 96 & 96 & 96 & 336.2 & 366.7 & \xmark \\
{\tt main}  & 190 & 132 & 183 & \xmark & \xmark & \xmark \\
{\tt makeMaps}  & 16 & 16 & 16 & 14.5 & 15.8 & 17.9 \\
{\tt med3}  & 14 & 14 & 14 & 3.8 & 4.1 & 4.2 \\
{\tt moveToFrontCodeAndSend}  & 9 & 9 & 9 & 15.1 & 16.0 & 15.5 \\
{\tt mySIGSEGVorSIGBUScatcher}  & 35 & 23 & 23 & 178.3 & \xmark & \xmark \\
{\tt mySignalCatcher}  & 10 & 10 & 10 & 16.9 & 18.9 & 25.2 \\
{\tt panic}  & 13 & 13 & 13 & 32.3 & 36.1 & 30.3 \\
{\tt qSort3}  & 297 & 297 & 363 & \xmark & \xmark & \xmark \\
{\tt randomiseBlock}  & 35 & 37 & 38 & 155.1 & 177.9 & \xmark \\
{\tt recvDecodingTables}  & 199 & 193 & 295 & 2539.8 & 2690.8 & \xmark \\
{\tt sendMTFValues}  & 691 & 692 & 832 & \xmark & \xmark & \xmark \\
{\tt setDecompressStructureSizes}  & 79 & 79 & 81 & 426.1 & 351.8 & 345.0 \\
{\tt setGlobalCRC}  & 3 & 3 & 3 & 2.8 & 3.0 & 3.1 \\
{\tt showFileNames}  & 8 & 8 & 8 & 15.6 & 18.4 & 17.0 \\
{\tt simpleSort}  & 194 & 185 & 215 & \xmark & \xmark & \xmark \\
{\tt sortIt}  & 409 & 406 & 421 & \xmark & \xmark & \xmark \\
{\tt spec\_compress}  & 11 & 11 & 11 & 16.0 & 16.0 & 16.0 \\
{\tt spec\_getc}  & 29 & 29 & 29 & 40.6 & 43.7 & 46.6 \\
{\tt spec\_init}  & 48 & 49 & 49 & 120.4 & 134.1 & 123.7 \\
{\tt spec\_initbufs}  & 9 & 9 & 9 & 11.0 & 9.6 & 11.0 \\
{\tt spec\_load}  & 110 & 105 & 105 & 512.4 & 499.8 & 524.1 \\
{\tt spec\_putc}  & 29 & 29 & 29 & 52.5 & 51.2 & 57.3 \\
{\tt spec\_read}  & 44 & 46 & 46 & 133.5 & 131.2 & 166.7 \\
{\tt spec\_reset}  & 16 & 16 & 16 & 21.8 & 20.1 & 23.4 \\
{\tt spec\_rewind}  & 5 & 5 & 5 & 3.4 & 3.3 & 3.5 \\
{\tt spec\_uncompress}  & 10 & 10 & 10 & 15.0 & 16.4 & 14.3 \\
{\tt spec\_ungetc}  & 45 & 48 & 48 & 176.1 & 188.2 & 183.7 \\
{\tt spec\_write}  & 34 & 34 & 34 & 73.3 & 77.0 & 73.8 \\
{\tt testStream}  & 195 & 194 & 196 & 1619.5 & \xmark & \xmark \\
{\tt uncompressOutOfMemory}  & 14 & 14 & 14 & 48.6 & 50.5 & 45.8 \\
{\tt uncompressStream}  & 169 & 174 & 176 & 1010.5 & \xmark & \xmark \\
{\tt undoReversibleTransformation\_fast}  & 221 & 223 & 248 & 1794.0 & 1836.8 & \xmark \\
{\tt undoReversibleTransformation\_small}  & 273 & 271 & 281 & \xmark & \xmark & \xmark \\
{\tt vswap}  & 27 & 27 & 27 & 63.3 & 61.1 & 54.0 \\

   \bottomrule

   \caption{\label{tab:bzip2_full}List of {\tt bzip2} functions with their assembly lines of code (ALOC) and equivalence check times (in seconds) for the three Clang/LLVM compiler configurations ({\tt O1-}, {\tt O1}, {\tt O2}).\\\xmark{} denotes equivalence check failure for that function-compiler pair.}
\end{longtable}
\end{footnotesize}

\clearpage

\subsection{Using a translation validator for checking alignment}
\label{app:alignmentCheck}
A translation validator has more applications than just compiler validation.
For example, compilers often use higher alignment factors than
those necessitated by the C standard, e.g., the ``{\tt long long}'' type is
often aligned at eight-byte boundaries to reduce cache misses.
This is easily checked by changing the well-formedness condition
for alignment (\cref{sec:trans}) to reflect the higher alignment
value.  Using our first set of benchmarks (containing different programming
patterns),
we validated that all the three production compilers ensure that {\tt long long}
variables are eight-byte aligned for these benchmarks.  In contrast, using
the validator, we found
that the ACK compiler \cite{ack_compiler} only ensures four-byte
alignment.

\subsection{Full source code for discussed benchmarks}
\label{app:benchmarks_source_code}
We provide the full source code of the benchmarks from \cref{tab:benchmarks} in \cref{fig:app_vilcce,fig:app_rod,fig:app_acll,fig:app_mp} below (the source code for {\tt fib} is already listed in \cref{fig:example1c}).

\noindent
The loops of validated {\tt bzip2} benchmarks are shown in \cref{fig:bzip2_src}.

\newpage

\begin{figure}
\begin{subfigure}[t]{.45\textwidth}
\begin{myexamplesmall}
int vsl(int n)
{
  if (n <= 0)
    return 0;
  int v[n];
  for (int i = 0; i < n; ++i) {
    v[i] = i*(i+1);
  }
  return v[0]+v[n-1];
}
\end{myexamplesmall}
\end{subfigure}
\begin{subfigure}[t]{.45\textwidth}
\begin{myexamplesmall}
int vcu(int n, int k)
{
   int a[n];
   if (k > 0 && k <= n) {
      a[0] = 0;
      a[k-1] = 10;
      return a[0];
   }
   return 0;
}
\end{myexamplesmall}
\end{subfigure}

\vfill

  \begin{subfigure}[t]{.45\textwidth}
\begin{myexamplesmall}
// substitute ~$\mathscr{N}$~ with 1, 2, 3
// to obtain vil1, vil2, vil3
int vil~$\mathscr{N}$~(unsigned n)
{
  int r = 0;
  for (unsigned i = 1; i < n; ++i) {
    int v1[4*i], v2[4*i], ~{\ldots}~, v~$\mathscr{N}$~[4*i];
    r += foo~$\mathscr{N}$~(v1,v2,~{\ldots}~,v~$\mathscr{N}$~, i);
  }
  return r;
}
\end{myexamplesmall}
\end{subfigure}

\vfill

\begin{subfigure}[t]{.45\textwidth}
\begin{myexamplesmall}
int vilcc(int n)
{
  int ret = 0;
  int i = 1;
  while (i < n) {
    char t[i];
    if (init(t, i) < 0)
      continue;
    ret += t[i-1];
    ++i;
  }
  return ret;
}
\end{myexamplesmall}
\end{subfigure}
\begin{subfigure}[t]{.45\textwidth}
\begin{myexamplesmall}
int vilce(int n)
{
  int ret = 0;
  int i = 1;
  while (i < n) {
    char t[i];
    if (init(t, i) < 0)
      break;
    ret += t[i-1];
    ++i;
  }
  return ret;
}
\end{myexamplesmall}
\end{subfigure}

\caption{\label{fig:app_vilcce}Benchmarks with VLAs.}
\end{figure}

\begin{figure}
\begin{subfigure}[t]{.45\textwidth}
\begin{myexamplesmall}
#include <alloca.h>
int as(int n)
{
  if (n < 1) {
    return 0;
  }
  int* p = alloca(n*sizeof(n));
  for (int i = 0; i < n; ++i) {
    p[i] = i*i;
  }
  return p[0]*p[n-1];
}
\end{myexamplesmall}
\end{subfigure}
\begin{subfigure}[t]{.45\textwidth}
\begin{myexamplesmall}
int ac(char* s, int fd, int* a)
{
  int n;
  if (!s || (n = strlen(s)) <= 0)
    return 0;
  if (!a) {
    a = alloca(sizeof(int)*n);
  }
  for (int i = 0; i < n; ++i) {
    a[i] = s[i] + 32;
  }
  return write(fd, a, n);
}
\end{myexamplesmall}
\end{subfigure}
\begin{subfigure}[t]{.45\textwidth}
\begin{myexamplesmall}
#include <alloca.h>
int n;
int all()
{
  typedef struct lln {
    int data;
    struct lln* next;
  } Node;
  if (n > 4096)
    return 0;
  Node* hd = 0;
  for (int i = 0; i < n; ++i) {
    Node* t = alloca(sizeof(Node));
    t->data = next_data();
    t->next = hd; hd = t;
  }
  Node* t = hd;
  int ret = 0;
  while (t != 0) {
    ret += t->data;
    t = t->next;
  }
  return ret;
}
\end{myexamplesmall}
\end{subfigure}
\caption{\label{fig:app_acll}Benchmarks with use of {\tt alloca}}
\end{figure}

\begin{figure}
\begin{subfigure}[c]{.48\textwidth}
\begin{myexamplesmall}
const int cts[] = { 0x66, 0x65, 0x67, 0x60 };
int rod(int n)
{
  char zz[] = "0123456789";
  printf("Scanning %d chars", n);
  char t[n];
  scanf("%s",t);
  int ret = 0;
  for (int i = 0, j = 0; i < n; ++i) {
    printf("Round #...\n", i);
    zz[j] ^= t[i];
    if (++j >= sizeof zz) j = 0;
  }
  ret += zz[0] + cts[n%((sizeof cts)/sizeof(cts[0]))];
  printf("Returning %d", ret);
  return ret;
}
\end{myexamplesmall}
\end{subfigure}
\caption{\label{fig:app_rod}{\tt rod} with mixed use of VLA and address-taken variable}
\end{figure}

\begin{figure}
\begin{subfigure}[c]{.48\textwidth}
\begin{myexamplesmall}
#include <stdarg.h>
void minprintf(char *fmt, ...)
{
  va_list ap;
  char *p, *sval;
  int ival;
  va_start(ap, fmt);
  for (p = fmt; *p; p++) {
    check(p);
    if (*p != '%') {
      putchar(*p);
      continue;
    }
    switch (*++p) {
      case 'd':
        ival = va_arg(ap, int);
        print_int(ival);
        break;
      case 's':
        for (sval = va_arg(ap, char*); *sval; sval++)
          putchar(*sval);
        break;
      default:
        break;
    }
  }
  va_end(ap);
}
\end{myexamplesmall}
\end{subfigure}
\caption{\label{fig:app_mp}{\tt minprintf} with variable argument list.  Adapted from K\&R}
\end{figure}

\begin{figure}
\begin{subfigure}[b]{.44\textwidth}
\begin{subfigure}[b]{\textwidth}
\begin{myexamplescriptsz}
void recvDecodingTables() {
  unsigned char inUse16[16];
  for (~$\ldots$~) { /* write:inUse16 ~$\ldots$~ */ }
  for (~$\ldots$~) { /* ~$\ldots$~ */ }
  for (~$\ldots$~) { /* read:inUse16 ~$\ldots$~*/
    for (~$\ldots$~) { /* ~$\ldots$~ */ }
  }
  for (~$\ldots$~) { while (~$\ldots$~) {/* ~$\ldots$~ */ }  }
  { unsigned char pos[6];
    for (~$\ldots$~) { /* write:pos ~$\ldots$~ */ }
    for (~$\ldots$~) { /* read,write:pos ~$\ldots$~ */
      while (~$\ldots$~) {/* ~$\ldots$~ */ }
    }
  }
  for (~$\ldots$~) {
    for (~$\ldots$~) {
      while (~$\ldots$~) {/* ~$\ldots$~ */ }
    }
  }
  for (~$\ldots$~) { for (~$\ldots$~) { /* ~$\ldots$~ */ } }
}
\end{myexamplescriptsz}
\caption{\label{fig:rDT}Loops in {\tt recvDecodingTables()}}
\end{subfigure}
\\
\end{subfigure}
\begin{subfigure}[b]{.52\textwidth}
\begin{subfigure}[b]{\textwidth}
\begin{myexamplescriptsz}
void generateMTFValues() {
  unsigned char yy[256];
  for (~$\ldots$~) { /* ~$\ldots$~ */ }
  for (~$\ldots$~) { /* write:yy ~$\ldots$~ */ }
  for (~$\ldots$~) { /* read,write:yy ~$\ldots$~ */
    while (~$\ldots$~) {/* ~$\ldots$~ */ }
    while (~$\ldots$~) {/* ~$\ldots$~ */ }
  }
  while (~$\ldots$~) {/* ~$\ldots$~ */ }
}
\end{myexamplescriptsz}
\caption{\label{fig:gMV}Loops in {\tt generateMTFValues()}}
\end{subfigure}
\begin{subfigure}[b]{\textwidth}
\begin{myexamplescriptsz}
void undoReversibleTransformation_fast() {
  int cftab[257];
  for (~$\ldots$~) { /* write:cftab ~$\ldots$~ */ }
  for (~$\ldots$~) { /* read,write:cftab ~$\ldots$~ */ }
  for (~$\ldots$~) { /* read,write:cftab ~$\ldots$~ */ }
  if (~$\ldots$~) { while (~$\ldots$~) for (~$\ldots$~) { /* ~$\ldots$~ */ } }
  else { while (~$\ldots$~) for (~$\ldots$~) { /* ~$\ldots$~ */ } }
}
\end{myexamplescriptsz}
\caption{\label{fig:uRTf}Loops in {\tt undoReversibleTransformation\_fast()}}
\end{subfigure}
\\
\end{subfigure}
\caption{\label{fig:bzip2_src}Structure of {\tt bzip2}'s functions}
\end{figure}

\end{document}

%% file: common_macros.tex
\newcommand{\toolName}{Dynamo}%

\newcommand{\store}{{\tt st}}%
\newcommand{\select}{{\tt sel}}%
\newcommand{\selectMath}{\mathrm{\tt sel}}%
\newcommand{\size}{{\tt sz}}
\newcommand{\sizeMath}{\mathrm{\tt sz}}

\newcommand*\inv[1]{\tikz[baseline=(char.base)]{
\node[shape=rounded rectangle,draw,inner sep=2pt] (char) {\scriptsize {\tt #1}};}}
\newcommand*\pred[1]{\raisebox{0.3ex}{\fboxsep=0.55ex\relax\fbox{\scriptsize {\tt #1}}}}
\newcommand*\apred[1]{\raisebox{0.3ex}{{\fboxsep=1pt\relax\fbox{\fboxsep=1.6pt\relax\fbox{\scriptsize {\tt #1}}}}}}

\newcommand{\ourIR}{\ensuremath{\mathtt{LLVM}_d}}

\newcommand{\refines}{\ensuremath{\sqsupseteq}}

\definecolor{myastral}{RGB}{46,116,181}
\definecolor{myolive}{named}{olive}
\definecolor{mygreen}{rgb}{0,0.6,0.2}
\definecolor{mygray}{rgb}{0.5,0.5,0.5}
\definecolor{myred}{rgb}{0.8,0,0.2}

\newcommand{\Tuple}[1]{\ensuremath{\vv{#1}}}

\newcommand{\intervalSym}{\ensuremath{i}}
\newcommand{\asetSym}{\ensuremath{\Sigma}}
\newcommand{\asetType}{\ensuremath{2^{\bv{32}}}}
\newcommand{\ii}[2]{\ensuremath{\intervalSym^{#2}_{#1}}}
\newcommand{\il}[2]{\ensuremath{\asetSym^{#2}_{#1}}}

\newcommand{\PC}[2]{\ensuremath{p^{#2}_{#1}}}

\newcommand{\Vcall}[1]{\ensuremath{\mathtt{fcall}(#1)}}
\newcommand{\VallocA}[1]{\ensuremath{\mathtt{allocBegin}(#1)}}
\newcommand{\VallocB}[1]{\ensuremath{\mathtt{allocEnd}(#1)}}
\newcommand{\Vdealloc}[1]{\ensuremath{\mathtt{dealloc}(#1)}}
\newcommand{\Vret}[1]{\ensuremath{\mathtt{ret}(#1)}}

\newcommand{\VallocASym}{{\tt allocBegin}}
\newcommand{\VallocBSym}{{\tt allocEnd}}

\definecolor{somered}{HTML}{E3242B}
\definecolor{somegrey}{HTML}{3F3F4E}
\definecolor{someviolet}{HTML}{B026FF}
\definecolor{somelightgrey}{HTML}{E2E1E4}
\definecolor{somedarkgrey}{HTML}{C9C9CC}
\definecolor{brightyellow}{HTML}{FEFEF0}
\definecolor{brightgreen}{HTML}{DFF0D8}

\newcommand{\memType}{\ensuremath{\bv{32} \rightarrow \bv{8}}}

\newcommand{\roundedbox}[2]{%
  \tikz[baseline=(a.base)]{
    \node[rounded corners, inner sep=1pt, fill=#1, font=\normalfont] (a) {#2};
  }%
}
\newcommand{\Shaded}{\roundedbox{somelightgrey}{shaded}}
\newcommand{\Ihalt}[1]{\ensuremath{\mathtt{halt}(#1)}}
\newcommand{\IhaltN}{\ensuremath{\text{{\texttt{\textcolor{somegrey}{\textbf{halt}($\emptyset$)}}}}}}
\newcommand{\IhaltW}{\ensuremath{\text{{\texttt{\textcolor{somered}{\textbf{halt}($\mathscr{W}$)}}}}}}
\newcommand{\IhaltU}{\ensuremath{\text{\texttt{\textcolor{someviolet}{\textbf{halt}($\mathscr{U}$)}}}}}
\newcommand{\Iread}[1]{\ensuremath{\mathtt{rd}(#1)}}
\newcommand{\Iwrite}[1]{\ensuremath{\mathtt{wr}(#1)}}
\newcommand{\HLBox}[1]{\roundedbox{somelightgrey}{\ensuremath{#1}}}
\newcommand{\IreadB}[1]{\ensuremath{\texttt{\textbf{rd}}(#1)}}
\newcommand{\IwriteB}[1]{\roundedbox{somedarkgrey}{\ensuremath{\texttt{\textbf{wr}}(#1)}}}
\newcommand{\Ichoose}[1]{\ensuremath{\theta(#1)}}
\newcommand{\IchooseM}{\Ichoose{\memType}}

\newcommand{\initAddrSets}[1]{\ensuremath{\mathtt{addrSets}_{#1}()}}
\newcommand{\roMem}[3]{\ensuremath{\mathtt{ROM}_{#1}^{#2}(#3)}}

\newcommand{\rdAcc}[1]{\ensuremath{\il{#1}{\mathtt{rd}}}}
\newcommand{\wrAcc}[1]{\ensuremath{\il{#1}{\mathtt{wr}}}}

\newcommand{\Iio}[2]{\ensuremath{\mathtt{io}(\World{P},\mathtt{#2},#1)}}

\newcommand{\HighlightMath}[1]{%
  \tikz[baseline=(a.base)]{
    \node[draw, dotted, rounded corners, inner sep=1pt, fill=brightgreen, font=\normalfont] (a) {\ensuremath{#1}};
  }%
}

\newcommand{\prjSym}{\ensuremath{\pi}}
\newcommand{\prjM}[2]{\ensuremath{\prjSym_{#1}(#2)}}
\newcommand{\updM}[3]{\ensuremath{\mathtt{upd}_{#1}(#2,\allowbreak #3)}}
\newcommand{\prjMEq}[3]{\ensuremath{#2=_{#1}#3}}

\newcommand{\bv}[1]{\ensuremath{\mathtt{i}_{#1}}}

\newcommand\llangle{\langle\!|}
\newcommand\rrangle{\reflectbox{\ensuremath{\llangle}}}
\newcommand{\Itempl}[1]{\ensuremath{\llangle #1 \rrangle}}

\newcommand{\xmark}{\ding{55}}%

\newcommand{\Trace}[1]{\ensuremath{T_{#1}}}
\newcommand{\TraceP}[1]{\ensuremath{T^{\prime}_{#1}}}
\newcommand{\World}[1]{\ensuremath{\Omega_{#1}}}

\newcommand{\LBi}[1]{\ensuremath{\mathtt{lb}(#1)}}
\newcommand{\UBi}[1]{\ensuremath{\mathtt{ub}(#1)}}

\newcommand{\LBr}[1]{\ensuremath{\mathtt{lb}(#1)}}
\newcommand{\UBr}[1]{\ensuremath{\mathtt{ub}(#1)}}

\newcommand{\LB}[1]{\ensuremath{\mathtt{lb}.{#1}}}
\newcommand{\UB}[1]{\ensuremath{\mathtt{ub}.{#1}}}
\newcommand{\LSz}[1]{\ensuremath{\mathtt{lstSz}.{#1}}}
\newcommand{\Sz}[1]{\ensuremath{\mathtt{sz}.{#1}}}
\newcommand{\Empty}[1]{\ensuremath{\mathtt{em}.{#1}}}

\newcommand{\VretVal}[1]{\ensuremath{\triangle_{#1}(\mathtt{eax},\mathtt{edx})}}
\newcommand{\VgetVal}[2]{\ensuremath{\triangledown_{#1}(#2)}}

\newcommand{\BasedOnSym}{\ensuremath{\beta}}
\newcommand{\BasedOn}[1]{\ensuremath{\BasedOnSym({#1})}}
\newcommand{\BasedOnM}[1]{\ensuremath{\BasedOnSym_{M}({#1})}}
\newcommand{\BasedOnMS}[1]{\ensuremath{\BasedOnSym_{M}^*({#1})}}
\newcommand{\CalleeRegions}{\ensuremath{\beta^*}}
\newcommand{\BasedOnOp}{\ensuremath{\BasedOnSym{}^{\mathtt{op}}}}

\newcommand{\Assign}{\ensuremath{\coloneqq}}

\newcommand{\Overlap}{\ensuremath{\mathtt{ov}}}
\newcommand{\Memalloc}{\ensuremath{\mathcal{L}}}
\newcommand{\MemallocP}[1]{\ensuremath{\Memalloc_{#1}}}

\newcommand{\TRule}[1]{\ensuremath{(\textsc{#1})}}

\newcommand{\ghost}[1]{\tikz[baseline=(char.base)]{%
\node[draw,rounded corners=2pt,inner sep=2pt,densely dotted] (char) {\ensuremath{\scriptstyle #1}};
}}

\newcommand{\g}{\ensuremath{g}}

\newcommand{\Y}{\ensuremath{Y}}
\newcommand{\y}{\ensuremath{y}}
\newcommand{\yv}{\ensuremath{\mathtt{vrdc}}}

\newcommand{\z}{\ensuremath{z}}
\newcommand{\Z}{\ensuremath{Z}}
\newcommand{\zl}{\ensuremath{zl}}
\newcommand{\za}{\ensuremath{za}}
\newcommand{\Zl}{\ensuremath{Z_l}}
\newcommand{\Za}{\ensuremath{Z_a}}

\newcommand{\ilzs}[2]{\ensuremath{\il{#1}{#2} \mathord{\mid}^{s}}}
\newcommand{\ilzv}[2]{\ensuremath{\il{#1}{#2} \mathord{\mid}^{v}}}
\newcommand{\ilZv}[1]{\ensuremath{\il{#1}{\Zl{}} \mathord{\mid}^{v}}}

\newcommand{\spE}{\ensuremath{\mathtt{stk}_e}}

\newcommand{\stkE}{\ensuremath{\mathtt{cs}_e}}

\newcommand{\spV}[1]{\ensuremath{\mathtt{sp}.{#1}}}

\newcommand{\updMCSym}{\ensuremath{\mathtt{cwrite}}}

\newcommand{\EmitX}[1]{#1}
\newcommand{\Emit}[1]{\adjustbox{bgcolor=somelightgrey}{\ensuremath{#1}}}

\newcommand\hcancel[2][black]{\setbox0=\hbox{$#2$}%
\rlap{\raisebox{.45\ht0}{\textcolor{#1}{\rule{\wd0}{1pt}}}}#2} 

\newcommand{\heap}[0]{\ensuremath{hp}}
\newcommand{\stk}[0]{\ensuremath{stk}}
\newcommand{\free}[0]{\ensuremath{\mathtt{free}}}

\newcommand{\progA}{\ensuremath{\mathbb{A}}}
\newcommand{\progC}{\ensuremath{\mathbb{C}}}

\newcommand{\Ak}{\ensuremath{A}}
\newcommand{\AkP}{\ensuremath{A'}}
\newcommand{\dAk}{\ensuremath{\dot{A}}}
\newcommand{\ddAk}{\ensuremath{\ddot{A}}}
\newcommand{\ddAkP}{\ensuremath{\ddot{A'}}}

\newcommand{\Ck}{\ensuremath{C}}
\newcommand{\CkP}{\ensuremath{C'}}
\newcommand{\Xk}{\ensuremath{X}}
\newcommand{\XkP}{\ensuremath{X'}}

\newcommand{\updMA}[3]{\ensuremath{\updMCSym(#1, \lambda x {.} #2, #3)}}

\newcommand{\Inv}[1]{\ensuremath{\phi_{#1}}}
\newcommand{\IXk}[0]{\ensuremath{\Phi_{\Xk}}}
\newcommand{\IXkP}[0]{\ensuremath{\Phi_{\XkP}}}

\newcommand{\EW}{\ensuremath{\mathscr{W}}}
\newcommand{\EU}{\ensuremath{\mathscr{U}}}
\newcommand{\WP}[1]{\ensuremath{\EW_{#1}}}

\newcommand{\UAk}{\ensuremath{\mathscr{U}_{\ddAk}}}
\newcommand{\UAkP}{\ensuremath{\mathscr{U}_{\ddAkP}}}
\newcommand{\UCk}{\ensuremath{\mathscr{U}_{\Ck}}}
\newcommand{\WAk}{\ensuremath{\mathscr{W}_{\ddAk}}}
\newcommand{\WCk}{\ensuremath{\mathscr{W}_{\Ck}}}

\newcommand{\PPath}[1]{\ensuremath{\xi_{#1}}}
\newcommand{\PPathN}[2]{\ensuremath{\xi_{#1}^{#2}}}

\newcommand{\PPathsN}[2]{\ensuremath{\vv{\xi}^{#2}_{#1}}}

\newcommand{\APath}{\PPath{\ddAk}}
\newcommand{\APathO}{\PPathN{\ddAk}{o}}
\newcommand{\APathN}[1]{\PPathN{\ddAk}{#1}}
\newcommand{\APathP}[0]{\PPathN{\ddAk}{\prime}}
\newcommand{\CPath}{\ensuremath{\xi_{\Ck}}}
\newcommand{\CPathN}[1]{\ensuremath{\xi^{#1}_{\Ck}}}
\newcommand{\CPathP}[0]{\PPathN{\Ck}{\prime}}
\newcommand{\CPPathN}[1]{\ensuremath{\xi^{#1}_{\CkP}}}

\newcommand{\APathsP}[0]{\PPathsN{\ddAk}{\prime}}
\newcommand{\APathsN}[1]{\PPathsN{\ddAk}{#1}}

\newcommand{\CPathsP}[0]{\PPathsN{\Ck}{\prime}}
\newcommand{\CPathsN}[1]{\PPathsN{\Ck}{#1}}

\newcommand{\XEdgeN}[1]{\ensuremath{e^{#1}_{\Xk}}}
\newcommand{\XPEdgeN}[1]{\ensuremath{e^{#1}_{\XkP}}}
\newcommand{\PathT}[2]{\ensuremath{#1 \twoheadrightarrow{} #2}}
\newcommand{\XEdge}{\ensuremath{e_{\Xk}}}

\newcommand{\NP}[1]{\ensuremath{\mathcal{N}_{#1}}}
\newcommand{\NNP}[1]{\ensuremath{\mathcal{N}^{\bcancel{UW}}_{#1}}}
\newcommand{\NCk}{\NP{\Ck}}
\newcommand{\NXk}{\NP{\Xk}}
\newcommand{\NXkP}{\NP{\XkP}}
\newcommand{\EP}[1]{\ensuremath{\mathcal{E}_{#1}}}
\newcommand{\EXk}{\EP{\Xk}}
\newcommand{\EXkP}{\EP{\XkP}}
\newcommand{\DXk}{\ensuremath{\mathcal{D}_{\Xk}}}
\newcommand{\DXkP}{\ensuremath{\mathcal{D}_{\XkP}}}
\newcommand{\XEdgeT}[3]{\ensuremath{#1 \raisebox{-0.5ex}{$\xrightarrow[]{\text{$#2$}}$} #3}}

\newcommand{\Pathset}[1]{\ensuremath{{\langle\xi\rangle}_{#1}}}

\newcommand{\proofObligation}{O}
\newcommand{\intervalContainedInAddrSet}[3]{\ensuremath{\mathtt{intrvlInSet}(#1, #2, #3)}}
\newcommand{\intervalContainedInAddrSetAndAligned}[4]{\ensuremath{\mathtt{intrvlInSet}_{#4}(#1,\allowbreak #2, #3)}}

\newcommand{\espmin}{\ensuremath{SP_{min}}}
\newcommand{\zlunion}{\ensuremath{Zlv_{U}}}
\newcommand{\NS}{\ensuremath{B}}
\newcommand{\STACK}{\ensuremath{S}}

\newcommand{\HP}{\ensuremath{HP(\APath)}}
\newcommand{\CL}{\ensuremath{CL(\APath)}}
\newcommand{\CS}{\ensuremath{CS(\APath)}}

\newcommand{\f}{\ensuremath{f}}
\newcommand{\F}{\ensuremath{F}}

\newcommand{\Gro}{\ensuremath{G_{r}}}
\newcommand{\Grw}{\ensuremath{G_{w}}}
\newcommand{\Fro}{\ensuremath{\F_{r}}}
\newcommand{\Frw}{\ensuremath{\F_{w}}}
\newcommand{\pthcover}[4]{\ensuremath{\{#1{1},#1{2},\ldots,#1{m}\} \langle{}#3,#4\rangle{}}}

\newcommand{\hoareTriple}[3]{\ensuremath{\{#1\}#2\{#3\}}}

\newcommand{\Load}[1]{Load$_{#1}$}
\newcommand{\Store}[1]{Store$_{#1}$}
\newcommand{\Entry}[1]{Entry$_{#1}$}
\newcommand{\Call}[1]{Call$_{#1}$}
\newcommand{\Ret}[1]{Ret$_{#1}$}

\newcommand{\pathDet}[3]{\ensuremath{[#1]_{#2}^{#3}}}
\newcommand{\pathcond}[1]{\ensuremath{pathcond(#1)}}
\newcommand{\CPathD}{\pathDet{\CPath{}}{\DXk}{e_{\Xk}}}
\newcommand{\CPathDN}[1]{\ensuremath{\pathDet{\CPathN{#1}}{\DXk}{e_{\Xk}^{#1}}}}

\newcommand{\steq}[2]{#1=_{st}#2}
\newcommand{\stprefixSym}{\leq_{st}}
\newcommand{\stprefix}[2]{#1\stprefixSym{}#2}
\newcommand{\eT}[1]{e(#1)}
\newcommand{\neT}[1]{\tilde{e}(#1)}

\newcommand{\Rall}{\ensuremath{R}}
\newcommand{\R}{\ensuremath{\vv{r}}}
\newcommand{\RN}[1]{\ensuremath{\vv{r_{#1}}}}

\newcommand{\execT}[3]{\ensuremath{#1 \mathrel{\downarrow_{#2}} #3}}

%% file: oopsla24.bbl

\begin{thebibliography}{27}


\ifx \showCODEN    \undefined \def \showCODEN     #1{\unskip}     \fi
\ifx \showDOI      \undefined \def \showDOI       #1{#1}\fi
\ifx \showISBNx    \undefined \def \showISBNx     #1{\unskip}     \fi
\ifx \showISBNxiii \undefined \def \showISBNxiii  #1{\unskip}     \fi
\ifx \showISSN     \undefined \def \showISSN      #1{\unskip}     \fi
\ifx \showLCCN     \undefined \def \showLCCN      #1{\unskip}     \fi
\ifx \shownote     \undefined \def \shownote      #1{#1}          \fi
\ifx \showarticletitle \undefined \def \showarticletitle #1{#1}   \fi
\ifx \showURL      \undefined \def \showURL       {\relax}        \fi
\providecommand\bibfield[2]{#2}
\providecommand\bibinfo[2]{#2}
\providecommand\natexlab[1]{#1}
\providecommand\showeprint[2][]{arXiv:#2}

\bibitem[z3b(2024)]%
        {z3bugreport_model2024}
 \bibinfo{year}{2024}\natexlab{}.
\newblock \bibinfo{title}{Z3 bug report for an unsound model}.
\newblock
  \bibinfo{howpublished}{\url{https://github.com/Z3Prover/z3/issues/7132}}.
\newblock


\bibitem[Andersen(1994)]%
        {andersen94programanalysis}
\bibfield{author}{\bibinfo{person}{Lars~Ole Andersen}.}
  \bibinfo{year}{1994}\natexlab{}.
\newblock \bibinfo{booktitle}{\emph{Program Analysis and Specialization for the
  {C} Programming Language}}.
\newblock \bibinfo{type}{{T}echnical {R}eport}.
\newblock


\bibitem[Churchill et~al\mbox{.}(2019)]%
        {semalign}
\bibfield{author}{\bibinfo{person}{Berkeley Churchill}, \bibinfo{person}{Oded
  Padon}, \bibinfo{person}{Rahul Sharma}, {and} \bibinfo{person}{Alex Aiken}.}
  \bibinfo{year}{2019}\natexlab{}.
\newblock \showarticletitle{Semantic Program Alignment for Equivalence
  Checking}. In \bibinfo{booktitle}{\emph{Proceedings of the 40th ACM SIGPLAN
  Conference on Programming Language Design and Implementation}} (Phoenix, AZ,
  USA) \emph{(\bibinfo{series}{PLDI 2019})}. \bibinfo{publisher}{ACM},
  \bibinfo{address}{New York, NY, USA}, \bibinfo{pages}{1027--1040}.
\newblock
\showISBNx{978-1-4503-6712-7}
\urldef\tempurl%
\url{https://doi.org/10.1145/3314221.3314596}
\showDOI{\tempurl}


\bibitem[Gupta et~al\mbox{.}(2020)]%
        {oopsla20}
\bibfield{author}{\bibinfo{person}{Shubhani Gupta}, \bibinfo{person}{Abhishek
  Rose}, {and} \bibinfo{person}{Sorav Bansal}.}
  \bibinfo{year}{2020}\natexlab{}.
\newblock \showarticletitle{Counterexample-Guided Correlation Algorithm for
  Translation Validation}.
\newblock \bibinfo{journal}{\emph{Proc. ACM Program. Lang.}}
  \bibinfo{volume}{4}, \bibinfo{number}{OOPSLA}, Article
  \bibinfo{articleno}{221} (\bibinfo{date}{Nov.} \bibinfo{year}{2020}),
  \bibinfo{numpages}{29}~pages.
\newblock
\urldef\tempurl%
\url{https://doi.org/10.1145/3428289}
\showDOI{\tempurl}


\bibitem[Henning(2000)]%
        {spec:cint2000}
\bibfield{author}{\bibinfo{person}{John~L. Henning}.}
  \bibinfo{year}{2000}\natexlab{}.
\newblock \showarticletitle{{S}{P}{E}{C} {C}{P}{U}2000: Measuring {C}{P}{U}
  performance in the new millenium}.
\newblock \bibinfo{journal}{\emph{IEEE Computer}} \bibinfo{volume}{33},
  \bibinfo{number}{7} (\bibinfo{date}{July} \bibinfo{year}{2000}),
  \bibinfo{pages}{28--35}.
\newblock


\bibitem[Kang et~al\mbox{.}(2018)]%
        {crellvm18}
\bibfield{author}{\bibinfo{person}{Jeehoon Kang}, \bibinfo{person}{Yoonseung
  Kim}, \bibinfo{person}{Youngju Song}, \bibinfo{person}{Juneyoung Lee},
  \bibinfo{person}{Sanghoon Park}, \bibinfo{person}{Mark~Dongyeon Shin},
  \bibinfo{person}{Yonghyun Kim}, \bibinfo{person}{Sungkeun Cho},
  \bibinfo{person}{Joonwon Choi}, \bibinfo{person}{Chung-Kil Hur}, {and}
  \bibinfo{person}{Kwangkeun Yi}.} \bibinfo{year}{2018}\natexlab{}.
\newblock \showarticletitle{Crellvm: Verified Credible Compilation for LLVM}.
  In \bibinfo{booktitle}{\emph{Proceedings of the 39th ACM SIGPLAN Conference
  on Programming Language Design and Implementation}} (Philadelphia, PA, USA)
  \emph{(\bibinfo{series}{PLDI 2018})}. \bibinfo{publisher}{ACM},
  \bibinfo{address}{New York, NY, USA}, \bibinfo{pages}{631--645}.
\newblock
\showISBNx{978-1-4503-5698-5}
\urldef\tempurl%
\url{https://doi.org/10.1145/3192366.3192377}
\showDOI{\tempurl}


\bibitem[Kasampalis et~al\mbox{.}(2021)]%
        {llvm_tv21}
\bibfield{author}{\bibinfo{person}{Theodoros Kasampalis},
  \bibinfo{person}{Daejun Park}, \bibinfo{person}{Zhengyao Lin},
  \bibinfo{person}{Vikram~S. Adve}, {and} \bibinfo{person}{Grigore Ro\c{s}u}.}
  \bibinfo{year}{2021}\natexlab{}.
\newblock \showarticletitle{Language-Parametric Compiler Validation with
  Application to LLVM}. In \bibinfo{booktitle}{\emph{Proceedings of the 26th
  ACM International Conference on Architectural Support for Programming
  Languages and Operating Systems}} (Virtual, USA)
  \emph{(\bibinfo{series}{ASPLOS 2021})}. \bibinfo{publisher}{Association for
  Computing Machinery}, \bibinfo{address}{New York, NY, USA},
  \bibinfo{pages}{1004–1019}.
\newblock
\showISBNx{9781450383172}


\bibitem[Kernighan and Ritchie(1988)]%
        {knr}
\bibfield{author}{\bibinfo{person}{Brian~W. Kernighan} {and}
  \bibinfo{person}{Dennis~M. Ritchie}.} \bibinfo{year}{1988}\natexlab{}.
\newblock \bibinfo{booktitle}{\emph{The C Programming Language}
  (\bibinfo{edition}{2nd} ed.)}.
\newblock \bibinfo{publisher}{Prentice Hall Professional Technical Reference}.
\newblock
\showISBNx{0131103709}


\bibitem[Lee et~al\mbox{.}(2021)]%
        {alive2_llvm_mem_model}
\bibfield{author}{\bibinfo{person}{Juneyoung Lee}, \bibinfo{person}{Dongjoo
  Kim}, \bibinfo{person}{Chung-Kil Hur}, {and} \bibinfo{person}{Nuno~P.
  Lopes}.} \bibinfo{year}{2021}\natexlab{}.
\newblock \showarticletitle{An SMT Encoding of LLVM's Memory Model for Bounded
  Translation Validation}. In \bibinfo{booktitle}{\emph{Computer Aided
  Verification}}, \bibfield{editor}{\bibinfo{person}{Alexandra Silva} {and}
  \bibinfo{person}{K.~Rustan~M. Leino}} (Eds.). \bibinfo{publisher}{Springer
  International Publishing}, \bibinfo{address}{Cham},
  \bibinfo{pages}{752--776}.
\newblock
\showISBNx{978-3-030-81688-9}


\bibitem[Leroy(2006)]%
        {compcert}
\bibfield{author}{\bibinfo{person}{Xavier Leroy}.}
  \bibinfo{year}{2006}\natexlab{}.
\newblock \showarticletitle{Formal certification of a compiler back-end, or:
  programming a compiler with a proof assistant}. In
  \bibinfo{booktitle}{\emph{33rd ACM symposium on Principles of Programming
  Languages}}. \bibinfo{publisher}{ACM Press}, \bibinfo{pages}{42--54}.
\newblock
\urldef\tempurl%
\url{http://gallium.inria.fr/~xleroy/publi/compiler-certif.pdf}
\showURL{%
\tempurl}


\bibitem[Leroy and Blazy(2008)]%
        {compcertMemModel}
\bibfield{author}{\bibinfo{person}{Xavier Leroy} {and}
  \bibinfo{person}{Sandrine Blazy}.} \bibinfo{year}{2008}\natexlab{}.
\newblock \showarticletitle{Formal Verification of a C-like Memory Model and
  Its Uses for Verifying Program Transformations}.
\newblock \bibinfo{journal}{\emph{J. Autom. Reason.}} \bibinfo{volume}{41},
  \bibinfo{number}{1} (\bibinfo{year}{2008}), \bibinfo{pages}{1--31}.
\newblock
\urldef\tempurl%
\url{https://doi.org/10.1007/s10817-008-9099-0}
\showDOI{\tempurl}


\bibitem[Lopes et~al\mbox{.}(2021)]%
        {alive2}
\bibfield{author}{\bibinfo{person}{Nuno~P. Lopes}, \bibinfo{person}{Juneyoung
  Lee}, \bibinfo{person}{Chung-Kil Hur}, \bibinfo{person}{Zhengyang Liu}, {and}
  \bibinfo{person}{John Regehr}.} \bibinfo{year}{2021}\natexlab{}.
\newblock \showarticletitle{Alive2: Bounded Translation Validation for LLVM}.
  In \bibinfo{booktitle}{\emph{Proceedings of the 42nd ACM SIGPLAN
  International Conference on Programming Language Design and Implementation}}
  (Virtual, Canada) \emph{(\bibinfo{series}{PLDI 2021})}.
  \bibinfo{publisher}{Association for Computing Machinery},
  \bibinfo{address}{New York, NY, USA}, \bibinfo{pages}{65–79}.
\newblock
\showISBNx{9781450383912}
\urldef\tempurl%
\url{https://doi.org/10.1145/3453483.3454030}
\showDOI{\tempurl}


\bibitem[Maleki et~al\mbox{.}(2011)]%
        {tsvc}
\bibfield{author}{\bibinfo{person}{Saeed Maleki}, \bibinfo{person}{Yaoqing
  Gao}, \bibinfo{person}{Maria~J. Garzar\'{a}n}, \bibinfo{person}{Tommy Wong},
  {and} \bibinfo{person}{David~A. Padua}.} \bibinfo{year}{2011}\natexlab{}.
\newblock \showarticletitle{An Evaluation of Vectorizing Compilers}. In
  \bibinfo{booktitle}{\emph{Proceedings of the 2011 International Conference on
  Parallel Architectures and Compilation Techniques}}
  \emph{(\bibinfo{series}{PACT '11})}. \bibinfo{publisher}{IEEE Computer
  Society}, \bibinfo{address}{Washington, DC, USA}, \bibinfo{pages}{372--382}.
\newblock
\showISBNx{978-0-7695-4566-0}
\urldef\tempurl%
\url{https://doi.org/10.1109/PACT.2011.68}
\showDOI{\tempurl}


\bibitem[Menendez et~al\mbox{.}(2016)]%
        {aliveFP}
\bibfield{author}{\bibinfo{person}{David Menendez}, \bibinfo{person}{Santosh
  Nagarakatte}, {and} \bibinfo{person}{Aarti Gupta}.}
  \bibinfo{year}{2016}\natexlab{}.
\newblock \showarticletitle{Alive-FP: Automated Verification of Floating Point
  Based Peephole Optimizations in LLVM}. \bibinfo{pages}{317--337}.
\newblock
\showISBNx{978-3-662-53412-0}
\urldef\tempurl%
\url{https://doi.org/10.1007/978-3-662-53413-7_16}
\showDOI{\tempurl}


\bibitem[Namjoshi and Zuck(2013)]%
        {namjoshi13}
\bibfield{author}{\bibinfo{person}{KedarS. Namjoshi} {and}
  \bibinfo{person}{LenoreD. Zuck}.} \bibinfo{year}{2013}\natexlab{}.
\newblock \showarticletitle{Witnessing Program Transformations}.
\newblock In \bibinfo{booktitle}{\emph{Static Analysis}},
  \bibfield{editor}{\bibinfo{person}{Francesco Logozzo} {and}
  \bibinfo{person}{Manuel Fähndrich}} (Eds.). \bibinfo{series}{Lecture Notes
  in Computer Science}, Vol.~\bibinfo{volume}{7935}.
  \bibinfo{publisher}{Springer Berlin Heidelberg}, \bibinfo{pages}{304--323}.
\newblock
\showISBNx{978-3-642-38855-2}
\urldef\tempurl%
\url{https://doi.org/10.1007/978-3-642-38856-9_17}
\showDOI{\tempurl}


\bibitem[Necula(2000)]%
        {tvi}
\bibfield{author}{\bibinfo{person}{George~C. Necula}.}
  \bibinfo{year}{2000}\natexlab{}.
\newblock \showarticletitle{Translation Validation for an Optimizing Compiler}.
  In \bibinfo{booktitle}{\emph{Proceedings of the ACM SIGPLAN 2000 Conference
  on Programming Language Design and Implementation}} (Vancouver, British
  Columbia, Canada) \emph{(\bibinfo{series}{PLDI '00})}.
  \bibinfo{publisher}{ACM}, \bibinfo{address}{New York, NY, USA},
  \bibinfo{pages}{83--94}.
\newblock
\showISBNx{1-58113-199-2}
\urldef\tempurl%
\url{https://doi.org/10.1145/349299.349314}
\showDOI{\tempurl}


\bibitem[Rose and Bansal(2024)]%
        {oopsla24_artifact}
\bibfield{author}{\bibinfo{person}{Abhishek Rose} {and} \bibinfo{person}{Sorav
  Bansal}.} \bibinfo{year}{2024}\natexlab{}.
\newblock \bibinfo{booktitle}{\emph{{Artifact for paper "Modeling Dynamic
  (De)Allocations of Local Memory for Translation Validation"}}}.
\newblock
\urldef\tempurl%
\url{https://doi.org/10.5281/zenodo.10797459}
\showDOI{\tempurl}


\bibitem[Sewell et~al\mbox{.}(2013)]%
        {tv_oskernel}
\bibfield{author}{\bibinfo{person}{Thomas Arthur~Leck Sewell},
  \bibinfo{person}{Magnus~O. Myreen}, {and} \bibinfo{person}{Gerwin Klein}.}
  \bibinfo{year}{2013}\natexlab{}.
\newblock \showarticletitle{Translation Validation for a Verified OS Kernel}.
  In \bibinfo{booktitle}{\emph{Proceedings of the 34th ACM SIGPLAN Conference
  on Programming Language Design and Implementation}} (Seattle, Washington,
  USA) \emph{(\bibinfo{series}{PLDI '13})}. \bibinfo{publisher}{Association for
  Computing Machinery}, \bibinfo{address}{New York, NY, USA},
  \bibinfo{pages}{471–482}.
\newblock
\showISBNx{9781450320146}
\urldef\tempurl%
\url{https://doi.org/10.1145/2491956.2462183}
\showDOI{\tempurl}


\bibitem[Sharma et~al\mbox{.}(2013)]%
        {ddec}
\bibfield{author}{\bibinfo{person}{Rahul Sharma}, \bibinfo{person}{Eric
  Schkufza}, \bibinfo{person}{Berkeley Churchill}, {and} \bibinfo{person}{Alex
  Aiken}.} \bibinfo{year}{2013}\natexlab{}.
\newblock \showarticletitle{Data-driven Equivalence Checking}. In
  \bibinfo{booktitle}{\emph{Proceedings of the 2013 ACM SIGPLAN International
  Conference on Object Oriented Programming Systems Languages \&\#38;
  Applications}} (Indianapolis, Indiana, USA) \emph{(\bibinfo{series}{OOPSLA
  '13})}. \bibinfo{publisher}{ACM}, \bibinfo{address}{New York, NY, USA},
  \bibinfo{pages}{391--406}.
\newblock
\showISBNx{978-1-4503-2374-1}
\urldef\tempurl%
\url{https://doi.org/10.1145/2509136.2509509}
\showDOI{\tempurl}


\bibitem[Steensgaard(1996)]%
        {steensgaard}
\bibfield{author}{\bibinfo{person}{Bjarne Steensgaard}.}
  \bibinfo{year}{1996}\natexlab{}.
\newblock \showarticletitle{Points-to analysis in almost linear time}. In
  \bibinfo{booktitle}{\emph{Proceedings of the 23rd ACM SIGPLAN-SIGACT
  symposium on Principles of programming languages}}. \bibinfo{pages}{32--41}.
\newblock


\bibitem[Stepp et~al\mbox{.}(2011)]%
        {stepp_eqsat_llvm11}
\bibfield{author}{\bibinfo{person}{Michael Stepp}, \bibinfo{person}{Ross Tate},
  {and} \bibinfo{person}{Sorin Lerner}.} \bibinfo{year}{2011}\natexlab{}.
\newblock \showarticletitle{Equality-based Translation Validator for LLVM}. In
  \bibinfo{booktitle}{\emph{Proceedings of the 23rd International Conference on
  Computer Aided Verification}} (Snowbird, UT)
  \emph{(\bibinfo{series}{CAV'11})}. \bibinfo{publisher}{Springer-Verlag},
  \bibinfo{address}{Berlin, Heidelberg}, \bibinfo{pages}{737--742}.
\newblock
\showISBNx{978-3-642-22109-5}
\urldef\tempurl%
\url{http://dl.acm.org/citation.cfm?id=2032305.2032364}
\showURL{%
\tempurl}


\bibitem[Sun et~al\mbox{.}(2016)]%
        {emi16}
\bibfield{author}{\bibinfo{person}{Chengnian Sun}, \bibinfo{person}{Vu Le},
  {and} \bibinfo{person}{Zhendong Su}.} \bibinfo{year}{2016}\natexlab{}.
\newblock \showarticletitle{Finding Compiler Bugs via Live Code Mutation}. In
  \bibinfo{booktitle}{\emph{Proceedings of the 2016 ACM SIGPLAN International
  Conference on Object-Oriented Programming, Systems, Languages, and
  Applications}} (Amsterdam, Netherlands) \emph{(\bibinfo{series}{OOPSLA
  2016})}. \bibinfo{publisher}{Association for Computing Machinery},
  \bibinfo{address}{New York, NY, USA}, \bibinfo{pages}{849–863}.
\newblock
\showISBNx{9781450344449}
\urldef\tempurl%
\url{https://doi.org/10.1145/2983990.2984038}
\showDOI{\tempurl}


\bibitem[Tanenbaum et~al\mbox{.}(1983)]%
        {ack_compiler}
\bibfield{author}{\bibinfo{person}{Andrew~S. Tanenbaum}, \bibinfo{person}{Hans
  van Staveren}, \bibinfo{person}{E.~G. Keizer}, {and}
  \bibinfo{person}{Johan~W. Stevenson}.} \bibinfo{year}{1983}\natexlab{}.
\newblock \showarticletitle{A Practical Tool Kit for Making Portable
  Compilers}.
\newblock \bibinfo{journal}{\emph{Commun. ACM}} \bibinfo{volume}{26},
  \bibinfo{number}{9} (\bibinfo{date}{sep} \bibinfo{year}{1983}),
  \bibinfo{pages}{654–660}.
\newblock
\showISSN{0001-0782}
\urldef\tempurl%
\url{https://doi.org/10.1145/358172.358182}
\showDOI{\tempurl}


\bibitem[Tristan et~al\mbox{.}(2011)]%
        {tristan_tv_eqsat11}
\bibfield{author}{\bibinfo{person}{Jean-Baptiste Tristan},
  \bibinfo{person}{Paul Govereau}, {and} \bibinfo{person}{Greg Morrisett}.}
  \bibinfo{year}{2011}\natexlab{}.
\newblock \showarticletitle{Evaluating Value-graph Translation Validation for
  LLVM}. In \bibinfo{booktitle}{\emph{Proceedings of the 32Nd ACM SIGPLAN
  Conference on Programming Language Design and Implementation}} (San Jose,
  California, USA) \emph{(\bibinfo{series}{PLDI '11})}.
  \bibinfo{publisher}{ACM}, \bibinfo{address}{New York, NY, USA},
  \bibinfo{pages}{295--305}.
\newblock
\showISBNx{978-1-4503-0663-8}
\urldef\tempurl%
\url{https://doi.org/10.1145/1993498.1993533}
\showDOI{\tempurl}


\bibitem[Zaks and Pnueli(2008)]%
        {covac}
\bibfield{author}{\bibinfo{person}{Anna Zaks} {and} \bibinfo{person}{Amir
  Pnueli}.} \bibinfo{year}{2008}\natexlab{}.
\newblock \showarticletitle{CoVaC: Compiler Validation by Program Analysis of
  the Cross-Product}. In \bibinfo{booktitle}{\emph{Proceedings of the 15th
  International Symposium on Formal Methods}} (Turku, Finland)
  \emph{(\bibinfo{series}{FM '08})}. \bibinfo{publisher}{Springer-Verlag},
  \bibinfo{address}{Berlin, Heidelberg}, \bibinfo{pages}{35--51}.
\newblock
\showISBNx{978-3-540-68235-6}
\urldef\tempurl%
\url{https://doi.org/10.1007/978-3-540-68237-0_5}
\showDOI{\tempurl}


\bibitem[Zhao et~al\mbox{.}(2012)]%
        {llvm_verify_zhao12}
\bibfield{author}{\bibinfo{person}{Jianzhou Zhao}, \bibinfo{person}{Santosh
  Nagarakatte}, \bibinfo{person}{Milo~M.K. Martin}, {and}
  \bibinfo{person}{Steve Zdancewic}.} \bibinfo{year}{2012}\natexlab{}.
\newblock \showarticletitle{Formalizing the LLVM Intermediate Representation
  for Verified Program Transformations}. In
  \bibinfo{booktitle}{\emph{Proceedings of the 39th Annual ACM SIGPLAN-SIGACT
  Symposium on Principles of Programming Languages}} (Philadelphia, PA, USA)
  \emph{(\bibinfo{series}{POPL '12})}. \bibinfo{publisher}{ACM},
  \bibinfo{address}{New York, NY, USA}, \bibinfo{pages}{427--440}.
\newblock
\showISBNx{978-1-4503-1083-3}
\urldef\tempurl%
\url{https://doi.org/10.1145/2103656.2103709}
\showDOI{\tempurl}


\bibitem[Zhao et~al\mbox{.}(2013)]%
        {ssa_verify_zhao13}
\bibfield{author}{\bibinfo{person}{Jianzhou Zhao}, \bibinfo{person}{Santosh
  Nagarakatte}, \bibinfo{person}{Milo~M.K. Martin}, {and}
  \bibinfo{person}{Steve Zdancewic}.} \bibinfo{year}{2013}\natexlab{}.
\newblock \showarticletitle{Formal Verification of SSA-based Optimizations for
  LLVM}. In \bibinfo{booktitle}{\emph{Proceedings of the 34th ACM SIGPLAN
  Conference on Programming Language Design and Implementation}} (Seattle,
  Washington, USA) \emph{(\bibinfo{series}{PLDI '13})}.
  \bibinfo{publisher}{ACM}, \bibinfo{address}{New York, NY, USA},
  \bibinfo{pages}{175--186}.
\newblock
\showISBNx{978-1-4503-2014-6}
\urldef\tempurl%
\url{https://doi.org/10.1145/2491956.2462164}
\showDOI{\tempurl}


\end{thebibliography}
